\documentclass[11pt]{article}
\usepackage{mathrsfs}
\usepackage{amssymb,physics,bm,fullpage,mathtools,amsthm}
\usepackage[numbers,sort]{natbib}
\usepackage{graphicx}
\usepackage{tikz}
\usepackage{subcaption}
\usepackage[dvipsnames]{xcolor}
\usepackage{algpseudocode,times,thmtools,thm-restate,booktabs,float}
\usepackage[ruled,vlined,linesnumbered]{algorithm2e}
\SetKwProg{Procedure}{Procedure}{}{}
\MakeRobust{\Call}

\usepackage{enumitem}
\setlist{listparindent=\parindent,parsep=0pt,itemsep=1em}
\setlist[itemize]{label=$-$,noitemsep}
\setlist[enumerate]{itemsep=1mm}
\setlist[description]{leftmargin=\parindent}

\usepackage{tocloft}
\newcommand{\SwapSens}{\mathsf{SwapSens}}
\newcommand{\SwapClo}{\mathsf{SwapClo}}
\newcommand{\SwapDist}{\mathsf{SwapDist}}
\newcommand{\val}{\mathsf{val}}
\newcommand{\opt}{\mathsf{opt}}
\newcommand{\EMD}{\mathsf{EMD}}

\DeclareMathOperator{\poly}{poly}

\theoremstyle{definition}
\newtheorem{theorem}{Theorem}[section]
\newtheorem{lemma}[theorem]{Lemma}
\newtheorem{corollary}[theorem]{Corollary}
\newtheorem{proposition}[theorem]{Proposition}

\theoremstyle{definition}
\newtheorem{definition}[theorem]{Definition}



\usepackage{hyperref}
\definecolor{newcolor}{hsb}{0.6,1,0.75}
\hypersetup{
  pdftitle={Sensitivity and Locality Lower Bounds via Low-Soundness Label Cover},
  pdfauthor={Noah Fleming, Max Hopkins, Yuichi Yoshida},
	colorlinks=true,
	urlcolor=newcolor,
	linkcolor=newcolor,
	citecolor=newcolor,
	unicode
}
\usepackage{cleveref}
\crefname{claim}{claim}{claims}
\Crefname{claim}{Claim}{Claims}

\newcommand{\E}{\mathop{\mathbf E}}
\newcommand{\F}{\mathbb{F}}

\newcommand{\red}{\operatorname{red}}

\allowdisplaybreaks

\title{Sensitivity and Locality Lower Bounds via Low-Soundness Label Cover}
\author{
Noah Fleming\thanks{Supported by the Swedish Research Council under grant number 2025-06762.} \\
Lund  \\
\texttt{noah.fleming@cs.lth.se}
\and
Max Hopkins\thanks{Supported by NSF Grant Number DMS-2424441} \\
IAS \\
\texttt{nmhopkin@ias.edu}
\and
Yuichi Yoshida\thanks{Supported by JSPS KAKENHI Grant Number JP24K02903.} \\
National Institute of Informatics \\
\texttt{yyoshida@nii.ac.jp}
}

\begin{document}

\pagenumbering{roman}
\maketitle

\begin{abstract}
We study the sensitivity of approximation algorithms, which measures
how much their output distributions change under a single input edit.
Can weaker approximation guarantees make their solutions more stable?
In this work, we show that, for Set Cover, Dominating Set, Min-Rep, and minimum
unweighted directed 3-spanner, even a $\log^\beta n$ approximation
guarantee requires sensitivity $n^{\Omega(1)}$, for some
absolute constant $\beta>0$ and all sufficiently large $n$.
The lower bounds hold even for randomized algorithms with unrestricted
computation.
Here $n$ is the universe size for Set Cover and the number of vertices
for the graph problems.

On bounded-degree graphs, our polynomial sensitivity lower bounds imply
logarithmic locality lower bounds in the randomized \mbox{LOCAL} model
and its non-signaling generalization.
For Dominating Set, fix any
$\varepsilon\in(0,1)$ and a sufficiently large constant maximum-degree
bound $\Delta$. A $(1+\varepsilon)\ln\Delta$ approximation
can be achieved with sensitivity and \mbox{LOCAL} round complexity
independent of $n$.
We show that improving the ratio to $(1-\varepsilon)\ln\Delta$ requires
sensitivity $n^{\Omega_{\varepsilon,\Delta}(1)}$, and hence locality
$\Omega_{\varepsilon,\Delta}(\log n)$.
For minimum unweighted directed 3-spanner, we show that every constant
approximation factor requires logarithmic locality, even on DAGs of
constant maximum total degree.

The key ingredient is a polynomial sensitivity lower bound for
a variant of low-soundness Label Cover.
Fleming and Yoshida established such a lower bound when the required
fraction of satisfied constraints on satisfiable instances is close
to one. We extend this to inverse-polylogarithmic fractions, enabling
the reductions to the problems above.
The transformations used to lower this fraction can spread one input
edit across many constraints, making the transferred lower bound vacuous.
We overcome this obstacle by adapting the
Impagliazzo--Kabanets--Wigderson direct-product construction
and Dinur--Harsha iterated composition to control the propagation
of input edits, and by designing recovery procedures that are
stable under changes to both the supplied labeling and the source instance.
\end{abstract}

\thispagestyle{empty}
\newpage
\tableofcontents
\thispagestyle{empty}
\clearpage
\pagestyle{plain}
\pagenumbering{arabic}

\section{Introduction}

A desirable property of an optimization algorithm
is that similar inputs lead to similar outputs~\cite{varma2023average}. Unfortunately, in many traditional optimization problems, a single input edit can force optimal outputs to differ in many coordinates.
A natural approach to improving stability is to allow both approximation
and randomization, giving the algorithm more freedom in choosing its
output.\footnote{Unless stated otherwise, approximation and value
guarantees are in expectation; see \Cref{sec:pre} for the formal definitions.}
Can this freedom prevent a single input edit from forcing many output
changes, even when computation is unrestricted?

We measure stability by \emph{sensitivity}~\cite{yoshida2021sensitivity}: the maximum earth mover's
distance (EMD) between output distributions on inputs differing by one
allowed edit. EMD is the minimum expected Hamming distance over all
couplings, where the Hamming distance counts changed labels or changes in set
membership.\footnote{The allowed edits are specified for each problem.} Thus a
sensitivity lower bound $S$ means that some single input edit forces
at least $S$ output changes in expectation under every coupling.

Sensitivity captures robustness of the solution itself, beyond its
objective value. Such robustness matters when inputs are incomplete
or perturbed, especially when the output is used in subsequent
computation~\cite{varma2023average,yoshida2021sensitivity}.
Sensitivity is also relevant to exact data deletion, where the output
after a deletion must have the same distribution as recomputation on
the remaining input~\cite{ginart2019making,ullah2021stability,ullah2023adaptive}, and in distributed computing \cite{fleming2026sensitivity}, where LOCAL algorithms on bounded-degree graphs are automatically low-sensitivity.

Fleming and Yoshida~\cite{fleming2026sensitivity} used probabilistically
checkable proof (PCP) constructions to prove polynomial sensitivity
lower bounds for approximation algorithms for problems including
Minimum Vertex Cover and Max 3SAT. A key step is a lower bound for Label Cover:
for some constant $\varepsilon\in(0,1)$, Fleming and Yoshida show that satisfying at least a
$1-\varepsilon$ fraction of constraints on every
satisfiable instance requires sensitivity $n^{\Omega(1)}$, where $n$
is the number of vertices.

However, classical hardness-of-approximation reductions to further problems,
such as Set Cover~\cite{feige1998threshold}, start from `low-soundness'
Label Cover. In particular, recovering a Label Cover solution from an approximate
Set Cover solution may guarantee only a small (even sub-constant) fraction of satisfied
constraints, whereas Fleming and Yoshida's sensitivity lower bound
requires a fraction close to one.
This leads to the central question of this paper, also raised by
Fleming and Yoshida:
\begin{quote}
\emph{Can we prove polynomial sensitivity lower bounds for low-soundness
Label Cover, where an algorithm need only satisfy a tiny fraction of
constraints on satisfiable instances?}
\end{quote}
The usual parallel-repetition argument~\cite{raz1995parallel} does not
establish this directly, because one source edit may change so many
repeated constraints that the resulting sensitivity lower bound
becomes vacuous.

\subsection{Our Results}

We answer this question affirmatively, even when the required fraction
is inverse-polylogarithmic in the instance size.
As applications, we obtain new sensitivity lower bounds for Set Cover,
Dominating Set, Min-Rep, and minimum unweighted directed 3-spanner,
as well as new locality lower bounds in the LOCAL and non-signaling models for Dominating Set and
directed 3-spanner.

We give two constructions of low soundness Label Cover, each serving a different purpose. The first construction, based on Impagliazzo--Kabanets--Wigderson (IKW)
direct-product amplification~\cite{impagliazzo2009new}, keeps the left label alphabet small enough
to obtain Dominating Set sensitivity lower bounds for approximation factors near
$\ln\Delta$, where $\Delta$ is the maximum-degree bound.
The second construction, based on Dinur--Harsha (DH) iterated composition~\cite{dinur2013composition}, has larger alphabet but improved size, resulting in stronger sensitivity lower bounds even for polylogarithmic approximation
ratios.
We first state the Label Cover guarantees that make these results possible,
then give their consequences for covering and network-design problems.

\subsubsection{Low-Soundness Label Cover}

A Label Cover instance consists of a bipartite (multi)graph
$(U,V,E)$, finite label alphabets $\Sigma_U,\Sigma_V$, and a projection constraint $f_e:\Sigma_U\to\Sigma_V$ for each edge $e=(u,v)$. Formally, we will use \emph{left-predicate Label Cover}, which also specifies a predicate
$P_u:\Sigma_U\to\{0,1\}$ at each left vertex. A labeling of the vertices satisfies the edge
$e=(u,v)$ when its left label passes $P_u$ and projects to its right
label under $f_e$.  Its \emph{value} is the fraction of satisfied edges, and the value of the instance is the maximum value across all possible assignments. The instance is called satisfiable if some labeling has value $1$.
An input edit changes one projection or one left predicate, and output
distance counts changed vertex labels.

Our constructions transform the `source' binary constraint systems (2-CSPs) used in
Fleming and Yoshida's lower bound~\cite{fleming2026sensitivity} into `target' Label Cover
instances in a way that preserves satisfiability and admits a randomized recovery procedure. In particular, on a satisfiable source instance, the recovery procedure uses (randomized) solutions to the transformed target instance with expected value above some small specified threshold to produce back a source assignment whose expected
value is close to one. Following the PCP literature, we call such constructions
\emph{low-soundness} because the required target value is small.

Fleming and Yoshida's source lower bound applies to algorithms attaining
value close to one on every satisfiable member of their family. A low-sensitivity
algorithm attaining the required target value, combined with sufficiently
stable recovery, would violate that lower bound. We therefore control both
the number of target edits caused by a source edit and the sensitivity of
recovery. \Cref{sec:intro-technical} explains the transfer in detail.

We first prove a sensitivity lower bound using the IKW
construction~\cite{impagliazzo2009new}. Its small left
alphabet lets us obtain Dominating Set sensitivity lower bounds for approximation
factors proportional to the logarithm of the maximum degree.
The parameter $k$ below controls the soundness threshold, how small expected value we need for recovery to the initial source instance.

\begin{theorem}[Sensitivity lower bound via IKW; informal version of \Cref{thm:intro-label-cover-sensitivity-IKW}]
\label{thm:informal-label-cover-ikw}
For every sufficiently large $N$ and sufficiently large integer
$k=o(\log N)$, there is a family of instances with $n=N^{\Theta(k)}$
vertices such that any randomized algorithm for
left-predicate Label Cover that attains expected value at least
$\exp(-C\sqrt{k})$ on every satisfiable input has sensitivity
$n^{\Omega(1/k)}$. Here $C>0$ is an absolute constant. The instances are
biregular, with degrees and alphabet sizes bounded by $2^{O(k)}$.
\end{theorem}

As $k$ grows, the required value decreases, but so does the exponent in
the IKW sensitivity lower bound. To obtain an exponent independent of the
required value, we instead adapt the DH construction~\cite{dinur2013composition}. This leads to a truly polynomial sensitivity lower bound even
for inverse-polylogarithmic value guarantees, at the cost of a larger
left alphabet.

\begin{theorem}[Sensitivity lower bound via Dinur--Harsha; informal version of \Cref{thm:intro-label-cover}]
\label{thm:informal-label-cover-dh}
There are absolute constants $B,\delta>0$ such that, for every sufficiently
large $n$ and every $1<g\le\log^B n$, any randomized algorithm for
left-predicate Label Cover that attains expected value at least $1/g$
on every satisfiable input has sensitivity $\Omega(n^\delta)$.
This holds on instances with $n$ left vertices and at most
$n\log^{O(1)} n$ right vertices, with left and right degrees and right
alphabet size bounded by $g^{O(1)}$, and left alphabet size bounded by
$\exp(g^{O(1)})$.
\end{theorem}

\Cref{tab:intro-lc-engines} compares the two constructions at the same
required value.

\begin{table}[ht!]
\centering
\caption{Sensitivity lower bounds for left-predicate Label Cover.
For IKW, $n$ counts all vertices and $g=\exp(C\sqrt{k})$;
the bound holds at the sizes and in the parameter range of
\Cref{thm:intro-label-cover-sensitivity-IKW}.
For DH, $n$ counts left vertices, and the bound holds for all sufficiently
large $n$ and $1<g\le\log^B n$ as in \Cref{thm:intro-label-cover}.
The DH exponent $\delta>0$ is independent of $g$.}
\label{tab:intro-lc-engines}
\small
\setlength{\tabcolsep}{4pt}
\begin{tabular}{@{}p{0.16\linewidth}p{0.15\linewidth}p{0.23\linewidth}p{0.38\linewidth}@{}}
\toprule
Construction & Value guarantee & Sensitivity & Degrees and alphabet sizes \\
\midrule
IKW & $1/g$ & $n^{\Omega(1/\log^2 g)}$
& Degrees and alphabets $\exp(O(\log^2 g))$ \\
\addlinespace
Dinur--Harsha & $1/g$ & $\Omega(n^\delta)$
& Degrees and right alphabet $g^{O(1)}$; left alphabet $\exp(g^{O(1)})$ \\
\bottomrule
\end{tabular}
\end{table}

\subsubsection{Covering and Network Design}

For Set Cover, Dominating Set, Min-Rep, and minimum unweighted
directed 3-spanner, we prove that even a $\log^\beta n$ approximation
requires sensitivity $n^{\Omega(1)}$, for some absolute
constant $\beta>0$.
The bounds apply to randomized algorithms with unrestricted computation
and hold for all sufficiently large $n$, where $n$ is the universe size
for Set Cover and the number of vertices for the graph problems.

For Dominating Set and directed 3-spanner, our sensitivity lower bounds
also yield lower bounds for distributed approximation.
By the sensitivity-to-locality bound of Fleming and
Yoshida~\cite{fleming2026sensitivity}, polynomial sensitivity implies
logarithmic locality on graphs of constant maximum degree (c.f., \Cref{sec:intro-technical}).
In the LOCAL model, vertices communicate with their neighbors in
synchronous rounds, with unrestricted message size and local computation;
the number of rounds is the locality.
Our locality lower bounds hold even in the more general
\emph{non-signaling} model, where the joint output distribution on
every vertex set depends only on that set's radius-$t$ input
neighborhood at locality $t$
(\Cref{subsec:non-signaling-locality}).
We highlight two applications.

\paragraph{Dominating Set: stability and locality near $\ln\Delta$.}
Dominating Set asks for a smallest vertex set such that every vertex
either belongs to the set or has a neighbor in it.
An input edit inserts or deletes one edge, and output distance counts
changes in vertex membership.
Fix any $\varepsilon\in(0,1)$ and a sufficiently large constant
maximum-degree bound $\Delta$.
A $(1+\varepsilon)\ln\Delta$ approximation can be attained with
sensitivity and round complexity independent of $n$ in the CONGEST model,
which restricts LOCAL messages to $O(\log n)$ bits.
This follows from the algorithm of Kuhn, Moscibroda, and
Wattenhofer~\cite{kuhn2006price}
combined with randomized rounding
(\Cref{subsec:ds-logarithmic-approximation}).
In contrast, we show that improving the ratio to $(1-\varepsilon)\ln\Delta$
requires sensitivity
$n^{\Omega_{\varepsilon,\Delta}(1)}$ and locality
$\Omega_{\varepsilon,\Delta}(\log n)$, with the locality lower bound
holding even in the non-signaling model
(\Cref{thm:intro-dominating-set-IKW,thm:intro-locality-dominating-set-IKW}).
Together, these bounds identify $\ln\Delta$, to leading order,
as the approximation threshold for achieving sensitivity and
round complexity independent of $n$ at constant maximum degree.

To compare with previous locality lower bounds, let $\alpha$ denote
the approximation ratio and $t$ the locality, which in LOCAL is
the number of communication rounds.
Chang and Li~\cite{chang2023complexity} prove the LOCAL lower bound
\[
    \alpha=1+\eta
    \quad\Longrightarrow\quad
    t=\Omega\!\left(\frac{\log n}{\eta}\right)
\]
for sufficiently small constants $\eta>0$.
For every constant $c\ge1$, the KMW
LOCAL lower bound~\cite{kuhn2016local,coupette2021breezing} gives
\[
    \alpha=\log^c\Delta
    \quad\Longrightarrow\quad
    t=\Omega\!\left(\min\!\left\{
        \sqrt{\frac{\log n}{c\log\log n}},
        \frac{\log\Delta}{c\log\log\Delta}
    \right\}\right).
\]
At constant $\Delta$, the expression inside this bound is independent
of $n$ for all sufficiently large $n$.
Our result retains logarithmic dependence on $n$ while allowing
approximation ratios approaching $\ln\Delta$ from below:
\[
    \alpha=(1-\varepsilon)\ln\Delta
    \quad\Longrightarrow\quad
    t=\Omega_{\varepsilon,\Delta}(\log n),
\]
for every fixed $\varepsilon\in(0,1)$ and sufficiently large constant
$\Delta$, even in the non-signaling model.
\Cref{sec:dominating-set} gives the full degree-dependent bounds and
comparisons.

\paragraph{Directed 3-spanner: locality of constant-factor approximation.}
A directed 3-spanner retains a subset of the input arcs such that every
input arc has a directed replacement path of length at most three.
This feasibility condition can be checked within a constant-radius
neighborhood. We show that approximating the minimum number of retained
arcs within any fixed constant factor nevertheless requires information
from logarithmic distances, even with unrestricted message size and
local computation.

We measure sensitivity under insertion or deletion of one arc, with
output distance counting changed arc memberships. In the distributed
formulation, communication is bidirectional on the underlying undirected
graph, and each vertex selects incident input arcs whose union is the
output spanner.
For every constant approximation factor $\alpha\ge1$, our sensitivity
lower bounds imply
\[
    t=\Omega_\alpha(\log n),
\]
already on weakly connected DAGs of constant maximum total degree
depending on $\alpha$ (\Cref{cor:spanner-constant-rounds}).
This lower bound on $t$ holds in LOCAL and even in the non-signaling model.

Together with the $(1+\varepsilon)$-approximation algorithm of
Censor-Hillel and Dory~\cite{censorhillel2021spanner}, our lower bounds
place the LOCAL complexity of every constant approximation factor
greater than one between logarithmic and polylogarithmic in $n$.
Their lower bounds concern directed unweighted $k$-spanner for $k\ge5$ in CONGEST,
and weighted 2-spanner in LOCAL. These do not give a LOCAL lower bound
for unweighted directed 3-spanner.
\Cref{sec:directed-spanner,sec:spanner-distributed} give the full
parameter-dependent bounds and comparisons.

\paragraph{Set Cover and Min-Rep.}
The reductions to Dominating Set and directed 3-spanner pass through
Set Cover and Min-Rep, respectively, and also yield sensitivity lower
bounds for these two problems.
Set Cover asks for a minimum number of sets covering a given universe.
A Min-Rep instance is a bipartite graph whose vertices on each side are
partitioned into groups; the goal is to select as few vertices as
possible so that every pair of groups joined by an edge has an edge
with both endpoints selected.
An input edit changes one set--element incidence for Set Cover, with
the universe and set indices fixed, or one edge for Min-Rep, with the
vertices and groups fixed. Output distance counts changes in selected
sets or vertices, respectively.

For both problems, our DH construction gives
\[
    \log^\beta n\text{-approximation}
    \quad\Longrightarrow\quad
    n^{\Omega(1)}\text{ sensitivity}
\]
for some absolute constant $\beta>0$ and all sufficiently large $n$,
where $n$ is the universe size for Set Cover and the number of vertices
for Min-Rep (\Cref{thm:intro-set-cover,cor:minrep-dh}).
The Set Cover bound holds even with $O(n)$ sets.
\Cref{sec:set-cover} also gives IKW-based bounds when set sizes and
element frequencies are bounded.

\subsection{Technical Overview}
\label{sec:intro-technical}

All our constructions start from the same 2-CSP family of Fleming and
Yoshida~\cite{fleming2026sensitivity}. A 2-CSP has one variable at each
vertex and a binary constraint on each edge. At each sufficiently large size $N$, the family
uses a common constant-degree regular graph on $N$ vertices and a
common constant-size alphabet;
only the constraint relations vary. For absolute constants
$\eta_{\mathrm{src}},c_{\mathrm{src}}>0$, producing assignments of expected value
at least $1-\eta_{\mathrm{src}}$ on its satisfiable members requires sensitivity
$N^{c_{\mathrm{src}}}$ (\Cref{lem:fy-source}). To prove
\Cref{thm:intro-label-cover-sensitivity-IKW,thm:intro-label-cover}, we
recover source assignments with this value guarantee from
Label Cover labelings whose expected value exceeds a small threshold.

We must control how recovery responds to changes in both the supplied
labeling and the source instance. All value guarantees below concern
satisfiable source instances.

\subsubsection{Transferring Sensitivity through Recovery}

Each of our reductions consists of a transformation $\mathcal T$ from source
instances $I$ to target instances $\mathcal T(I)$. To control how the sensitivity changes through these reductions, we equip each instance $I$ with a recovery procedure as follows. Given an algorithm $A$ for the target problem, the recovery procedure may use
several independent outputs of $A$ on $\mathcal T(I)$, together with
additional randomness and the source instance $I$. We therefore view the recovery procedure as a map $\mathcal R_I$ from distributions over
target outputs to distributions over source outputs. Write
$\mu_I=\operatorname{Law}(A(\mathcal T(I)))$ for the target output
distribution, where $\operatorname{Law}$ denotes the distribution of a
random variable. Applying the recovery procedure to outputs of $A$
defines a source algorithm $B$ satisfying
\[
\operatorname{Law}(B(I))=\mathcal R_I(\mu_I).
\]

We bound the sensitivity of $B$ in terms of
four quantities:
\begin{itemize}
\item $S_A$ is the sensitivity of the target algorithm $A$:
the maximum EMD between its output distributions on target
instances differing by one allowed edit.
\item $C_{\mathcal T}$ bounds the number of target edits caused by one
source edit.
\item $C_{\mathcal R}$ bounds the factor by which recovery can
increase the EMD between two target output distributions,
with the source instance $I$ unchanged.
For every source instance $I$ and all distributions $\mu,\nu$ over its
target range,
\[
\EMD(\mathcal R_I(\mu),\mathcal R_I(\nu))
\le C_{\mathcal R} \cdot \EMD(\mu,\nu).
\]
\item $D_{\mathcal R}$ bounds the EMD caused by replacing $I$ by a
neighboring source instance while supplying the same output distribution.
For every neighboring pair $I,\widetilde I$ and every distribution $\mu$
over their common target range,
\[
\EMD(\mathcal R_I(\mu),\mathcal R_{\widetilde I}(\mu))
\le D_{\mathcal R}.
\]
\end{itemize}
For neighboring source instances $I,\widetilde I$, the target instances
$\mathcal T(I),\mathcal T(\widetilde I)$ can be joined by a sequence of
at most $C_{\mathcal T}$ single target edits. Coupling the outputs of $A$
along this sequence gives
$\EMD(\mu_I,\mu_{\widetilde I})\le C_{\mathcal T}S_A$.

Compare the recovered distributions in two steps. First keep the recovery
map $\mathcal R_I$ fixed and change the supplied distribution from
$\mu_I$ to $\mu_{\widetilde I}$. Then keep $\mu_{\widetilde I}$ fixed and
replace $\mathcal R_I$ by $\mathcal R_{\widetilde I}$.
The triangle inequality therefore bounds the sensitivity
$\operatorname{Sens}(B)$ of $B$:
\[
 \operatorname{Sens}(B)\le
C_{\mathcal T}C_{\mathcal R}S_A+D_{\mathcal R}.
\]
If recovery gives $B$ the approximation guarantee required by a source
sensitivity lower bound $L$, then
$C_{\mathcal T}C_{\mathcal R}S_A\ge L-D_{\mathcal R}$.
The multiplicative loss is the product $C_{\mathcal T}C_{\mathcal R}$: critically $C_{\mathcal T}$ will often be quite large and we will require recovery $C_{\mathcal R} < 1$ that actually contracts EMD to compensate for the transform causing
many target edits. We also need $D_{\mathcal R}<L$ for the transferred
lower bound to be positive.

\Cref{lem:law-sensitivity-pullback} states the formal transfer bound;
\Cref{lem:law-recovery-composition} bounds the losses under composition.

\subsubsection{Sensitivity Transfer through Direct-Product Amplification (\Cref{thm:intro-label-cover-sensitivity-IKW})}

Let $\Phi$ be the source 2-CSP, with variable set $X$ and constraint
set $E$, and put $k'=\lfloor\sqrt{k}\rfloor$.
The encoding map $\mathcal T_{\mathrm{IKW}}^k$ follows the IKW
direct-product construction~\cite{impagliazzo2009new} and sends $\Phi$ to a
Label Cover instance with left vertices $E^k$ and right vertices $X^{k'}$.
A label at a left vertex $(e_1,\ldots,e_k)$ assigns values to the
$2k$ endpoint positions. Its left predicate requires that repeated
occurrences of a variable receive the same value and that all $k$
source constraints are satisfied.
A label at a right vertex assigns values to its $k'$ variable positions.
Each edge checks that the left label projected onto selected endpoint
positions agrees with the right label.
Using the IKW agreement analysis, we
construct a randomized recovery procedure that can find a nearly
satisfying source assignment even when the supplied labeling has value
only $\varepsilon=\exp(-\Theta(\sqrt{k}))$.
Increasing $k$ thus lets us handle weaker value guarantees.
Our contribution here is to make this recovery stable under changes
to both the supplied labeling and the source instance.

The obstacle for sensitivity is that one source constraint appears in
many tuples. Changing its relation affects only the left predicates
at tuples containing it; the graph and edge projections stay unchanged.
The edit bound $C_{\mathcal T}$ can nevertheless be very large.
We therefore need a small $C_{\mathcal R}$ to keep the product
$C_{\mathcal T}C_{\mathcal R}$ small.

\paragraph{Bounding changes in the decoded assignment.}
To apply the agreement analysis on satisfiable source instances, we
first replace each violating left label by a label satisfying its
predicate, using only the label and predicate at that vertex.
This cannot decrease value and, for a fixed source instance,
cannot increase Hamming distance between labelings.

The core decoding step chooses a right vertex uniformly at random
and uses its label on the variables in that tuple.
To recover each source variable outside that tuple, it samples edges
incident to the chosen right vertex whose left tuples contain the variable,
and reads its value from
a left label whose projection agrees with the chosen right label,
if one is found. The agreement analysis shows that a sufficient fraction
of right vertices lead to good source assignments.
For sensitivity, we arrange the sampling so that, averaged over the
recovered variable and the chosen right vertex, the left endpoint of
each sampled edge is uniform among all left vertices.
Two executions using the same randomness can then
disagree at a variable only if they read a changed label.

To control changes to right labels as well, we use the standard
right-degree reduction and its recovery
rule~\cite[Theorem~5.4]{dinur2009composition}, replacing each right vertex with one
copy per incident edge and using an expander to redistribute the edge
incidences among these copies (\Cref{alg:degree_reduction}).
The recovery rule chooses each original right label from a uniformly chosen copy.
A change to one copy's label can affect
a decoding trial only when that copy supplies the label of the chosen
right vertex. Writing $N=|X|$ for the number of source variables,
our coupling analysis, together with the right-vertex copying,
gives a factor $N/|E|^k$ in the bound on $C_{\mathcal R}$, with additional
factors from sampling and candidate selection.
On the transformation side, one source edit changes $O(k|E|^k/N)$ left
predicates, giving $C_{\mathcal T}\le O(k|E|^k/N)$.
Thus the factors $|E|^k$ and $N$ cancel in
$C_{\mathcal T}C_{\mathcal R}$.

We must also bound $D_{\mathcal R}$, the effect of changing the source
instance while keeping the supplied labeling fixed.
At the decoding stage, a source edit affects recovery only through the
replacement of violating left labels. The replacement can change only
at left tuples containing the edited constraint, so the same sampling
argument bounds this contribution to $D_{\mathcal R}$.

\paragraph{Selecting a good candidate without introducing instability.}
Suppose the target algorithm guarantees expected value at least
$4\varepsilon$. To achieve the required expected source value of
$1-\eta_{\mathrm{src}}$, we perform $O(\varepsilon^{-2})$
independent trials so that a nearly satisfying candidate exists with
sufficiently high constant probability.
Selecting the candidate with the best value could be unstable: two
assignments may differ on many variables while their values are almost
tied. We instead choose a threshold uniformly from a short interval
near one and return the first candidate whose value exceeds it.
Under a shared threshold, a small change in value changes a candidate's
eligibility only with small probability. This also controls the effect
of changing a source constraint, since such an edit changes the value
of any fixed assignment by at most $1/|E|=O(1/N)$.

We show that, after including these recovery steps, both the multiplicative loss
$C_{\mathcal T}C_{\mathcal R}$ and the additive loss $D_{\mathcal R}$
are $N^{o(1)}$ when $k=o(\log N)$
(\Cref{prop:ikw-recovery}). The source lower bound therefore remains
$N^{\Omega(1)}$. The limitation is the size of the amplified instance:
it has $n_{\mathrm{LC}}=N^{\Theta(k)}$ vertices, so this becomes
$n_{\mathrm{LC}}^{\Omega(1/k)}$. Thus stronger amplification weakens
the sensitivity exponent in the final size. On the other hand, labels
describe only $O(k)$ source symbols, giving the small alphabets that
benefit our Dominating Set application.

\subsubsection{Sensitivity Transfer through Iterated Composition (\Cref{thm:intro-label-cover})}

Our adaptation of the DH construction gives a polynomial sensitivity lower bound for
a Label Cover value guarantee of $1/g$, with an exponent independent of $g$,
even when $g$ is polylogarithmic in the instance size.
Our construction uses standard ingredients from algebraic PCPs,
including polynomial encodings and low-degree testing.
We first reduce the source 2-CSP to a Label Cover instance with
polynomially many vertices. This reduction already permits recovery
from labelings of inverse-polylogarithmic value.
However, the degrees and alphabet sizes of the resulting Label Cover
instance still depend on the source size. In our subsequent reductions,
we create sets or representatives for each label, so large alphabets
produce large instances and greater sensitivity loss.
To remove this dependence on the source size, we use the
Dinur--Harsha composition framework~\cite{dinur2013composition},
together with alphabet and degree reduction, to obtain degree and
alphabet-size bounds depending only on $g$.
Each round replaces the current local constraints by tests on auxiliary
proofs, then applies alphabet and degree reduction to prepare the
resulting instance for the next round. Throughout the iteration, we
maintain decreasing upper bounds on the query complexity and the size
of each local checking circuit.
Each composition step raises the value needed for recovery, so we
start with a threshold sufficiently below $1/g$ to allow repeated composition.
For sensitivity transfer, we additionally need to control how changes
to the source instance or the supplied labeling affect recovery,
and how these losses accumulate under composition.

\paragraph{1. Recovering a good assignment when local answers are ambiguous.}
The initial reduction maps a source 2-CSP to a Label Cover instance.
A satisfying source assignment induces a satisfying labeling by
extending the assignment to a polynomial, assigning its restrictions
to left vertices, and assigning its point evaluations to right vertices.
The constraints check agreement between these labels and satisfaction
of groups of source constraints.
The groups are chosen so that an assignment violating a constant
fraction of source constraints rarely satisfies a whole group.

At a small acceptance rate, the local answers need not identify a
unique polynomial, but the Arora--Sudan low-degree
test~\cite{arora2003improved} implies there is a short list of global
polynomials accounting for most of the accepted comparisons.
This list guarantee does not by itself provide a stable recovery procedure:
small changes to the supplied labeling may change the list.
Our contribution at this stage is a stable recovery procedure that
scores every polynomial in a fixed space common to all inputs.
Each polynomial's score is the fraction of left answers that
agree with its restriction and pass their constraint checks.
The list guarantee ensures that some polynomial has a large score
whenever enough tests pass. The choice of constraint groups ensures
that every polynomial above a suitable score threshold represents a
nearly satisfying source assignment.

Each score is an average over the left vertices. If there are $A$
left vertices, changing one left label changes every candidate's score
by at most $1/A$; right labels do not enter the scores.
Changing one source constraint affects only a small fraction of the
terms in each average. To select a polynomial stably, we use
a random threshold, as in the IKW recovery, and take the first qualifying
polynomial in an order fixed across inputs. Only a bounded number of
polynomials can have scores near the threshold. Under a shared threshold,
the score perturbation and the bound on the number of near-threshold
polynomials control the probability that the selected polynomial changes.
We repeat recovery on independent algorithm outputs and return the
assignment decoded from the first successful trial. Sufficiently many
trials give expected source value at least $1-\eta_{\mathrm{src}}$
(\Cref{thm:ldrc-free-initializer}).

\paragraph{2. Simplifying a test without changing where it reads.}
For composition, we view the constraint at each old left vertex $u$
as a circuit $C_u$ on the labels of its right neighbors, with one input
coordinate per incident edge. The circuit accepts a tuple exactly when
some left label satisfying $P_u$ projects to its entries along the
corresponding edges.
We replace this circuit by tests on an auxiliary proof intended to
encode an accepted tuple. A decodable PCP~\cite{dinur2013composition}
provides such a proof as a string over a finite alphabet. Each test
reads a few proof symbols and attempts to recover a specified coordinate
of the encoded tuple, or rejects. Thus the test for an edge $(u,v)$
attempts to recover the old right label at $v$.

We use the DH composition rule~\cite[Section~4]{dinur2013composition}
to combine these tests into a new Label Cover instance. Each new left
vertex combines tests that decode the label of the same old right
vertex, and its predicate requires all of them to return the same label.
Each proof position used by the tests becomes a right vertex, whose
label specifies the symbol at that position. A new left label specifies
proposed symbols for the positions queried by the tests being combined.
Projection constraints check that these proposed symbols agree with
the corresponding right labels. Thus the number of queries controls the length
of the new left labels, without requiring any label to contain an
entire proof. The soundness guarantee lets us recover a good labeling
of the old instance from one satisfying sufficiently many new constraints.
\Cref{fig:composition-overview}(a) illustrates the construction.
\begin{figure}[t!]
\centering
\begin{tikzpicture}[x=1cm,y=1cm,>=stealth,
  every node/.style={font=\small},
  oldvertex/.style={circle,draw,minimum size=7mm,inner sep=1pt},
  symbol/.style={circle,draw,fill=white,minimum size=12mm,inner sep=1pt,font=\scriptsize},
  recovery/.style={->,dashed,thick}]
\node[anchor=west,font=\small\bfseries] at (0,6.25)
  {(a) One new left vertex combines tests for the same old right label};
\node[font=\footnotesize] at (7.5,5.75)
  {Old vertex $v$ has neighbors $u_1,u_2$. Both tests try to decode the label at $v$.};
\node[font=\footnotesize] at (7.5,5.3)
  {Fix the test randomness $\rho$; test $i$ for $(u_i,v)$ queries positions $z_{i1},z_{i2}$.};
\node at (1.3,4.75) {Old graph};
\node at (6.5,4.75) {One new left vertex $(v,\rho)$};
\node at (12.85,4.75) {Neighbors of $(v,\rho)$};

\node[oldvertex,label={[font=\footnotesize]above:$u_1$}] (u1) at (0.4,3.75) {};
\node[oldvertex,label={[font=\footnotesize]above:$u_2$}] (u2) at (0.4,2.45) {};
\node[oldvertex,label={[font=\footnotesize]above:$v$}] (v) at (2.0,3.1) {};
\draw (u1) -- (v) -- (u2);
\node[symbol] (newleft) at (6.5,3.1)
  {$w=(w_1,w_2)$};
\draw[->,thick] (2.7,3.1) -- (newleft.west);
\node[draw,rounded corners,fill=gray!4,anchor=north,
  text width=4.9cm,inner sep=5pt,align=center,font=\footnotesize] at (6.5,1.3)
  {\textbf{Predicate at $(v,\rho)$}\\[2pt]
   For $i=1,2$, write $w_i=(a_{i1},a_{i2})$.\\
   Run test $i$ on $w_i$ using $\rho$.\\
   Accept iff the outputs agree\\
   and are not $\bot$.};

\draw[rounded corners,fill=blue!4] (10.65,2.65) rectangle (15.05,4.4);
\draw[rounded corners,fill=blue!4] (10.65,0.4) rectangle (15.05,2.2);
\node[symbol] (p11) at (11.6,3.4) {$p_{u_1}[z_{11}]$};
\node[symbol] (p12) at (14.1,3.4) {$p_{u_1}[z_{12}]$};
\node[symbol] (p21) at (11.6,1.15) {$p_{u_2}[z_{21}]$};
\node[symbol] (p22) at (14.1,1.15) {$p_{u_2}[z_{22}]$};
\node[font=\footnotesize] at (11.6,4.18) {$(u_1,z_{11})$};
\node[font=\footnotesize] at (14.1,4.18) {$(u_1,z_{12})$};
\node[font=\footnotesize] at (11.6,1.93) {$(u_2,z_{21})$};
\node[font=\footnotesize] at (14.1,1.93) {$(u_2,z_{22})$};
\draw (newleft.east) -- (p11.west);
\draw (newleft.east) -- (10.2,2.5) -- (14.1,2.5) -- (p12.south);
\draw (newleft.east) -- (p21.west);
\draw (newleft.east) -- (9.8,1.7) -- (10.2,0.15) -- (14.1,0.15) -- (p22.south);
\node[anchor=north,align=center,font=\footnotesize] at (12.85,-0.1)
  {Each edge checks $a_{ih}=p_{u_i}[z_{ih}]$\\
   for $i,h\in\{1,2\}$.};

\draw[gray!50] (0,-1.3) -- (15.05,-1.3);
\node[anchor=west,font=\small\bfseries] at (0,-1.85)
  {(b) Which new labels determine each recovered old label?};
\node[draw,rounded corners,fill=blue!4,minimum width=6.2cm,
  minimum height=0.85cm,align=center] (rightlabels) at (3.2,-2.75)
  {All new right labels at $(u,z)$\\form the proof string $p_u$};
\node[draw,rounded corners,minimum width=4.4cm,minimum height=0.85cm,align=center]
  (oldleft) at (12.85,-2.75) {Recovered label\\at old left vertex $u$};
\draw[recovery] (rightlabels) -- node[above,font=\footnotesize]
  {recover from $p_u$} (oldleft);
\node[draw,rounded corners,minimum width=6.2cm,
  minimum height=0.85cm,align=center] (leftlabel) at (3.2,-4.1)
  {One new left label at $(v,\rho)$\\with $\rho$ sampled uniformly};
\node[draw,rounded corners,minimum width=4.4cm,minimum height=0.85cm,align=center]
  (oldright) at (12.85,-4.1) {Recovered label\\at old right vertex $v$};
\draw[recovery] (leftlabel) -- node[above,font=\footnotesize]
  {run the decoding tests} (oldright);
\node[align=center,font=\footnotesize] at (7.5,-5.05)
  {Changing one new right label can affect recovery only at its old vertex $u$.\\
   Changing one new left label can affect recovery only at its old vertex $v$.};
\end{tikzpicture}
\caption{Composition and recovery, illustrated with two old edges
$(u_1,v),(u_2,v)$ and two queries per test.
In (a), $u_1,u_2$ are neighbors of $v$ in the old graph;
they index the answer strings $w_1,w_2$ within the single label at $(v,\rho)$.
The entry $a_{ih}$ is a proposed answer to query $h$ of test $i$.
Vertex names are outside the circles, and new vertex labels are inside;
the shaded boxes on the right group symbols of the same proof.
Lines without arrowheads are graph edges; the solid arrow denotes composition.
Query positions, and hence the new edges, stay fixed when old constraints change.
In (b), dashed arrows show recovery dependencies, with the same random
choices used when comparing two labelings. For an old constraint edit at $u$,
keep the new labeling fixed: recovery can change only at $u$ and its old neighbors.}
\label{fig:composition-overview}
\end{figure}
 
Changing the queried positions would therefore change which right
vertices a new left vertex is adjacent to. To preserve the graph under
a source edit, we modify the decoder to make these positions independent of the circuit $C_u$,
a property called \emph{query-obliviousness}.
To achieve this query-obliviousness, we use a fixed universal circuit
whose input contains an encoded
circuit description and a candidate tuple of right labels.
The universal circuit checks that the description is a valid codeword,
decodes it, and checks that the described circuit accepts the candidate tuple.
We apply a decodable PCP to this fixed universal circuit.
To recover a coordinate of the candidate tuple, the decoder also attempts
to recover a uniformly sampled symbol of the encoded description.
We compare this symbol with the corresponding symbol in the encoding
of the description of $C_u$. This comparison checks that the decoded
description corresponds to $C_u$.
In this decoder, only this comparison depends on $C_u$, so the query
locations remain fixed.
Distinct descriptions have encodings that
agree on few positions, making acceptance under a circuit different from $C_u$ unlikely.

To control the sensitivity of recovery as well, we design the procedure
in \Cref{fig:composition-overview}(b) to have local dependencies.
For each old left vertex $u$, assemble the new right labels at
$(u,z)$ into a proof $p_u$. We use the decoder's list of tuples accepted
by $C_u$ to recover a label at $u$: each candidate tuple corresponds
to a left label satisfying $P_u$ that projects to it.
This step depends only on $p_u$ and the old local constraint at $u$.
Independently, sample the test randomness $\rho$ uniformly. For each old
right vertex $v$, run the local decoders on the symbols in the left label
at $(v,\rho)$ in the composed instance. Their common output supplies
the recovered label when all decoders accept and agree.

With the same recovery randomness, changing one label in the composed
instance affects at most one recovered label, giving
$C_{\mathcal R}\le 1$ for this composition step.
If an old constraint at $u$ changes while the labeling of the composed
instance is fixed, only recovery at $u$ and its
neighbors can change. Both bounds are independent of proof length
(\Cref{thm:composition}).

\paragraph{3. Reducing local circuit size under composition.}
Reading few proof symbols does not by itself make a test small:
the circuit that processes them could still be large. This distinction
matters because the next composition step takes that circuit as its
input. We analyze the DH algebraic
decoder~\cite[Theorem~6.5]{dinur2013composition} to bound its local circuit size.

At each stage, for input circuits of size at most $M$, we choose a decoder
whose number of queries and local decoding circuit size are both at most
$M^\alpha$, for a sufficiently small absolute constant $\alpha\in(0,1)$.
This saving is large enough that the local circuit-size bound still
decreases after the subsequent alphabet and degree reductions.

\paragraph{4. Keeping the accumulated loss small.}
For the initial encoding, the stable scoring procedure used in the IKW construction gives
polylogarithmic multiplicative and additive sensitivity losses. We now bound the additional losses from iterated composition.
Let $n_0$ be the initial size parameter, which is polynomial in the source
size. We show that, at stage $i$, the factors controlling size growth and sensitivity
loss are at most $\log^{O(4^{-i})} n_0$. Since
$\sum_{i\ge0}4^{-i}=4/3$, their products remain polylogarithmic in $n_0$.
Thus the final instance size is polynomial in the source size, with an
exponent independent of $g$, and the accumulated transfer parameters satisfy
$C_{\mathcal T}C_{\mathcal R},D_{\mathcal R}\le\log^{O(1)} n_0$
(\Cref{prop:stage-update,lem:iterative-family}).
Consequently, the
sensitivity lower bound remains polynomial in the final instance size.

For a target value $1/g$, we choose a stopping stage at which the
local tests have size bounded in terms of $g$ alone, while the
recovery threshold remains at most $1/g$
(see the proof of \Cref{thm:intro-label-cover}).

\subsubsection{From Label Cover to Covering Problems}

We adapt classical reductions from Label Cover to covering and network-design
problems. In the parameter ranges used below, recovery from an
$\alpha$-approximate Set Cover or Min-Rep solution guarantees expected
Label Cover value $\Omega(\alpha^{-2})$
(\Cref{lem:sc-soundness,lem:minrep-recovery}).
For growing $\alpha$, this guarantee does not suffice to apply Fleming
and Yoshida's sensitivity lower bound, which requires Label Cover
value close to one. In contrast, our low-soundness bounds apply in this range, and
stable recovery maps let us transfer sensitivity.
The reductions follow two branches, with solutions recovered in the
reverse direction:
\begin{center}
\begin{tikzpicture}[every node/.style={font=\small},>=stealth]
\node (lc) at (0,0) {Label Cover};
\node (sc) at (3,0.55) {Set Cover};
\node (mr) at (3,-0.55) {Min-Rep};
\node (ds) at (7,0.55) {Dominating Set};
\node (sp) at (7,-0.55) {Directed 3-Spanner};
\draw[->] (lc) -- (sc);
\draw[->] (lc) -- (mr);
\draw[->] (sc) -- (ds);
\draw[->] (mr) -- (sp);
\end{tikzpicture}
\end{center}

The covering branch uses the classical reductions from Label Cover
to Set Cover~\cite{lund1994hardness,feige1998threshold} and from Set Cover
to bounded-degree Dominating Set~\cite{chlebik2008approximation}.
We show that these reductions admit stable recovery while preserving
the leading logarithmic approximation coefficient.

Start with a biregular Label Cover instance on $(U,V,E)$, with left and
right degrees $d_L,d_R$. Recall that the classical Set Cover construction
has one set for each left vertex--label pair and one copy of a universe $B$
for each right vertex. Each right label indexes a partition of $B$ into
$d_R$ parts, supplied by compatible left labels on the incident edges;
a satisfying labeling gives a cover of size $|U|$.

For our local recovery, we need a right label supported by enough incident
edges to be found by local sampling. We use a quantitative strengthening
of the partition guarantee in~\cite[Lemma~3.2]{feige1998threshold}:
a small cover must contain a constant fraction of the parts of one
partition, whereas the cited lemma guarantees at least two parts.
We show that the standard probabilistic construction gives this stronger
guarantee for covers of size sufficiently below $d_R\ln|B|$, when $|B|$ is sufficiently
large relative to $|\Sigma_V|d_R$ (\Cref{lem:m-l-set-system}).

Given a cover, for each left vertex $u$, let $L_u$ contain the labels $y$
whose sets $(u,y)$ are selected. Recovery labels each $u$ uniformly from
$L_u$. Independently, for each right vertex $v$, it samples a uniform
incident edge $e=(u,v)$ and a fresh uniform label $y\in L_u$, then assigns
$v$ the right label $f_e(y)$ prescribed by the edge $e$.
The strengthened partition guarantee ensures enough agreement between these independent choices to obtain the expected value bound above.

This sampling is also stable. Adding or removing a set
at $u$ changes only $L_u$. Under a coupling using the same sampled edges
and identical label samples wherever the lists agree, recovery can change
only at $u$ and at right vertices whose sampled edge is incident to $u$.
The expected number of changed labels is therefore at most $1+d_L/d_R$.
The recovery does not use the predicates, and a projection
swap can change only its right endpoint's recovered label.
A source edit changes at most $d_L|\Sigma_U||B|$ incidences
(\Cref{lem:sc-recovery}).

Each set has size at most $d_L|B|$, and each element belongs to at most
$d_R|\Sigma_U|$ sets. To keep both quantities at most a prescribed bound
$\Gamma$, for any constant $\varepsilon\in(0,1)$ we choose
the IKW parameters so that the degrees and alphabet sizes are
$\Gamma^{o(1)}$. This lets us bound both set sizes and element frequencies
by $\Gamma$ while taking $\ln|B|=(1-o(1))\ln\Gamma$.
A sufficiently small $\varepsilon$-dependent margin then yields the
$(1-\varepsilon)\ln\Gamma$-approximation lower bound
(\Cref{thm:app-set-cover-IKW}).

We then pass from Set Cover to Dominating Set using the classical
incidence-graph reduction. 
Recovering a cover requires only local replacements of selected element-vertices by incident sets, giving constant sensitivity loss (\Cref{prop:ds-exact-size-lift}).

We choose the set-size and frequency bounds proportional to $\Delta$
so that the reduction keeps the maximum degree at most $\Delta$ and
extends a cover to a dominating set of size at most
$(1+O(\varepsilon))$ times the cover size. The logarithms of these source bounds are
$\ln\Delta-O_\varepsilon(1)$, so the reduction preserves the leading
coefficient and yields the sensitivity lower bound for
$(1-\varepsilon)\ln\Delta$-approximation
(\Cref{thm:intro-dominating-set-IKW}).
For polylogarithmic approximation ratios, we instead use the DH
construction. Taking $g=\Theta(\alpha^2)$ preserves a polynomial sensitivity
lower bound with an absolute exponent;
the DH degree and alphabet bounds yield approximation factor
$\log^\beta \Delta$ for a sufficiently small constant $\beta>0$
(\Cref{thm:intro-dominating-set}).

\subsubsection{From Label Cover through Min-Rep to Directed Spanners}

Our spanner application follows the approach of encoding Min-Rep
solutions by short paths, as in classical spanner hardness
reductions~\cite{dinitz2012largegirth}. We adapt this approach to directed
3-spanners with bounded degree. The key difficulties are to encode the
Min-Rep constraints using paths of length three and to amplify the cost
of selecting representatives while preserving the sensitivity lower bound.

\paragraph{Adapting Label Cover to Min-Rep.}
In the reduction to Min-Rep, the labels of each source vertex become a
group of representatives. Edges join pairs whose labels satisfy the
source constraints; we call such pairs compatible. We make as many
copies of each group as the degree of its source vertex.
To recover a labeling, independently at each vertex choose a uniform
group copy and a uniform selected representative in that copy.
We show that an $\alpha$-approximate solution yields expected Label Cover value at least
$\alpha^{-2}$ on a satisfiable source instance, and that changing one
selected representative affects at most one recovered label
(\Cref{lem:minrep-recovery}).

\paragraph{Obtaining directed 3-spanners with bounded degree.}
Before introducing repetitions, add a source $s_i$ for each left group
$A_i$ and a sink $t_j$ for each right group $B_j$.
Include selection arcs $s_i\to a$ for every $a\in A_i$ and
$b\to t_j$ for every $b\in B_j$, and orient each compatibility edge
from its left representative to its right representative.
For each pair of groups $A_i,B_j$ containing a compatible pair, add an
arc $s_i\to t_j$. We call this a \emph{demand arc}: a 3-spanner must
either retain it or provide a path $s_i\to a\to b\to t_j$ through
compatible representatives $a\in A_i$ and $b\in B_j$.
The two selection arcs on this path correspond to selecting $a$ and $b$.

Every input arc must be spanned, including unused selection arcs.
We therefore add one auxiliary vertex $h_i$ per left group and
$\bar h_j$ per right group, with arcs providing the two-arc routes
$s_i\to h_i\to a$ and $b\to\bar h_j\to t_j$ for all representatives
in the respective groups. 
These two-arc routes span omitted selection arcs, while any demand path using one has at least four arcs. 
Using a separate auxiliary vertex for each group avoids a high-degree
vertex shared by all groups.

\paragraph{Preserving sensitivity when copying vertices.}
We next increase the cost of selecting representatives to absorb the
cost of the compatibility and auxiliary arcs.
To this end, for an integer $x$, replace each source and sink by copies $s_i^p,t_j^p$,
indexed by $p\in[x]$, while sharing the representatives and auxiliary
vertices among the copies. Each copy inherits the selection and auxiliary
arcs of the original vertex. For each pair of groups containing a
compatible pair, add demand arcs $s_i^p\to t_j^p$ only between copies
with the same index $p$.

Let $r$ be the number of groups. On a satisfiable source instance,
choosing one representative per group according to a satisfying labeling
uses $xr$ selection arcs across the copies. Adding all compatibility and
auxiliary arcs gives a 3-spanner. We choose $x$ large enough that these
additional arcs number $O(xr)$ as well, so the spanner has $O(xr)$ arcs.
The choice of $x$ is unchanged when an input edge is edited.

Given a spanner $H$, recovery chooses $p\in[x]$ uniformly and selects
every representative incident to a selection arc of $H$ with index $p$.
For each demand arc $s_i^p\to t_j^p$ in $H$, recovery also selects a
compatible pair in $A_i\times B_j$, using only the compatibility edges
between these groups to choose the pair. Each selection or demand arc in $H$ is used
for only one value of $p$ and contributes at most two representatives,
so the expected number selected is at most $2|H|/x$.
In particular, an $\alpha$-approximate spanner yields a Min-Rep solution
of expected size $O(\alpha r)$ on a satisfiable source instance.

We must also show that this copying preserves the sensitivity lower bound.
A compatibility-edge edit changes at most $x$ demand arcs in addition
to the compatibility arc itself, giving $C_{\mathcal T}=O(x)$.
For a fixed input, couple recovery from
two spanners using the same sampled index $p$. Each changed selection or
demand arc then affects at most two representatives, and only with
probability $1/x$, giving $C_{\mathcal R}=O(1/x)$.
Thus $C_{\mathcal T}C_{\mathcal R}=O(1)$, independently of $x$.
With the supplied spanner held fixed, an input edit changes
the chosen compatible pair for at most one pair of groups, contributing
an additive loss $D_{\mathcal R}=O(1)$.

Combining the reductions, an $\alpha$-approximate spanner gives
a Label Cover value guarantee of $\Omega(\alpha^{-2})$. For DH, the size increase and
sensitivity losses are at most $\exp(g^{O(1)})$. Setting
$g=\Theta(\alpha^2)$ and taking $\alpha=\log^\beta n$ for a sufficiently
small constant $\beta>0$ makes these factors $n^{o(1)}$, so a polynomial
sensitivity lower bound is preserved.
\Cref{thm:spanner-distributed-parameterized}
gives the corresponding locality lower bounds.

\paragraph{From sensitivity to locality.}
Finally, we convert the sensitivity lower bounds into locality lower bounds.
Fleming and Yoshida~\cite{fleming2026sensitivity} show that, for
vertex-set outputs on graphs of maximum degree $\Delta$, a sensitivity
lower bound $S$ implies non-signaling locality
\[
    t\ge \Omega\!\left(\frac{\log S}{\log(\Delta+1)}\right)-O(1).
\]
Intuitively, non-signaling allows outputs before and after an edge edit
to be coupled so that changes are confined to a neighborhood of the
edited edge. Bounded degree limits this neighborhood's growth
(\Cref{subsec:non-signaling-locality}). For directed spanners, counting
changed incident arcs introduces only another degree factor
(\Cref{lem:spanner-nonsignaling}).

\subsection{Related Work}

\paragraph{Algorithmic sensitivity.}
Sensitivity has been studied
for network centrality~\cite{murai2019sensitivity}, graph approximation
algorithms~\cite{varma2023average,yoshida2021sensitivity,yoshida2026low},
clustering~\cite{yoshida2022average}, and dynamic
programming~\cite{kumabe2022average}.
For every constant $\varepsilon\in(0,1)$, maximum matching in general
graphs admits a randomized $(1-\varepsilon)$-approximation with
polylogarithmic worst-case sensitivity~\cite{yoshida2027polylogarithmic}.
Recent results for Max Cut include worst-case sensitivity bounds for
approximations arbitrarily close to the Goemans--Williamson ratio and
algorithms with low average sensitivity on expanders admitting a cut
containing nearly all edges~\cite{he2027low}.

On the lower-bound side, Fleming and Yoshida~\cite{fleming2026sensitivity}
proved a sensitivity
lower bound for Label Cover algorithms attaining value close to one
on satisfiable instances.
Locally testable codes yield linear sensitivity lower bounds for
near-optimal approximation of satisfiable CSPs such as Max E3LIN2,
as well as lower bounds for Clique and maximum
$k$-coverage~\cite{yoshida2026ltc}.
We extend the Label Cover lower bound to
inverse-polylogarithmic value guarantees while retaining a polynomial
sensitivity lower bound with a constant exponent, enabling the covering
and spanner applications.

Our deletion interpretation uses the output-distribution requirement of
exact unlearning~\cite{ginart2019making,ullah2021stability,ullah2023adaptive}:
after deletion, the updated output must have the same distribution as a
fresh run on the remaining input. The joint distribution of the outputs
before and after an update is a coupling, so their expected Hamming
distance is at least their EMD, even if the algorithm retains auxiliary
state. For single-edge or single-arc edits, take the graph containing the
changed edge or arc as the initial input. Symmetry of EMD then gives the
same lower bound for deletion. This bounds output changes under exact
deletion; it does not by itself give a lower bound on update time.

\paragraph{PCP constructions with low soundness.}
PCP constructions with low soundness are central in PCP-based hardness
of approximation~\cite{haastad2001some}.  Constructions include the
two-query PCP of Moshkovitz and Raz~\cite{moshkovitz2008two},
Dinur--Harsha composition~\cite{dinur2013composition}, parallel
repetition and direct-product testing
\cite{raz1995parallel,impagliazzo2009new}, and derandomized parallel
repetition via structured PCPs~\cite{dinur2011derandomized}.
More recent constructions obtain low-soundness PCPs using
high-dimensional expanders~\cite{bafna2024quasilinear,o2025low}.
To transfer sensitivity lower bounds, we bound both the number of target
constraints changed by one source edit and the change in recovery when
either the supplied labeling or the source instance changes. We prove
these bounds for our adaptations of the IKW and DH constructions.
We do not know whether constructions based on parallel repetition or
high-dimensional expanders admit comparable sensitivity-transfer bounds.

\paragraph{Covering reductions and distributed Dominating Set.}
The Set Cover reductions originate in classical approximation hardness:
Lund and Yannakakis~\cite{lund1994hardness} introduced the set-system
construction, Feige~\cite{feige1998threshold} used it for Set Cover, and
Chleb\'{i}k and Chleb\'{i}kov\'{a}~\cite{chlebik2008approximation} treated
bounded-degree Dominating Set. We show that these reductions transfer our
Label Cover sensitivity lower bounds to Set Cover and Dominating Set.

For distributed Dominating Set, Chang and Li~\cite{chang2023complexity}
proved an $\Omega(\log n)$ LOCAL lower bound for a constant
approximation factor.  Kuhn, Moscibroda, and Wattenhofer (KMW)
\cite{kuhn2016local} gave degree-dependent upper and lower bounds;
Coupette and Lenzen~\cite{coupette2021breezing} gave a streamlined
presentation of the lower-bound construction.
The KMW lower bound for the linear-programming relaxation of Dominating
Set also holds in the non-signaling model
\cite{gavoille2009can,balliu2025new}.

\paragraph{Spanner approximation.}
Min-Rep is a standard intermediate problem in spanner hardness
reductions; Dinitz, Kortsarz, and Raz~\cite{dinitz2012largegirth} use
Label Cover with large girth to obtain hardness for undirected spanners.
Our directed construction permits explicit control of the three-arc
paths, the effect of individual arc edits, and the maximum degree.

Distributed approximation algorithms for directed spanners include those of
Censor-Hillel and Dory~\cite{censorhillel2021spanner} and Dinitz and
Nazari~\cite{dinitz2018distributednetwork}.

\subsection{Organization of the Paper}

\Cref{sec:pre} gives the formal definitions and the common 2-CSP source.
Part~I develops the sensitivity-transfer framework and the degree and
alphabet reductions.  Part~II proves the IKW Label Cover theorem,
\Cref{thm:intro-label-cover-sensitivity-IKW}.  Part~III proves the DH
theorem, \Cref{thm:intro-label-cover}, through initialization by polynomial
restrictions, query-oblivious composition, and iteration.
Part~IV transfers the bounds to Set Cover and Dominating Set and gives
locality consequences for Dominating Set, with a comparison to prior
distributed bounds.
Part~\ref{part:minrep-spanner} develops the second application through
Min-Rep and directed 3-spanner, then derives the non-signaling and LOCAL
lower bounds for directed 3-spanner.
 \section{Preliminaries}\label{sec:pre}

For a positive integer $n$, let $[n]$ denote the set $\{1,2,\ldots,n\}$.
We use bold symbols such as $\bm\pi$ to denote random assignments.  For a random variable $X$, write
$\mathsf{Law}(X)$ for its probability distribution.

For the minimization problems considered later, a solution is \emph{feasible}
if it satisfies the constraints of the problem, and an instance is
\emph{feasible} if it admits such a solution.
For a feasible instance $I$, write $\operatorname{cost}_I$ for its objective
function and $\opt(I)$ for the minimum cost of a feasible solution.
A randomized algorithm $A$ is an $\alpha$-approximation if, for every feasible
instance $I$, it outputs a feasible solution almost surely and satisfies
\[
    \E[\operatorname{cost}_I(A(I))]\le \alpha \cdot \opt(I).
\]

\subsection{Label Cover}

For the Dinur-Harsha route we will work with the following extension of the Label Cover problem which allows for additional predicates on the left vertices. This will allow us to maintain sensitivity. The usual Label Cover problem is recovered by taking each predicate $P_u$ to be identically $1$.

\begin{definition}[Left-predicate Label Cover instance]\label{def:admissible-game}
    A \emph{left-predicate Label Cover instance} is a tuple
    \[
        I=(U,V,E,\Sigma_U,\Sigma_V,\mathcal P,F),
        \qquad
        \mathcal P=\{P_u:\Sigma_U\to\{0,1\}\}_{u\in U},
        \qquad
        F=\{f_e:\Sigma_U\to\Sigma_V\}_{e\in E},
    \]
    where $G=(U,V,E)$ is a bipartite multigraph, $\Sigma_U$ and $\Sigma_V$ are finite alphabets, each $P_u$ is a unary predicate on left labels.
    Parallel edges are distinct elements of $E$.  Thus, when we write
    $e=(u,v)$, the symbol $e$ identifies a particular parallel edge with
    endpoints $u$ and $v$, and the $f_e \in F$ are called \emph{projections}.
    An assignment is a pair $\pi=(\pi_U,\pi_V)$ with $\pi_U:U\to\Sigma_U$ and $\pi_V:V\to\Sigma_V$.
    An edge $e=(u,v)$ is satisfied by $\pi$ if
    \[
        P_u(\pi_U(u))=1
        \qquad\text{and}\qquad
        f_e(\pi_U(u))=\pi_V(v).
    \]
    Thus a left label that fails its predicate satisfies none of its incident edges.
    For an assignment $\pi$, set $\val_I(\pi):=1$ when $E=\varnothing$;
    when $E\ne\varnothing$, let
    \[
        \val_I(\pi)
        :=
        \frac{1}{|E|}
        \cdot
        |\bigl\{e=(u,v)\in E : P_u(\pi_U(u))=1 \text{ and } f_e(\pi_U(u))=\pi_V(v)\bigr\}|
    \]
    denote the fraction of satisfied edges, and let  $\val(I):=\max_\pi \val_I(\pi)$.
    When the predicates are clear from context, we may abbreviate the tuple by omitting $\mathcal P$.
\end{definition}

For two assignments $\pi,\tilde \pi:X  \to \Sigma$ over the same domain $X$, let $d_{\mathrm H}(\pi,\tilde \pi) = |\{x \in X : \pi(x) \neq \tilde \pi(x)\}|$ denote their Hamming distance.
For two probability distributions $P,Q$ on the same space $\mathcal X$, the \emph{earth mover's distance} (EMD) between $P$ and $Q$ is
\[\EMD(P,Q) := \min_{\Gamma \in \Pi(P,Q)} \E_{(X,Y)\sim \Gamma}\bigl[d_H(X,Y)\bigr],\]
where $\Pi(P,Q)$ denotes the set of couplings of $P$ and $Q$.

\subsection{Sensitivity}

Throughout our reductions, we will measure sensitivity for a (left-predicate) Label Cover instances by swapping one projection or one unary predicate, keeping the underlying graph fixed. 
Let $I=(U,V,E,\Sigma_U,\Sigma_V,\mathcal P,F)$ be a left-predicate Label Cover instance.
For $e \in E$ and $f: \Sigma_U \to \Sigma_V$, let $I^{e \gets f}$ denote the instance obtained from $I$ by replacing the projection $f_e$ with $f$.
Similarly, for $u \in U$ and a predicate $Q:\Sigma_U \to \{0,1\}$, $I^{u \gets Q}$ denotes the instance obtained  by replacing $P_u$ with $Q$.

The \emph{swap sensitivity} of a randomized algorithm $A$ on $I$ is 
\[
    \SwapSens(A, I)
    :=
    \max\left\{
        \begin{aligned}
            &\max_{e \in E}\max_{f:\Sigma_U\to \Sigma_V}
            \EMD\qty(A(I), A(I^{e\gets f})), \\
            &\max_{u \in U}\max_{Q:\Sigma_U\to\{0,1\}}
            \EMD\qty(A(I), A(I^{u\gets Q}))
        \end{aligned}
    \right\},
\]
and for a family of instances $\mathcal I$,  $\SwapSens(A,\mathcal I) := \max_{I \in \mathcal I}\SwapSens(A,I)$.
This requires $A$ to be defined on $\mathcal I$ and on all one-swap
neighbors of its members, whether or not those neighbors lie in $\mathcal I$. We say that a family of left-predicate Label Cover instances $\mathcal I$ is \emph{swap-closed} if for any instance
$I \in \mathcal I$, $e \in E$, and $u \in U$,
every projection swap $I^{e\gets f}$ and every predicate swap $I^{u\gets Q}$ also belongs to $\mathcal I$. For a family of instances $\mathcal I$, its \emph{swap-closure}  $\SwapClo(\mathcal I)$ is the family of instances obtained by repeatedly swapping projections and/or left predicates in instances from~$\mathcal I$.

For instances $I$ and $\tilde I$ on the same underlying graph and the same alphabets, the \emph{swap distance} $\SwapDist(I,\tilde I)$ is the minimum number of projection swaps and predicate swaps needed to transform $I$ into $\tilde I$.
For a swap-closed family $\mathcal I$, an algorithm $A$ on $\mathcal I$, and
$I,\tilde I\in\mathcal I$, the triangle inequality gives
$\EMD(A(I),A(\tilde I))\le\SwapSens(A,\mathcal I)\cdot\SwapDist(I,\tilde I)$. The following lemma allows us to restrict attention to neighboring instances when bounding sensitivity.

\begin{lemma}
\label{lem:neighboring-witness}
    Let $\mathcal I$ be a family of instances equipped and $A$ be an algorithm on $\mathcal I$.
    If
    $\EMD(A(I),A(\widetilde I))\ge \lambda \cdot \SwapDist(I,\widetilde I)$
    for some $I,\widetilde I\in \mathcal I$ with $\SwapDist(I,\widetilde I)=k>0$, then every shortest swap path
    \[
        I=I^{(0)},I^{(1)},\ldots,I^{(k)}=\widetilde I,
        \qquad I^{(t)}\in\mathcal I\quad(0\le t\le k),
    \]
    contains a consecutive pair satisfying
    $\EMD\bigl(A(I^{(t)}),A(I^{(t+1)})\bigr)\ge \lambda$.
    Such a path exists, in particular, when $\mathcal I$ is swap-closed.
\end{lemma}
\begin{proof}
    Repeatedly applying the triangle inequality,
    \[
        \EMD(A(I),A(\widetilde I))
        \le
        \sum_{t=0}^{k-1}
        \EMD\bigl(A(I^{(t)}),A(I^{(t+1)})\bigr).
    \]
    Hence, the average of the $k$ terms is at least
    $\EMD(A(I),A(\widetilde I))\ge \lambda k $,
    so one has value at least $\lambda$.
\end{proof}

Finally, for a family of instances $\mathcal I$, let  $\mathcal I_{\mathrm{sat}}\subseteq \mathcal I$ denote the instances that are satisfiable. Let $\SwapSens_s(\mathcal I)$ be the minimum of $\SwapSens(A,\mathcal I)$ over all algorithms $A$ such that for every $I\in \mathcal I_{\mathrm{sat}}$, if $\bm \pi \sim A(I)$ then $\E[\val_I(\bm \pi)] \ge s$.

\paragraph{The Single-edge-edit Model.}

In the single-edge-edit model, an edit inserts or deletes one edge while
leaving the vertex sets, alphabets, and left predicates unchanged.
Let $A$ also be defined on instances obtained by such edits.
If $I^{e\setminus}$ denotes the instance obtained by deleting $e$, then a
projection swap $I\to I^{e\gets f}$ is realized by a path of two edge edits,
$I\to I^{e\setminus}\to I^{e\gets f}$.  The triangle inequality for EMD
therefore gives
\[
\begin{aligned}
    &\max_{e \in E}\max_{f:\Sigma_U\to \Sigma_V}
      \EMD\qty(A(I), A(I^{e\gets f}))
    &\qquad\le
      2\max_{\substack{
        J\in\{I\}\cup\{I^{e\setminus}:e\in E\};\\
        \widetilde J\text{ is obtained from }J\\
        \text{by one edge insertion or deletion}}}
      \EMD(A(J),A(\widetilde J)),
\end{aligned}
\]
where an inserted edge may carry an arbitrary projection.
In particular, every swap-sensitivity lower bound that is proved using only projection swaps implies the same lower bound in the single-edge-edit model up to a factor~$2$.

\subsection{The Base 2-CSP}\label{sec:common-source}

Both our low-soundness constructions start from the bounded-degree 2-CSPs in
the proof of Fleming and Yoshida's Label Cover sensitivity lower bound \cite{fleming2026sensitivity}.
A 2-CSP consists of a multigraph $(X,E)$, an alphabet $\Sigma$, and a
relation $R_e\subseteq\Sigma^2$ for each edge $e$.  We fix an ordering of
the two endpoint incidences of each edge; parallel edges and self-loops
are allowed, and degrees count endpoint incidences.  An assignment
$a:X\to\Sigma$ satisfies $e=(x,y)$ if $(a(x),a(y))\in R_e$.
Its value is the fraction of satisfied edges, with value~$1$ when
$E=\varnothing$. We use the same
swap-sensitivity definitions as above, with a swap replacing one relation
on its fixed pair of endpoint incidences.  The value guarantee is required
only on satisfiable members of the swap-closed family.

The construction in \cite[proof of Theorem~5.2]{fleming2026sensitivity}, yields a family of regular 2-CSPs
 with constant degree and alphabet-size, and a polynomial swap-sensitivity
lower bound for algorithms achieving expected value at least $1-\eta$ on
satisfiable instances, for some constant $\eta\in(0,1)$.
We adapt it by a padding argument to obtain exactly $N$ variables for every
sufficiently large $N$, and making the degree even by duplicating edges.

\begin{lemma}[Fleming--Yoshida sensitivity lower bound]\label{lem:fy-source}
There are constants $\eta_{\mathrm{src}},c_{\mathrm{src}}\in(0,1)$, an
even integer $d_0\ge2$, and a constant alphabet $\Sigma_{\mathrm{src}}$
such that, for every sufficiently large $N$, there is a swap-closed family
$\mathcal F_N$ of 2-CSPs on a common $d_0$-regular multigraph with $N$
variables and alphabet $\Sigma_{\mathrm{src}}$ satisfying
\[
 \SwapSens_{1-\eta_{\mathrm{src}}}(\mathcal F_N)
 \ge N^{c_{\mathrm{src}}}.
\]
\end{lemma}
\begin{proof}
Fleming and Yoshida's reductions preserve satisfiability. Within each family
they produce, all instances have the same variable set, edge multiset,
alphabet, and ordering of the endpoints of each edge; only the constraint
relations $R_e$ vary. Their sensitivity argument applies to algorithms
defined on the whole swap closure and required to satisfy the value
guarantee only on its satisfiable members.

The numbers of variables attained by the construction form
an unbounded increasing sequence whose successive ratios are bounded by an
absolute constant $C$. Thus, for every sufficiently large $N$, we can choose
a family with $M$ variables such that $N/C\le M\le N$.
Duplicate every edge to make the degree $d_0$ even; this preserves values
and loses at most a factor of two in the sensitivity lower bound.
Append $N-M$ variables, each carrying $d_0/2$ always-satisfied self-loops.
The graph remains $d_0$-regular, and restriction to the original variables
is $1$-Lipschitz in Hamming distance. The expected fraction of violated
constraints on the original instance is at most $N/M\le C$ times that on
the padded instance. Choosing $\eta_{\mathrm{src}}$ sufficiently small
therefore preserves a lower bound of $\Omega(M^c)=\Omega(N^c)$ for some
constant $c>0$. Taking $0<c_{\mathrm{src}}<c$ proves the lemma for
sufficiently large $N$.
\end{proof}

\subsection{Non-signaling Locality}
\label{subsec:non-signaling-locality}

We compare input graphs on a common vertex set $V$.
For a graph $G=(V,E)$ with vertex and edge inputs, let $N_t^G(S)$ denote
the radius-$t$ input neighborhood of $S\subseteq V$, retaining vertex names
and all input information initially available at vertices in the neighborhood.
For example, in the directed spanner setting, this information includes each
vertex's incident arcs and their endpoints.
Let $\{A(G)\}_G$ be a family indexed by the input graphs under consideration,
where $A(G)$ is a probability distribution over vectors $(y_v)_{v\in V}$
of vertex outputs. For $S\subseteq V$, write $A(G)|_S$ for the joint marginal
distribution of $(y_v)_{v\in S}$.
We call $\{A(G)\}_G$ a \emph{$t$-local non-signaling algorithm} if, for every
pair of inputs $G,\widetilde G$ as above and every $S\subseteq V$,
\[
    N_t^G(S)=N_t^{\widetilde G}(S)
    \quad\Longrightarrow\quad
    A(G)|_S=A(\widetilde G)|_S.
\]
In particular, the output distributions of every $t$-round randomized LOCAL
algorithm form a $t$-local non-signaling family.
The model also contains quantum-LOCAL and bounded-dependence algorithms
at the same locality~\cite{gavoille2009can,akbari2025online,coiteux2024no}.

We use the following sensitivity-to-locality consequence of this definition~\cite{fleming2026sensitivity}.
Suppose two graphs $G,\widetilde G$ on a common vertex set $V$, each of
maximum degree at most $\Delta\ge1$, differ by one input edit supported on a
constant-size set $X\subseteq V$ (the endpoints in the case of an edge edit).
Let $W\subseteq V$ be the union of the vertex sets of
$N_{t+1}^G(X)$ and $N_{t+1}^{\widetilde G}(X)$. The choice of radius $t+1$
in the definition of $W$ ensures that the radius-$t$ input neighborhoods of
$V\setminus W$ coincide in the two graphs, even when vertices at the boundaries
of these neighborhoods initially know all their incident edges.
For a $t$-local non-signaling algorithm, the joint output marginals on
$V\setminus W$ are identical, so they admit a coupling that agrees
outside $W$.
Consequently, the earth mover's distance between the two output distributions, under Hamming distance, is at most
\[
    |W| \le O((\Delta+1)^{t+1}).
\]
Thus a sensitivity lower bound $S>0$ implies a locality lower bound
\begin{equation}\label{eq:sensitivity-locality}
    t\ge c\,\frac{\log S}{\log(\Delta+1)}-C
\end{equation}
for absolute constants $c,C>0$.

\subsection{Local Circuits and Paths of Predicate Swaps}

Two of the key transformations,  alphabet reduction and iterated composition, that underlie our PCPs require bounds on
the circuits that test whether neighboring right labels are the projections of a left label satisfying its predicate. In this section we describe the encoding of these circuits, and the measure of complexity that we will use, which we call their local description complexity.

We encode alphabet symbols by fixed binary strings and fix a binary circuit
encoding that explicitly lists all inputs and outputs, gates, and
connections and includes an end marker. Circuit complexity is the bit length
of this encoding. Padding appends ignored zero bits after the end marker to
reach a prescribed description length, leaving the computed function and
number of inputs unchanged.

\begin{definition}[Local description complexity]
\label{def:local-description-instance}
For a left-predicate Label Cover instance $I$ and $u\in U$, put
$k_u=\deg_I(u)$ and order the incident edges at $u$ as
$e_{u,1},\ldots,e_{u,k_u}$, with right endpoints $v_{u,j}\in V$.
Let $\mathcal N_I(u)\subseteq\Sigma_V^{k_u}$ be the set of projection
tuples of left labels satisfying $P_u$:
\[
 \mathcal N_I(u):=
 \{(f_{e_{u,1}}(a),\ldots,f_{e_{u,k_u}}(a)):
                       a\in\Sigma_U,\ P_u(a)=1\}.
\]
The instance has \emph{local description complexity at most $S$} if each set
$\mathcal N_I(u)$ is decided by a circuit
$C_u:\Sigma_V^{k_u}\to\{0,1\}$ of complexity at most $S$.
The circuits are chosen vertexwise: when a predicate or projection at
$u$ changes, the chosen circuit at every other left vertex is retained.
\end{definition}

We also need the local description complexity bound to hold at intermediate
instances when left predicates are changed one at a time.
A \emph{path of predicate swaps} is a sequence of instances with fixed graph,
alphabets, and projections in which each step replaces one left predicate.
If two such instances have local description complexity at most $S$,
replacing their differing predicates one at a time gives a path with the
same bound at every instance: at each left vertex $u$, use the circuit from
the endpoint instance whose predicate is used at $u$.
 
\part{Sensitivity Transfer and Reduction Tools}

In this section we develop the tools which we will need in order to bound the changes in sensitivity through our PCP reductions. We then develop tools for reducing the alphabet and degree of Label Cover instances.

\section{Sensitivity Transfer for Output Distributions}\label{sec:degree-reduction-region}

We formalize the mechanism for transferring sensitivity lower bounds through reductions that we outlined in \Cref{sec:intro-technical}. A reduction $\mathcal T$ maps instances of a source problem (such as our base 2-CSP or Label Cover) to a target problem, typically with better parameters. Let $\mathcal X$ and $\mathcal Y$ be instances of the source and target problems respectively, then $\mathcal T:\mathcal X\to\mathcal Y$. Each problem-instance $I \in \mathcal X$ has a domain and a range $X_I$, and denote by $Y_I$ the range of $\mathcal T(I)$. We require that neighboring instances $I$ have the same range so that we can compare them.

In order to control the change in sensitivity through $\mathcal T$, we equip each instance $I \in \mathcal X$ with a randomized \emph{recovery map} $\mathcal R_I$. Since we would like to bound the sensitivity of randomized algorithms solving our instances, our recovery algorithm will map a distribution  on $Y_I$ (thought of as the output of some randomized algorithm on $Y$) to a distribution on $X_I$.

\begin{definition}
\label{def:law-sensitivity-pullback-data}
With the setup above, we call
$(C_{\mathcal T},C_{\mathcal R},D_{\mathcal R})$
\emph{sensitivity-transfer parameters} if the following hold.
\begin{enumerate}
  \item For all neighboring instances $I,\widetilde I\in\mathcal X$, the target
  instances $\mathcal T(I), \mathcal T(\tilde I)$ have a common range and
  \[
    \SwapDist\bigl(\mathcal T(I),
                   \mathcal T(\widetilde I)\bigr)
    \le C_{\mathcal T}.
  \]
  \item For every $I$, the recovery map $\mathcal R_I$ is
  $C_{\mathcal R}$-Lipschitz with respect to EMD: for all
  $\mu,\nu\in\mathfrak D(Y_I)$,
  \[
    \EMD\bigl(\mathcal R_I(\mu),
              \mathcal R_I(\nu)\bigr)
    \le C_{\mathcal R} \cdot \EMD(\mu,\nu).
  \]
  We will refer to this throughout as the \emph{Lipschitz bound}.
  \item For all neighboring instances $I,\widetilde I$ and every distribution $\mu$ on
  their common target range,
  \[
    \EMD\bigl(\mathcal R_I(\mu),
              \mathcal R_{\widetilde I}(\mu)\bigr)
    \le D_{\mathcal R}.
  \]
\end{enumerate}
The second condition changes the distribution supplied to recovery while holding the
source instance fixed; the third changes the source instance while holding
that distribution fixed.
\end{definition}

\begin{theorem}\label{lem:law-sensitivity-pullback}
Let $\mathcal X$ be a swap-closed family of source instances and let $\mathcal T$ and $\mathcal R_I$ be a reduction and recovery maps with sensitivity-transfer parameters
$(C_{\mathcal T},C_{\mathcal R},D_{\mathcal R})$.  Let $\mathcal Y'$ be a
swap-closed family of target instances containing $\mathcal T(\mathcal X)$,
and let $A$ be an algorithm on $\mathcal Y'$.  Then $\mathcal R_I$ transforms the algorithm $A$ to an algorithm $B$ for the source instance, whose output distribution is given by
\[
  \mathsf{Law}(B(I))
  :=\mathcal R_I
    \bigl(\mathsf{Law}(A(\mathcal T(I)))\bigr).
\]
Then
\[
  \SwapSens(B,\mathcal X)
  \le C_{\mathcal T}C_{\mathcal R}
       \SwapSens(A,\mathcal Y')+D_{\mathcal R}.
\]
Consequently, if $C_{\mathcal T}C_{\mathcal R}>0$ and
$\SwapSens(B,\mathcal X)\ge L$, then
\[
  \SwapSens(A,\mathcal Y')
  \ge \frac{L-D_{\mathcal R}}
           {C_{\mathcal T}C_{\mathcal R}}.
\]
\end{theorem}

\begin{proof}
For $I\in\mathcal X$ and any one-swap neighbor $\widetilde I$, swap-closure
gives $\widetilde I\in\mathcal X$.  A shortest swap path between the target
instances lies in $\mathcal Y'$ and gives
\[
  \EMD\bigl(A(\mathcal T(I)),
            A(\mathcal T(\widetilde I))\bigr)
  \le C_{\mathcal T} \cdot \SwapSens(A,\mathcal Y').
\]
Apply condition~2 of \Cref{def:law-sensitivity-pullback-data}
to the two target output distributions, keeping $I$ fixed, and then
change the recovery from $I$ to $\widetilde I$ while holding the
second distribution fixed.  The triangle inequality proves the first bound;
rearranging proves the second.
\end{proof}

A randomized recovery procedure acting on one target output is called
a \emph{pointwise recovery}. We use $R_I$ both for this procedure and
for its associated map $R_I$ from elements in $Y_I$ to distributions on $X_I$.
For a distribution $\mu$ on $Y_I$, write 
\[
  (R_I)_\#\mu
  :=\sum_{y\in Y_I}\mu(y)R_I(y)
\]
for the distribution induced by $R_I$, also called the pushforward
of $\mu$ through $R_I$.
Equivalently, first draw a target output from $\mu$ and then run
$R_I$ on that output with fresh internal randomness.
When the input is a random target output $\bm Y$, we write $R_I(\bm Y)$
for the resulting random source output. Thus,
\[
 \mathsf{Law}(R_I(\bm Y))=(R_I)_\#\mathsf{Law}(\bm Y).
\]
We reserve $\mathcal R_I$ for
maps from target output distributions to source output distributions.
A pointwise recovery induces the distribution map
$\mathcal R_I(\mu)=(R_I)_\#\mu$.

\begin{lemma}[Lifting a pointwise recovery to distributions]
\label{lem:recovery-lifting}
A pointwise recovery map $R_I$ that satisfies
\begin{align*}
 \EMD(R_I(y),R_I(\widetilde y))
 &\le C_{\mathcal R}d_{\mathrm H}(y,\widetilde y),\\
 \EMD(R_I(y),R_{\widetilde I}(y))
 &\le D_{\mathcal R}
\end{align*}
for every $I$ and all $y,\widetilde y$ in the first line, and for all
neighboring instances $I,\widetilde I$ and every $y$ in their common input domain in
the second,
induces a recovery on distributions $\mu\mapsto(R_I)_\#\mu$ with the same two constants.
\end{lemma}
\begin{proof}
Fix input distributions $\mu,\nu$ on the range of $\mathcal T(I)$.
Couple them optimally and, conditional on a coupled input pair,
use the pointwise coupling for the recovery maps.  This proves
\[
 \EMD\bigl((R_I)_\#\mu,(R_I)_\#\nu\bigr)
 \le C_{\mathcal R} \cdot \EMD(\mu,\nu).
\]
Using the same input sample for neighboring instances proves the corresponding
bound on the change in recovery.
\end{proof}

\paragraph{Composing Recovery Maps}

The next lemma bounds the sensitivity-transfer parameters under composition.
Suppose $\Phi\mapsto\mathcal T(\Phi)$ has a recovery $\mathcal R_\Phi$ on distributions with
parameters $(C_{\mathcal T},C_{\mathcal R},D_{\mathcal R})$.
Let $\mathcal J$ be a swap-closed
family containing all $\mathcal T(\Phi)$, with one common assignment
domain. Let $\mathcal T_{\mathrm{step}}$ be a further transformation
defined on every instance $I\in\mathcal J$. We apply $\mathcal T$ first
and $\mathcal T_{\mathrm{step}}$ second, obtaining
$\mathcal T_{\mathrm{step}}(\mathcal T(\Phi))$.
For each $I\in\mathcal J$, let $R_I$ be a pointwise recovery that takes
an output of $\mathcal T_{\mathrm{step}}(I)$ and produces a random output
for $I$.
Thus recovery proceeds in the reverse order: first apply
$(R_{\mathcal T(\Phi)})_\#$, then $\mathcal R_\Phi$.
We state the lemma for a pointwise recovery in the second stage
because the composition, alphabet-reduction, and degree-reduction
steps used later all have recoveries of this form; the accumulated
recovery $\mathcal R_\Phi$ need not be induced by a pointwise recovery.

Suppose $\mathcal T_{\mathrm{step}}$ and $R_I$ have respective bounds
$t,r,d$ for instance changes, input perturbations, and changes in the
recovery rule, as follows. For all neighboring
instances $I,\widetilde I\in\mathcal J$, the instances $\mathcal T_{\mathrm{step}}(I)$ and
$\mathcal T_{\mathrm{step}}(\widetilde I)$ have a common range,
\[
 \SwapDist(\mathcal T_{\mathrm{step}}(I),\mathcal T_{\mathrm{step}}(\widetilde I))\le t,
 \qquad
 \EMD(R_I(y),R_{\widetilde I}(y))
 \le d,
\]
where the second inequality holds for every $y$ in that domain, and for every
$I$ and all $y,\widetilde y$ in the range of $\mathcal T_{\mathrm{step}}(I)$,
\[
 \EMD(R_I(y),R_I(\widetilde y))
 \le r\,d_{\mathrm H}(y,\widetilde y).
\]
For a distribution $\mu$ on assignments to
$\mathcal T_{\mathrm{step}}(\mathcal T(\Phi))$, define the recovered
distribution on assignments to $\Phi$ by
\[
 \mathcal R^+_\Phi(\mu)
 :=\mathcal R_\Phi
       \bigl((R_{\mathcal T(\Phi)})_\#\mu\bigr).
\]
\begin{lemma}[Composing recoveries]
\label{lem:law-recovery-composition}
Under the hypotheses above, the composed transformation
$\mathcal T_{\mathrm{step}}\circ\mathcal T$ with recovery
$\mathcal R^+_\Phi$ has sensitivity-transfer parameters satisfying
\begin{align*}
 C_{\mathcal T}^{+}
 &\le C_{\mathcal T}t,\\
 C_{\mathcal R}^{+}
 &\le C_{\mathcal R}r,\\
 D_{\mathcal R}^{+}
 &\le D_{\mathcal R}
       +C_{\mathcal T}C_{\mathcal R}d.
\end{align*}
\end{lemma}

\begin{proof}
By \Cref{lem:recovery-lifting}, the $r$-Lipschitz bound for $R_I$ lifts to
distributions. Composing it with the $C_{\mathcal R}$-Lipschitz bound for
$\mathcal R_\Phi$ gives
$C_{\mathcal R}^+\le C_{\mathcal R}r$.
For neighboring instances $\Phi,\widetilde\Phi$, applying
$\mathcal T_{\mathrm{step}}$ to a path of at most
$C_{\mathcal T}$ swaps from $\mathcal T(\Phi)$ to
$\mathcal T(\widetilde\Phi)$ gives a path of at most $C_{\mathcal T}t$
swaps between their images. Thus $C_{\mathcal T}^+\le C_{\mathcal T}t$.

Along this path from $\mathcal T(\Phi)$ to $\mathcal T(\widetilde\Phi)$,
at each step $I,\widetilde I$,
apply the bound $\EMD(R_I(y),R_{\widetilde I}(y))\le d$ and average over
$y\sim\mu$. The triangle inequality gives
\[
 \EMD\bigl((R_{\mathcal T(\Phi)})_\#\mu,
           (R_{\mathcal T(\widetilde\Phi)})_\#\mu\bigr)
 \le C_{\mathcal T}
       d.
\]
Apply the $C_{\mathcal R}$-Lipschitz bound for $\mathcal R_\Phi$ to these two
distributions, then replace $\mathcal R_\Phi$ by
$\mathcal R_{\widetilde\Phi}$ at cost at most $D_{\mathcal R}$.
The triangle inequality gives
$D_{\mathcal R}^{+}\le C_{\mathcal R}C_{\mathcal T}d+D_{\mathcal R}$.
\end{proof}

\paragraph{Stable Selection}

Our recovery procedures select among candidate outputs by taking the first
candidate whose score meets a uniformly random threshold. This rule is
inspired by related procedures from replicable learning~\cite{ImpLPS22}.
We first bound the probability that perturbing the scores changes the
selected index, then deduce a Hamming-distance bound when the candidates
are assignments.

Let $J\subseteq[0,1]$ be an interval of length $w>0$.
For ordered scores $q_1,\ldots,q_R\in[0,1]$ and a threshold $\theta\in J$,
select the smallest index $i$ with $q_i\ge\theta$, or index $0$ if no
score meets the threshold. Index $0$ represents a default output.

\begin{lemma}[Random-threshold selection]
\label{lem:random-threshold-selection}
For two score lists $(q_i)_{i=1}^R$ and $(\widetilde q_i)_{i=1}^R$
in the setting above, let $K$ and $\widetilde K$ be the selected indices
when both lists use the same uniform threshold $\boldsymbol\theta\in J$.
Let $A$ be the set of indices $i$ for which the interval between $q_i$
and $\widetilde q_i$ intersects $J$:
\[
 A:=\{i\in[R]:[\min\{q_i,\widetilde q_i\},
             \max\{q_i,\widetilde q_i\}]\cap J\ne\varnothing\}.
\]
Then
\[
 \Pr[K\ne\widetilde K]
 \le \frac1w\sum_{i\in A}|q_i-\widetilde q_i|.
\]
Moreover, for a finite alphabet $\Sigma$ and a positive integer $N$,
let $z_0,z_1,\ldots,z_R,\widetilde z_1,\ldots,\widetilde z_R\in\Sigma^N$
be candidate assignments with a common default $\widetilde z_0:=z_0$.
Then
\[
 \E d_{\mathrm H}(z_K,\widetilde z_{\widetilde K})
 \le \sum_{i=1}^R d_{\mathrm H}(z_i,\widetilde z_i)
    +\frac Nw\sum_{i\in A}|q_i-\widetilde q_i|.
\]
\end{lemma}

\begin{proof}
For each $i$, the tests $q_i\ge\boldsymbol\theta$ and
$\widetilde q_i\ge\boldsymbol\theta$ can have different outcomes only when
$\boldsymbol\theta$ lies between $q_i$ and $\widetilde q_i$, an event of
probability at most $|q_i-\widetilde q_i|/w$ and probability zero for
$i\notin A$. If all tests have the same outcomes, then $K=\widetilde K$.
A union bound proves the probability estimate.
When $K=\widetilde K$, the Hamming distance between the selected
assignments is at most $\sum_{i=1}^R d_{\mathrm H}(z_i,\widetilde z_i)$;
when $K\ne\widetilde K$, it is at most $N$. Taking expectations and
using the probability estimate proves the second inequality.
\end{proof}

\section{Right-Degree Reduction}\label{subsection:degree-reduction-region}

We use the standard expander-based right-degree reduction described by
Dinur and Harsha~\cite[Theorem~5.4]{dinur2009composition}, which replaces
each right vertex by a cloud of copies of the prescribed degree.  Their
proof also gives the randomized recovery used here: retain the left
labels and recover each original right label by uniformly sampling a
copy in its cloud.  We formulate this reduction for left-predicate
Label Cover.

Our contribution is to bound
the number of target constraints affected by a source edit and the
expected Hamming distances between recovered labelings, separately on
the left and right sides.  These bounds allow us to transfer sensitivity
lower bounds through degree reduction (\Cref{lem:degree-reduction-region}).

\begin{definition}[Expanders]\label{def:expander}
An undirected $d$-regular multigraph $G=(V,E)$ on $n$ vertices is an
$[n,d,\lambda]$-expander if the normalized adjacency operator has spectral norm at
most $\lambda$ on $\bm 1^\perp$.  Parallel edges are counted with
multiplicity, and a loop contributes two to both its diagonal adjacency entry
and the degree of its endpoint.
\end{definition}
\begin{proposition}[Explicit expanders of every size]\label{pro:ramanujan}
For every $n\ge1$ and every even $d\ge6$, there is a $\poly(n,d)$-time deterministic algorithm that constructs an $[n,d,\lambda_d]$-expander, where
\[
  \lambda_d
  :=\frac{2\sqrt{d/2-1}}{d/2}
  <1
  \qquad\text{and}\qquad
  \lambda_d\le \frac{2\sqrt2}{\sqrt d}.
\]
\end{proposition}
\begin{proof}
Put $q:=d/2$.  Cohen~\cite{cohen2016ramanujan} gives a deterministic
polynomial-time construction of a $q$-regular bipartite Ramanujan multigraph
on $2n$ vertices, with classes $[n]_L$ and $[n]_R$.  Let $B$ be its
biadjacency matrix.  The Ramanujan bound gives
$\lVert Bx\rVert_2,\lVert B^{\mathsf T}x\rVert_2
 \le2\sqrt{q-1}\lVert x\rVert_2$ for every $x\perp\bm 1$.

Identify $i_L$ and $i_R$ for every $i\in[n]$.  The resulting undirected
multigraph has adjacency matrix $A:=B+B^{\mathsf T}$ and degree $2q=d$;
edges between $i_L$ and $i_R$ become loops under the convention in
\Cref{def:expander}.  For every $x\perp\bm 1$,
\[
  \lVert Ax\rVert_2
  \le \lVert Bx\rVert_2+\lVert B^{\mathsf T}x\rVert_2
  \le 4\sqrt{q-1}\,\lVert x\rVert_2.
\]
Dividing by $d=2q$ gives the claimed value of $\lambda_d$.
\end{proof}

Let $d \geq 6$ be the desired right degree.  Degree reduction replaces each right vertex $v$ of degree $D_v$ by a cloud of $D_v$ new vertices, each of degree $d$.  A $d$-regular expander on the cloud determines which copies of the original edge incidences are attached to each new vertex.  The recovery map samples one cloud vertex to recover the original label at $v$.
Expansion limits the advantage of assigning different labels to different
cloud vertices: the value of any target labeling is at most the expected
value of its recovered labeling plus the expander parameter $\lambda$.
\Cref{fig:DR} illustrates the construction, which we now define formally.

Fix an even integer $d\ge 6$ and a parameter $\lambda \in [0,1)$.
Let $I = (U,V,E,\Sigma_U,\Sigma_V,\mathcal P,F)$ be a left-predicate Label Cover instance with no isolated right vertices, and write $G=(U,V,E)$ for its bipartite multigraph.
For each right vertex $v \in V$, let $D_v$ denote its degree and set $\delta_V := \min_{v\in V} D_v$.
A \emph{family of $(d,\lambda)$-expanders} for $I$ is a choice, for each
$v\in V$, of a $[D_v,d,\lambda]$-expander
$H_v=([D_v],E'')$.
Once such a family is fixed, we define the transformation
$\mathcal T_{\mathrm{DR}}(I; d)$ as follows.
We fix an arbitrary ordering of the edge incidences at each right vertex $v\in V$.  The maps $e_G(v,\cdot):[D_v]\to E$ and $u_G(v,\cdot):[D_v]\to U$
return, on input $i\in[D_v]$, the $i$-th incident edge and its left
endpoint, respectively.
Let $\Gamma_{H_v}:[D_v]\times[d]\to[D_v]$ map $(i,k)$ to the $k$-th neighbor of $i$ in $H_v$, where
parallel edges and the two incidences of each loop appear with multiplicity
in the neighbor list.

\begin{algorithm}[H]
\caption{Right-Degree Reduction ${\cal T}_{\mathrm{DR}}$} \label{alg:degree_reduction}
    \KwIn{A left-predicate Label Cover instance $I=(U,V,E,\Sigma_U,\Sigma_V,\mathcal P,F)$, an integer $d \geq 6$, a family of $(d, \lambda)$-expanders, and an ordering of the edge incidences at each right vertex}
    \KwOut{An instance $I'=(U,V',E',\Sigma_U,\Sigma_V,\mathcal P,F')$ with right degree $d$.}
    \For{$v \in V$}{
        Let $D_v$ be the degree of $v$\;
        Let $H_v$ be the chosen $[D_v,d, \lambda]$-expander for $v$\;
        Add to $V'$ a vertex $(v,i)$ for each $i \in [D_v]$\;
        \For{$i \in [D_v]$}{
            \For{$k \in [d]$}{
                Let $j:=\Gamma_{H_v}(i,k)$\;
                Add a new edge incidence $e'$ from $u_G(v,j)$ to $(v,i)$\;
                Set $f_{e'}:=f_{e_G(v,j)}$\;
            }
        }
    }
    \Return $I'$
\end{algorithm}

By construction, $|V'|=|E|$.  Since each $H_v$ is $d$-regular, each
original edge is copied $d$ times, always with the same left endpoint.
Thus the degree of each $u\in U$ in the output graph is
$d\cdot\deg_G(u)$, counting parallel edges with multiplicity.

For an input Label Cover instance $I$, the recovery map $R_{\mathrm{DR},I}$ converts an assignment $\pi'=(\pi'_U,\pi'_V)$ for $I'=\mathcal T_{\mathrm{DR}}(I; d)$ into a random assignment $\bm \pi$ for $I=(U,V,E,\Sigma_U,\Sigma_V,\mathcal P,F)$.
It copies the left labels, i.e., sets $\bm\pi_U(u):=\pi'_U(u)$ for each $u\in U$.
For each right vertex $v\in V$, it samples an index $i_v\in [D_v]$ uniformly and independently, and sets $\bm\pi_V(v):=\pi'_V(v,i_v)$.
When $I$ is clear, we abbreviate this map by $R_{\mathrm{DR}}$.

\begin{figure}
    \centering
        \begin{subfigure}[b]{0.4\textwidth}
    \centering
        \begin{tikzpicture}[scale=1.1]
            \tikzset{inner sep=0,outer sep=3}

            \tikzstyle{a}=[inner sep=4pt, inner ysep=4pt,outer sep=0.5pt,
            draw=black!40!white, fill=Cerulean!10!white, very thick, rounded corners=6pt, align=center]

            \tikzstyle{g}=[inner sep=4pt, inner ysep=4pt,outer sep=0.5pt,
            draw=black!40!white, fill=green!10!white, very thick, rounded corners=6pt, align=center]

            \tikzstyle{b}=[inner sep=4pt, inner ysep=4pt,outer sep=0.5pt,
            draw=black!20!white, fill=Cerulean!10!white, thick, align=center]
            \tikzstyle{e}=[inner sep=1pt, inner ysep=1pt,outer sep=0.5pt,
            draw=Orange!70!white, fill=yellow!10!white, dashed, thick, rounded corners=6pt, align=center]

             \tikzstyle{f}=[inner sep=1pt, inner ysep=1pt,outer sep=0.5pt,
            draw=Orange!70!white, fill=red!4!white, dashed, thick, rounded corners=6pt, align=center]
                \draw[e] (-4,-0.5) -- (-4,0.5) -- (-3,0.5) -- (-3,-0.5) -- cycle;

                \draw[e] (-4,0.7) -- (-4,1.7) -- (-3,1.7) -- (-3,0.7) -- cycle;

                \draw[e] (-4,-1.7) -- (-4,-0.7) -- (-3,-0.7) -- (-3,-1.7) -- cycle;

                \draw[e] (-4,-2.9) -- (-4,-1.9) -- (-3,-1.9) -- (-3,-2.9) -- cycle;

                \node[a, circle] (L1) at (-2.5,1.2) {$u_1$};

                \tikzstyle{c}=[inner sep=3pt, inner ysep=2pt,outer sep=0.5pt,
            draw=black!40!white, fill=Cerulean!10!white, very thick, rounded corners=6pt, align=center]

                \node[a, circle] (L2) at (-2.5,0) {$u_2$};

                \node[a, circle] (L3) at (-2.5,-1.2) {$u_3$};

                \node[a, circle] (L4) at (-2.5,-2.4) {$u_4$};

                \node[a, circle] (R1) at (0,0) {$v$};

                \draw[a] (L1) -- (R1);
                \draw[a] (L2) -- (R1);
                \draw[a] (L3) -- (R1);
                \draw[a] (L4) -- (R1);

                \draw[f] (0,-3) -- (-2,-3) -- (-2,-5) -- (0,-5) -- cycle;

                \node at (-1.5,0.2) {\tiny $f_{(u_2,v)}$};

                \node[rotate=29] at (-1.5,-0.55) {\tiny $f_{(u_3,v)}$};

                \node[rotate=45] at (-1.5,-1.25) {\tiny $f_{(u_4,v)}$};

                \node[rotate=-25] at (-1.5,0.9) {\tiny $f_{(u_1,v)}$};

                \node at (-3.5,1.2) {$P_{u_1}$};

                \node at (-3.5,0) {$P_{u_2}$};

                \node at (-3.5,-1.2) {$P_{u_3}$};

                \node at (-3.5,-2.4) {$P_{u_4}$};

                \node[c, circle] (E2) at (-0.5,-3.5) {\tiny $2$};
                \node[c, circle] (E1) at (-1.5,-3.5) {\tiny $1$};

                \node[c, circle] (E4) at (-0.5,-4.5) {\tiny $4$};

                \node[c, circle] (E3) at (-1.5,-4.5) {\tiny $3$};

                \draw[a] (E1) --(E2);
                \draw[a] (E2) --(E3);
                \draw[a] (E3) --(E4);
                \draw[a] (E4) --(E1);

                \node at (-2.4,-4) {$H_v$};

                \draw [draw=black!80!white,  thick] plot [smooth, tension=1.4] coordinates {(-0.3, 0.3) (-0.5, -0.1) (-0.2, -0.5)};

                \node at (-0.3,0.6) {$D_v$};
            \end{tikzpicture}
            \caption{The neighborhood of $v$ before degree reduction and auxiliary graph $H_v$.}
\end{subfigure}
    \begin{subfigure}[b]{0.4\textwidth}
    \centering
        \begin{tikzpicture}[scale=1.1]
            \tikzset{inner sep=0,outer sep=3}

            \tikzstyle{a}=[inner sep=4pt, inner ysep=4pt,outer sep=0.5pt,
            draw=black!40!white, fill=Cerulean!10!white, very thick, rounded corners=6pt, align=center]

            \tikzstyle{c}=[inner sep=1pt, inner ysep=2pt,outer sep=0.5pt,
            draw=black!40!white, fill=Cerulean!10!white, very thick, rounded corners=6pt, align=center]

            \tikzstyle{g}=[inner sep=4pt, inner ysep=4pt,outer sep=0.5pt,
            draw=black!40!white, fill=green!10!white, very thick, rounded corners=6pt, align=center]

            \tikzstyle{b}=[inner sep=4pt, inner ysep=4pt,outer sep=0.5pt,
            draw=black!20!white, fill=Cerulean!10!white, thick, align=center]
            \tikzstyle{e}=[inner sep=1pt, inner ysep=1pt,outer sep=0.5pt,
            draw=Orange!70!white, fill=yellow!10!white, dashed, thick, rounded corners=6pt, align=center]

             \tikzstyle{f}=[inner sep=1pt, inner ysep=1pt,outer sep=0.5pt,
            draw=Orange!70!white, fill=red!4!white, dashed, thick, rounded corners=6pt, align=center]
                 \draw[e] (-4,-0.5) -- (-4,0.5) -- (-3,0.5) -- (-3,-0.5) -- cycle;

                \draw[e] (-4,0.7) -- (-4,1.7) -- (-3,1.7) -- (-3,0.7) -- cycle;

                \draw[e] (-4,-1.7) -- (-4,-0.7) -- (-3,-0.7) -- (-3,-1.7) -- cycle;

                \draw[e] (-4,-2.9) -- (-4,-1.9) -- (-3,-1.9) -- (-3,-2.9) -- cycle;

                \node[a, circle] (L1) at (-2.5,1.2) {$u_1$};

                \node[a, circle] (L2) at (-2.5,0) {$u_2$};

                \node[a, circle] (L3) at (-2.5,-1.2) {$u_3$};

                \node[a, circle] (L4) at (-2.5,-2.4) {$u_4$};

                \draw[f] (1.1,1.8) -- (-0.1,1.8) -- (-0.1,-3) -- (1.1,-3) -- cycle;

                \node at (-3.5,1.2) {$P_{u_1}$};

                \node at (-3.5,0) {$P_{u_2}$};

                \node at (-3.5,-1.2) {$P_{u_3}$};

                \node at (-3.5,-2.4) {$P_{u_4}$};

                \node[c, circle] (E2) at (0.5,0) {\tiny $(v,2)$};
                \node[c, circle] (E1) at (0.5,1.2) {\tiny $(v,1)$};

                \node[c, circle] (E3) at (0.5,-1.2) {\tiny $(v,3)$};

                \node[c, circle] (E4) at (0.5,-2.4) {\tiny $(v,4)$};

                \draw[a] (E1) -- (L2);
                \draw[a] (E1) -- (L4);

                \draw[a] (E2) -- (L1);
                \draw[a] (E2) -- (L3);

                \draw[a] (E3) -- (L2);
                \draw[a] (E3) -- (L4);

                \draw[a] (E4) -- (L1);
                \draw[a] (E4) -- (L3);

                 \node[rotate=-21] at (-1.6,1) {\tiny $f_{(u_1,v)}$};

                \node[rotate=-48] at (-1.8,0.6) {\tiny $f_{(u_1,v)}$};

                \node[rotate=21] at (-.3,1.05) {\tiny $f_{(u_2,v)}$};

                \node[rotate=21] at (-.3,-0.15) {\tiny $f_{(u_3,v)}$};

                \node[rotate=21] at (-.3,-1.35) {\tiny $f_{(u_4,v)}$};

                \node[rotate=50] at (-.1,0.7) {\tiny $f_{(u_4,v)}$};

                \node[rotate=-21] at (-1.6,-0.2) {\tiny $f_{(u_2,v)}$};

                \node[rotate=-21] at (-1.6,-1.4) {\tiny $f_{(u_3,v)}$};
            \end{tikzpicture}
            \vspace{3em}
            \caption{The portion of the Label Cover instance corresponding to $v$ after degree reduction.}
\end{subfigure}
    \caption{Schematic of degree reduction at a right vertex $v$, using a 2-regular auxiliary graph $H_v$ to illustrate the edge routing. The reduction uses an expander of the prescribed degree.}
    \label{fig:DR}
\end{figure}

\begin{lemma}[Sensitivity transfer through right-degree reduction]\label{lem:degree-reduction-region}
Let $d\ge6$ be even and let $\lambda\in[0,1)$.  Let $\mathcal I$ be a
swap-closed family of left-predicate Label Cover instances on a common
underlying graph and alphabets, with minimum right degree $\delta_V\ge1$.  For
each right vertex $v$, let $D_v$ be its common degree across the family, and
fix a $[D_v,d,\lambda]$-expander and an ordering of the source edge
incidences at $v$. Put
\[
 \mathcal I':=\SwapClo(\mathcal T_{\mathrm{DR}}(\mathcal I;d)).
\]
The family $\mathcal I'$ has right degree $d$, and
$\mathcal T_{\mathrm{DR}}$ maps satisfiable instances to satisfiable
instances.

For every $I\in\mathcal I$ and every labeling $\pi'$ of
$\mathcal T_{\mathrm{DR}}(I;d)$, recovery satisfies
\begin{equation}\label{eq:degree-reduction-pointwise-value}
 \val_{\mathcal T_{\mathrm{DR}}(I;d)}(\pi')
 \le
 \E\!\left[\val_I\!\left(R_{\mathrm{DR},I}(\pi')\right)\right]
 +\lambda.
\end{equation}

Fix $\varepsilon>0$ and let $A'$ be an algorithm on $\mathcal I'$ whose
output has expected value at least $\varepsilon+\lambda$ on every
satisfiable member.  Define
\[
 A(I):=R_{\mathrm{DR},I}
       \bigl(A'(\mathcal T_{\mathrm{DR}}(I;d))\bigr).
\]
Then $A$ has expected value at least $\varepsilon$ on every satisfiable member
of $\mathcal I$.  Moreover, suppose that for all neighboring instances
$J,\widetilde J\in\mathcal I'$, there is one coupling
$(\bm\pi',\widetilde{\bm\pi}')$ of the two output distributions such that
\[
 \E[d_{\mathrm H}(\bm\pi'_U,\widetilde{\bm\pi}'_U)]\le x_U,
 \qquad
 \E[d_{\mathrm H}(\bm\pi'_V,\widetilde{\bm\pi}'_V)]\le x_V.
\]
For all neighboring instances $I,\widetilde I\in\mathcal I$, there is then a coupling
$(\bm\pi,\widetilde{\bm\pi})$ of the output distributions of
$A(I)$ and $A(\widetilde I)$ satisfying
\[
 \E[d_{\mathrm H}(\bm\pi_U,\widetilde{\bm\pi}_U)]\le d x_U,
 \qquad
 \E[d_{\mathrm H}(\bm\pi_V,\widetilde{\bm\pi}_V)]
 \le \frac{d}{\delta_V}x_V.
\]
\end{lemma}
\begin{proof}
If $(\pi_U,\pi_V)$ satisfies $I$, then the assignment
$\pi'_U:=\pi_U$ and $\pi'_V(v,i):=\pi_V(v)$ satisfies
$\mathcal T_{\mathrm{DR}}(I;d)$.

We next prove \eqref{eq:degree-reduction-pointwise-value}.  Fix an assignment $\pi'$ to
$\mathcal T_{\mathrm{DR}}(I;d)$.  For each $v\in V$, write
$e_G(v,j)=(u_{v,j},v)$ for $j\in[D_v]$, and, for $\sigma\in\Sigma_V$, set
\begin{align*}
 X_{v,\sigma}
 &:={}\{j\in[D_v]:
       P_{u_{v,j}}(\pi'_U(u_{v,j}))=1
       \text{ and }
       f_{e_G(v,j)}(\pi'_U(u_{v,j}))=\sigma\},\\
 Y_{v,\sigma}
 &:={}\{i\in[D_v]:\pi'_V(v,i)=\sigma\},\\
 \alpha_{v,\sigma}&:=|X_{v,\sigma}|/D_v,
 &\beta_{v,\sigma}&:=|Y_{v,\sigma}|/D_v.
\end{align*}
Let $M_v$ be the normalized adjacency matrix of $H_v$.  The fraction of
transformed edges in the cloud of $v$ that are satisfied is
\[
 s_v
 :=\frac1{D_v}\sum_{\sigma\in\Sigma_V}
    \mathbf 1_{Y_{v,\sigma}}^{\mathsf T}
    M_v\mathbf 1_{X_{v,\sigma}},
\]
whereas the expected fraction of the original incidences at $v$ satisfied
after recovery is
\[
 r_v:=\sum_{\sigma\in\Sigma_V}
          \alpha_{v,\sigma}\beta_{v,\sigma}.
\]
To compare $s_v$ and $r_v$, center the two indicator vectors at their
means.  Both centered vectors lie in $\bm 1^\perp$, so the expander norm
bound gives, for each $\sigma\in\Sigma_V$,
\[
 \begin{aligned}
 &\frac1{D_v}\mathbf1_{Y_{v,\sigma}}^{\mathsf T}
 M_v\mathbf1_{X_{v,\sigma}}
 -\alpha_{v,\sigma}\beta_{v,\sigma}\\
 &\quad =
 \frac1{D_v}
 (\mathbf1_{Y_{v,\sigma}}-\beta_{v,\sigma}\bm1)^{\mathsf T}
 M_v
 (\mathbf1_{X_{v,\sigma}}-\alpha_{v,\sigma}\bm1)
 \\
 &\quad \le\lambda\sqrt{\alpha_{v,\sigma}\beta_{v,\sigma}}.
 \end{aligned}
\]
Summing over $\sigma$ and applying Cauchy--Schwarz gives
\[
 s_v-r_v
 \le\lambda\sum_\sigma\sqrt{\alpha_{v,\sigma}\beta_{v,\sigma}}
 \le\lambda\sqrt{\sum_\sigma\alpha_{v,\sigma}}
             \sqrt{\sum_\sigma\beta_{v,\sigma}}
 \le\lambda,
\]
since $\sum_\sigma\alpha_{v,\sigma}\le1$ and
$\sum_\sigma\beta_{v,\sigma}=1$.
Averaging with weights $D_v/|E|$ gives
\eqref{eq:degree-reduction-pointwise-value}.
This calculation uses the adjacency matrix with multiplicity, so it also
covers loops and parallel edges under the convention in
\Cref{def:expander}.  Averaging over the output distribution of
$A'(\mathcal T_{\mathrm{DR}}(I;d))$ gives
$\E[\val_I(A(I))]\ge\varepsilon$ for every satisfiable $I\in\mathcal I$.

Fix neighboring instances $I,\widetilde I\in\mathcal I$.  A projection swap changes
exactly the $d$ copied projections, while a left-predicate swap changes only
that predicate.  Hence the transformed instances are joined in $\mathcal I'$
by a swap path
\[
 \mathcal T_{\mathrm{DR}}(I;d)=J_0,J_1,\ldots,
 J_m=\mathcal T_{\mathrm{DR}}(\widetilde I;d)
 \qquad\text{with }m\le d.
\]
For each consecutive pair, choose the coupling in the hypothesis.  Repeatedly
glue these couplings.  The coordinatewise triangle inequalities then give an
endpoint coupling $(\bm\pi',\widetilde{\bm\pi}')$ satisfying
\[
 \E[d_{\mathrm H}(\bm\pi'_U,\widetilde{\bm\pi}'_U)]\le m x_U,
 \qquad
 \E[d_{\mathrm H}(\bm\pi'_V,\widetilde{\bm\pi}'_V)]\le m x_V.
\]

Conditioned on this endpoint pair, couple the recoveries by using the same
uniform index $i_v\in[D_v]$ at every $v$.  The left labels are copied, and
\[
 \Pr[\bm\pi_V(v)\ne\widetilde{\bm\pi}_V(v)
      \mid\bm\pi',\widetilde{\bm\pi}']
 \le \frac{1}{\delta_V}
       \bigl|\{i\in[D_v]:
          \bm\pi'_V(v,i)\ne\widetilde{\bm\pi}'_V(v,i)\}\bigr|.
\]
Summing over $v$ and averaging proves the two bounds in the statement.  The
same coupling also shows that the pointwise recovery is
$1$-Lipschitz under total Hamming distance.  For a fixed target labeling
$\pi'$, the distribution of $R_{\mathrm{DR},I}(\pi')$ is independent of
the projections and predicates of $I$.
\end{proof}
 \section{Alphabet and Degree Reduction}\label{sec:alphabet-reduction}
This section gives a sensitivity-preserving right alphabet reduction
technique for Label Cover.

After recalling the required coding-theoretic facts, we define the
transformation $\mathcal T_{\mathrm{AR}}$ and its recovery maps
$R_{\mathrm{AR},I}$.
We then establish bounds on size, degrees, alphabets, value recovery,
and the sensitivity-transfer parameters
$C_{\mathcal T},C_{\mathcal R},D_{\mathcal R}$
from \Cref{def:law-sensitivity-pullback-data},
for use in \Cref{sec:dh-iteration}.
Alphabet reduction leaves each left predicate $P_u$ unchanged.
We omit $\mathcal P$ from the instance notation.

\subsection{Alphabet Reduction}
A mapping $\mathcal C: \Sigma \to \hat \Sigma^k$ is called a \emph{code} with minimum relative distance $1-\delta$ if for every $a \neq b\in\Sigma$ the strings $\mathcal C(a)$ and $\mathcal C(b)$ differ in at least $(1-\delta)k$ coordinates.

\begin{proposition}[Code with polynomial-size encoding and membership circuits]\label{pro:dinur-remark-5.6}
    For every $0 < \delta < 1$ and alphabet $\Sigma$, there exists a code $\mathcal C: \Sigma \to \hat \Sigma^k$ with relative distance $1-\delta$ where $|\hat \Sigma|= O(\delta^{-2})$ and $k= O(\delta^{-2}\log |\Sigma|)$.
    The code may be chosen explicitly so that encoding, testing whether a word
    is a codeword, and decoding valid codewords have circuits of size
    $\poly(k,\log|\Sigma|,\delta^{-1})$.
\end{proposition}
\begin{proof}
    Remark~5.4 of~\cite{dinur2013composition} obtains the stated alphabet,
    length, and distance by concatenating a rate-optimal outer code with a
    Reed--Solomon inner code.  To obtain the circuit bounds asserted here,
    choose both component codes to be linear and put their generator matrices in
    systematic form, so that each component message appears unchanged in
    designated coordinates of its codeword.  Identify $\Sigma$ with the
    first $|\Sigma|$ strings in the resulting message space.

    To encode, multiply by the component generator matrices.  To test
    whether a word is a codeword, first check the linear equations defining
    the two component codes, given by their parity-check matrices, and then
    check that the systematic message represents an element of $\Sigma$.  For a
    valid codeword, decoding reads the systematic message and inverts this
    identification.  Standard Boolean circuits for finite-field arithmetic,
    matrix--vector multiplication, and comparison have size polynomial in the
    block length, message length, and field description length.  Their
    composition has size $\poly(k,\log|\Sigma|,\delta^{-1})$.
\end{proof}

For each alphabet $\Sigma$ and distance parameter $0<\delta<1$, fix an explicit
code from \Cref{pro:dinur-remark-5.6}, together with its encoding and decoding
circuits.

Suppose $\mathcal C: \Sigma_V \to \hat \Sigma_V^k$ is a code with (relative) distance $1-\eta^3$ for some $0<\eta<1$ satisfying the guarantees of \autoref{pro:dinur-remark-5.6} for $\delta=\eta^3$.
We now describe our alphabet reduction transformation, which we denote by $\mathcal T_{\mathrm{AR}}$.
The transformation maps Label Cover instances $I = (U,V,E,\Sigma_U,\Sigma_V,F)$ to Label Cover instances $I'= (U,V \times [k],E',\Sigma_U,\hat \Sigma_V,F')$.  For every edge incidence $e=(u,v)\in E$ and every $i\in[k]$, the edge set $E'$ contains a distinct incidence $(e,i)$ from $u$ to $(v,i)$, with projection
\[
    f'_{(e,i)}(a):=\mathcal C(f_e(a))_i.
\]
Writing $\mathcal C_i(b):=\mathcal C(b)_i$ for the $i$th coordinate map,
this projection is $\mathcal C_i\circ f_e$.

\begin{figure}[H]
    \centering
        \begin{subfigure}[b]{0.4\textwidth}
    \centering
        \begin{tikzpicture}[scale=1.1]
            \tikzset{inner sep=0,outer sep=3}

            \tikzstyle{a}=[inner sep=4pt, inner ysep=4pt,outer sep=0.5pt,
            draw=black!40!white, fill=Cerulean!10!white, very thick, rounded corners=6pt, align=center]
            \tikzstyle{b}=[inner sep=4pt, inner ysep=4pt,outer sep=0.5pt,
            draw=black!20!white, fill=Cerulean!10!white, thick, align=center]
            \tikzstyle{e}=[inner sep=1pt, inner ysep=1pt,outer sep=0.5pt,
            draw=Orange!70!white, fill=yellow!10!white, dashed, thick, rounded corners=6pt, align=center]
                \draw[e] (-4,-0.5) -- (-4,0.5) -- (-3,0.5) -- (-3,-0.5) -- cycle;
                \node[a, circle] (L1) at (-2.5,0) {$u$};
                \node[a, circle] (R1) at (0,0) {$v$};

                \draw[a] (L1) -- (R1);
                \node at (-1.3,0.2) {\small $f_e$};

                \node at (-3.5,0) {$P_u$};
            \end{tikzpicture} \vspace{4em}
            \caption{Edge incidence $e=(u,v)$ before ${\cal T}_{AR}$}
\end{subfigure}
    \begin{subfigure}[b]{0.4\textwidth}
    \centering
        \begin{tikzpicture}[scale=1.1]
            \tikzset{inner sep=0,outer sep=3}
            \tikzstyle{e}=[inner sep=1pt, inner ysep=1pt,outer sep=0.5pt,
            draw=Orange!70!white, fill=yellow!10!white, dashed, thick, rounded corners=6pt, align=center]
            \tikzstyle{a}=[inner sep=4pt, inner ysep=4pt,outer sep=0.5pt,
            draw=black!40!white, fill=Cerulean!10!white, very thick, rounded corners=6pt, align=center]
            \tikzstyle{b}=[inner sep=1pt, inner ysep=1pt,outer sep=0.5pt,
            draw=black!40!white, fill=Cerulean!10!white, very thick, rounded corners=6pt, align=center]

                 \draw[e] (-4,1) -- (-4,2) -- (-3,2) -- (-3,1) -- cycle;

                \node at (-3.5,1.5) {$P_u$};

                \node[a, circle] (L1) at (-2.5,1.5) {$u$};
                \node[b, circle] (R1) at (0,0) {\tiny $(v,1)$};
                \node[b, circle] (R2) at (0,1) {\tiny $(v,2)$};
                \node at (0,2.1) {\large $\vdots$};
                \node[b, circle] (R3) at (0,3) {\tiny $(v,k)$};
                \draw[a] (L1) -- (R1);
                \draw[a] (L1) -- (R2);
                \draw[a] (L1) -- (R3);

                \node[rotate = 33] at (-1.4,2.5) {\tiny $\mathcal C_k\circ f_e$};
                \node[rotate = -7] at (-1.2,1.5) {\tiny $\mathcal C_2\circ f_e$};
                \node[rotate = -30] at (-1.4,0.5) {\tiny $\mathcal C_1\circ f_e$};
            \end{tikzpicture}
            \caption{The incidences $(e,i)$ after ${\cal T}_{AR}$}
\end{subfigure}
    \caption{Alphabet reduction ${\cal T}_{AR}$ on one edge incidence.}
\end{figure}

For an instance $I=(U,V,E,\Sigma_U,\Sigma_V,F)$ we write
$\bar d_U := \frac{|E|}{|U|}$, $\bar d_V := \frac{|E|}{|V|}$, $\Delta_V := \max_{v\in V} \deg(v)$, $\Delta_U := \max_{u\in U} \deg(u)$, and $\delta_V := \min_{v\in V} \deg(v)$
for the average left and right degrees, the maximum right and left degrees, and the minimum right degree, respectively.

For recovery, assume that the source instance has no isolated right vertices,
so that $\delta_V\ge1$.
The labels on the $k$ reduced copies of a source right vertex need not
form a codeword.  We recover original
right labels from the projections of the left labels; the code distance
will relate agreement with the reduced right labels to the value after
this recovery.  In particular, recovery itself reads only the left labels.
For a source instance $I$, write $R_I$ for $R_{\mathrm{AR},I}$ within
this subsection. Given an assignment $\pi'$ for
$I'=\mathcal T_{\mathrm{AR}}(I)$, the recovery procedure $R_I$ constructs
an assignment $\bm\pi$ for $I$ as follows.
For each $v \in V$ and $b \in \Sigma_V$, define
$p_v(b) := \Pr_{e=(u,v)\text{ incident to } v}[f_e(\pi'(u)) = b]$.
Here the probability is over a uniformly random edge incident to $v$ (counting multiplicity).
We include projections of all left labels in $p_v$, even those failing
their predicates.  This makes recovery independent of the predicates;
the value analysis below counts only projected labels $f_e(\pi'(u))$
with $P_u(\pi'(u))=1$.
We set
\[
    \bm \pi(u) = \pi'(u)
    \qquad\text{for every }u \in U.
\]
For each $v \in V$, we then sample $\bm \pi(v)$ independently from the distribution $p_v$.

The recovery map $R_I$ depends on the source instance $I$ through
the projections $f_e$ used to define $p_v$.
To bound $C_{\mathcal R}$, we vary the input assignment $\pi'$ with $I$
held fixed; to bound $D_{\mathcal R}$, we change $I$ by one swap with
$\pi'$ held fixed.

\begin{lemma}[Sensitivity transfer through alphabet reduction]\label{lem:alphabet-reduction}
    Let $\mathcal I$ be a swap-closed family of Label Cover instances $I=(U,V,E,\Sigma_U,\Sigma_V,F)$ with average left degree $\bar d_U$, average right degree $\bar d_V$, maximum left degree $\Delta_U$, and minimum right degree $\delta_V\ge1$.
    Fix $\eta\in(0,1/4)$ and choose a code
    $\mathcal C:\Sigma_V\to\hat\Sigma_V^k$ from
    \Cref{pro:dinur-remark-5.6} with relative distance $1-\eta^3$.
    Using this code, let $\mathcal I'=\mathcal T_{\mathrm{AR}}(\mathcal I)$.
    Its right alphabet satisfies $|\hat\Sigma_V|=O(\eta^{-6})$, and the
    block length satisfies $k=O(\eta^{-6}\log|\Sigma_V|)$.
    Then the following hold.
    \begin{enumerate}
        \item If $I \in \mathcal I$ is satisfiable, then $\mathcal T_{\mathrm{AR}}(I)$ is satisfiable.
        \item For every $I\in\mathcal I$ and every possibly random labeling
        $\bm\pi'$ of $I'=\mathcal T_{\mathrm{AR}}(I)$,
        \[
            \E[\val_I(R_I(\bm\pi'))]
            \ge\bigl(\E[\val_{I'}(\bm\pi')]\bigr)^2-\eta^3.
        \]
        \item The transformation $\mathcal T_{\mathrm{AR}}$ and the induced
        recovery maps on distributions
        $\mathcal R_I(\mu):=(R_I)_\#\mu$
        have sensitivity-transfer parameters
        \[
            C_{\mathcal T}=k,\qquad
            C_{\mathcal R}=1+\frac{\Delta_U}{\delta_V},\qquad
            D_{\mathcal R}=\frac{1}{\delta_V}.
        \]
    \end{enumerate}
\end{lemma}
\begin{proof}
The alphabet and code-length bounds follow from
\Cref{pro:dinur-remark-5.6} with $\delta=\eta^3$.
A satisfying labeling $\pi$ extends to the reduced instance by
$\pi'(u)=\pi(u)$ and $\pi'(v,i)=\mathcal C(\pi(v))_i$.

For Item~2, fix a labeling $\pi'$ of $I'=\mathcal T_{\mathrm{AR}}(I)$.
For each $v\in V$, define the subdistribution $x_v:\Sigma_V\to[0,1]$ by
\[
 x_v(b):=\Pr_{e=(u,v)\text{ incident to }v}
       [P_u(\pi'(u))=1,\ f_e(\pi'(u))=b].
\]
Thus $x_v(b)\le p_v(b)$ and $\sum_bx_v(b)\le1$.
For a uniform code coordinate $i\in[k]$, put
\[
 Z_v(i):=\sum_bx_v(b)\mathbf1[\pi'(v,i)=\mathcal C(b)_i],
 \qquad \alpha_v:=\E_i Z_v(i),\qquad
 \nu_v:=\sum_bx_v(b)p_v(b).
\]
The reduced value and expected recovered value are the averages of
$\alpha_v$ and $\nu_v$, respectively, with weights $\deg(v)/|E|$.
Distinct codewords agree on at most an $\eta^3$ fraction of coordinates, so
\[
 \alpha_v^2\le\E_i Z_v(i)^2
 \le\sum_bx_v(b)^2+\eta^3\sum_{b\ne b'}x_v(b)x_v(b')
 \le\nu_v+\eta^3.
\]
Averaging over $v$ and the input labeling, Jensen's inequality proves Item~2.

A projection swap changes at most $k$ reduced projections, and a predicate
swap changes one reduced predicate, so $C_{\mathcal T}=k$ is valid.
For $C_{\mathcal R}$ and $D_{\mathcal R}$, implement recovery by sampling a uniform incident edge
at each $v$ and returning its projected left label. Couple two recoveries
using the same sampled incidences. For instances $I,\widetilde I$ with
projections $f_e,\widetilde f_e$ and labelings $\pi',\widetilde\pi'$ of
their reduced instances, this gives
\[
 \EMD(R_I(\pi'),R_{\widetilde I}(\widetilde\pi'))
 \le d_{\mathrm H}(\pi'_U,\widetilde\pi'_U)
 +\sum_{e=(u,v)\in E}
 \frac{\mathbf1[f_e(\pi'_U(u))\ne
                  \widetilde f_e(\widetilde\pi'_U(u))]}{\deg(v)}.
\]
For $I=\widetilde I$, the sum is at most
$(\Delta_U/\delta_V)d_{\mathrm H}(\pi'_U,\widetilde\pi'_U)$.
Together with the first term, this gives the pointwise Lipschitz constant
$1+\Delta_U/\delta_V$.
For a fixed input labeling, one projection swap contributes at most
$1/\delta_V$, whereas a predicate swap leaves the sampling rule unchanged.
By \Cref{lem:recovery-lifting}, these two pointwise bounds give the claimed
$C_{\mathcal R}$ and $D_{\mathcal R}$ for the induced maps $\mathcal R_I$.
\end{proof}

\subsection{Combined Alphabet-and-Degree Reduction}

We now combine alphabet reduction with right-degree reduction.  Every family
to which this combined map is applied has a common graph and alphabets and no
isolated right vertices.  The expanders of every size from
\Cref{pro:ramanujan} allow us to apply degree reduction regardless of the
minimum right degree.
For the recovery thresholds recorded below, we henceforth assume $0<\varepsilon<1/16$, so that $3\sqrt{\varepsilon}$ and $4\sqrt{\varepsilon}$ both lie in $(0,1)$.
Let $I = (U, V, E, \Sigma_U, \Sigma_V, \mathcal P, F )$ be a left-predicate Label Cover instance with $|U| = n$, $|V| = m$, average left degree $\bar d_U$, average right degree $\bar d_V$, maximum right degree $\Delta_V$, maximum left degree $\Delta_U$, and minimum right degree $\delta_V$.
\Cref{tab:reduced-label-cover} tracks how these parameters change when we first apply the alphabet-reduction map induced by a code $\mathcal C : \Sigma_V \to \hat\Sigma_V^k$ of relative distance $1 - \eta^3$, and then run the degree-reduction step based on $[D_v, d, \lambda]$-expanders tailored to the right degree $D_v$ of each vertex $v \in V$.
The last row records required expected values and is read from right to
left: attaining a threshold after a reduction allows recovery to attain
at least the preceding threshold before that reduction.

\begin{table}[t!]
    \centering
    \caption{Parameters after the alphabet-then-degree reduction with $0<\varepsilon<1/16$, $\eta=\varepsilon/4$, $k = O(\varepsilon^{-6}\log |\Sigma_V|)$, and $d = \Theta(\varepsilon^{-1})$ chosen so that $\lambda_d\le\sqrt\varepsilon$. Degrees are averages unless indicated otherwise.}\label{tab:reduced-label-cover}
    \begin{tabular}{|l|c|c|c|}
    \hline Parameter & $I$ & alphabet $\rightarrow \hat\Sigma_V$ & degree $\rightarrow d$ \\
    \hline \# left vertices & $n$ & $n$ & $n$ \\
    \# right vertices & $m$ & $O(\varepsilon^{-6} m \log |\Sigma_V|)$ & $O(\varepsilon^{-6} m \bar d_V \log |\Sigma_V|)$ \\
    Average left degree & $\bar d_U$ & $O(\varepsilon^{-6} \bar d_U \log |\Sigma_V|)$ & $O(\varepsilon^{-7} \bar d_U \log |\Sigma_V|)$ \\
    Average right degree & $\bar d_V$ & $\bar d_V$ & $\Theta(\varepsilon^{-1})$ \\
    Left alphabet & $\Sigma_U$ & $\Sigma_U$ & $\Sigma_U$ \\
    Right alphabet & $\Sigma_V$ & $\hat\Sigma_V$ ($|\hat\Sigma_V| = O(\varepsilon^{-6})$) & $\hat\Sigma_V$ \\
    Recovery threshold & $\varepsilon$ & $3\sqrt{\varepsilon}$ & $4\sqrt{\varepsilon}$ \\
\hline
    \end{tabular}
\end{table}

We choose $\mathcal C$ from \Cref{pro:dinur-remark-5.6} with
$\delta=\eta^3$ and $\eta=\varepsilon/4$.  The resulting right alphabet
$\hat\Sigma_V$ has size $O(\varepsilon^{-6})$, and the block length is
$k=O(\varepsilon^{-6}\log|\Sigma_V|)$. Since
$(3\sqrt\varepsilon)^2-(\varepsilon/4)^3\ge\varepsilon$,
the value-recovery bound in \Cref{lem:alphabet-reduction} gives recovery
threshold $3\sqrt\varepsilon$.

For $0<\varepsilon<1/16$ and an absolute constant $c_{\mathrm{AR}}$ to be
chosen below, set $d_\varepsilon:=2\lfloor c_{\mathrm{AR}}/(2\varepsilon)\rfloor$.
Define $\red_{\varepsilon}(I)$ by first applying alphabet
reduction with $\eta=\varepsilon/4$, transporting the left predicates
unchanged, and then applying right-degree reduction with target degree
$d_\varepsilon$. Write $R_{\varepsilon,I}$ for the recovery
obtained by composing the two recovery maps in the reverse order.

We use the following standing convention whenever the combined reduction is
applied to a family with a common graph and alphabets.  Fix one code for the
common right alphabet.  For each right degree $D$ attained in the
alphabet-reduced graph, use at every vertex of degree $D$ the same expander
$H_D$ of target degree $d_\varepsilon$, returned by
\Cref{pro:ramanujan}.  Also fix the edge-incidence
orderings across the family.  Thus
the degree-reduction recovery depends only on its input assignment and these
fixed data; the combined recovery depends on the source instance only through
the alphabet-reduction recovery.

\begin{lemma}[Alphabet-then-degree reduction]\label{lem:degree-alphabet-reduction}
    There are sufficiently large absolute constants $c_{\mathrm{AR}},C_{\mathrm{red}}>0$ for which the following holds.  Let $0<\varepsilon<1/16$ and $\mathcal I$ be a swap-closed family of left-predicate Label Cover instances $I=(U,V,E,\Sigma_U,\Sigma_V,\mathcal P,F)$ on a common graph and common alphabets, with no isolated right vertices, $|U|=n$, $|V|=m$, average left degree $\bar d_U$, average right degree $\bar d_V$, maximum right degree $\Delta_V$, maximum left degree $\Delta_U$, and minimum right degree $\delta_V$.
    Then every instance $\red_{\varepsilon}(I)=(U',V',E',\Sigma'_U,\Sigma'_V,\mathcal P',F')$ satisfies
    \[
        |U'|=n,
        \qquad
        |V'|=O(\varepsilon^{-6} m\bar d_V \log |\Sigma_V|),
        \qquad
        \bar d_{U'}=O(\bar d_U\varepsilon^{-7} \log |\Sigma_V|),
    \]
    has right degree $d_\varepsilon=\Theta(\varepsilon^{-1})$ and right alphabet $\hat\Sigma_V$ of size $O(\varepsilon^{-6})$.
    The recovery $R_{\varepsilon,I}$ and the transformation
    satisfy the following guarantees.
    \begin{enumerate}
        \item If $I$ is satisfiable, then $\red_{\varepsilon}(I)$ is satisfiable.
        \item If a possibly random assignment $\bm\pi'$ for $\red_{\varepsilon}(I)$ satisfies
        $\E[\val_{\red_{\varepsilon}(I)}(\bm\pi')]
        \ge 4\sqrt{\varepsilon}$, then
        \[
            \E\!\left[\val_I\!\left(
            R_{\varepsilon,I}(\bm\pi')\right)\right]
            \ge\varepsilon.
        \]
        \item For all $I,\tilde I\in\mathcal I$, we have
        \[
            \SwapDist\bigl(\red_{\varepsilon}(I),\red_{\varepsilon}(\tilde I)\bigr)
            \le C_{\mathrm{red}}\varepsilon^{-7}\log |\Sigma_V|\cdot \SwapDist(I,\tilde I).
        \]
        \item For every $I\in\mathcal I$ and every two assignments $\pi',\tilde \pi'$ on $\red_{\varepsilon}(I)$, we have
        \[
            \EMD\bigl(R_{\varepsilon,I}(\pi'),R_{\varepsilon,I}(\tilde \pi')\bigr)
            \le \qty(1+\frac{\Delta_U}{\delta_V}) d_{\mathrm H}(\pi',\tilde \pi').
        \]
        \item If $I,\tilde I\in\mathcal I$ differ by one swap, then for every assignment $\pi'$ on the common domain of $\red_{\varepsilon}(I)$ and $\red_{\varepsilon}(\tilde I)$, we have
        \[
            \EMD\bigl(R_{\varepsilon,I}(\pi'),R_{\varepsilon,\tilde I}(\pi')\bigr)
            \le \frac{1}{\delta_V}.
        \]
    \end{enumerate}
\end{lemma}
\begin{proof}
Set $\eta=\varepsilon/4$ and choose $c_{\mathrm{AR}}$ large enough that
$d=d_\varepsilon=\Theta(\varepsilon^{-1})\ge6$ and
$\lambda_d\le\sqrt\varepsilon$ in \Cref{pro:ramanujan}.
The alphabet-reduction code has length
$k=O(\varepsilon^{-6}\log|\Sigma_V|)$ and alphabet size
$O(\varepsilon^{-6})$. The two reductions give
\[
 |U'|=n,\qquad |V'|=k|E|=km\bar d_V,\qquad
 \bar d_{U'}=dk\bar d_U,
\]
and right degree $d$, proving the structural bounds.
Both reductions preserve satisfiability. Their value-recovery guarantees
in \Cref{lem:degree-reduction-region,lem:alphabet-reduction} give
$4\sqrt\varepsilon\to3\sqrt\varepsilon\to\varepsilon$ in expectation.
A source swap causes at most $kd$ target swaps, proving Item~3.

The degree-reduction recovery copies the left labels, and the
alphabet-reduction recovery reads only these labels. Thus the composite
recovery has the same distribution as applying the alphabet-reduction
sampling rule directly to the final left labels. Applying the
$C_{\mathcal R}$ and $D_{\mathcal R}$ bounds from
\Cref{lem:alphabet-reduction} to point masses proves Items~4 and~5.
\end{proof}

For a family $\mathcal I$, write
\[
  \red_{\varepsilon}(\mathcal I)
  :=\{\red_{\varepsilon}(I):I\in\mathcal I\}
\]
for its pointwise image before swap closure.
Using the same code, expanders, and incidence orderings along a source
swap path gives
\[
  \red_{\varepsilon}(\SwapClo(\mathcal I))
  \subseteq
  \SwapClo(\red_{\varepsilon}(\mathcal I)).
\]
These choices and the common input graph determine one common output
graph and common alphabets.

For later composition rounds in the Dinur-Harsha framework, we also need circuits that decide whether
the right labels around a left vertex $u$ belong to $\mathcal N_I(u)$:
whether they are the projections of one left label satisfying its
predicate, as in \Cref{def:local-description-instance}.
The next lemma bounds the description lengths of these circuits after reduction.
The bound depends nearly linearly on the existing description length $S$,
so repeated reductions do not raise $S$ to a new constant power at each round.

\begin{lemma}[Local circuits and paths of predicate swaps after reduction]
\label{lem:red-local-description}
There are absolute constants $C,c_{\rm red}\ge1$ with the following
property.  Suppose $I$ has no isolated vertices, local description
complexity at most $S$, left degree at most $\Delta_U$, and right
alphabet $\Sigma_V$.  For $0<\varepsilon<1/16$, set
\[
 T:=\Delta_U\bigl(\varepsilon^{-1}\log(2|\Sigma_V|)\bigr)^{c_{\rm red}}.
\]
Then $\red_\varepsilon(I)$ has local description complexity at most
\begin{equation}\label{eq:red-additive-local-complexity}
 C(S+T)\log^{c_{\rm red}}(2+S+T).
\end{equation}
It has no isolated vertices.  If a family has common graph, alphabets,
and edge projections, the reduced family also has common graph, alphabets,
and edge projections.  In particular, one predicate swap changes exactly one
transported predicate and leaves the recovery map unchanged:
\begin{equation}\label{eq:red-predicate-parameters}
 C_{\mathcal T,\rm red}^{\rm pred}=1,
 \qquad D_{\mathcal R,\rm red}^{\rm pred}=0,
 \qquad C_{\mathcal R,\rm red}\le1+\Delta_U/\delta_V.
\end{equation}
The bounds \eqref{eq:red-additive-local-complexity} and
\eqref{eq:red-predicate-parameters} hold along a path of predicate swaps
whenever every instance on the path has local description complexity at
most $S$.
\end{lemma}
\begin{proof}
Let $k=O(\varepsilon^{-6}\log(2|\Sigma_V|))$ be the code length and
$d=\Theta(\varepsilon^{-1})$ the degree parameter.  At one left vertex,
the reduced neighborhood has one group of $dk$ input symbols for each
original incidence.  First check equality among the $d$ routed copies of
each code coordinate.  For the resulting length-$k$ word, use the
circuit for codeword membership and decode its unique original symbol.
Do this for every original incidence, and apply the original circuit
for membership in $\mathcal N_I(u)$ \emph{once} to the decoded tuple,
where $u$ is the left vertex under consideration.

By \Cref{pro:dinur-remark-5.6}, the work per original incidence is
polynomial in $k$, $\log|\Sigma_V|$, and $\varepsilon^{-1}$.
Choose $c_{\rm red}$ sufficiently large that the number of added gates and
the total input and output lengths of the added circuits are at most $T$.
The combined circuit has $O(S+T)$ gates and terminals.  Encoding its wiring
gives a total description length of
$O((S+T)\log^{c_{\rm red}}(2+S+T))$, proving
\eqref{eq:red-additive-local-complexity}.

Each left degree is multiplied by $dk$, and every reduced right vertex
has degree $d$, so no vertex becomes isolated.  Fixing
the code, cloud expanders, and incidence orderings fixes the graph and
projections throughout the family.  The left predicate is transported
unchanged.  The alphabet-reduction recovery samples from the empirical
distributions $p_v$ of projected labels, which do not depend on the predicates; the
right-degree-reduction recovery also does not inspect them.  Thus a
predicate-only swap leaves the composite recovery exactly unchanged.
The Lipschitz bound is the one in
\Cref{lem:degree-alphabet-reduction}.  Each vertexwise circuit construction
uses only the corresponding original circuit, so the same reasoning
applies to every intermediate instance on a path of predicate swaps.
\end{proof}

\part{The IKW Route to Low-Soundness Label Cover}

In this part we prove the sensitivity lower bound for Label Cover in
\Cref{thm:intro-label-cover-sensitivity-IKW}
using the direct-product analysis of Impagliazzo, Kabanets, and
Wigderson~\cite{impagliazzo2009new}, along with new recovery and degree reduction procedures.

\section{The IKW Direct-Product Construction}\label{sec:ikw}

We amplify a constraint system by using assignments to tuples of
$k$ constraints as labels and imposing consistency constraints on overlaps. Increasing
$k$ lets us obtain sensitivity lower bounds for algorithms with smaller
value guarantees. Here $n$ denotes the total number of vertices
in the final Label Cover instance.

The direct-product analysis of Impagliazzo, Kabanets, and
Wigderson~\cite{impagliazzo2009new} shows that sufficiently frequent
agreement between local assignments yields an assignment to the original
variables. We construct a randomized recovery procedure based on their analysis and combine it
with degree reduction in \Cref{sec:ikw-sensitivity-analysis} to prove
the sensitivity lower bound of
\Cref{thm:intro-label-cover-sensitivity-IKW}.

In the theorem below, $N$ counts variables of the original 2-CSP.

\begin{theorem}[Formal version of \Cref{thm:informal-label-cover-ikw}]
\label{thm:intro-label-cover-sensitivity-IKW}
There exist constants $C,a,b,\eta>0$ and $k_0\in\mathbb N$ with the following
property.  Let $k=k(N)$ be any integer-valued function such that
$k(N)\ge k_0$ and $k(N)=o(\log N)$.  For all sufficiently large $N$, there is
a swap-closed family of left-predicate Label Cover instances whose members
have a common number $n$ of vertices satisfying
\[
    N^{ak}\le n\le N^{bk}.
\]
Every algorithm defined on this family that, on each satisfiable member,
outputs a possibly random assignment of expected value at least
$\exp(-C\sqrt{k})$ has sensitivity at least $n^{\eta/k}$.  Moreover, the
instances are biregular, and both alphabet sizes and the maximum left and
right degrees are bounded by $2^{O(k)}$.
\end{theorem}

\subsection{The Tuple Label Cover Construction}

We first give a biregular variant of the direct-product construction of
Impagliazzo, Kabanets, and Wigderson
(IKW)~\cite[Section~6.3]{impagliazzo2009new}.
Biregularity ensures that a uniformly sampled edge has a uniformly
distributed endpoint on each side. We use this uniformity to relate
constraint satisfaction to averages over right vertices and to bound
the effect of label changes on the randomized recovery in
\Cref{sec:ikw-sensitivity-analysis}.
Each left vertex is an ordered tuple of edges of the original
constraint graph. Each incident Label Cover edge specifies
which endpoint coordinates are selected and their order
in the right vertex.

Let $\Phi$ be a 2-CSP on a fixed $d$-regular multigraph $G=([N],E)$ with
alphabet $\Sigma$; self-loops are allowed. Degrees count endpoint
incidences, so $M:=|E|=dN/2$. Fix an ordering of the two endpoint
incidences of each edge $e$. Write $\operatorname{end}_c(e)$ for the
vertex of incidence $(e,c)$, where $c\in\{0,1\}$.
For a self-loop $e=(v,v)$, the two incidences are distinct, although
both have vertex $v$.

For $k\ge4$, put
\[
 k':=\lfloor\sqrt{k}\rfloor,\qquad r:=k-k',\qquad
 p_k:=(k)_{k'}=\frac{k!}{(k-k')!}.
\]
We define the encoding map $\mathcal T_{\mathrm{IKW}}^k$ from 2-CSPs
to left-predicate Label Cover instances as follows. The instance
$\mathcal T_{\mathrm{IKW}}^k(\Phi)$ has vertex sets and alphabets
\[
 U:=E^k,\qquad V:=[N]^{k'},\qquad
 \Sigma_U:=\Sigma^{2k},\qquad \Sigma_V:=\Sigma^{k'}.
\]
For a left vertex $S=(e_1,\ldots,e_k)$, index the coordinates of a label
$y\in\Sigma_U$ by $(h,c)\in[k]\times\{0,1\}$.
The coordinate $y_{(h,c)}$ specifies the symbol assigned to
$\operatorname{end}_c(e_h)$. The predicate $P_S$ requires that
all coordinates corresponding to the same vertex of $G$ receive the same symbol and
that all $k$ constraints of $\Phi$ indexed by $e_1,\ldots,e_k$ are satisfied.
A right label assigns a symbol to each of the $k'$ coordinates of its
vertex tuple.

For every $S\in U$, injection $\iota:[k']\hookrightarrow[k]$, and
$b\in\{0,1\}^{k'}$, add an edge $\omega=(S,\iota,b)$ to
$\mathcal T_{\mathrm{IKW}}^k(\Phi)$, with right endpoint
\[
 A(\omega):=
 \bigl(\operatorname{end}_{b_i}(e_{\iota(i)})\bigr)_{i\in[k']}.
\]
The projection $f_\omega$ selects the corresponding coordinates of the left label:
\[
f_\omega(y)
=
\left(y_{(\iota(i),b_i)}\right)_{i\in[k']}.
\]
A satisfying assignment to $\Phi$ induces a satisfying labeling of
$\mathcal T_{\mathrm{IKW}}^k(\Phi)$ by restricting that assignment to each tuple.

Let $D_U$ and $D_V$ denote the left and right degrees of
$\mathcal T_{\mathrm{IKW}}^k(\Phi)$, respectively, and let $E_k$ be its edge set. We have
\begin{equation}\label{eq:ikw-degrees}
 \begin{gathered}
 |U|=M^k,\qquad |V|=N^{k'},\\
 D_U=p_k2^{k'},\qquad D_V=p_kd^{k'}M^r,\qquad
 |E_k|=D_U|U|=D_V|V|.
 \end{gathered}
\end{equation}
For each right vertex $A$, there are $p_k$ choices of $\iota$,
$d^{k'}$ choices of endpoint incidences at its coordinates,
and $M^r$ choices of edges in the remaining positions.
Thus every right vertex has degree $D_V=p_kd^{k'}M^r$.

We next bound the number of predicate swaps needed to transform
the output instance when one constraint of $\Phi$ changes.
Let $\Phi$ and $\widetilde\Phi$ differ only in the relation on
an edge $e_\star$. Their images have the same graph and projections,
and their left predicates can differ only at tuples containing
$e_\star$. Hence
\begin{equation}\label{eq:ikw-predicate-swaps}
 \SwapDist\bigl(\mathcal T_{\mathrm{IKW}}^k(\Phi),
               \mathcal T_{\mathrm{IKW}}^k(\widetilde\Phi)\bigr)
 \le F_k:=M^k-(M-1)^k
 \le \frac{k|U|}{M}.
\end{equation}
This will give the bound $C_{\mathcal T}\le F_k$ for the full
transformation, since the subsequent right-degree reduction
preserves each left predicate.

\subsection{Conditional Agreement on the Original Variables}

Following the ideas of \cite{impagliazzo2009new}, we show that for many right vertices $A$ there is a corresponding assignment
to $\Phi$ that extends the label at $A$ to all of $[N]$ and is close, on average,
to the neighboring left labels whose projections agree with the label at $A$.
This agreement will imply both a value guarantee and a recovery map
with bounded sensitivity.

Let $N(A):=\{\omega\in E_k:A(\omega)=A\}$ be the set of edges
incident to $A$ in $\mathcal T_{\mathrm{IKW}}^k(\Phi)$.  For a labeling $\pi=(\pi_U,\pi_V)$, define
\begin{equation}\label{eq:ikw-cons}
 \mathrm{Cons}_A:=
 \{\omega=(S,\iota,b)\in N(A):f_\omega(\pi_U(S))=\pi_V(A)\},
 \qquad
 \mu_A:=\frac{|\mathrm{Cons}_A|}{|N(A)|}.
\end{equation}
We call an edge $\omega\in N(A)$ consistent if $\omega\in\mathrm{Cons}_A$;
thus consistency means that the projection of the left label equals
the right label. This condition alone does not require the left label
to satisfy $P_S$.
For an assignment $g:[N]\to\Sigma$ to $\Phi$ and an edge
$\omega=(S,\iota,b)$, let $\operatorname{err}(\omega,g)$ be the fraction
of the $r=k-k'$ constraint positions of $S$ not selected by $\iota$
at which the left label $\pi_U(S)$ disagrees with $g$:
\[
 \operatorname{err}(\omega,g):=
 \frac1r\sum_{h\notin\operatorname{im}\iota}
 \mathbf1[\pi_U(S)|_h\ne g|_{e_h}],
\]
where $\pi_U(S)|_h$ is the pair of symbols at position $h$, and $g|_{e_h}$
is the pair of symbols assigned by $g$ to the endpoints of $e_h$ in the fixed order.

We first analyze labelings whose left labels all satisfy their predicates.
When $\Phi$ is satisfiable, every $P_S$ is satisfied by at least one label, so
replacing a left label that violates $P_S$ by one satisfying $P_S$
cannot decrease value.
We make this replacement before applying the next lemma in
\Cref{subsec:ikw-sensitivity}.

\begin{lemma}[Conditional agreement for the tuple construction]\label{lem:ikw-tuple-decoding}
There are absolute constants $c_0,c_1,c_2>0$ and
$C_{\mathrm{IKW}},C_{\mathrm{agr}}>0$ with the following property.
Suppose $0<\eta<1/100$, $0<\varepsilon<1/4$,
and
\begin{equation}\label{eq:ikw-decoding-regime}
 k\ge C_{\mathrm{IKW}}\eta^{-2}\log^2(2/\eta),\qquad
 k=o(\log N),\qquad
 \log(1/\varepsilon)\le c_0\eta\sqrt{k}.
\end{equation}
For all sufficiently large $N$, if $P_S(\pi_U(S))=1$ for every $S\in U$
and $\val_{\mathcal T_{\mathrm{IKW}}^k(\Phi)}(\pi)\ge\varepsilon$, then
there is a set $\mathcal G_\pi\subseteq V$ of uniform measure at least
$c_1\varepsilon$ such that every $A=(a_1,\ldots,a_{k'})\in\mathcal G_\pi$ has distinct
coordinates, satisfies $\mu_A\ge c_2\varepsilon$, and admits an
assignment $g_A:[N]\to\Sigma$ with
\begin{equation}\label{eq:ikw-conditional-agreement}
 g_A(a_i)=\pi_V(A)_i\quad(i\in[k']),\qquad
 \E_{\omega\in\mathrm{Cons}_A}
       \operatorname{err}(\omega,g_A)\le C_{\mathrm{agr}}\eta.
\end{equation}
\end{lemma}
\begin{proof}
We apply the IKW agreement analysis to random assignments on
$k'$- and $k$-subsets of the original variable set $[N]$, obtained
from the right and left labels, respectively.
After conditioning on the sampled vertices being distinct, we identify
the projection test with the test that compares these assignments on
a uniform $k'$-subset of a uniform $k$-set.
The IKW analysis then yields an assignment $g_A$ close to the local
assignments consistent with $\pi_V(A)$.
Finally, we use the bound on disagreement at individual vertices to
bound $\operatorname{err}$, which compares pairs of endpoint symbols.

Sample a uniform edge $\omega=(S,\iota,b)\in E_k$, with
$S=(e_1,\ldots,e_k)$ and right endpoint $A=A(\omega)=(a_1,\ldots,a_{k'})$.
At each position $h=\iota(i)$ selected for comparison, set $c_h=b_i$;
at every remaining position $h\notin\operatorname{im}\iota$, choose
$c_h\in\{0,1\}$ independently and uniformly.  Thus we choose one end of
each edge $e_h$, obtaining vertices
\[
 z_h:=\operatorname{end}_{c_h}(e_h)\quad(h\in[k]),\qquad
 z_{\iota(i)}=\operatorname{end}_{b_i}(e_{\iota(i)})=a_i.
\]
The uniform choice of $\omega$ makes $S,\iota,b$ independent and uniform.
For any $\iota$, the pairs $(e_h,c_h)$ are therefore independent uniform
incidences, giving each vertex probability $d/(2M)=1/N$.
Hence $z_1,\ldots,z_k$ are independent uniform vertices, also independent
of $\iota$.

To obtain a $k$-element vertex set, condition on the event
$\mathsf{Distinct}$ that all $z_h$ are distinct.  We will later remove
this conditioning using the bounds
\begin{equation}\label{eq:ikw-representative-collisions}
 \Pr[\neg\mathsf{Distinct}]\le\binom{k}{2}/N,\qquad
 \Pr[\neg\mathsf{Distinct}\mid A]=O(k^2/N)
\end{equation}
for every $A$ with distinct coordinates.  Here conditioning on $A$ means
$A(\omega)=A$; the second bound counts collisions among the remaining
vertices and with the vertices of $A$.

Conditional on $\mathsf{Distinct}$, $W:=\{z_1,\ldots,z_k\}$ is a uniform
$k$-set, and $A$ is a uniform ordering of a uniform $k'$-subset of $W$.
For each $k'$-subset $Q\subseteq[N]$, $C_1(Q)$ is a random assignment
$Q\to\Sigma$: choose a uniform ordering $A=(a_1,\ldots,a_{k'})$
of $Q$ and assign $\pi_V(A)_i$ to $a_i$.
We also write $\pi_V(A)$ for this assignment on
$\underline A:=\{a_1,\ldots,a_{k'}\}$, retaining the order $A$ in the analysis.
For each $k$-subset $W\subseteq[N]$, $C_2(W)$ is a random assignment
$W\to\Sigma$: independently choose an incidence uniformly among the $d$ incidences
at each vertex of $W$, uniformly order the resulting edge occurrences
into $S$, and return $\pi_U(S)$ restricted to $W$.
This restriction is well defined because $P_S$ requires equal symbols
at all occurrences of the same vertex.

Choose a uniform $k$-set $W$ and a uniform $k'$-subset of $W$,
and generate assignments using $C_1$ on that subset and $C_2$ on $W$
with independent randomness.
Agreement of these assignments on the subset corresponds to satisfaction
of the projection constraint on $\omega$, conditional on $\mathsf{Distinct}$.
Indeed, a uniform incidence conditioned on its vertex is uniform among
the $d$ incidences there, independently at each position.
Conditional on $W,A$, the vertex order is also uniform: each order
determines exactly one injection $\iota$ placing $A$ in it.
Thus $C_2(W)$ has the same distribution as the restricted left label
and can be generated without using $A$.

Until we remove the conditioning on $\mathsf{Distinct}$, probabilities over $A$ use
the uniform distribution on ordered tuples with distinct coordinates.
Put $\mu_A^\circ:=\Pr[\omega\in\mathrm{Cons}_A\mid A,\mathsf{Distinct}]$,
and call $A$ \emph{good} if $\mu_A^\circ\ge\varepsilon/4$.
Since every left label satisfies its predicate, an edge is satisfied
if and only if it is consistent.
By \eqref{eq:ikw-representative-collisions}, the probability that $\omega$
is satisfied conditional on $\mathsf{Distinct}$ is at least $\varepsilon/2$
for large $N$.
Thus $\E_A\mu_A^\circ\ge\varepsilon/2$, and good right vertices have measure
at least $\varepsilon/4$.

We next compare two local assignments on an additional shared set $H$
when both agree with $\pi_V(A)$ on $\underline A$.

For a fixed $A$, choose a uniform $k'$-set $H\subseteq[N]\setminus\underline A$
and, conditional on $H$, independent uniform $(k-2k')$-sets
$B_1,B_2\subseteq[N]\setminus(\underline A\cup H)$.
Each set $\underline A\cup H\cup B_i$ has size $k$; let
$w_i=C_2(\underline A\cup H\cup B_i)$, using independent randomness for
the two executions of $C_2$.
For $\alpha,\gamma>0$, call a good $A$ \emph{excellent} if
\[
 \Pr\bigl[
 w_1|_{\underline A}=w_2|_{\underline A}=\pi_V(A),\
 d_{\mathrm H}(w_1|_H,w_2|_H)>\alpha k'
 \bigr]\le\gamma.
\]
Here $\alpha$ is the threshold for the disagreement fraction on $H$,
and $\gamma$ (to be set) bounds the probability of this joint event over
$H,B_1,B_2$ and the randomness of $C_2$.

We use two IKW bounds, viewing $C_1,C_2$ as the two oracles in
\cite[Section~6.3]{impagliazzo2009new} and including their randomness
as in \cite[Section~3.5]{impagliazzo2009new}.
First, \cite[Lemma~3.6]{impagliazzo2009new} gives
\begin{equation}\label{eq:ikw-excellence-mass}
 \Pr_A[A\text{ is good but not excellent}]
 \le \frac{\exp(-\Omega(\alpha k'))}{\gamma}.
\end{equation}
To prove this bound for ordered tuples $A$, first choose a uniform
$2k'$-set $T=\underline A\cup H$,
then $B_1,B_2$ outside $T$ and the randomness of $C_2$,
then split $T$ uniformly into $\underline A,H$ and order $\underline A$.
The assignments $w_1,w_2$ are already fixed.
If they disagree on more than $\alpha k'$ vertices of $T$, the chance
that $\underline A$ avoids all disagreements is at most $\exp(-\Omega(\alpha k'))$.
The joint event defining excellence requires this avoidance;
Markov's inequality gives the factor $1/\gamma$.

For the second bound, let $B\subseteq[N]\setminus\underline A$ be a
uniform $r$-set, and let $w=C_2(\underline A\cup B)$.
Define $g_A(x)$, for $x\notin\underline A$, as a most likely value of
$w(x)$ conditional on $x\in B$ and $w|_{\underline A}=\pi_V(A)$.
Break ties arbitrarily, use a fixed default symbol
$\sigma_0\in\Sigma$ when this event has probability zero,
and set $g_A(a_i)=\pi_V(A)_i$.
Thus $g_A$ is the coordinatewise plurality of consistent local assignments,
including the randomness of $C_2$.
By \cite[Lemma~3.8]{impagliazzo2009new}, extended to arbitrary
alphabets by \cite[Lemma~5.1]{impagliazzo2009new}, there are absolute constants
$C_*,c_*>0$ such that every excellent $A$ satisfies
\begin{equation}\label{eq:ikw-decoding-tail}
 \Pr_{B,w}\!\left[
   \frac{d_{\mathrm H}(w|_B,g_A|_B)}r>C_*\alpha
   \,\middle|\,w|_{\underline A}=\pi_V(A)
 \right]\le C_*\frac{\gamma}{\varepsilon^2},
\end{equation}
provided $\alpha\ge C_*(k'/k)\log(4/\varepsilon)$ and
$\gamma\le c_*\varepsilon^3$.
Here $\mu_A^\circ\ge\varepsilon/4$, and
the remaining sets have size $r=k-k'$; replacing $r$ by $k$ in the
parameter bound only changes absolute constants since $r=\Theta(k)$.

Set $\alpha=\eta$ and $\gamma=c\min\{\eta,\varepsilon\}\varepsilon^2$
for a sufficiently small absolute $c>0$.
Then $\gamma\le c_*\varepsilon^3$, and for each excellent $A$ the
exceptional fraction of consistent local assignments in \eqref{eq:ikw-decoding-tail}
is at most $\eta$.
The lower bound on $\alpha$ follows from $k'/k\le1/\sqrt{k}$ and
\eqref{eq:ikw-decoding-regime}, with $c_0$ sufficiently small.
For the fraction of good but nonexcellent right vertices in \eqref{eq:ikw-excellence-mass}
to be at most $\varepsilon/20$, it suffices that
\[
 \alpha k'\ge C\log\frac{20}{\varepsilon\gamma}
   =O\bigl(\log(1/\eta)+\log(1/\varepsilon)\bigr)
\]
for a sufficiently large absolute $C$.
Since $\alpha k'=\Theta(\eta\sqrt{k})$, this follows by taking
$C_{\mathrm{IKW}}$ large and $c_0$ small in \eqref{eq:ikw-decoding-regime}.
Let $\mathcal G_\pi\subseteq V$ be the set of excellent right vertices
with distinct coordinates.
Its measure among distinct tuples is at least
$\varepsilon/4-\varepsilon/20=\varepsilon/5$.
Distinct tuples occupy a $1-O(k'^2/N)$ fraction of $V$, so
$\mathcal G_\pi$ has uniform measure $\Omega(\varepsilon)$ in $V$.

Fix $A\in\mathcal G_\pi$.  For uniform $B$ of size $r$ outside
$\underline A$ and $w=C_2(\underline A\cup B)$, the tail bound
\eqref{eq:ikw-decoding-tail} gives
\begin{equation}\label{eq:ikw-vertex-agreement}
 \E_{B,w}\!\left[
 \frac{d_{\mathrm H}(w|_B,g_A|_B)}r
 \,\middle|\,w|_{\underline A}=\pi_V(A)\right]
 \le C_*\eta+\eta=O(\eta).
\end{equation}
Return to the experiment sampling $\omega$ and the extra endpoint choices.
Let $Z_A\in[0,1]$ be the fraction of positions
$h\notin\operatorname{im}\iota$ where the symbol at $z_h$ in $\pi_U(S)$
differs from $g_A(z_h)$.
Under $\mathsf{Distinct}$, these vertices form
$B=W\setminus\underline A$, and the restriction $w=\pi_U(S)|_W$ has
the distribution generated by $C_2(W)$. Thus
$Z_A=d_{\mathrm H}(w|_B,g_A|_B)/r$, so \eqref{eq:ikw-vertex-agreement}
bounds its expectation conditional on consistency and $\mathsf{Distinct}$.

To restore uniform extra endpoint choices conditional on consistency, we now
remove the conditioning on $\mathsf{Distinct}$.
By \eqref{eq:ikw-representative-collisions},
$|\mu_A-\mu_A^\circ|=O(k^2/N)$, so $\mu_A\ge\varepsilon/8$ for large $N$.
Conditioning on consistency increases the collision probability by at most
$1/\mu_A$; since $Z_A\le1$, we obtain
\[
 \E[Z_A\mid\omega\in\mathrm{Cons}_A,A]
 \le O(\eta)+\frac{\Pr[\neg\mathsf{Distinct}\mid A]}{\mu_A}
 =O(\eta)+O\!\left(\frac{k^2}{N\varepsilon}\right)=O(\eta).
\]
Indeed, \eqref{eq:ikw-decoding-regime} gives
$k^2/(N\varepsilon\eta)\le O(k^{5/2})\exp(O(\sqrt{k}))/N=o(1)$
because $k=o(\log N)$.
The extra endpoint choices are independent of consistency.
Whenever $\pi_U(S)|_h\ne g_A|_{e_h}$, a uniform endpoint detects a
disagreement with probability at least $1/2$.
Therefore $\E_{\omega\in\mathrm{Cons}_A}\operatorname{err}(\omega,g_A)
\le2\E[Z_A\mid\omega\in\mathrm{Cons}_A,A]=O(\eta)$.
Together with $g_A(a_i)=\pi_V(A)_i$ and $\mu_A\ge\varepsilon/8$,
this proves the lemma.
\end{proof}

\section{Randomized Decoding and Sensitivity Analysis}\label{sec:ikw-sensitivity-analysis}

Using the conditional agreement in \Cref{lem:ikw-tuple-decoding},
we construct a randomized recovery procedure with two guarantees.
First, when the original 2-CSP is satisfiable,
an algorithm attaining the required value on the constructed Label Cover
instance yields an algorithm attaining high value on that 2-CSP.
Second, we bound the recovery parameters
$C_{\mathcal R}$ and $D_{\mathcal R}$ from
\Cref{def:law-sensitivity-pullback-data}.
For fixed $\Phi$, the recovery map from Label Cover labeling distributions
to assignment distributions is $C_{\mathcal R}$-Lipschitz with respect to EMD.
With the supplied labeling distribution held fixed, $D_{\mathcal R}$
bounds the EMD between the recovered output distributions when one
constraint relation of $\Phi$ changes.
Recovery may use several independent samples from the supplied
labeling distribution.

We apply right-degree reduction to $\mathcal T_{\mathrm{IKW}}^k(\Phi)$,
replacing each right vertex by copies. We recover each original right label
by choosing one copy uniformly. A change to one copy can affect recovery
only when that copy is chosen. Combining
this observation with bounds on how often the decoder reads each left
label gives the sensitivity-transfer bound.

\subsection{The Decoder and Its Value Guarantee}\label{subsec:ikw-randomized-decoder}

Recall that each edge $\omega=(S,\iota,b)\in E_k$ has a left
endpoint $S=(e_1,\ldots,e_k)\in U=E^k$.  To recover the value at a variable $x$ of
$\Phi$, the decoder chooses an incident edge $e$ of $G$ and samples
$\omega$ together with a position $h\notin\operatorname{im}\iota$
such that $e_h=e$.  Different positions $h$ count separately even when
they contain the same edge.  For $A\in V$ and $e\in E$, the sampling set is
\[
 N(A;e):=
 \{(\omega,h)\in N(A)\times[k]:\omega=(S,\iota,b),\
             h\notin\operatorname{im}\iota,\ e_h=e\}.
\]

Use the same default symbol $\sigma_0$ as above.
For parameters $\varepsilon,\eta$ as in \eqref{eq:ikw-decoding-regime}
and a labeling $\pi=(\pi_U,\pi_V)$, the decoder chooses a uniform
$A=(a_1,\ldots,a_{k'})\in V$.
For each $x\in[N]$, it independently does the following.
\begin{enumerate}
 \item If $x=a_i$ for some $i$, return $\pi_V(A)_i$, using the least such
 index when necessary.
 \item Otherwise choose a uniform endpoint incidence at $x$, with
 underlying edge $e$, and sample
 \[
 L:=\left\lceil C_{\mathrm{dec}}\varepsilon^{-1}\log(2/\eta)\right\rceil
 \]
 independent uniform pairs $(\omega_\ell,h_\ell)\in N(A;e)$,
 with $S_\ell$ the left endpoint of $\omega_\ell$.
 \item Choose an index $\ell_*$ uniformly among those with
 $\omega_{\ell_*}\in\mathrm{Cons}_A$, and return the symbol of
 $\pi_U(S_{\ell_*})$ in edge position $h_{\ell_*}$ at the endpoint
 chosen in Step~2.
 If there is no such index, return $\sigma_0$.
\end{enumerate}
Here $C_{\mathrm{dec}}$ is a sufficiently large absolute constant.
Denote the resulting assignment by $\bm h_A:[N]\to\Sigma$.

Each uniform sample from $N(A;e)$ in Step~2 can be generated as follows.
Choose $\iota$ and one of the $d$ endpoint incidences at each coordinate
$a_i$ of $A$ independently and uniformly.
The incidence at $a_i$ specifies $(e_{\iota(i)},b_i)$.
Choose $h$ uniformly outside $\operatorname{im}\iota$,
set $e_h=e$, and fill the other unselected positions with independent
uniform edges of $E$.

For coupling purposes, sample the unused randomness also when $x$ is
a coordinate of $A$.

\begin{lemma}[Coupling bound for the randomized decoder]\label{claim:IKW-sensitivity}
Let $\operatorname{dist}$ denote normalized Hamming distance.
For two labelings $\pi=(\pi_U,\pi_V)$ and
$\widetilde\pi=(\widetilde\pi_U,\widetilde\pi_V)$ on the common graph of
$\mathcal T_{\mathrm{IKW}}^k(\Phi)$ with its fixed projections,
use the same right vertex $A$, endpoint incidences, and sampled pairs
$(\omega_\ell,h_\ell)$ in both decoders.  For each $x$, use a shared
uniform random ordering of $[L]$ to select the first successful index
in each execution.
The corresponding decoded assignments $\bm h_A$ and
$\widetilde{\bm h}_A$ satisfy
\[
 \E\operatorname{dist}(\bm h_A,
                       \widetilde{\bm h}_A)
 \le
 \operatorname{dist}(\pi_V,\widetilde\pi_V)
 +L\operatorname{dist}(\pi_U,\widetilde\pi_U).
\]
The bound also holds in expectation when $(\pi,\widetilde\pi)$ is drawn
from any coupling of two distributions of labelings.
\end{lemma}
\begin{proof}
The two symbols recovered at $x$ agree unless the labels at $A$ differ
or some sampled left label differs.  Averaging over a uniform
vertex $x\in[N]$, a uniform endpoint incidence makes $e$ uniform on $E$.
Placing this uniform edge in a uniform unselected position and choosing
all other unselected edges uniformly recovers the original distribution
on $N(A)$, together with a uniform unselected position.
Averaging over $A$ therefore makes every sampled $\omega_\ell$ a uniform
edge of $E_k$.  By biregularity, its left tuple is uniform on $U$.
The union bound proves the assertion.
\end{proof}

To prove the decoder's value guarantee, we must control the probability
that a sample in Step~2 is consistent with the label at $A$.
Sampling from $N(A;e)$ fixes one unselected edge occurrence to $e$,
so this probability may differ from $\mu_A$.
The following lemma shows that it remains close to $\mu_A$
for most edges $e$.
We defer its proof to the end of this subsection.

\begin{lemma}[Conditioning on one coordinate]\label{lem:ikw-conditioning-sampler}
Let $\mathcal X$ be a finite set of size $m$, let $r\ge1$, and let
$f:\mathcal X^r\to[0,1]$ have mean $\mu>0$ under the uniform product
distribution.  Let $B\sim\operatorname{Unif}(\mathcal X^r)$ and define
$q:\mathcal X\to[0,1]$ by
\[
 q(x):=\frac1r\sum_{h=1}^r\E[f(B)\mid B_h=x].
\]
For every constant $\theta\in(0,1)$,
\[
 \frac1m\bigl|\{x\in\mathcal X:|q(x)-\mu|>\theta\mu\}\bigr|
 \le C_\theta\frac{\log(1/\mu)}r
\]
for a constant $C_\theta$ depending only on $\theta$.
\end{lemma}

\begin{lemma}[Value of a decoded assignment]\label{claim:ikw-analysis}
There is an absolute constant $C_{\mathrm{val}}>0$ such that, under the
hypotheses of \Cref{lem:ikw-tuple-decoding}, every $A\in\mathcal G_\pi$
satisfies
\[
 \Pr\bigl[\val_\Phi(\bm h_A)
                  \ge1-C_{\mathrm{val}}\eta\bigr]\ge0.99.
\]
\end{lemma}
\begin{proof}
Fix $A\in\mathcal G_\pi$ and the assignment $g_A$ from
\Cref{lem:ikw-tuple-decoding}.
We first show that $g_A$ has value at least $1-O(\eta)$.
We then show that, with probability at least $0.99$, the decoded assignment
$\bm h_A$ differs from $g_A$ on at most an $O(\eta)$-fraction of the variables.
Since $\Phi$ is regular, this implies that $\bm h_A$ also has value
at least $1-O(\eta)$.

Define $q_A:E\to[0,1]$ by
\[
 q_A(e):=\Pr_{(\omega,h)\in N(A;e)}[\omega\in\mathrm{Cons}_A].
\]
For $B\in E^r$, let $f_A(B)$ be the probability of consistency after
averaging over $\iota$ and the endpoint incidences at $A$, with $B$ placed in
increasing order of the unselected positions.  Its mean is $\mu_A$,
and $q_A(e)$ is the average of its conditional means when one coordinate
of $B$ is fixed to $e$.  \Cref{lem:ikw-conditioning-sampler} gives an
exceptional set $E_{\mathrm{irr}}\subseteq E$ with
\begin{equation}\label{eq:ikw-irregular-edges}
 b_A:=|E_{\mathrm{irr}}|/M
 \le C\frac{\log(1/\mu_A)}r=O(\eta),
 \qquad
 \frac{\mu_A}{2}\le q_A(e)\le\frac{3\mu_A}{2}
 \quad(e\notin E_{\mathrm{irr}}).
\end{equation}
The estimate $b_A=O(\eta)$ uses $\mu_A\ge c_2\varepsilon$ and
\eqref{eq:ikw-decoding-regime}.

Every right vertex has degree $D_V$ from \eqref{eq:ikw-degrees}.
Writing $n_e:=|N(A;e)|$, the construction gives
\begin{equation}\label{eq:ikw-marked-counts}
 |N(A)|=D_V,\qquad n_e=D_V r/M.
\end{equation}
Let $\mathcal D_A$ be the distribution on pairs
$(\omega,h)\in\mathrm{Cons}_A\times[k]$ with
$h\notin\operatorname{im}\iota$, where $\omega=(S,\iota,b)$, obtained by
first choosing $\omega$ uniformly in $\mathrm{Cons}_A$ and then an
unselected position $h$ uniformly.  Its mass at each such pair
$(\omega,h)$ is $1/(\mu_A D_V r)$, and its marginal at $e_h=e$ is
$q_A(e)/(M\mu_A)$.  If $g_A$ violates $e$, every left label satisfying
the predicate of a tuple containing $e$ disagrees with $g_A$ there,
because the predicate requires that $e$ be satisfied.  Thus
\eqref{eq:ikw-conditional-agreement} and
\eqref{eq:ikw-irregular-edges} imply
\begin{equation}\label{eq:ikw-recovered-base-value}
 1-\val_\Phi(g_A)\le b_A+2C_{\mathrm{agr}}\eta=O(\eta).
\end{equation}

We next bound $\E\operatorname{dist}(\bm h_A,g_A)$.
Since $\mu_A\ge c_2\varepsilon$, for a nonexceptional edge $e$,
the probability that none of the $L$ sampled pairs is consistent is at most
\[
 (1-q_A(e))^L\le\exp(-L\mu_A/2)\le\eta.
\]
The samples are independent and uniform, and Step~3 selects uniformly
among successful indices.  Thus, conditional on success, every consistent
pair in $N(A;e)$ is equally likely to be selected.
For the sampling rule in Steps~2--3, average over a uniform vertex of
$G$ and a uniform endpoint incidence there, including samples unused
when the vertex lies in $A$.  The probability of successfully selecting
a particular pair $(\omega,h)\in N(A;e)$ with
$\omega\in\mathrm{Cons}_A$ is
\[
 \frac{1-(1-q_A(e))^L}{M q_A(e)n_e}
 \le \frac1{M q_A(e)n_e}.
\]
Summing these probabilities over all $e$ and consistent pairs gives the
success probability, which may be less than one.
By \eqref{eq:ikw-marked-counts}, the upper bound divided by the mass
under $\mathcal D_A$ is
\[
 \frac{1/(M q_A(e)n_e)}{1/(\mu_A D_V r)}
 =\frac{\mu_A}{q_A(e)}\le2
 \qquad(e\notin E_{\mathrm{irr}}).
\]
Consequently \eqref{eq:ikw-conditional-agreement} bounds the probability
of returning a symbol different from $g_A(x)$ through a nonexceptional
edge by $2C_{\mathrm{agr}}\eta$.  For $x$ equal to a coordinate of $A$, the decoder
returns $g_A(x)$ by definition, so
\[
 \E\operatorname{dist}(\bm h_A,g_A)
 \le b_A+2C_{\mathrm{agr}}\eta+\eta=O(\eta).
\]
Markov's inequality bounds this distance by a sufficiently large
constant times $\eta$ with probability at least $0.99$.
Changing the labels at a $\rho$-fraction of the vertices of a regular 2-CSP changes
the assignment's value by at most $2\rho$.  Together with
\eqref{eq:ikw-recovered-base-value}, this proves the lemma.
\end{proof}

\begin{proof}[Proof of \Cref{lem:ikw-conditioning-sampler}]
Let $P$ be the uniform distribution on $\mathcal X^r$ and define a
distribution $Q$ on the same space by $Q(z):=P(z)f(z)/\mu$ for
$z\in\mathcal X^r$.  For probability distributions on a common finite
set, use natural logarithms to define
$D(Q\|P):=\sum_z Q(z)\ln(Q(z)/P(z))$, with zero terms omitted.
Since $f\le1$, we have $D(Q\|P)\le\ln(1/\mu)$.
For $h\in[r]$, let $Q_h$ be the $h$-th coordinate marginal of $Q$,
a distribution on $\mathcal X$, and let
$\overline Q=r^{-1}\sum_{h=1}^rQ_h$ be their average on $\mathcal X$.  Entropy subadditivity and convexity of
relative entropy give
\[
 D(\overline Q\|\operatorname{Unif}(\mathcal X))
 \le\frac1r\sum_hD(Q_h\|\operatorname{Unif}(\mathcal X))
 \le\frac{\ln(1/\mu)}r.
\]
Moreover, $m\overline Q(x)=q(x)/\mu$.  Since the latter quantities
average to one, the relative entropy on the left equals
\[
 \frac1m\sum_x
 \left(\frac{q(x)}\mu\ln\frac{q(x)}\mu-\frac{q(x)}\mu+1\right).
\]
The function $z\ln z-z+1$, extended continuously at zero, has a positive
minimum outside $[1-\theta,1+\theta]$.  This proves the bound.
\end{proof}

\subsection{Recovery after Right-Degree Reduction}\label{subsec:ikw-sensitivity}

We apply the right-degree reduction of
\Cref{subsection:degree-reduction-region} to
$\mathcal T_{\mathrm{IKW}}^k(\Phi)$, replacing each right vertex by a cloud
of copies. To recover the label of an original right vertex, we choose
one copy uniformly.
If $D_V$ is the original right degree and $V$ the original right vertex
set, choosing a uniform vertex of $V$ and then a uniform copy of that
vertex selects any particular copy with probability $1/(D_V|V|)=1/|E_k|$.
Thus one changed right label contributes at most $N/|E_k|$ to the expected
Hamming distance between the $N$-variable assignments in one run of the decoder.

\begin{proposition}[Sensitivity-transfer parameters for the IKW route]
\label{prop:ikw-recovery}
There is an absolute constant $K\ge100$ such that the following holds.
Let $\eta,\varepsilon,k,N$ satisfy \eqref{eq:ikw-decoding-regime}, with
$0<\eta<1/K$ and $0<\varepsilon<1/4$.
Consider $2$-CSPs $\Phi$ on the common $d$-regular multigraph
$G=([N],E)$, with the common alphabet $\Sigma$ and endpoint orderings
from the tuple construction; only their constraint relations vary.
Choose an even right degree $d_\star\ge6$ with expander parameter
$\lambda_\star\le\varepsilon$ as in \Cref{def:expander}.
Use the same expander and incidence ordering on each cloud as $\Phi$
varies, and define the full transformation by
\[
 \mathcal T(\Phi):=
 \mathcal T_{\mathrm{DR}}(\mathcal T_{\mathrm{IKW}}^k(\Phi);d_\star).
\]
There is a recovery map $\mathcal R_\Phi$
mapping distributions of labelings
of $\mathcal T(\Phi)$ to distributions of assignments to $\Phi$.
If $\Phi$ is satisfiable and an input distribution $\mu$ has expected
value at least $4\varepsilon$,
then $\mathcal R_\Phi(\mu)$ has expected value at least $1-K\eta$ on $\Phi$.
The sensitivity-transfer parameters
(\Cref{def:law-sensitivity-pullback-data}) satisfy
\begin{equation}\label{eq:ikw-transfer-parameters}
 C_{\mathcal T}C_{\mathcal R},\ D_{\mathcal R}
 \le O\!\left(
 \frac{k\log^2(2/\eta)}{\eta\varepsilon^3}
 \right).
\end{equation}
As the constraint relations of $\Phi$ vary, all images have the same graph,
alphabets, and projections. The images of
2-CSPs differing in one relation are connected by paths of predicate swaps.
\end{proposition}
\begin{proof}
To sample from $\mathcal R_\Phi(\mu)$, we generate independent candidate
assignments to $\Phi$ and select one using a random value threshold.
Each trial generates one candidate as follows. Draw a labeling of
$\mathcal T(\Phi)$ from $\mu$ and apply $R_{\mathrm{DR}}$ from
\Cref{subsection:degree-reduction-region} to obtain a labeling of
$\mathcal T_{\mathrm{IKW}}^k(\Phi)$.
The value guarantee in \Cref{claim:ikw-analysis} requires every left label
to satisfy its predicate. We therefore replace each left label that violates
its predicate by the first label satisfying that predicate in an ordering
independent of $\Phi$, if one exists,
and otherwise leave it unchanged. Run the randomized decoder of
\Cref{subsec:ikw-randomized-decoder} on the resulting labeling to obtain
the candidate assignment to $\Phi$.

The replacement of left labels cannot decrease value, is $1$-Lipschitz
in each left coordinate, and depends on $\Phi$ only through the
corresponding predicate.

If $\Phi$ is satisfiable, every left predicate is satisfied by at least one label.
Thus the replacement ensures that all left labels satisfy their predicates.
The value-recovery bound in \Cref{lem:degree-reduction-region} shows that
the labeling supplied to the decoder has expected value at
least $3\varepsilon$.  Its value is at least $\varepsilon$ with probability
at least $\varepsilon$.  Conditional on any such labeling,
\Cref{lem:ikw-tuple-decoding,claim:ikw-analysis} imply that the trial returns an
assignment to $\Phi$ of value at least $1-C_{\mathrm{val}}\eta$ with probability
at least $a\varepsilon$, for an absolute $a>0$.
Hence a trial succeeds with probability at least $a\varepsilon^2$.

Run
\[
 T:=\left\lceil\frac1{a\varepsilon^2}\log\frac1\eta\right\rceil
\]
independent trials.  We select among their candidate assignments by a
random value threshold: a small change in a candidate's value then changes
its eligibility only with small probability.
\Cref{lem:random-threshold-selection} will bound the sensitivity of this
selection rule.  Choose $t$ independently of the trials and uniformly in
$[C_{\mathrm{val}},2C_{\mathrm{val}}]$ and return the first candidate
of value at least $1-t\eta$, or the constant assignment
$x\mapsto\sigma_0$ if no candidate passes. With probability at least $1-\eta$, some trial
succeeds, in which case a candidate passes and its value is at least
$1-2C_{\mathrm{val}}\eta$.  Choosing
$K\ge\max\{100,4C_{\mathrm{val}}+2\}$ therefore gives expected value
at least
\[
 (1-\eta)(1-2C_{\mathrm{val}}\eta)\ge1-K\eta.
\]

We next bound the expected Hamming distance between the candidate
assignments in one coupled trial.
The right vertex set $V'$ of $\mathcal T(\Phi)$ consists of $D_V$ copies
of each vertex in $V$, so $|V'|=|E_k|$.
For fixed $\Phi$ and two labelings of $\mathcal T(\Phi)$, let
$H_U,H_{V'}$ be their Hamming distances on $U$ and $V'$, respectively.
For each vertex in $V$, choose the same random copy in both recoveries
and use the coupling from \Cref{claim:IKW-sensitivity}.
Let $\bm g,\widetilde{\bm g}:[N]\to\Sigma$ be the two candidate
assignments obtained in one coupled trial.  They satisfy
\begin{equation}\label{eq:ikw-one-trial-coupling}
 \begin{aligned}
 \E d_{\mathrm H}(\bm g,\widetilde{\bm g})
 &\le N\left(\frac{L H_U}{|U|}
             +\frac{H_{V'}}{D_V|V|}\right)\\
 &\le\frac{NL}{|U|}(H_U+H_{V'}),
 \end{aligned}
\end{equation}
where $L=\lceil C_{\mathrm{dec}}\varepsilon^{-1}\log(2/\eta)\rceil$,
$D_V|V|=|E_k|$, and $|E_k|\ge|U|$.
If instead one relation of $\Phi$ changes while the
labeling of $\mathcal T(\Phi)$ is held fixed, $R_{\mathrm{DR}}$ is unchanged,
and the replacement of left labels can differ only at tuples containing
the changed edge. There are $F_k=M^k-(M-1)^k$ such tuples by
\eqref{eq:ikw-predicate-swaps}.
The same coupling bounds the expected Hamming distance between the two
candidate assignments by $NLF_k/|U|$.

For any neighboring 2-CSPs $\Phi,\widetilde\Phi$ and assignments $g,\widetilde g$,
regularity gives
\[
 |\val_\Phi(g)-\val_{\widetilde\Phi}(\widetilde g)|
 \le\frac{2}{N}d_{\mathrm H}(g,\widetilde g)+\frac1M.
\]
A candidate $g\in\Sigma^N$ recovered from $\Phi$ has score
$\val_\Phi(g)\in[0,1]$. Apply the Hamming-distance bound in
\Cref{lem:random-threshold-selection} to the $T$ coupled trials,
conditional on their candidate assignments, with threshold interval
$[1-2C_{\mathrm{val}}\eta,1-C_{\mathrm{val}}\eta]$ of width
$C_{\mathrm{val}}\eta$. For fixed $\Phi$ and labeling distributions
$\mu,\nu$ on $\mathcal T(\Phi)$, draw each trial's pair of labelings
independently from a coupling attaining $\EMD(\mu,\nu)$, and omit the
$1/M$ term.  When one relation of $\Phi$ changes, use the same labeling
drawn from $\mu$ in both recoveries and include that term. Averaging the
conditional Hamming-distance bound over the coupled trials gives
\[
 \begin{aligned}
 C_{\mathcal R}
 &\le O\!\left(\frac{TNL}{\eta|U|}\right),\\
 D_{\mathcal R}
 &\le O\!\left(\frac{TNLF_k}{\eta|U|}
                    +\frac{TN}{\eta M}\right).
 \end{aligned}
\]
These bounds hold on all 2-CSP instances $\Phi$; only the value guarantee
above used satisfiability.

Right-degree reduction retains each left predicate, so the path of predicate
swaps from \eqref{eq:ikw-predicate-swaps} still has length at most $F_k$.
Thus $C_{\mathcal T}\le F_k\le k|U|/M$.
Multiplying the preceding bound for $C_{\mathcal R}$ by $C_{\mathcal T}$,
and using $N/M=2/d\le2$, proves
\eqref{eq:ikw-transfer-parameters}.
\end{proof}

\subsection{Proof of the Label Cover Consequence}\label{subsec:ikw-consequence}

Both the multiplicative factor $C_{\mathcal T}C_{\mathcal R}$ and the
additive term $D_{\mathcal R}$ in the sensitivity-transfer bound are
$N^{o(1)}$, where $N$ is the number of variables of $\Phi$.
We apply this bound to the 2-CSP family of \Cref{lem:fy-source} and
express the resulting sensitivity in terms of the final number of vertices.

\begin{proof}[Proof of \Cref{thm:intro-label-cover-sensitivity-IKW}]
Use the family $\mathcal F_N$ from \Cref{lem:fy-source}, with degree
$d_0$, alphabet $\Sigma_{\mathrm{src}}$, and sensitivity
$\SwapSens_{1-\eta_{\mathrm{src}}}(\mathcal F_N)\ge N^{c_{\mathrm{src}}}$.
Let $K$ be the constant in \Cref{prop:ikw-recovery}, and fix
\[
 \eta:=\frac{\eta_{\mathrm{src}}}{2K},\qquad
 C:=\frac{c_0\eta}{2},\qquad
 \varepsilon_k:=\frac14\exp(-C\sqrt{k}),
\]
where $c_0$ is from \Cref{lem:ikw-tuple-decoding}.
Choose $k_0$ large enough that the lower bound on $k$ in
\eqref{eq:ikw-decoding-regime} and the inequality
$\log(1/\varepsilon_k)=C\sqrt{k}+\log4\le c_0\eta\sqrt{k}$
hold for every $k\ge k_0$.
The remaining condition $k=o(\log N)$ is a hypothesis of the theorem.

By \Cref{pro:ramanujan}, choose an even
$d_\star=\exp(O(\sqrt{k}))$ large enough that
$\lambda_\star\le\varepsilon_k$. Define the swap-closed family
\[
 \mathcal I_{N,k}:=\SwapClo\bigl(
 \{\mathcal T_{\mathrm{DR}}(\mathcal T_{\mathrm{IKW}}^k(\Phi);d_\star):
                                      \Phi\in\mathcal F_N\}\bigr).
\]
Let $\mathcal A$ be an algorithm on $\mathcal I_{N,k}$ with expected value at least
$4\varepsilon_k=\exp(-C\sqrt{k})$ on every satisfiable member.
Applying recovery to its outputs gives expected value at least
$1-K\eta>1-\eta_{\mathrm{src}}$
on every satisfiable $\Phi\in\mathcal F_N$.
For this constant $\eta$ and $k=o(\log N)$,
\eqref{eq:ikw-transfer-parameters} gives
\[
 C_{\mathcal T}C_{\mathcal R},\ D_{\mathcal R}
 \le O(k\varepsilon_k^{-3})=N^{o(1)}.
\]
The sensitivity-transfer theorem \Cref{lem:law-sensitivity-pullback}
therefore yields
\[
 \SwapSens(\mathcal A,\mathcal I_{N,k})
 \ge\frac{N^{c_{\mathrm{src}}}-D_{\mathcal R}}
          {C_{\mathcal T}C_{\mathcal R}}
 \ge N^{c_{\mathrm{src}}/2}
\]
for all sufficiently large $N$.

It remains to check the degrees, alphabet sizes, and number of vertices.
By \eqref{eq:ikw-degrees}, the original left degree is
$D_U=p_k2^{k'}\le(2k)^{k'}$.
Degree reduction makes
the right degree $d_\star$ and the left degree $D_Ud_\star$.
The graph stays biregular, both alphabets are unchanged, and
\[
 \begin{gathered}
 D_Ud_\star=\exp(O(\sqrt{k}\log k))\le2^{O(k)},\qquad
 d_\star=\exp(O(\sqrt{k})),\\
 |\Sigma_U|=|\Sigma_{\mathrm{src}}|^{2k}=2^{O(k)},\qquad
 |\Sigma_V|=|\Sigma_{\mathrm{src}}|^{k'}=2^{O(k)}.
 \end{gathered}
\]
Finally,
\[
 n=|U|+|V'|=M^k(1+D_U)=N^{\Theta(k)},
\]
since $M=d_0N/2$ and $k=o(\log N)$.  Choose absolute $a,b>0$ with
$N^{ak}\le n\le N^{bk}$.  The sensitivity bound is then at least
$n^{c_{\mathrm{src}}/(2bk)}$.  Taking the sensitivity exponent in
\Cref{thm:intro-label-cover-sensitivity-IKW} to be $c_{\mathrm{src}}/(2b)$ proves all
its assertions.
\end{proof}

\part{The Dinur--Harsha Route to Low-Soundness Label Cover}
\label{part:DH}

In this part we prove the sensitivity lower bound for Label Cover in
\Cref{thm:intro-label-cover}
using the composition framework of Dinur and
Harsha~\cite{dinur2013composition}.

We start from the same 2-CSP family in \Cref{lem:fy-source}, for which
satisfying almost all constraints requires polynomial sensitivity. Our
goal is to retain a polynomial sensitivity lower bound when the required
fraction of satisfied
constraints is only $1/g$, for $g$ potentially a function of $n$, while obtaining the degree and alphabet bounds
in the theorem.

In \Cref{sec:robust,sec:robust-recovery}, we encode assignments to these 2-CSPs
by polynomials and construct tests that check parts of this encoding.
The construction lets us convert an algorithm with a small value
guarantee on satisfiable instances into one that satisfies almost all
constraints of the original 2-CSP. To transfer the sensitivity lower
bound, we also control how the recovered output distribution changes
when either the supplied output distribution or one constraint of the original 2-CSP
changes.

The degrees and alphabet sizes of this initial construction still depend
on the number of variables in the original 2-CSP. To reduce them, in \Cref{sec:dh-composition} we view predicates as 
large circuits deciding which tuples of neighboring right labels can be
realized by a left label satisfying its predicate. These are then replaced by
smaller circuits that query a small number of proof symbols.
The graph specifying which labels are compared
remains fixed when the 2-CSP constraints change, allowing us to bound the
effect of each such change.

Finally, \Cref{sec:dh-iteration} repeats this construction together with
alphabet and degree reduction. We control both the size of the checking
circuits and the accumulated losses in the sensitivity-transfer bound
of \Cref{sec:degree-reduction-region}. This yields degree and alphabet
bounds depending only on $g$, while preserving a sensitivity lower bound
that is polynomial in the final instance size, with an exponent
independent of $g$.

\section{Constraint Groups and Low-Degree Testing}\label{sec:robust}

We prepare two ingredients for converting an algorithm with a small
value guarantee on satisfiable instances into one that satisfies
almost all constraints of the 2-CSP family in \Cref{lem:fy-source}.

First, we group the constraints of these 2-CSPs so that, for any assignment
violating a constant fraction of all constraints, a uniformly chosen group is
unlikely to contain only satisfied constraints. Second, we state a
polynomial-testing theorem: agreement between low-degree polynomials on
lines and proposed values at points can be explained, apart from a small
error, by a short list of global polynomials. In
\Cref{sec:robust-recovery}, these ingredients let us identify a small
collection of assignments with high value on the original 2-CSP from
local polynomial labels.

We use the following guarantee for the Arora--Sudan low-degree test,
stated in terms of a list of global polynomials.
A polynomial on an affine line is represented in any fixed affine
parameterization; equality of two such restrictions does not depend on this
choice.

\begin{proposition}[Arora--Sudan low-degree test]
\label{prop:as-line-list}
There are absolute constants $C_0,C_1>0$ such that, for a finite field
$\F$, whenever $q=|\F|\ge C_0((D+1)/\kappa)^{C_1}$ and $0<\kappa\le 1/2$, every
$f:\F^m\to\F$ has a list
$\mathcal L_f\subseteq\{P\in\F[x_1,\ldots,x_m]:\deg P\le D\}$ of at
most $4/\kappa$ polynomials with the following property.  For every assignment
$\ell\mapsto B_\ell$ of polynomial functions $B_\ell:\ell\to\F$ of
degree at most $D$ to affine lines $\ell\subseteq\F^m$,
\begin{equation}\label{eq:as-line-list}
 \Pr_{\ell,\,z\in\ell}
 \bigl[B_\ell(z)=f(z)\;\wedge\;
       B_\ell\notin\{P|_\ell:P\in\mathcal L_f\}\bigr]
 \le 2\kappa,
\end{equation}
where the line and the point on it are uniform.  The list depends only on $f$,
not on the assignment of line polynomials.
\end{proposition}
\begin{proof}
This is the remark following Theorem~17 in the
\href{https://people.csail.mit.edu/madhu/papers/1997/arora-journ.pdf}{August 1, 2001 author version}
of \cite{arora2003improved}, obtained there by using Lemma~14 in place of
Lemma~13.  The list may be taken to contain all polynomials of degree at most $D$
having agreement at least $\kappa/2$ with $f$; the same field-size
condition, after enlarging the absolute constants, bounds its size by
$4/\kappa$.  Lemma~14 applies to every assignment $\ell\mapsto B_\ell$
with this fixed list.
\end{proof}

For later use, let $\beta_f(\ell)$ be the maximum agreement with $f|_\ell$
of a line polynomial of degree at most $D$ outside the listed restrictions,
with maximum zero if this set is empty.  Choosing a maximizer separately on
each line in \eqref{eq:as-line-list} gives
\begin{equation}\label{eq:as-max-outside-list}
 \E_\ell\beta_f(\ell)\le 2\kappa.
\end{equation}
This observation allows different representations of the same line to carry
different answers: every answer outside the fixed list is bounded by the same
maximum.

We next choose groups of constraint indices so that an assignment
violating a constant fraction of constraints is unlikely to satisfy
every constraint in a uniformly chosen group. The bound on how often
each index occurs will also control the number of tests affected by
changing one constraint relation.

\begin{lemma}[Constraint groups from expander walks]\label{lem:walk-bundles}
Let $E$ be a nonempty finite set of constraint indices, and fix
$0<\eta<1/2$ and $0<p\le 1/4$.  There is a fixed list of ordered groups
$S_1,\ldots,S_M\in E^r$ of constraint indices, with repetitions allowed, such that
\[
 r=1+\left\lceil\frac{2}{\eta}\log\frac{8}{p}\right\rceil,
 \qquad M=|E|p^{-O_\eta(1)},
\]
every $e\in E$ belongs to at most $rM/|E|$ groups, and, for every
$W\subseteq E$ with $|W|\ge\eta |E|$,
\begin{equation}\label{eq:bundle-avoidance}
 \Pr_{j\in[M]}[S_j\text{ contains no member of }W]\le p/8.
\end{equation}
The list depends only on $E$, $\eta$, and $p$, not on the constraint relations.
\end{lemma}
\begin{proof}
Fix a constant even degree $d_{\rm exp}$ large enough that
\Cref{pro:ramanujan} gives an expander on $E$ with normalized nontrivial
spectral norm at most $1/2$.  Enumerate all length-$(r-1)$ walks, retaining
the choice of edge incidence at each step, so $M=|E|d_{\rm exp}^{r-1}$.  Every walk has the same weight
in the stationary walk distribution.  Regularity implies that exactly $M/|E|$
indexed walks visit a prescribed vertex at each position, proving the
occurrence bound.

Let $T\in\mathbb R^{E\times E}$ be the normalized adjacency matrix and
$A=E\setminus W$.  For a vector $x\in\mathbb R^E$ supported on $A$, decomposition into the constant vector
and its orthogonal complement gives
\[
 x^{\mathsf T}Tx
 \le \tfrac12\|x\|_2^2+\frac{1}{2|E|}\Bigl(\sum_{e\in E}x_e\Bigr)^2
 \le (1-\eta/2)\|x\|_2^2.
\]
The principal submatrix $T_A\in\mathbb R^{A\times A}$ is symmetric and
entrywise nonnegative.  For $y\in\mathbb R^A$,
$|y^{\mathsf T}T_Ay|\le |y|^{\mathsf T}T_A|y|$; applying the preceding
bound to the zero extension of $|y|$ to $E$ gives
$\|T_A\|\le1-\eta/2$.  Write $\mathbf1_A\in\mathbb R^A$ for the
all-ones vector.  Hence the probability that the stationary walk
stays in $A$ is at most
$|E|^{-1}\mathbf1_A^{\mathsf T}T_A^{r-1}\mathbf1_A
\le(1-\eta/2)^{r-1}\le p/8$.
\end{proof}

\section{Initialization by Polynomial Restrictions}
\label{sec:robust-recovery}

We construct the initial Label Cover family from which the later
composition steps start, using the base 2-CSP family from \Cref{lem:fy-source}.
On satisfiable 2-CSPs, an algorithm with a small value guarantee
for the constructed Label Cover instances can be converted into one that
satisfies almost all constraints of the original 2-CSP.
We bound the changes in both the transformed instance and the recovered
output distribution, so that the 2-CSP sensitivity lower bound transfers.

We encode assignments to the 2-CSP variables by polynomial interpolation.
Each left label is a polynomial intended to describe part of this encoding;
each right label gives a value at one point. The tests from the preceding section
use the Arora--Sudan guarantee to relate agreement between these labels
to a short list of candidate global polynomials. For sensitivity transfer,
we combine the constraint groups with a recovery over a common polynomial space chosen
independently of the constraint relations and supplied labels. We score
every polynomial by the fraction of left labels that pass their tests
and agree with its restrictions.
We return the first polynomial in a common order whose score meets a
random threshold. For labelings of sufficiently large value, the
short-list guarantee ensures that some polynomial has a sufficiently
large score. Changes to labels or relations affect the scores, and
the random threshold lets us bound the resulting change in recovery.

Fix the common graph $G=(X,E)$ of the family $\mathcal F_N$ in
\Cref{lem:fy-source}, with alphabet $\Sigma_{\mathrm{src}}$.
As in that lemma, $N=|X|$ is the number of variables, so
$|E|=d_0N/2$ is the number of constraints. Each edge has a fixed ordering
of its endpoints, and parallel edges retain distinct indices.
For relations $b=(b_e)_{e\in E}$, with
$b_e\subseteq\Sigma_{\mathrm{src}}\times\Sigma_{\mathrm{src}}$, let
$\Phi_b\in\mathcal F_N$ have relation $b_e$ on edge $e$.
An assignment $a:X\to\Sigma_{\mathrm{src}}$ satisfies an edge
$e=(x,y)$ exactly when $(a(x),a(y))\in b_e$; the endpoints may coincide.
Only the relations $b_e$ vary across this family. We choose all query
locations from the fixed graph $G$, so changing one relation will affect
local tests while leaving the constructed graph unchanged.

\begin{theorem}[Initialization by polynomial restrictions]\label{thm:ldrc-free-initializer}
For all sufficiently large $N$ and every
$(\log |E|)^{-1}\le\gamma\le1/2$, the family $\mathcal F_N$ on
$G=(X,E)$ with alphabet $\Sigma_{\mathrm{src}}$ admits a transformation
$b\mapsto\widehat J_b$ and a recovery map
$\mathcal R_{b,\gamma}$ from distributions on full labelings of
$\widehat J_b$ to distributions on $\Sigma_{\mathrm{src}}^X$, with the
following properties.
The instances have common graph, alphabets, and projections; all vertices
have positive degree; and satisfiability is preserved.  For every
distribution $\mu$ on full labelings of $\widehat J_b$, if
\[
 \E_{Y\sim\mu}\val_{\widehat J_b}(Y)\ge\gamma,
\]
then the recovered distribution has expected value greater than
$1-\eta_{\rm src}$ on $\Phi_b$.  The number of vertices on each side is
$N^{O(1)}$, and
both degrees, the right alphabet size, the logarithm of the left alphabet
size, and the local description complexity
(\Cref{def:local-description-instance}) are
$\poly(\log N,\gamma^{-1})$.
The sensitivity-transfer parameters satisfy
\begin{equation}\label{eq:surface-initializer-transfer}
 C_{\mathcal T}C_{\mathcal R},\ D_{\mathcal R}
 \le O\!\left(\gamma^{-5}\log\frac2\gamma\right).
\end{equation}
If $b$ and $b'$ differ in one relation, their images are connected by
paths of predicate swaps of length at most $C_{\mathcal T}$ with the same
local-description bound. The constants depend only on the constants in
\Cref{lem:fy-source}.
\end{theorem}

We prove the theorem in the remainder of this section.  Recovery is allowed
to enumerate a finite set of polynomials; no polynomial-time recovery guarantee
is needed for the sensitivity lower bound.  The graph and local circuits of
$\widehat J_b$ have polynomial-size descriptions.

\subsection{The Polynomial Restriction Tests}

Each group of source constraints determines a polynomial map with two
parameters. At designated parameter values, the map visits the variables
in the group, so a left polynomial label can be checked against the
source constraints. Away from these designated values, randomizing the
map produces uniform affine lines, on which the Arora--Sudan test compares
the left label with right labels at individual points. This connects the
source constraints to the global polynomials supplied by the list guarantee.
We define the maps and labels before collecting the field-size conditions.

Fix a constant $0<\eta<\eta_{\rm src}/2$, the allowed fraction of violated
source constraints after recovery. We use $\tau$ as the value threshold
before right-degree reduction, $\kappa$ as the error parameter in
\Cref{prop:as-line-list}, and $L_{\max}$ as its list-size bound. We choose
a score threshold $p$ small enough that value at least $\tau$ forces
some polynomial on this list to agree with accepted left labels on at
least a $p$ fraction of vertices.
Set
\begin{equation}\label{eq:surface-threshold-parameters}
 \tau=\gamma/4,\qquad \kappa=\tau/16,\qquad
 L_{\max}=\lceil64/\tau\rceil,\qquad
 p=\frac{\tau}{2L_{\max}}=\Theta(\gamma^2).
\end{equation}
Apply \Cref{lem:walk-bundles} to $E$ with parameters $\eta,p$ to obtain
constraint groups $S_1,\ldots,S_M\in E^r$. The avoidance bound
\eqref{eq:bundle-avoidance} will ensure that
any polynomial agreeing with accepted left labels on more than a $p/8$
fraction decodes to source value greater than $1-\eta$.
Put $k=2r$, the number of variable occurrences in each group.
We place the $N$ source variables in a grid of size $h^m\ge N$ and encode
their values by a polynomial of total degree at most $D=m(h-1)$. Choose
\[
 h=\lceil\log N\rceil,\qquad
 m=\left\lceil\frac{\log N}{\log h}\right\rceil,\qquad D=m(h-1).
\]
Work over $\F=\operatorname{GF}(q)$, where $q$ is a power of two
with $q>D$, $q>k$, and $q\ge\max\{h,|\Sigma_{\mathrm{src}}|\}$; further lower bounds
on $q$ are specified below.
Fix $H\subseteq\F$ of size $h$, and an injection $\iota:X\hookrightarrow H^m$.
Fix an injection $\jmath:\Sigma_{\mathrm{src}}\hookrightarrow\F$ and a
map $\beta:\F\to\Sigma_{\mathrm{src}}$ with
$\beta(\jmath(\sigma))=\sigma$ for every $\sigma\in\Sigma_{\mathrm{src}}$;
outside the image of $\jmath$, use a fixed default symbol.
Tensor-product interpolation, with the unused grid positions assigned zero,
extends every encoded assignment $\jmath\circ a$ to a polynomial of total
degree at most $D$.  Let
$\mathcal P_D=\{P\in\F[x_1,\ldots,x_m]:\deg P\le D\}$ be the finite
space of such polynomials.  For $P\in\mathcal P_D$, define the assignment
$a_P:X\to\Sigma_{\mathrm{src}}$ by
\begin{equation}\label{eq:decoded-source-assignment}
 a_P(x)=\beta(P(\iota(x))).
\end{equation}
Since $D<q$, different formal polynomials in $\mathcal P_D$ define different
functions.

Choose distinct $t_1,\ldots,t_k\in\F$, and write
$T_0=\{t_1,\ldots,t_k\}\subseteq\F$. Define $Z\in\F[t]$ by
$Z(t)=\prod_{s=1}^k(t-t_s)$.
For each group $S_j$, list the embedded variable locations of its $r$ constraints,
retaining all repetitions, as $z_{j,1},\ldots,z_{j,k}\in H^m\subseteq\F^m$.
Interpolate a vector-valued polynomial $C_j:\F\to\F^m$ of degree less than $k$ with
$C_j(t_s)=z_{j,s}$.  For $a\in\F^m$ and $v\in\F^m\setminus\{0\}$, define
$\Gamma_{j,a,v}:\F^2\to\F^m$ by
\begin{equation}\label{eq:surface-patch}
 \Gamma_{j,a,v}(t,u)=C_j(t)+Z(t)a+uv.
\end{equation}
Its designated values are $\Gamma_{j,a,v}(t_s,0)=z_{j,s}$.
For every fixed $j$ and $t\notin T_0$, varying $u$ traces a uniform affine
line when $a,v$ are uniform: $C_j(t)+Z(t)a$ is a uniform point, and $v$ is an
independent uniform nonzero direction.

The initial instance $J_b$ has vertex sets
$U(J_b)=[M]\times\F^m\times(\F^m\setminus\{0\})$ and $V(J_b)=\F^m$.
Write $\omega=(j,a,v)\in U(J_b)$ for a left vertex and $z\in V(J_b)$
for a right vertex.  A left label $B_\omega\in\F[t,u]$ is a bivariate polynomial with
$\deg_t B_\omega\le kD$ and $\deg_u B_\omega\le D$, represented by its
coefficient vector.  A right label belongs to $\F$.  For each
$(t,u)\in(\F\setminus T_0)\times\F$, add an edge incidence from $\omega$
to $\Gamma_\omega(t,u)$ projecting $B_\omega$ to $B_\omega(t,u)$.
Distinct parameter pairs give distinct incidences even when their endpoints
coincide.  The predicate $Q_{b,\omega}$ checks all $r$ constraints in
$S_j$: for its $\ell$th edge $e$, it requires
\[
 \bigl(\beta(B_\omega(t_{2\ell-1},0)),
       \beta(B_\omega(t_{2\ell},0))\bigr)\in b_e.
\]
The graph, alphabets, and projections are independent of $b$.

Writing $Q=q^m$ and denoting the left and right degrees of $J_b$ by
$d_L,d_R$, we have
\begin{equation}\label{eq:surface-counts}
 A:=|U(J_b)|=MQ(Q-1),\quad |V(J_b)|=Q,\quad
 d_L=q(q-k),\quad d_R=M(Q-1)d_L.
\end{equation}
Indeed, for fixed $z,j,v,t,u$ with $t\notin T_0$, the equation
$\Gamma_{j,a,v}(t,u)=z$ has exactly one solution $a$.  Thus the graph is
biregular and has no isolated vertices.  If $a_*$ satisfies $\Phi_b$, let
$P$ interpolate $\jmath\circ a_*$.  The labels $B_\omega=P\circ\Gamma_\omega$
and $f=P$ satisfy $J_b$: the composition has bidegree at most $(kD,D)$,
and every designated symbol equals $\jmath(a_*(x))$, so applying
$\beta$ recovers $a_*(x)$.  This proves completeness.

We now choose the field size. The line test requires the first bound
below; the remaining bounds ensure that evaluations determine the
polynomials and that distinct candidates rarely agree. Take the least
power of two $q$ satisfying
\begin{equation}\label{eq:surface-field-size}
 \begin{gathered}
 q\ge C_0((D+1)/\kappa)^{C_1},\qquad q\ge\max\{h,|\Sigma_{\mathrm{src}}|\},\qquad q>D,\qquad q-k>kD,\\
 \frac{L_{\max}kD}{q-k}\le\frac{\tau}{8},\qquad
 \frac Dq\le\frac{p^2}{1000}.
 \end{gathered}
\end{equation}
Here $C_0,C_1$ are the constants in \Cref{prop:as-line-list}.
The bound $L_{\max}kD/(q-k)\le\tau/8$ controls agreement after fixing $t$:
if a bivariate left label differs from all restrictions of the at most
$L_{\max}$ candidate global polynomials, it can agree identically in $u$
with one of these restrictions for at most a $\tau/8$ fraction of the
allowed $t$-values. The bound $D/q\le p^2/1000$ limits
the fraction of points on which two distinct global polynomials agree.
The bounds $q>D$ and $q-k>kD$ ensure that evaluations determine the
global and bivariate polynomials, respectively.
These requirements give $q=\poly(D,k,\gamma^{-1})$.

The inverse-polylogarithmic lower bound on $\gamma$ gives
\begin{equation}\label{eq:surface-domain-size}
 q=\log^{O(1)} N,\qquad
 q^m=\exp\!\left(O(\log\log N)\cdot
                  O\!\left(\frac{\log N}{\log\log N}\right)\right)
     =N^{O(1)}.
\end{equation}

\begin{lemma}[From line agreement to bivariate restrictions]\label{lem:surface-list}
For every right labeling $f:\F^m\to\F$ there are
$P_1,\ldots,P_L\in\mathcal P_D$, with $L\le L_{\max}$, such that, for
every assignment of left polynomials,
\begin{equation}\label{eq:surface-list-error}
 \Pr_{\omega,t,u}\bigl[
 B_\omega(t,u)=f(\Gamma_\omega(t,u)),\quad
 B_\omega\notin\{P_i\circ\Gamma_\omega:i\in[L]\}\bigr]
 \le\tau/4.
\end{equation}
Here $\omega$ is uniform on the left side, $t$ is uniform outside $T_0$,
and $u$ is uniform in $\F$.  The condition on $B_\omega$ refers to
equality of bivariate polynomials.
\end{lemma}
\begin{proof}
Use the list from \Cref{prop:as-line-list} with parameter $\kappa$.
For fixed $j,t\notin T_0$, varying $u$ traces a uniform affine line, and
$B_\omega(t,\cdot)$ has degree at most $D$.  Even if different left vertices
supply different polynomial answers on the same geometric line, their agreement when
outside the list is bounded by $\beta_f(\ell)$ after changing to a canonical
line parameterization.  Thus \eqref{eq:as-max-outside-list}, averaged over
$j,t$, bounds the probability of agreement with $f$ when the line polynomial is outside the list
by $2\kappa=\tau/8$.

Fix $\omega$ for which $B_\omega$ differs from every $P_i\circ\Gamma_\omega$.
For each $i$, the nonzero bivariate polynomial
$B_\omega-P_i\circ\Gamma_\omega$ has at least one nonzero coefficient when
expanded in powers of $u$.  That coefficient has degree at most $kD$ in $t$.
Consequently, for this $i$, the two polynomials become identical as
polynomials in $u$ for at most $kD$ values of $t$.
A union bound gives at most $L_{\max}kD/(q-k)\le\tau/8$ for the
fraction of $t$-values at which the restricted polynomial is in the list.
Adding the two bounds proves the claim.
\end{proof}

\subsection{Scores for Candidate Polynomials and Lipschitz Bounds for Recovery}

For each global polynomial $P\in\mathcal P_D$, we score the fraction of left vertices
$\omega$ for which $B_\omega$ satisfies the left predicate
$Q_{b,\omega}$ and equals $P\circ\Gamma_\omega$.  The predicate checks
all $r$ constraints of $\Phi_b$ in $S_j$, where $\omega=(j,a,v)$.
These scores will
identify assignments with high value on $\Phi_b$ and bound the change
in recovery when either the labels or the relations $b_e$ change.
For a left labeling $B=(B_\omega)_{\omega\in U(J_b)}$, define
$s_{b,B}:\mathcal P_D\to[0,1]$ by
\begin{equation}\label{eq:full-patch-score}
 s_{b,B}(P)=\frac1A\sum_\omega
 \mathbf1\bigl[Q_{b,\omega}(B_\omega)=1,\quad
                B_\omega=P\circ\Gamma_\omega\bigr].
\end{equation}
The recovery will use these scores on all of $\mathcal P_D$, rather than
restricting its candidates to the list determined by the right labeling
$f$. That list supplies a high-scoring candidate when the labeling has
large value. For sensitivity, we separately bound the number of
high-scoring polynomials in the entire space $\mathcal P_D$, so that the
perturbation bound does not depend on the size of this space. Neither
the scores nor the recovery rule use $f$.

\begin{lemma}[Score guarantees]\label{lem:surface-scores}
For every initial instance $J_b$ and every labeling $(B,f)$ of $J_b$,
the following hold.
\begin{enumerate}
 \item If $\val_{J_b}(B,f)\ge\tau$, some $P\in\mathcal P_D$ has
       $s_{b,B}(P)\ge p$.
 \item Every $P\in\mathcal P_D$ with $s_{b,B}(P)>p/8$ satisfies
       $\val_{\Phi_b}(a_P)>1-\eta$.
 \item The number of $P\in\mathcal P_D$ with $s_{b,B}(P)\ge p/4$ is at
       most $B_{\rm list}:=\lceil20/p\rceil$.
\end{enumerate}
\end{lemma}
\begin{proof}
For the first assertion, let $P_1,\ldots,P_L$ be the list supplied by
\Cref{lem:surface-list} for the right labeling $f$.
At least a $3\tau/4$ fraction of all edges are satisfied and have a left
endpoint $\omega$ with $B_\omega=P_i\circ\Gamma_\omega$ for some $i\in[L]$.
Since each left vertex has the same degree,
\[
 \sum_{i=1}^L s_{b,B}(P_i)\ge3\tau/4,
 \qquad
 \max_i s_{b,B}(P_i)\ge\frac{3\tau}{4L_{\max}}\ge p.
\]
If a left vertex $\omega=(j,a,v)$ contributes to $s_{b,B}(P)$, the assignment
$a_P$ satisfies every constraint in $S_j$.  Every group $S_j$ indexes the same number
$Q(Q-1)$ of left vertices.  If $a_P$ violated at least $\eta |E|$ constraints,
\eqref{eq:bundle-avoidance} would therefore imply $s_{b,B}(P)\le p/8$.
This proves the first two assertions.

For distinct $P,P'\in\mathcal P_D$, fix any $t_*\notin T_0$.  Equality of their full
restrictions to $\Gamma_{j,a,v}$ implies equality at
$C_j(t_*)+Z(t_*)a$.  This is a uniform point of $\F^m$ when $a$ is uniform,
so Schwartz--Zippel gives
\begin{equation}\label{eq:surface-collision}
 \Pr_\omega[P\circ\Gamma_\omega=P'\circ\Gamma_\omega]
 \le D/q\le p^2/1000.
\end{equation}
Let $A_P\subseteq U(J_b)$ be the set of left vertices counted in \eqref{eq:full-patch-score} and
put $\alpha=p/4$.  Suppose there is a set $\mathcal C\subseteq\mathcal P_D$
of $K=\lceil4/\alpha\rceil$ candidates, each with score at least $\alpha$.
The first two inclusion--exclusion terms give
\[
 1\ge\frac1A\left|\bigcup_{P\in\mathcal C} A_P\right|
 \ge K\alpha-\binom K2\frac{p^2}{1000}>1.
\]
The last inequality uses $K\alpha\ge4$ and $Kp<17$.  This contradiction
bounds the number of candidates by $K\le20/p$, as required.
\end{proof}

Fix a total order on $\mathcal P_D$ independently of $b$.  A recovery
trial on a left labeling $B$ samples $\theta$ uniformly from $[p/2,p]$
and outputs $R_b(B;\theta)\in\mathcal P_D\cup\{\bot\}$: the first
polynomial $P$ with $s_{b,B}(P)\ge\theta$, or $\bot$ if none exists.
By \Cref{lem:surface-scores}, for every $B$, each nonfailure output $P$
satisfies $\val_{\Phi_b}(a_P)>1-\eta$. If some right labeling together
with $B$ gives value at least $\tau$, the trial never fails.

\begin{lemma}[Perturbation bounds for recovery from one labeling]\label{lem:surface-perturbation}
Using the same uniform threshold $\theta$ in both runs, for fixed $b$
and arbitrary left labelings $B,B'$, and for $b,b'$ differing in one
relation with $B$ fixed, respectively, we have
\begin{align}
 N\Pr[R_b(B;\theta)\ne R_b(B';\theta)]
 &\le C_{\mathcal R}^{\rm one}d_{\rm H}(B,B'),
 &C_{\mathcal R}^{\rm one}&:=\frac{8NB_{\rm list}}{pA},
 \label{eq:surface-one-cr}\\
 N\Pr[R_b(B;\theta)\ne R_{b'}(B;\theta)]
 &\le D_{\mathcal R}^{\rm one},
 &D_{\mathcal R}^{\rm one}&:=\frac{8NB_{\rm list}r}{p|E|}.
 \label{eq:surface-one-dr}
\end{align}
Changing one relation $b_e$ changes at most $C_{\mathcal T}:=rA/|E|$ left predicates of $J_b$, and
\begin{equation}\label{eq:surface-one-product}
 C_{\mathcal T}C_{\mathcal R}^{\rm one},\ D_{\mathcal R}^{\rm one}
 \le O(r/p^2).
\end{equation}
\end{lemma}
\begin{proof}
Let $s,s':\mathcal P_D\to[0,1]$ be the two score functions being compared, with
$\|s-s'\|_\infty\le\nu$. By \Cref{lem:surface-scores}, at most
$2B_{\rm list}$ polynomials have score at least $p/4$ under either
function. Apply the probability bound in
\Cref{lem:random-threshold-selection} to the two score lists in the same
polynomial order, with index $0$ representing $\bot$ and threshold
interval $[p/2,p]$. Only those polynomials can have score intervals
meeting $[p/2,p]$. Different selected indices give different trial
outputs, so multiplying the probability bound by $N$ gives
\[
 N\Pr[\text{trial outputs differ}]\le4NB_{\rm list}\nu/p.
\]

Changing $d_{\rm H}(B,B')$ left labels changes each score by at most
$d_{\rm H}(B,B')/A$.
An edge $e\in E$ occurs in at most $rM/|E|$ groups, so changing $b_e$
changes at most $rA/|E|$ left predicates and each score by at most $r/|E|$.
These observations imply the stated (slightly looser) constants. Finally $N/|E|=2/d_0$
and $B_{\rm list}=O(1/p)$ give \eqref{eq:surface-one-product}.
\end{proof}

\subsection{Degree Reduction and Recovery from Independent Samples}

We reduce the right degree while preserving recovery from the left labels.
We then repeat recovery on independent input samples and return the
assignment from the first successful trial.

Choose an even $d_\gamma=\Theta(\gamma^{-2})$ for which
\Cref{pro:ramanujan} gives error $\lambda_\gamma\le\gamma/4$, and set
\[
 \widehat J_b=\mathcal T_{\rm DR}(J_b;d_\gamma).
\]
All expanders and incidence orderings are fixed across $b$.
Left labels and left predicates are unchanged.  We apply $R_b$ directly
to the left part of a labeling of $\widehat J_b$, ignoring every right label.

For a fixed labeling $Y$ of $\widehat J_b$,
\eqref{eq:degree-reduction-pointwise-value} gives
\[
 \val_{\widehat J_b}(Y)
 \le \E\val_{J_b}(R_{\mathrm{DR},J_b}(Y))+\lambda_\gamma.
\]
The degree-reduction recovery preserves all left labels.  Thus if
$\val_{\widehat J_b}(Y)\ge\gamma/2$, some right labeling of $J_b$ gives value at
least $\gamma/4=\tau$ with those same left labels.  Applying $R_b$ to these
left labels therefore returns a polynomial rather than $\bot$ almost surely.
We therefore need neither to execute $R_{\mathrm{DR},J_b}$ nor to use its
generic Lipschitz bound. Also, since degree reduction leaves left predicates
unchanged, changing $b_e$ still changes at most $rA/|E|$ left predicates of
$\widehat J_b$.

Fix a default assignment $a_0:X\to\Sigma_{\mathrm{src}}$ independently of
$b$. For any distribution $\mu$ over full labelings of $\widehat J_b$, take
\[
 R_{\rm amp}=\left\lceil\frac2\gamma\log\frac4{\eta_{\rm src}}\right\rceil
\]
independent samples $Y_1,\ldots,Y_{R_{\rm amp}}$ from $\mu$. Write $B_j$
for the left part of $Y_j$, and independently sample uniform thresholds
$\theta_j\in[p/2,p]$ to obtain $P_j=R_b(B_j;\theta_j)$. Return $a_{P_j}$
for the first $j$ with $P_j\ne\bot$, or $a_0$ if all trials fail.
This defines $\mathcal R_{b,\gamma}(\mu)$ on all input distributions.
The failure symbol is distinct from every polynomial, including those
that decode to $a_0$.

If $\E_{Y\sim\mu}\val_{\widehat J_b}(Y)\ge\gamma$, then
\[
 \Pr_{Y\sim\mu}[\val_{\widehat J_b}(Y)\ge\gamma/2]
 \ge\frac{\gamma}{2-\gamma}\ge\gamma/2.
\]
Each trial thus returns a polynomial with probability at least $\gamma/2$,
and independence bounds the probability of all trials failing by
$\exp(-\gamma R_{\rm amp}/2)\le\eta_{\rm src}/4$.
Every polynomial returned by a trial decodes to value greater than
$1-\eta$, so the output has expected value at least
\[
 (1-\eta_{\rm src}/4)(1-\eta)>1-\eta_{\rm src}.
\]

For two input distributions with $b$ fixed, take $R_{\rm amp}$
independent copies of an optimal input coupling and share the thresholds
within each pair. If all paired trial outputs agree in
$\mathcal P_D\cup\{\bot\}$, the final assignments agree. Their Hamming
distance is at most $N$ otherwise. Restriction to left labels contracts
Hamming distance, so a union bound and \Cref{lem:surface-perturbation}
give the first bound below.
When one relation changes, use the same input samples and shared
thresholds; the same argument gives the second bound:
\begin{equation}\label{eq:surface-law-transfer}
 C_{\mathcal R}\le R_{\rm amp}C_{\mathcal R}^{\rm one},\qquad
 D_{\mathcal R}\le R_{\rm amp}D_{\mathcal R}^{\rm one}.
\end{equation}
Together with \eqref{eq:surface-one-product}, $R_{\rm amp}=O(\gamma^{-1})$,
$p=\Theta(\gamma^2)$, and $r=O(\log(2/\gamma))$, this proves
\eqref{eq:surface-initializer-transfer}.

\paragraph{Structural parameters and local description.}
For $\widehat J_b$, we have
\begin{equation}\label{eq:surface-reduced-structure}
 \begin{gathered}
 |U(\widehat J_b)|=A,\qquad |V(\widehat J_b)|=Ad_L,\\
 \deg_U=d_\gamma d_L,\qquad \deg_V=d_\gamma,\qquad
 \Sigma_V=\F,\qquad
 |\Sigma_U|=q^{(kD+1)(D+1)}.
 \end{gathered}
\end{equation}
By \eqref{eq:surface-domain-size} and \eqref{eq:surface-counts}, the
number of vertices on each side is $N^{O(1)}$. Both degrees, the right
alphabet size, and the logarithm of the left alphabet size are
$\poly(\log N,\gamma^{-1})$.

For the local description complexity of
\Cref{def:local-description-instance}, we must test whether a tuple of
incident right labels at $\omega$ is the projection tuple of a left label satisfying
$Q_{b,\omega}$.  On the query
grid $(\F\setminus T_0)\times\F$ of $J_b$, the evaluations uniquely determine a
polynomial of bidegree at most $(kD,D)$, since $q-k>kD$ and $q>D$.
Interpolate from a fixed $(kD+1)\times(D+1)$ subgrid, check agreement on
all remaining grid positions, and check the $r$ constraints of $\Phi_b$ in
$S_j$ on the recovered polynomial, where $\omega=(j,a,v)$.
All interpolation matrices and points are hardwired.
The map $\beta$ compares a field element with the constantly many encoded
symbols, and each relation $b_e$ is a constant-size table.  Together with
finite-field arithmetic, these tests give a circuit of size
$\poly(q,k,D,\log q)$, independent of the large right degree.
After degree reduction, first check that the right labels at the endpoints
of the $d_\gamma$ copies of each incidence of $J_b$ agree, then apply this
circuit to the common label from each group of copies. The resulting size
is $\poly(q,d_\gamma,k,D,\log q)=\poly(\log N,\gamma^{-1})$.

Finally, for $b,b'$ differing in one relation, a path of predicate swaps
between $\widehat J_b$ and $\widehat J_{b'}$ has length at most $rA/|E|$.
Each local predicate on the path is chosen from
one endpoint and retains the common projections, so the same circuit bound
holds along the path.  Completeness follows by copying honest right labels
through the degree reduction.  This completes the proof of
\Cref{thm:ldrc-free-initializer}.
 \section{Query-Oblivious Decodable PCPs and Composition}
\label{sec:dh-composition}

At a left vertex of a Label Cover instance $I$, we must check whether the
neighboring right labels are the projections of a left label satisfying
its predicate.
We express this test as a circuit and replace it by smaller tests on an
auxiliary proof. These tests use a decoder that queries a small number of
proof symbols to recover a requested input symbol or reject.

For sensitivity transfer, the proof locations read by the decoder must
stay fixed when the circuit being checked changes. We obtain this property,
called query-obliviousness, by encoding the circuit description and
confining its effect to one local comparison. We also bound the size of
each decoding circuit. Using one right vertex per proof symbol then
produces a Label Cover instance with smaller local checking circuits and
a graph that is fixed across the family of input instances $I$. These properties are
needed to repeat the construction in \Cref{sec:dh-iteration}.
\Cref{fig:composition-overview} shows the new vertices and the
recovery dependencies: new right labels form a separate proof for each
old left vertex, while a new left label supplies answers to tests that
decode the same old right label.

\subsection{Decoder Guarantees Used in Composition}

We first state the decoder guarantees, then give the composition and
recovery, and finally prove the decoder lemmas. For a finite alphabet
$\Sigma$, the input to the decoder is a circuit
$C:\Sigma^k\to\{0,1\}$, with the encoding convention in \Cref{sec:pre}.

\begin{definition}[PCP decoder]\label{def:pcp-decoder}
A nonadaptive PCP decoder has a proof alphabet $\sigma$, proof length $m$,
and a finite seed set $\mathscr R$.  Given $C:\Sigma^k\to\{0,1\}$, an
index $j\in[k]$, and a uniform seed $\rho\in\mathscr R$, it produces a
query tuple $I_{C,j,\rho}\in[m]^q$ and a local circuit
$g_{C,j,\rho}:\sigma^q\to\Sigma\cup\{\bot\}$.
On a proof $p\in\sigma^m$, its output is
$g_{C,j,\rho}(p|_{I_{C,j,\rho}})$. For given size and alphabet parameters,
the local circuit complexity $s$ is the maximum encoding length of
$g_{C,j,\rho}$ over all admissible $C,j,\rho$.
A decoder is \emph{query-oblivious} if $I_{C,j,\rho}$ depends on $C$ only
through its input arity and the common size and alphabet parameters.
We use uniform finite seed sets to avoid rounding their cardinalities to
powers of two. A uniform seed can be sampled using
$\log|\mathscr R|+O(1)$ random bits in expectation.
\end{definition}

\begin{definition}[Robust decodable PCP]\label{def:dpcp}
For $w\in\sigma^q$ and $B\subseteq\sigma^q$, define
\[
 \operatorname{agr}(w,B):=
 \max_{z\in B}\frac{|\{h\in[q]:w_h=z_h\}|}{q},
\]
with agreement zero when $B=\varnothing$.  Query occurrences, including
repetitions of a proof position, are counted with multiplicity.
For a local circuit $g:\sigma^q\to\Sigma\cup\{\bot\}$, a set
$Y\subseteq\Sigma^k$, and $j\in[k]$, put
\[
 \operatorname{Bad}_Y(g,j)
 :=\big \{w\in\sigma^q:g(w)\notin\{\bot\}\cup\{y_j:y\in Y\} \big\}.
\]
Thus $\operatorname{Bad}_Y(g,j)$ consists of local answer words that
do not reject and decode to a symbol different from $y_j$ for every
$y\in Y$.
The quantity $\operatorname{agr}(w,\operatorname{Bad}_Y(g,j))$ is the
largest fraction of coordinates that $w$ shares with any such word;
it can be positive even when $g(w)=\bot$.

The decoder has perfect completeness if, for every $y$ with $C(y)=1$,
there is a proof that decodes $y_j$ for every $j,\rho$.  It has robust
soundness error $\delta$ and list bound $L$ if, for every proof $p$, there
is a possibly empty set $Y\subseteq C^{-1}(1)$ with $|Y|\le L$ such that,
for uniform independent $j\in[k]$ and $\rho\in\mathscr R$,
\begin{equation}\label{eq:robust-dpcp-soundness}
 \E_{j,\rho}\operatorname{agr}\bigl(p|_{I_{C,j,\rho}},
                    \operatorname{Bad}_Y(g_{C,j,\rho},j)\bigr)
 \le\delta.
\end{equation}
\end{definition}

In the composition of \Cref{sec:composition}, a left label supplies local
answer words that can match some of the queried proof symbols without
matching the whole query word.  The agreement bound controls the
contribution of these partial matches with words in
$\operatorname{Bad}_Y(g,j)$.
Robust soundness also implies ordinary list-decoding soundness for the same list $Y$:
\begin{equation}\label{eq:decodable-PCP-soundness}
 \Pr_{j,\rho}
 \big[g_{C,j,\rho}(p|_{I_{C,j,\rho}})\notin
       \{\bot\}\cup\{y_j:y\in Y\}\big]\le\delta.
\end{equation}
The following decoder supplies all the guarantees used in composition.
We derive it from the algebraic dPCP of
\cite[Theorem~6.5 and Section~6.1.2]{dinur2013composition}:
\Cref{sec:decoder-local-circuits} bounds the local circuit size, and
\Cref{sec:query-oblivization} makes the query locations independent of
which circuit is being checked.

\begin{lemma}[Query-oblivious robust decoder]\label{lem:dinur-harsha-decoder}
There are absolute constants $\alpha_*,\gamma_*,A_*>0$ such that, for
all sufficiently large $M$, all $M^{-\alpha_*}\le\delta\le1/2$, and
$|\Sigma|\le M^{\gamma_*}$, circuits $C:\Sigma^k\to\{0,1\}$ of
complexity at most $M$ have a query-oblivious robust dPCP with perfect
completeness and
\begin{equation}\label{eq:oblivious-decoder-parameters}
 L\le4/\delta,\qquad q,s\le M^{1/32},\qquad
 m,\ R:=|\mathscr R|,\ |\sigma|\le M^{A_*}.
\end{equation}
The parameters $L,q,s,m,\mathscr R,\sigma$ depend only on
$M,\delta,\Sigma$. The only dependence of the local circuit on $C$ is a
comparison with one symbol of a fixed encoding of its description.
\end{lemma}
The proof is given in \Cref{sec:query-oblivization}, after we use the
decoder to construct the composed Label Cover instance and its recovery.
The query bound $q$ controls the length of a new left label, and the
local circuit bound $s$ controls the cost of checking it. These are the
bounds needed for the next composition step; $m$ and $R$ control the
number of new vertices. Robust soundness with parameters $L,\delta$
will give recovered value at least $(\text{composed value}-\delta)/L$.

\subsection{Composition with One Proof Symbol per Right Vertex}
\label{sec:composition}

We use the Dinur--Harsha composition rule
\cite[Section~4]{dinur2013composition}, which compares tests decoding
the same old right label. We give the construction in Label Cover form
and bound changes to its instances and recovered labelings. For each left
vertex $u\in U$, with $k_u=\deg_I(u)$, let
$C_u:\Sigma_V^{k_u}\to\{0,1\}$ recognize $\mathcal N_I(u)$, the set of projection
tuples of left labels satisfying $P_u$, as defined in
\Cref{def:local-description-instance}.
Using one right vertex per proof symbol will make the local description
complexity depend on the number of queries rather than the proof length.

Fix such a family with common graph and alphabets, positive degrees, and local
circuits of complexity at most $M$.  Let $d_U$ and $d_V$ bound its left
and right degrees, respectively.  Let $\mathcal D$ be the decoder of
\Cref{lem:dinur-harsha-decoder}, or any query-oblivious robust decoder.
For an incidence $e=(u,v)$ of $I$, let $j(e)$ be its index at $u$.  Write
$I_{u,j,\rho}\in[m]^q$ for the query tuple, which can depend on $k_u$ but not on
$C_u$.  Define $I\circledast\mathcal D$ as follows.
\begin{enumerate}
 \item Its left vertex set is $U'=V\times\mathscr R$.
 Start with right vertices $(u,z)\in U\times[m]$ and delete positions $z$
 that are never queried at $u$, using the same deletion throughout the
 family.  Let $V'\subseteq U\times[m]$ be the remaining right vertex set.
 Deleted proof symbols are filled by a fixed default in recovery.
 \item Put $\Sigma'_V=\sigma$ and
 $\Sigma'_U=\sigma^{d_Vq}$.  A left label at $(v,\rho)$ has one row
 $w_e\in\sigma^q$ per incidence $e=(u,v)$ at $v$, with unused rows
 ignored.  Its predicate accepts exactly when all
 $g_{C_u,j(e),\rho}(w_e)$ agree on one non-$\bot$ symbol.
 \item For each such incidence and $h\in[q]$, put one edge from
 $(v,\rho)$ to $(u,(I_{u,j(e),\rho})_h)$, projecting to $(w_e)_h$.
 These are distinct edge incidences even if some endpoints coincide.
\end{enumerate}
Every right label is one proof \emph{symbol}; no vertex label contains
an entire proof.  The graph and projections are independent of the
checking circuits $C_u$ of $I$.

\begin{theorem}[Composition with a common output graph]\label{thm:composition}
Assume the input family has degree bounds $d_U,d_V$, positive degrees,
and common graph and alphabets, and use a query-oblivious robust decoder
with parameters $(q,m,R,\sigma,s,L,\delta)$.  The preceding transformation
has perfect completeness and an instance-specific recovery $R_I$ such
that, for every possibly random composed labeling $\bm Y$, the recovered
assignment $\bm\pi:=R_I(\bm Y)$ in $\Sigma_U^U\times\Sigma_V^V$ satisfies
\begin{equation}\label{eq:bundled-composition-pointwise-soundness}
 \E[\val_I(\bm\pi)]
 \ge \frac{\E[\val_{I\circledast\mathcal D}(\bm Y)]-\delta}{L}.
\end{equation}
If $I$ and $\widetilde I$ differ in one predicate or projection, their
composed instances differ only in at most $d_UR$ left predicates.
The transformation and recovery have sensitivity-transfer parameters satisfying
\begin{equation}\label{eq:symbol-composition-transfer}
 C_{\mathcal T,\rm Comp}=d_UR,\qquad
 C_{\mathcal R,\rm Comp}=1,\qquad
 D_{\mathcal R,\rm Comp}\le d_U+1.
\end{equation}
The structural bounds are
\begin{equation}\label{eq:symbol-composition-structure}
\begin{gathered}
 |U'|=|V|R,\quad |V'|\le|U|m,\quad |E'|=|E|Rq,\qquad
 d_{U'}\le d_Vq,\quad d_{V'}\le d_URq,\\
 |\Sigma'_V|=|\sigma|,\qquad
 \log|\Sigma'_U|=d_Vq\log|\sigma|.
\end{gathered}
\end{equation}
Both sides have positive degree after the common deletion above.  With
the decoder of \Cref{lem:dinur-harsha-decoder} and $d_V\le M$, the local
description complexity, left degree, and logarithm of the left alphabet size
are, for sufficiently large $M$, at most
\begin{equation}\label{eq:composition-local-contraction}
 C d_VM^{1/16}
\end{equation}
for an absolute constant $C$.  These local bounds do not involve the
whole proof length $m$.
\end{theorem}
\begin{proof}
Fix total orders on the finite candidate sets and label sets.  For each
$u$ and proof $p_u\in\sigma^m$, choose the first list
$Y_{I,u}(p_u)\subseteq\mathcal N_I(u)$ that satisfies robust soundness.
Define a representative map
$\operatorname{rep}_{I,u}:\mathcal N_I(u)\to\Sigma_U$ such that, for
each $y\in\mathcal N_I(u)$, the label $\operatorname{rep}_{I,u}(y)$ is
the first label in the fixed order that satisfies $P_u$ and has
projection tuple $y$.
The list and representative map depend only on $p_u$, $P_u$, and the
projections incident to $u$, so they are retained when these local data
are unchanged.

Given a composed labeling $Y$, let $p_u\in\sigma^m$ consist of its right
labels at $(u,z)$, filling deleted positions with the fixed default.
Recover the left label of $I$ at $u$ by
choosing a uniform member of $Y_{I,u}(p_u)$ and returning its
representative, or a fixed default if the list is empty.  Independently,
choose a uniform seed $\rho\in\mathscr R$.  At each $v$, apply the local
decoders to the rows of the composed left label at $(v,\rho)$.  If their
outputs agree on a non-$\bot$ symbol, recover that symbol at $v$;
otherwise use a fixed default.  This defines $R_I$ from labelings of
$I\circledast\mathcal D$ to distributions on
$\Sigma_U^U\times\Sigma_V^V$, the assignments to $I$.

\paragraph{Completeness.}
A satisfying labeling of $I$ supplies a tuple in $\mathcal N_I(u)$ at each $u$.
Choose an honest proof for that tuple and give its symbols to the right
vertices $(u,z)$.  At $(v,\rho)$, fill every row with the corresponding
queried symbols.  All decoders return the label of $I$ at $v$, every
left predicate accepts, and every projection is satisfied.

\paragraph{Soundness.}
Fix the composed labeling $Y$ and sample $\bm\pi\sim R_I(Y)$.
A uniform composed incidence is obtained by
choosing a uniform incidence $e=(u,v)$ of $I$, a uniform seed $\rho$,
and a uniform query occurrence $h$.  Let $z_{v,\rho}$ be the common
non-$\bot$ output when the left predicate at $(v,\rho)$ accepts.
Call the row $w_e$ at $(v,\rho)$ bad if the predicate at $(v,\rho)$ accepts but
$z_{v,\rho}\notin\{y_{j(e)}:y\in Y_{I,u}(p_u)\}$.
A bad row $w_e$ lies in the set
$\operatorname{Bad}_{Y_{I,u}(p_u)}$ from \Cref{def:dpcp}.
Consequently the fraction of matching queried coordinates in that row
is at most
\[
 \operatorname{agr}
 (p_u|_{I_{u,j(e),\rho}},
           \operatorname{Bad}_{Y_{I,u}(p_u)}(g_{C_u,j(e),\rho},j(e))).
\]
Conditionally on $u$, the incidence index $j(e)$ is uniform in $[k_u]$.
Robust soundness, averaged with weights $k_u/|E|$, therefore shows that
bad rows contribute at most $\delta$ to the composed value.
If the predicate at $(v,\rho)$ accepts and $w_e$ is not bad, some
candidate $y\in Y_{I,u}(p_u)$ satisfies $y_{j(e)}=z_{v,\rho}$.
The projection of the recovered left label along $e$ equals
$z_{v,\rho}$ with probability at least $1/L$.  Since a row's matching
fraction is at most one, averaging yields
\[
 \E[\val_I(\bm\pi)\mid Y]
 \ge L^{-1}(\val_{I\circledast\mathcal D}(Y)-\delta).
\]
Averaging over $Y$ proves the stated bound for expected value.  This argument
uses the robust soundness bound in \eqref{eq:robust-dpcp-soundness}, which
controls agreement with nonrejecting local words decoding outside the
list.  Ordinary list-decoding soundness alone would not bound the
partially matching bad rows.

\paragraph{Instance and recovery perturbations.}
Changing the predicate at $u$ or a projection incident to $u$ in $I$
changes only $C_u$.  It can therefore change
predicates only at $(v,\rho)$ with $v$ incident to $u$, at most $d_UR$
vertices.  Query-obliviousness keeps every edge and projection fixed.
Changing one composed right symbol changes the recovered left label
only at the corresponding vertex $u$ of $I$.  Changing one composed left label at
$(v,\rho_0)$ can change the recovered right label of $I$ only at $v$,
and only if the sampled seed is $\rho_0$.  These couplings give costs
at most one and $1/R$, respectively, and hence
$C_{\mathcal R,\rm Comp}=1$ by a coordinatewise hybrid.

If the predicate or an incident projection at $u$ in $I$ changes while
the composed labeling is fixed, use
the same sampled seed and the same candidate choices at every other
left vertex of $I$.  Only the recovered left label at $u$ and the
recovered right labels at its at most $d_U$ neighbors can change.
This gives $D_{\mathcal R,\rm Comp}\le d_U+1$, independently of $m$.

\paragraph{Local circuits and size.}
The bounds in \eqref{eq:symbol-composition-structure} follow directly from
the construction.  At a composed left vertex $(v,\rho)$, a tuple of incident right
labels belongs to $\mathcal N_{I\circledast\mathcal D}(v,\rho)$ exactly
when its rows $w_e\in\sigma^q$, indexed by incidences $e=(u,v)$ at $v$,
give the same non-$\bot$ value for every $g_{C_u,j(e),\rho}(w_e)$.
A circuit runs at most $d_V$ local
circuits of size $s$, reads $d_Vq$ proof symbols, and checks equality of
the decoded outputs.  Its encoding length is bounded by
\[
 O\bigl(d_V(s+q\log(2|\sigma|)+\log(2|\Sigma_V|))
       \log(2+d_V+q+s+\log|\sigma|)\bigr).
\]
Parallel incidences are separate input coordinates; the same test is
valid for arbitrary tuples, whether or not they arise from a global
labeling. Under \eqref{eq:oblivious-decoder-parameters}, $d_V\le M$, and
the input alphabet bound, the circuit encoding length, left degree, and
$\log|\Sigma'_U|$ are at most $C d_VM^{1/16}$ for sufficiently large $M$.
\end{proof}

\begin{table}[t]
\centering
\caption{Parameters of the composition in \Cref{thm:composition}.}
\label{tab:composition}
\begin{tabular}{@{}ll@{}}
\toprule
Quantity & Bound after composition \\
\midrule
Left / right vertices & $|V|R$ / at most $|U|m$ \\
Left / right degree & at most $d_Vq$ / at most $d_URq$ \\
Left / right alphabet & $\sigma^{d_Vq}$ / $\sigma$ \\
Expected recovered value & at least $(\text{expected composed value}-\delta)/L$ \\
Sensitivity-transfer parameters & $C_{\mathcal T}=d_UR$, $C_{\mathcal R}=1$, $D_{\mathcal R}\le d_U+1$ \\
Local description complexity (our decoder) & $O(d_VM^{1/16})$ \\
\bottomrule
\end{tabular}
\end{table}

\subsection{Local Circuits for the Algebraic Decoder}
\label{sec:decoder-local-circuits}

We now construct the decoder used above. The algebraic dPCP of
\cite[Theorem~6.5]{dinur2013composition} supplies robust soundness with
query bound $n^{1/8}$ for circuits of complexity at most $n$.
We bound the encoding length of its local decoding circuits by counting
their operations. We then choose the field size so that both the query
bound and the local circuit bound are at most $n^\nu$, for any constant
$\nu>0$. The latter bound is needed in
\eqref{eq:composition-local-contraction}.

\begin{lemma}[Query and local circuit bounds for the algebraic decoder]\label{lem:dh-small-local-circuits}
For every constant $\nu>0$, there are constants
$\alpha_\nu,\gamma_\nu,A_\nu>0$ such that, for all sufficiently large $n$,
all $n^{-\alpha_\nu}\le\delta\le1/2$, and every alphabet
$|\Sigma|\le n^{\gamma_\nu}$, circuits over $\Sigma$ of complexity at
most $n$ have a robust dPCP with perfect completeness and
\[
 L\le2/\delta,\qquad
 q\le n^\nu,\qquad s\le n^\nu,\qquad
 m,\ |\mathscr R|,\ |\sigma|\le n^{A_\nu}.
\]
Here $s$ is the local circuit complexity from
\Cref{def:pcp-decoder}.  The proof length, number of queries, seed set, and
proof alphabet can be chosen uniformly for the size and alphabet parameters;
the input arity may vary.  Query-obliviousness is not asserted.
\end{lemma}
\begin{proof}
Use the algebraic decoder in
\cite[Theorem~6.5 and Section~6.1.2]{dinur2013composition}.
Write $\F$ for its field, $Q=|\F|$ for the field size, $a$ for the
ambient dimension, and $d$ for the encoding degree.  The proof is a
table of values of a function $f:\F^a\to\F$, intended to be a
low-degree encoding of a satisfying input together with the auxiliary
proof used by the algebraic verifier.  Input coordinate $j$ is represented
by a designated point $x_j\in\F^a$.

For each requested coordinate $j$ and seed, the decoder reads values at
the points of a polynomial map $\Psi:\F^b\to\F^a$, where $b$ is an
absolute constant.
The image contains $x_j$, the affine plane used for the low-degree test,
and the curve containing the queries of the underlying algebraic
verifier.  This parameterized image is the \emph{manifold} in the cited
construction: it is formed by taking scalar linear combinations of
points on the plane and curve, and then combining these with $x_j$.
The query coordinates are indexed by its prescribed
parameter grid; distinct parameters may query the same proof position.
The parameterized formulation is described in
\cite[full version, Section~10.3, Construction~2]{moshkovitz2008two}:
edges are indexed by parameter tuples, including when two tuples have
the same image. We keep these query occurrences separately, so the
agreement in \Cref{def:dpcp} uses the same sampling of parameters.
The robust guarantee is the agreement inequality immediately after
Eq.~(6.1) in \cite[Section~6.1.2]{dinur2013composition}.
That inequality precedes the modification making query frequencies
equal. We omit this modification: our composition averages over
incidences at each old left vertex and does not require equal query
frequencies at proof positions.

To bound the local circuit complexity, fix the input circuit and seed.
The local computation has three parts.
\begin{enumerate}
 \item Test membership in the local polynomial space prescribed by the
 construction.  It contains $P\circ\Psi$ for encoding polynomials
 $P\in\F[x_1,\ldots,x_a]$ of degree at most $d$.  The individual degrees
 of its members in the $b$ parameters are at most
 $\Delta=\poly(a,d,\log n,\delta^{-1})$.
 The query grid uses coordinate sets obtained from $\F$ by excluding
 at most $\poly(a,d,\log n,\delta^{-1})$ prescribed values.  Choose $Q$ large
 enough that each coordinate set has more than $\Delta$ elements.
 Tensor-product interpolation then uniquely determines the local
 polynomial from these evaluations.  This polynomial lower bound on
 $Q$ is included in the field-size condition below.
 Fix a basis and an evaluation matrix for the permitted polynomial
 space.  Row
 reduction, done when constructing the local test, produces a parity-check
 matrix and a linear interpolation map.  Matrix--vector multiplication
 tests membership and recovers coefficients using $Q^{O(1)}$ field
 operations.  Write $T\in\F[t_1,\ldots,t_b]$ for the recovered
 polynomial in the $b$ parameters. For repeated proof positions, also
 check that their supplied values agree. These linear equations preserve
 honest answers and only remove accepting local words, so they cannot
 increase agreement with the set of bad local words.
 \item The construction prescribes parameter slices whose images under
 $\Psi$ are the plane and curve used by the verifier in
 \cite[Section~6.1.1]{dinur2013composition}.  Apply these parameter
 substitutions to $T$, giving a bivariate plane polynomial
 $T_{\rm pl}\in\F[s_1,s_2]$.  Let $d_{\rm pl}$ be the
 total-degree bound prescribed for plane answers by the base low-degree
 test; this bound is polynomial in $a,d,\log n,\delta^{-1}$ and may
 differ from the encoding degree $d$.
 Check that every monomial of total degree greater than $d_{\rm pl}$
 has coefficient zero.  The evaluation test in the first step already
 enforces agreement with all supplied point values.
 For the curve test, let $z_1,\ldots,z_t\in\F^a$ be the query points
 of the underlying nonrecursive algebraic verifier for this seed.
 Its construction puts these points on a curve
 $\Gamma:\F\to\F^a$, with $\Gamma(\lambda_i)=z_i$ at prescribed
 parameters $\lambda_i\in\F$; see \cite[Appendix~B.2]{arora2003improved}.
 Write $T_\Gamma:\F\to\F$ for the univariate polynomial obtained
 from $T$ by the parameter substitution for $\Gamma$.
 Evaluate $T_\Gamma(\lambda_i)$ and apply
 that verifier's acceptance test to the resulting $t$ answers.
 In particular, any values encoding sum-check messages are obtained
 from these evaluations of the supplied local polynomial, rather than
 from an additional transcript.
 Here $t$ and the degree of $\Gamma$ are polynomial in $a,d,\log n$.
 A degree-$d$ encoding restricted to $\Gamma$ has degree at most
 $d\deg\Gamma$.  The allowed degree of $T_\Gamma$ is bounded by
 a fixed polynomial in $a,d,\log n,\delta^{-1}$, so the evaluations
 also have polynomial cost.
 The nonrecursive acceptance test consists of polynomially many field
 equalities, polynomial evaluations, and sums over an encoding grid
 contained in $\F$.  Its cost is
 $\poly(a,d,Q,\log n,\delta^{-1})$ field operations.
 All circuit- and seed-dependent data are computed while constructing
 the test and are constants of the local circuit.  Thus the test acts
 on the recovered local answers; it does not search for a global
 satisfying extension.
 \item The construction also specifies a parameter tuple
 $\xi_j\in\F^b$ with $\Psi(\xi_j)=x_j$.  Evaluate $T(\xi_j)$ and test
 that the answer is in the input alphabet.  Evaluation is a fixed linear
 map on the coefficients of $T$.  We identify the alphabet with an
 initial segment of the field's fixed binary representation, so its membership test has
 complexity polynomial in $\log Q$.
\end{enumerate}
The parameter grid has $Q^{O(1)}$ points.  Under the field-size conditions
below, $d,\delta^{-1}\le Q$.  With these inequalities and
$a=O(\log n)$, there is therefore an absolute constant $c$,
independent of the field-size exponent chosen below, for which both the
number of queries and the binary encoding length of the local circuit are
bounded by
\[
 Q^c\log^c n.
\]
Field operations have Boolean circuits of size polynomial in $\log Q$;
encoding their wiring only increases the absolute constant $c$.  This
estimate concerns operations on the supplied local answers.  It is
stronger than, and is not inferred from, polynomial-time construction of
the whole verifier.

We next choose $Q$ so that the number of queries and local circuit
complexity are at most $n^\nu$ while allowing
$\delta\ge n^{-\alpha_\nu}$ for some $\alpha_\nu>0$.
The low-degree and algebraic-verifier conditions
in \cite[Theorem~17 and Appendix~B.2]{arora2003improved} require the field
size to exceed a fixed polynomial in $a,d,\log n,\delta^{-1}$.
The sampling and root-counting errors in
\cite[Section~6.1.2, Eq.~(6.1)]{dinur2013composition} impose conditions
of the same form, as do the interpolation requirements above.
Hence, for one sufficiently large absolute $K$,
all these conditions are met if
\begin{equation}\label{eq:decoder-field-size}
 Q\ge\bigl(a(d+1)\log(2n)\delta^{-1}\bigr)^K.
\end{equation}
For an encoding grid $H\subseteq\F$ with $h=|H|$, the grid $H^a$
has size $h^a=n^{O(1)}$ and the encoding has total degree $O(ah)$.
Thus, for any constant $\beta>0$, choosing the ambient dimension $a$ as a
sufficiently large constant permits $d\le n^\beta$.
Choose a constant $\theta>0$ with $c\theta<\nu/2$ and a power of two
$Q\in[n^\theta,2n^\theta]$.  Take $\beta,\alpha_\nu>0$ with
$K(\beta+\alpha_\nu)<\theta/2$, and then choose $a$ as above.
For $\delta\ge n^{-\alpha_\nu}$, \eqref{eq:decoder-field-size}
holds for all sufficiently large $n$.  The robust soundness and list
bound $L\le2/\delta$ are those of the cited decoder; decreasing $\theta$
changes the permitted $\alpha_\nu$ and polynomial proof-length exponent,
while the number $b$ of query parameters remains an absolute constant.
The local complexity bound $Q^c\log^c n$ is at most $n^\nu$.
The construction supports any input alphabet contained in the field.
Decrease $\gamma_\nu$ so that
$n^{\gamma_\nu}\le Q$.  The proof length and seed-set size are
polynomial in $n$.  Padding the basic arithmetization gives the same
proof length, number of queries, seed set, and proof alphabet for all
input circuits with the prescribed size and alphabet parameters.
This proves the lemma.
\end{proof}

\subsection{Making Query Locations Independent of the Circuit}
\label{sec:query-oblivization}

We prove \Cref{lem:dinur-harsha-decoder} by applying
\Cref{lem:dh-small-local-circuits} to a universal circuit. The universal
circuit checks an encoded circuit description supplied as part of its input.
The decoder then recovers one symbol of this encoding and compares it with the
corresponding symbol for the intended circuit. We show that this
modification preserves robust soundness while keeping the queries fixed.

\begin{proof}[Proof of \Cref{lem:dinur-harsha-decoder}]
Pad the description of each circuit $C$ to length $M$ and denote the resulting
$M$-bit string by $p_C$; invalid circuit descriptions are rejected.
Set $\zeta=\delta^2/16$, and fix the code from
\Cref{pro:dinur-remark-5.6}
\[
 E:\{0,1\}^M\longrightarrow\Gamma^K,
 \qquad K=O(M\delta^{-4}),\quad |\Gamma|=O(\delta^{-4}),
\]
whose distinct codewords have relative agreement at most $\zeta$.
Encoding, testing whether a word is a codeword, and decoding valid codewords have
polynomial-size circuits.

Since the circuit encoding explicitly lists all input terminals, every
circuit of complexity at most $M$ has input arity $k\le M$.
For each such arity, construct a \emph{fixed} universal circuit
$U_k:(\Sigma\times\Gamma)^{kK}\to\{0,1\}$.  Write its inputs as
$w_{j,t}=(a_{j,t},z_{j,t})$, $j\in[k]$, $t\in[K]$.
The circuit requires that $a_{j,t}$ be independent of $t$, that
$z_{j,t}$ be independent of $j$, and that the resulting string
$(z_t)_{t\in[K]}$ be a codeword $E(p)$ of a valid circuit description.
After decoding $p$, it evaluates the circuit described by $p$ on
$(a_j)_{j\in[k]}$ and requires acceptance.
A universal circuit is obtained, for example, by evaluating the gates
specified by $p$ in order and selecting each gate's inputs with fixed multiplexers.
Its size is polynomial in the description length $M$.  For a circuit
$C:\Sigma^k\to\{0,1\}$ with padded description
$p_C\in\{0,1\}^M$ and an input $y\in\Sigma^k$, define
\[
 w(C,y):=\bigl((y_j,E(p_C)_t)\bigr)_{j\in[k],\,t\in[K]}
       \in(\Sigma\times\Gamma)^{kK},
\]
where the coordinates are indexed by $(j,t)\in[k]\times[K]$.  Then,
with $C$ ranging over these circuits,
\begin{equation}\label{eq:universal-word-characterization}
 U_k^{-1}(1)=
 \left\{w(C,y):y\in\Sigma^k,\ C(y)=1\right\}.
\end{equation}
The universal circuit is independent of the particular $C$ to be decoded.

Under the coarse assumptions $\delta\ge M^{-1}$ and $|\Sigma|\le M$, all these
universal circuits have complexity at most $M^{c_U}$ for one absolute
$c_U$.  Fix an integer $C_U>c_U+1$ and pad their descriptions to length
$n_U=\lceil M^{C_U}\rceil$.  We will apply
\Cref{lem:dh-small-local-circuits} with
$\nu=1/(64C_U)$ and soundness error $\delta/2$ to these circuits;
let $\mathcal D_U$ denote the resulting decoder for $U_k$.  Write
$\mathscr R_U$ for its seed set and $\alpha_\nu,\gamma_\nu$ for its constants.  Choose positive
$\alpha_*,\gamma_*<1$ so small that
\[
 \alpha_*<C_U\alpha_\nu/2,
 \qquad \gamma_*+4\alpha_*<C_U\gamma_\nu/2.
\]
Then the soundness error $\delta/2$ of $\mathcal D_U$ is at least
$n_U^{-\alpha_\nu}$, and its input alphabet $\Sigma\times\Gamma$ has
size at most $n_U^{\gamma_\nu}$, for all sufficiently large $M$.
Thus the lemma applies.  The number of queries and local circuit complexity of
$\mathcal D_U$ are $O(M^{1/64})$; its other parameters are polynomial in $M$.

We now construct the decoder for $C$ by checking one symbol of $E(p_C)$.
On target coordinate $j$, independently choose $t\in[K]$ uniformly and
use $\mathcal D_U$ to decode coordinate $(j,t)$.  If its answer is
$(a,z)$, return $a$ exactly when $z=E(p_C)_t$, and return $\bot$
otherwise.  Use seed set $[K]\times\mathscr R_U$.
The queries use $U_k$, not $C$.  The additional local computation is one
comparison of $O(\log M)$ bits, so $q,s\le M^{1/32}$ for sufficiently
large $M$.  The parameters are common because the universal circuit
descriptions are padded to the same length. Completeness follows by encoding the table
in \eqref{eq:universal-word-characterization} for the given pair $(C,y)$.

It remains to check \emph{robust} soundness.  Fix an arbitrary proof $p\in\sigma^m$.
By robust soundness of $\mathcal D_U$, there is a list of satisfying
assignments $w(C_a,y^{(a)})$ to $U_k$, indexed by $a\in[L]$, where
$L\le4/\delta$.
The inputs $y^{(a)}$ whose associated circuit description equals $p_C$ form the set
\[
 Y_C:=\{y^{(a)}:a\in[L],\ p_{C_a}=p_C\}
       \subseteq C^{-1}(1)\subseteq\Sigma^k.
\]
Define the set of positions where $E(p_{C_a})$ agrees with $E(p_C)$
for some different description $p_{C_a}\ne p_C$ by
\[
 T_{\rm coll}:=\{t\in[K]:\exists a\in[L]\text{ with }p_{C_a}\ne p_C
                   \text{ and }E(p_{C_a})_t=E(p_C)_t\}.
\]
The code distance and a union bound give
\begin{equation}\label{eq:program-collision-bound}
 \Pr_t[t\in T_{\rm coll}]\le L\zeta\le\delta/4.
\end{equation}
For $t\notin T_{\rm coll}$, if a local answer word makes the decoder for
$C$ output outside $\{\bot\}\cup\{y_j:y\in Y_C\}$, then
$\mathcal D_U$ returns a pair different from the $(j,t)$-entry of every
listed table $w(C_a,y^{(a)})$.
Thus the set $\operatorname{Bad}_{Y_C}$ of nonrejecting local words decoding outside
the retained list is a subset of the corresponding set for
$\mathcal D_U$, for each $j,t$ and seed of $\mathcal D_U$.
On $T_{\rm coll}$, agreement is at most one.
Let $I\in[m]^q$ denote the query tuple for target coordinate $j$ and
seed $(t,\rho)$, where $\rho\in\mathscr R_U$.
Averaging over independent uniform $j\in[k]$, $t\in[K]$, and
$\rho\in\mathscr R_U$ proves
\[
 \E\operatorname{agr}(p|_I,\operatorname{Bad}_{Y_C})
 \le\delta/2+\Pr_t[t\in T_{\rm coll}]
 \le3\delta/4\le\delta.
\]
This also covers an empty retained list.  The list size remains at most
$4/\delta$, completing the proof.
\end{proof}
 \section{Iterated Composition}\label{sec:dh-iteration}

We now combine the initial family from \Cref{sec:robust-recovery} with
the construction of \Cref{sec:dh-composition}. We repeatedly replace
large local checking circuits by smaller ones, together with alphabet
and degree reduction, until the degree and alphabet bounds depend only
on the required value $1/g$.

The main issue is to control the entire sequence of reductions. We bound
the checking-circuit sizes so that each successive step remains applicable,
and we bound the total loss in sensitivity so that a polynomial lower
bound survives. These estimates prove \Cref{thm:intro-label-cover} with
a sensitivity exponent independent of $g$.
In the formal statement below, $n=|U|$ counts left vertices and $1/g$
is the required expected fraction of satisfied edges on satisfiable inputs.

\begin{theorem}[Formal version of \Cref{thm:informal-label-cover-dh}]
\label{thm:intro-label-cover}
There exist absolute constants $B,\delta>0$ and $n_{\min}\in\mathbb N$ such that,
for every integer $n\ge n_{\min}$ and every real $g$ with
$1<g\le\log^B n$, there is a swap-closed family of left-predicate Label
Cover instances with $|U|=n$.  Every algorithm defined on this family that,
on each satisfiable member, outputs a possibly random assignment of expected
value at least $1/g$ has sensitivity $\Omega(n^\delta)$.  Moreover, the
instances can be chosen so that $|V|\le n\log^{O(1)} n$, the left and right
degrees and the right alphabet size are $g^{O(1)}$, and
$|\Sigma_U|\le\exp(g^{O(1)})$.
\end{theorem}

\subsection{Transfer Parameters and Paths of Predicate Swaps}
\label{sec:dh-bookkeeping}

We track the three sensitivity-transfer parameters together with bounds
on the local description complexity along paths of predicate swaps, as defined in
\Cref{sec:pre}.  These paths let us apply the next decoder using one common
complexity bound at every intermediate instance.
For the existence statements we may also choose the circuit with the
lexicographically first minimum-length encoding for each set $\mathcal N_I(u)$ of projection
tuples of left labels satisfying $P_u$, so unchanged
local data always have unchanged chosen circuits. No bound on the time
needed to make this canonical choice is required for our sensitivity
lower bounds.

We retain a map $\mathcal T_i$ on the 2-CSP family
$\mathcal F_{N_{\mathrm{base}}}$ from \Cref{lem:fy-source} and
recovery maps $\mathcal R_{i,b}$.  Here $N_{\mathrm{base}}$ counts
the variables of $\Phi_b$. As in \Cref{sec:robust-recovery},
$b=(b_e)_{e\in E}$ records its constraint relations on the fixed edge set $E$.
We abbreviate $\mathcal T_i(\Phi_b)$ by $\mathcal T_i(b)$ and write
$\mathcal I_i=\{\mathcal T_i(\Phi):\Phi\in\mathcal F_{N_{\mathrm{base}}}\}$
for the image family, with common vertex sets $U_i,V_i$.
The map $\mathcal R_{i,b}$ takes distributions on labelings of
$\mathcal T_i(b)$ to distributions on $\Sigma_{\mathrm{src}}^X$, where
$X$ is the common variable set of the original 2-CSPs. Write
\[
 C_{\mathrm{mult},i}:=C_{\mathcal T}^{(i)}C_{\mathcal R}^{(i)},\qquad
 D_i:=D_{\mathcal R}^{(i)}.
\]
Here $C_{\mathrm{mult},i}$ is the multiplicative transfer loss and $D_i$
is the additive recovery bound accumulated through stage $i$.
Besides the sensitivity-transfer properties in
\Cref{def:law-sensitivity-pullback-data}, we keep a path of predicate swaps
of length at most $C_{\mathcal T}^{(i)}$ between the images of two 2-CSPs
differing in one relation, and a common bound on the local description
complexity of every instance on these paths. If one further transformation
replaces each predicate swap along these paths by a path of at most $t$
predicate swaps, has recovery Lipschitz bound $r$, and changes the recovered
distribution by EMD at most $d$ under one such source swap with the supplied
output held fixed, the composition rule is
\begin{equation}\label{eq:product-form-update}
 \begin{aligned}
 C_{\mathcal T}^{(i+1)}&\le C_{\mathcal T}^{(i)}t,
 & C_{\mathcal R}^{(i+1)}&\le C_{\mathcal R}^{(i)}r,\\
 C_{\mathrm{mult},i+1}&\le C_{\mathrm{mult},i}tr,
 & D_{i+1}&\le D_i+C_{\mathrm{mult},i}d.
 \end{aligned}
\end{equation}
This is \Cref{lem:law-recovery-composition} applied along the retained
paths. Explicitly, concatenating the transformed paths gives
$C_{\mathcal T}^{(i+1)}\le C_{\mathcal T}^{(i)}t$.
For the recovery bounds, couple recoveries along the same path while the
final distribution is fixed, and then apply $\mathcal R_{i,b}$ to the
recovered distributions.
The next decoder is applied only to the image family and the retained
paths, whose local circuits have the stated bound. In the final
sensitivity argument, the target algorithm is defined on the larger
family $\SwapClo(\mathcal I_i)$, the closure under both predicate and
projection swaps. We do not apply the decoder to this entire closure,
so its other instances need no local circuit bound.

\subsection{Stage-Zero Initialization by Polynomial Restrictions}
\label{sec:stage-zero-patches}

We choose the number $N_{\mathrm{base}}$ of 2-CSP variables so that every
instance in $\mathcal I_0$ has at most $n_0$ vertices on each side.
The degrees, right alphabet size, logarithm of the
left alphabet size, and local description complexity will then be
polylogarithmic in $n_0$, as will the multiplicative transfer loss
$C_{\mathrm{mult},0}$ and the additive bound
$D_0$.
For stage $i$, write $S_i$ for its local description complexity,
$d_{U,i},d_{V,i}$ for its maximum left and right degrees, and
$\Sigma_{U,i},\Sigma_{V,i}$ for its two alphabets.

\begin{proposition}[Initialization by polynomial restrictions]
\label{prop:init-stage-zero}
There are absolute constants $b_0,B_*,K_0>0$ such that, for all sufficiently
large $n_0$, the following holds. Put
\[
 \varepsilon_0=(\log n_0)^{-1},\quad
 N_{\mathrm{base}}=\lfloor n_0^{1/B_*}\rfloor,\quad
 \gamma_0=9\sqrt{\varepsilon_0}.
\]
There is a family $\mathcal I_0$ with a common graph, alphabets, and
projections, and an instance map from $\mathcal F_{N_{\mathrm{base}}}$, given by
\begin{equation}\label{eq:stage-zero-surface-chain}
 \Phi_b\longmapsto J_b
       \longmapsto\widehat J_b
         =\mathcal T_{\mathrm{DR}}(J_b;d_{\gamma_0})=\mathcal T_0(b).
\end{equation}
The number of vertices on each side is at most $n_0$, all vertices have positive degree,
and
\[
 S_0,\ d_{U,0},\ d_{V,0},\ |\Sigma_{V,0}|,
       \log|\Sigma_{U,0}|\le\varepsilon_0^{-b_0}.
\]
For every $b$, the recovery $\mathcal R_{0,b}$ maps any distribution on
labelings of $\mathcal T_0(b)$ with expected value at least
$9\sqrt{\varepsilon_0}$ to a distribution on assignments to $\Phi_b$
with expected value greater than $1-\eta_{\mathrm{src}}$.
Satisfiable 2-CSPs have satisfiable images. The transfer parameters obey
\begin{equation}\label{eq:stage-zero-surface-transfer}
 C_{\mathrm{mult},0},\ D_0
 \le K_0\varepsilon_0^{-5/2}\log(2/\varepsilon_0).
\end{equation}
Images of 2-CSPs differing in one relation admit paths of predicate swaps
of length at most $C_{\mathcal T}^{(0)}$, along which every instance has
local description complexity at most $\varepsilon_0^{-b_0}$.
\end{proposition}
\begin{proof}
Apply \Cref{thm:ldrc-free-initializer} with
$N=N_{\mathrm{base}}$ and $\gamma=\gamma_0$.
Choose $B_*$ larger than the absolute exponent in the theorem's
$N^{O(1)}$ bound on the number of vertices on each side;
then the number of vertices on each side is at most $n_0$, and
$(\log(d_0N_{\mathrm{base}}/2))^{-1}\le\gamma_0\le1/2$
for sufficiently large $n_0$.
Its local bounds are polynomial in $\log N_{\mathrm{base}}$ and
$\gamma_0^{-1}$, hence at most $\varepsilon_0^{-b_0}$ for one constant $b_0$.
Its remaining conclusions give completeness, value recovery, and the
paths of predicate swaps with this same local-description bound.
Substituting $\gamma_0=9\sqrt{\varepsilon_0}$ into
\eqref{eq:surface-initializer-transfer} proves
\eqref{eq:stage-zero-surface-transfer}.
\end{proof}

\subsection{Decreasing the Bound on Local Description Complexity}

One round first composes the current Label Cover instance with a decoder,
then reduces its right alphabet and degree. Recovery runs in the reverse
order. We choose the value thresholds first, then the circuit bounds
that make both transformations applicable at every round.

At stage $i$, recovery to the original 2-CSP requires expected value
$\chi_i:=9\sqrt{\varepsilon_i}$. Use the decoder of
\Cref{lem:dinur-harsha-decoder} with error
$\xi_i:=\varepsilon_i^{1/4}$ and list size at most $4/\xi_i$.
By \Cref{thm:composition}, expected value at least
\[
 \xi_i+(4/\xi_i)\chi_i=37\varepsilon_i^{1/4}
\]
in the composed instance suffices to recover value at least $\chi_i$.
We therefore take $\varepsilon_{i+1}=40\varepsilon_i^{1/4}$.
By \Cref{lem:degree-alphabet-reduction}, value at least
$9\sqrt{\varepsilon_{i+1}}$ after alphabet and degree reduction recovers
value at least $\varepsilon_{i+1}$ before that reduction. Thus the
successive lower bounds on expected recovered value are
\[
 9\sqrt{\varepsilon_{i+1}}
 \xrightarrow{\text{reduction recovery}}\varepsilon_{i+1}
 \xrightarrow{\text{composition recovery}}9\sqrt{\varepsilon_i}.
\]
For the small parameters used below, $\varepsilon_{i+1}>\varepsilon_i$:
the required value increases, while an upper bound
$\varepsilon_i^{-c}$ on local circuit size decreases, for a constant exponent
$c>0$ to be chosen. The remaining
choice of exponents ensures that the smaller circuit bound survives
alphabet and degree reduction.

To specify the bounds that each stage must preserve, we choose the
exponents before starting the iteration. Choose
\[
 a>\max\{b_0,7\}.
\]
The exponent $a$ controls the bounds on the right degree and right alphabet size. Let
$\alpha_*,\gamma_*$ be the constants in
\Cref{lem:dinur-harsha-decoder}, and let $c_{\mathrm{red}}$ be the constant
exponent in \Cref{lem:red-local-description}, increased once to be at
least $7$. The exponent $c$ will control the bounds on local description
complexity, left degree, and the logarithm of the left alphabet size.
Choose $c>b_0+a+1$ so large that
\begin{equation}\label{eq:fixed-exponent-choice}
 (c+1)\alpha_*>1/4,\qquad
 (c+1)\gamma_*>a,\qquad
 \frac c4+4a+c_{\mathrm{red}}+6<c.
\end{equation}
Such a choice exists: the coefficient of $c$ on the left of the last
inequality is $1/4$. The exponent $a$ and the decoder constants
have already been fixed and do not change when $c$ is chosen.

\begin{definition}[Stage data]\label{def:stage-direct-package}
For a parameter $\varepsilon_i$, a stage consists of the instance map,
recoveries on distributions, and paths of predicate swaps above, with the following invariants.
The graph, alphabets, and projections are common, all vertices have
positive degree, and
\begin{equation}\label{eq:separated-stage-invariants}
 \begin{aligned}
 d_{V,i},\ |\Sigma_{V,i}|&\le\varepsilon_i^{-a},\\
 S_i,\ d_{U,i},\ \log|\Sigma_{U,i}|&\le\varepsilon_i^{-c}.
 \end{aligned}
\end{equation}
Every instance on the retained paths has local description complexity at
most $\varepsilon_i^{-c}$.
If $\Phi_b$ is satisfiable, so is $\mathcal T_i(b)$. For every collection
of relations $b$, the recovery $\mathcal R_{i,b}$ maps any distribution
on labelings of $\mathcal T_i(b)$ with expected value at least
$\chi_i=9\sqrt{\varepsilon_i}$ to a distribution on assignments to
$\Phi_b$ with expected value at least $1-\eta_{\mathrm{src}}$.
\end{definition}

When applying the decoder, we use the complexity budget
$\lceil\varepsilon_i^{-(c+1)}\rceil$, leaving slack over the maintained
bound $\varepsilon_i^{-c}$.

\begin{proposition}[One-stage update]\label{prop:stage-update}
There is an absolute $c_0>0$ such that the following holds. Suppose stage
$i$ has the preceding data, put
\begin{equation}\label{eq:stage-parameter-update}
 \xi_i=\varepsilon_i^{1/4},\qquad
 \varepsilon_{i+1}=40\varepsilon_i^{1/4},\qquad
 M_i=\lceil\varepsilon_i^{-(c+1)}\rceil,
\end{equation}
and assume $\varepsilon_{i+1}\le c_0$. Let $\mathcal D_i$ be the
query-oblivious decoder from \Cref{lem:dinur-harsha-decoder} with complexity
bound $M_i$, error $\xi_i$, and input alphabet $\Sigma_{V,i}$.
Write $q_i,R_i,L_i$, and $\sigma_i$ for its number of queries, seed-set
cardinality, bound on the list size, and proof alphabet, respectively.
Applying each transformation to every instance of its input family, define
\[
 \mathcal J_{i+1}:=\mathcal I_i\circledast\mathcal D_i,
 \qquad
 \mathcal I_{i+1}:=\red_{\varepsilon_{i+1}}(\mathcal J_{i+1}).
\]
This gives stage-$(i+1)$ data with the same constants $a,c$.
For an absolute constant $b_{\mathrm{step}}\ge1$, the combined transformation has
transfer bounds $(t_i,r_i,d_i)$ along paths of predicate swaps and size
growth satisfying
\begin{equation}\label{eq:uniform-step-losses}
 t_i,r_i,d_i\le\varepsilon_i^{-b_{\mathrm{step}}},\qquad
 \max\{|U_{i+1}|,|V_{i+1}|\}
 \le\max\{|U_i|,|V_i|\}\varepsilon_i^{-b_{\mathrm{step}}}.
\end{equation}
The exponent $b_{\mathrm{step}}$ is independent of $i$, $n_0$, and the final target value.
\end{proposition}
\begin{proof}
We first verify the decoder conditions and the new local complexity.
The first two inequalities in \eqref{eq:fixed-exponent-choice} ensure,
for all sufficiently small $\varepsilon_i$, that
\[
 S_i\le M_i,\quad d_{V,i}\le M_i,\quad
 |\Sigma_{V,i}|\le M_i^{\gamma_*},\quad
 \xi_i\ge M_i^{-\alpha_*}.
\]
When the checking circuits at left vertices of a stage-$i$ instance have
different input arities, use the corresponding universal circuits from
\Cref{lem:dinur-harsha-decoder}; their padded description length and all decoder
parameters are common. Thus an incidence chosen uniformly at random at
a left vertex $u$ corresponds to a uniformly random input coordinate of
the checking circuit at $u$.

For the intermediate family $\mathcal J_{i+1}$, write $S^J_{i+1}$ for
a common bound on local description complexity, $d^J_{U,i+1}$ and
$\delta^J_{V,i+1}$ for its maximum left and minimum right degrees, and
$\Sigma^J_{U,i+1},\Sigma^J_{V,i+1}$ for its left and right alphabets.
By \eqref{eq:composition-local-contraction}, these quantities satisfy
\begin{equation}\label{eq:intermediate-local-bound}
 S^{J}_{i+1},\ d^{J}_{U,i+1},\ \log|\Sigma^J_{U,i+1}|
 \le C\varepsilon_i^{-a}M_i^{1/16},
 \qquad \log|\Sigma^J_{V,i+1}|\le A_*\log M_i.
\end{equation}
By \Cref{lem:degree-alphabet-reduction,lem:red-local-description}, reduction
makes the new right degree $O(\varepsilon_{i+1}^{-1})$ and right alphabet size
$O(\varepsilon_{i+1}^{-6})$. Since $a>7$, both are at most
$\varepsilon_{i+1}^{-a}$ after decreasing $c_0$.

The same lemmas and \eqref{eq:intermediate-local-bound} give
\begin{equation}\label{eq:explicit-local-contraction}
 \begin{split}
 \max\{S_{i+1},d_{U,i+1},\log|\Sigma_{U,i+1}|\}
 &\le C_1\varepsilon_i^{-a}M_i^{1/16}
          \varepsilon_{i+1}^{-c_{\mathrm{red}}}
          \log^{C_2} M_i\\
 &\le C_3\varepsilon_{i+1}^{-
            (c/4+4a+c_{\mathrm{red}}+5)}
 \le \varepsilon_{i+1}^{-c}.
 \end{split}
\end{equation}
For the second inequality, $M_i\le2\varepsilon_i^{-(c+1)}$ and
$\log^{C_2} M_i\le\varepsilon_i^{-1}$ for sufficiently small
$\varepsilon_i$. Substituting
$\varepsilon_i=(\varepsilon_{i+1}/40)^4$ gives an exponent of
$c/4+4a+c_{\mathrm{red}}+17/4$, below $c/4+4a+c_{\mathrm{red}}+5$.
The last inequality follows from \eqref{eq:fixed-exponent-choice}
after decreasing $c_0$ again to absorb the multiplicative constants.
The reduction leaves the left alphabet unchanged.
By \Cref{thm:composition,lem:red-local-description}, all instances in
$\mathcal I_{i+1}$ have the same graph, alphabets, and projections, and
all vertices have positive degree.

\paragraph{Value recovery.}
For each $b$, put $J_{i+1,b}=\mathcal T_i(b)\circledast\mathcal D_i$ and
$\mathcal T_{i+1}(b)=\red_{\varepsilon_{i+1}}(J_{i+1,b})$.
Write $Q_{i,b}$ for the recovery from labelings of $\mathcal T_{i+1}(b)$
to distributions on labelings of $J_{i+1,b}$ through alphabet-and-degree
reduction, and $S_{i,b}$ for the recovery from labelings of $J_{i+1,b}$
to distributions on labelings of $\mathcal T_i(b)$ through composition.
For a distribution $\mu$ on labelings of $\mathcal T_{i+1}(b)$, define
\[
 \mathcal R_{i+1,b}(\mu)
 :=\mathcal R_{i,b}
       \bigl((S_{i,b})_\#
                    (Q_{i,b})_\#\mu\bigr).
\]
To supply a sample requested by $\mathcal R_{i,b}$, draw $Y\sim\mu$,
then $Z\sim Q_{i,b}(Y)$, and return a sample from $S_{i,b}(Z)$.
Use fresh randomness for each draw.
Completeness is preserved by both transformations.
If $\mu$ has expected value at least $9\sqrt{\varepsilon_{i+1}}$, then
\Cref{lem:degree-alphabet-reduction} implies that $(Q_{i,b})_\#\mu$ has
expected value at least $\varepsilon_{i+1}$ on $J_{i+1,b}$.
The decoder has $L_i\le4/\xi_i$, and
\[
 \xi_i+L_i\chi_i
 \le\varepsilon_i^{1/4}+36\varepsilon_i^{1/4}
 \le\varepsilon_{i+1}.
\]
Hence \Cref{thm:composition} implies that
$(S_{i,b})_\#(Q_{i,b})_\#\mu$ has expected value at least $\chi_i$ on
$\mathcal T_i(b)$. Applying $\mathcal R_{i,b}$ therefore gives
$\mathcal R_{i+1,b}(\mu)$, a distribution on assignments to $\Phi_b$
with expected value at least $1-\eta_{\mathrm{src}}$.

\paragraph{Paths of predicate swaps and transfer.}
Composition sends one predicate swap at $u$ to at most $d_{U,i}R_i$
predicate swaps. Reduce these images with the same code and expanders.
By \eqref{eq:red-predicate-parameters}, reduction transports these
predicates unchanged, and its recovery rule is unchanged along this path,
so its additive recovery bound is zero. The combined step therefore has
\begin{equation}\label{eq:step-transfer-explicit}
 t_i=d_{U,i}R_i,\qquad
 r_i=1+\frac{d^J_{U,i+1}}{\delta^J_{V,i+1}},\qquad
 d_i\le d_{U,i}+1.
\end{equation}
The minimum intermediate right degree is at least one. The proof length
and number of seeds of $\mathcal D_i$ are polynomial in $M_i$, and
$d^J_{U,i+1}\le d_{V,i}q_i$. Consequently, $t_i,r_i,d_i$ are all at most
$\varepsilon_i^{-b_{\mathrm{step}}}$ for an absolute constant $b_{\mathrm{step}}$.

Replace the affected composed predicates one at a time. Every local
circuit on the resulting path comes from one of the adjacent images and
satisfies \eqref{eq:intermediate-local-bound}. The local reduction
construction is vertexwise, so \eqref{eq:explicit-local-contraction}
bounds the local description complexity of the image of each instance on
this path by $\varepsilon_{i+1}^{-c}$. Concatenation
along the retained stage-$i$ path gives the required stage-$(i+1)$ path.
Applying \eqref{eq:product-form-update} along that path proves the sensitivity-transfer bounds, including
\[
 C_{\mathrm{mult},i+1}\le C_{\mathrm{mult},i}t_ir_i,\qquad
 D_{i+1}
 \le D_i+C_{\mathrm{mult},i}d_i.
\]
Finally, combining the bounds on the numbers of vertices and edge incidences in
\eqref{eq:symbol-composition-structure} with the code-length and
degree bounds for reduction gives the size estimate in
\eqref{eq:uniform-step-losses} after increasing $b_{\mathrm{step}}$.
\end{proof}

\subsection{The Iterated Family and the Final Size}

We show that the accumulated size growth and sensitivity losses are
polylogarithmic in the initial size parameter $n_0$.  We then choose the
stopping stage for target expected value $1/g$ and pad to a prescribed
number of left vertices, retaining a sensitivity lower bound with a constant
positive exponent.

\begin{lemma}\label{lem:iterative-family}
Let $\varepsilon_0=(\log n_0)^{-1}$ and
$\varepsilon_{i+1}=40\varepsilon_i^{1/4}$. There are absolute constants
$c_{\mathrm{size}},c_{\mathrm{sens}},\delta_0>0$ such that, for every sufficiently
large $n_0$ and every $i$ with $\varepsilon_i\le c_0$, there are stage-$i$
data (\Cref{def:stage-direct-package}) with the following properties.
\begin{enumerate}
 \item The number of vertices on each side is at most
 $n_0\prod_{t=0}^i\varepsilon_t^{-c_{\mathrm{size}}}$.
 \item The transfer parameters satisfy
 $C_{\mathrm{mult},i},D_i\le
 \prod_{t=0}^i\varepsilon_t^{-c_{\mathrm{sens}}}$, and
 \[
 \SwapSens_{9\sqrt{\varepsilon_i}}(\SwapClo(\mathcal I_i))
 \ge \Omega\left(
       \frac{n_0^{\delta_0}}{
                    \prod_{t=0}^i\varepsilon_t^{-c_{\mathrm{sens}}}}\right).
 \]
\end{enumerate}
All exponents are independent of $i$.
\end{lemma}
\begin{proof}
\Cref{prop:init-stage-zero} gives the base case with the bounds in
\eqref{eq:separated-stage-invariants}, since $a,c>b_0$.
\Cref{prop:stage-update} gives all later
stages up to the stated stopping condition. Its size and transfer
factors are constant powers of $\varepsilon_t^{-1}$. To bound $D_i$
together with $C_{\mathrm{mult},i}$, put
$H_i=1+C_{\mathrm{mult},i}+D_i$. The update gives, after increasing a
constant exponent $b'$,
\[
 H_{i+1}\le H_i\varepsilon_i^{-b'}.
\]
Together with the initial bounds on the number of vertices on each side,
$C_{\mathrm{mult},0}$, and $D_0$
in \Cref{prop:init-stage-zero}, this proves the asserted product bounds
after increasing $c_{\mathrm{size}}$ and $c_{\mathrm{sens}}$.

For any constant $z>0$, the recursion implies
$\varepsilon_t\ge\varepsilon_0^{4^{-t}}$, and hence
\begin{equation}\label{eq:stage-product-polylog}
 \prod_{t=0}^i\varepsilon_t^{-z}
 \le\varepsilon_0^{-z\sum_{t=0}^i4^{-t}}
 \le\varepsilon_0^{-4z/3}=\log^{4z/3} n_0.
\end{equation}
In particular $D_i=o(N_{\mathrm{base}}^{c_{\mathrm{src}}})$ uniformly in the
allowed stages. Let $\mathcal A$ be an algorithm on the swap closure
with the stated value guarantee, and recover its output distribution to obtain
an algorithm on $\mathcal F_{N_{\mathrm{base}}}$. On every satisfiable
$\Phi_b$, this algorithm has expected value at least $1-\eta_{\mathrm{src}}$. Thus
\Cref{lem:fy-source,lem:law-sensitivity-pullback} give
\[
 \SwapSens(\mathcal A,\SwapClo(\mathcal I_i))
 \ge\frac{N_{\mathrm{base}}^{c_{\mathrm{src}}}-D_i}{C_{\mathrm{mult},i}}
 \ge\frac{N_{\mathrm{base}}^{c_{\mathrm{src}}}}{2C_{\mathrm{mult},i}}.
\]
Since $N_{\mathrm{base}}=\lfloor n_0^{1/B_*}\rfloor$, this proves Item~2,
with $\delta_0=c_{\mathrm{src}}/B_*$ and an adjusted absolute constant.
\end{proof}

\begin{proof}[Proof of \Cref{thm:intro-label-cover}]
Let $n$ be the prescribed final left-side size. Choose an absolute
$K>4c_{\mathrm{size}}/3+2$ and set
\[
 n_0=\left\lfloor\frac{n}{\log^K n}\right\rfloor,
 \qquad\varepsilon_0=(\log n_0)^{-1}.
\]
Fix $0<B<1/2$ and $1<g\le\log^B n$. Put
\[
 \tau=\min\{(81g^2)^{-1},c_0/2\}.
\]
For all sufficiently large $n$, uniformly over this range of $g$,
$\varepsilon_0\le\tau$. Let $i_*$ be the largest index with
$\varepsilon_{i_*}\le\tau$. The sequence is increasing while it is
below $\tau$, so such an index exists and is finite, and
\begin{equation}\label{eq:eps-window-final}
 (\tau/40)^4<\varepsilon_{i_*}\le\tau,
 \qquad 9\sqrt{\varepsilon_{i_*}}\le1/g.
\end{equation}
Apply \Cref{lem:iterative-family} to this index. The lower bound in
\eqref{eq:eps-window-final} shows that both degrees and the right
alphabet size are $g^{O(1)}$, and the left alphabet size is at most
$\exp(g^{O(1)})$. The constants are uniform even for $g$ close to one,
where these statements mean an absolute constant times the specified
power. By \eqref{eq:stage-product-polylog}, each side has at most
$n_0\log^{4c_{\mathrm{size}}/3} n_0\le n$ vertices.

Pad with isolated left vertices to make the left side exactly $n$,
using predicates identically one on these new vertices, and take the
swap closure of the resulting family. Value and satisfiability are
unchanged: both are defined by the edge constraints. The right side
is unchanged, hence has size at most $n$, in particular at most
$n\log^{O(1)} n$. An algorithm on the padded swap closure restricts
to an algorithm on the unpadded family with the same value and no larger
sensitivity, since original predicate and projection swaps commute with
padding and restriction contracts Hamming distance. The lower bound
from \Cref{lem:iterative-family}, together with
\eqref{eq:stage-product-polylog}, is
\[
 \Omega\left(\frac{n_0^{\delta_0}}{\log^{4c_{\mathrm{sens}}/3} n_0}\right)
 =\Omega\left(\frac{n^{\delta_0}}{\log^{O(1)} n}\right)
 =\Omega(n^\delta)
\]
for one constant $0<\delta<\delta_0$.
\end{proof}

\part{Covering Problems and Locality Lower Bounds}

In this part, we develop the covering applications of our two Label Cover
sensitivity lower bounds.  \Cref{sec:set-cover} proves the Set Cover
bounds, and \Cref{sec:dominating-set} transfers them to Dominating Set
before deriving the locality consequences and comparing them with
previous distributed upper and lower bounds.

\section{Sensitivity Lower Bounds for Set Cover}\label{sec:set-cover}
This section proves two sensitivity lower bounds for Set Cover.
The IKW route handles approximation ratios up to $\log\Gamma$,
where $\Gamma$ bounds both set size and element frequency, where the sensitivity lower bound scales weakly with $\Gamma$. The Dinur--Harsha route gives a polynomial sensitivity lower bound (independent of $\Gamma$) for approximation
ratios that are a small positive power of the logarithm of the universe size.
In both cases, we control set size and element frequency so that the
reduction in \Cref{sec:dominating-set} yields sensitivity lower bounds
for Dominating Set on graphs of bounded maximum degree.
Both proofs use an adaptation of the classical reduction from Label Cover
to Set Cover~\cite{lund1994hardness,feige1998threshold}, so we first establish
the common transformation and recovery statements.

In the \emph{Set Cover problem}, an instance is a pair $(U,\mathcal F=(S_1,\ldots,S_m))$, where $U$ is a finite universe and $\mathcal F$ is an indexed family of subsets of $U$.
The goal is to find an index set $\mathcal S \subseteq [m]$ of minimum size such that the corresponding subfamily $\{S_i:i\in\mathcal S\}$ covers $U$, i.e., $\bigcup_{i \in \mathcal S} S_i = U$.
We call an instance \emph{feasible} if some subfamily covers $U$, equivalently if $\bigcup_{i=1}^m S_i = U$.

For two instances $I=(U,\mathcal F=(S_1,\ldots,S_m))$ and $\tilde I=(U,\tilde{\mathcal F}=(\tilde S_1,\ldots,\tilde S_m))$ over the same universe and common index set $[m]$, we write
$\SwapDist(I,\tilde I)=\sum_{i=1}^{m} |S_i \triangle \tilde S_i|$.
That is, $\SwapDist(I,\tilde I)$ counts the number of \emph{single-incidence edits} (additions or deletions of one element--set membership) across all sets; equivalently, it is the Hamming distance between the two incidence matrices.
We identify a selected index set $\mathcal S \subseteq [m]$ with its indicator vector, and define the Hamming distance by
$d_{\mathrm H}(\mathcal S,\tilde{\mathcal S}) := |\mathcal S \triangle \tilde{\mathcal S}|$.
For a (possibly randomized) algorithm $A$ and an instance $I$, define its sensitivity with respect to incidence edits, denoted again by $\SwapSens(A,I)$, by
$\SwapSens(A,I) := \max_{\tilde I:\,\SwapDist(I,\tilde I)=1} \EMD(A(I),A(\tilde I))$,
where $\tilde I$ ranges over instances on the same universe and index set
that differ from $I$ by one incidence edit. The earth mover's distance
uses $d_{\mathrm H}$ as the underlying metric on outputs.
For a family $\mathcal I$, set $\SwapSens(A,\mathcal I) := \max_{I \in \mathcal I} \SwapSens(A,I)$.
For $\alpha \ge 1$, let $\SwapSens_\alpha(\mathcal I)$ be the minimum $\SwapSens(A,\mathcal I)$ over all algorithms $A$ defined on all Set Cover instances and satisfying the $\alpha$-approximation guarantee on every feasible $I\in\mathcal I$.

\paragraph{Parameter guide.}
In each lower-bound proof, $n$ is the prescribed Set Cover universe size.
We use $s$ for the source Label Cover size, defined locally: the total
number of vertices in the IKW proof and the number of left vertices in
the Dinur--Harsha proof.

\subsection{Reduction from Label Cover to Set Cover}

Following standard reductions from label cover to set cover we use the following set system of \cite{lund1994hardness,feige1998threshold}.
\begin{lemma}[{\cite{lund1994hardness,feige1998threshold}}]\label{lem:m-l-set-system}
Let $m,q\ge1$ and $b\ge2$ be integers, and let $\ell>0$ and
$0<\theta<1$ satisfy
\[
 b\exp\left(-\frac{\ell}{q(1-\theta)}\right)>mq\ln2.
\]
There is a universe $B$ of size $b$ with $m$ partitions into $q$ indexed parts
\[
 B=C_{x,1}\sqcup\cdots\sqcup C_{x,q},\qquad x\in[m],
\]
such that, if a collection of at most $\ell$ indexed parts covers $B$,
then the collection must contain more than $\theta q$ parts from one
partition.
\end{lemma}
The construction and proof follow \cite[Lemma 3.2]{feige1998threshold}, which makes the slightly weaker conclusion that any collection of $\ell$ indexed parts covering $B$ must have at least two parts from one partition. We need the stronger $\theta q$ bound for our recovery procedure. The proof is standard and deferred to the end of the section.

\paragraph{The Set Cover Instance}
Fix a left-predicate Label Cover instance
$I=(U,V,E,\Sigma_U,\Sigma_V,\mathcal P,F)$ whose multigraph is biregular
with positive left and right degrees $d_L,d_R$.
Fix a set system $(B,\{C_{i,j}\})$ as in \Cref{lem:m-l-set-system} with
$m=|\Sigma_V|$, $q=d_R$, and parameters $\ell>0$, $0<\theta<1$,
indexing its partitions by the right-alphabet $\Sigma_V$.
At each $v\in V$, order the incident edges as
$e_{v,1},\ldots,e_{v,d_R}$.
The Set Cover instance $\mathcal T_{\mathrm{SC}}(I)$ has universe
$V\times B$ and sets indexed by
\[
 \Lambda_I:=U\times\Sigma_U.
\]
We define $S_{(u,y)}$ in the following way. If $P_u(y)=0$, i.e., if label $y$ fails the predicate at $u$, we will leave $S_{(u,y)}$ empty. Otherwise, for each edge $e_{v,i}=(u,v)$ touching $u$, include $\{v\} \times C_{f_{e_{v,i}}(y),i} \subset V \times B$, i.e.,
\[
S_{(u,y)} \coloneqq \bigcup_{e_{v,i}=(u,v)} \{v\} \times C_{f_{e_{v,i}}(y),i} 
\]

\paragraph{The Recovery Map}

Given a target output for the above set cover instance, indexed by $\mathcal S\subseteq\Lambda_I$, let
\[
 L_u:=\{y\in\Sigma_U:(u,y)\in\mathcal S\}.
\]
The recovery $R_{\mathrm{SC},I}$ independently chooses, at each left
vertex $u$, a uniform label from $L_u$. At each right vertex $v$, it
independently chooses a uniform incident edge $e=(u,v)$, then a fresh
uniform label from $L_u$, and outputs its projection under $f_e$.\footnote{We remark one could instead re-use the label for $u$ chosen in the first step, but independence will make the recovery simpler to analyze.}
Whenever $L_u$ is empty, a fixed default label in $\Sigma_U$ is used.

\begin{lemma}[Size, completeness, and sensitivity]\label{lem:sc-recovery}
The construction has universe size $N_0$ and number of sets $m_0$ given by
\[
 N_0=|V||B|,\qquad m_0=|U||\Sigma_U|.
\]
Set sizes are at most $d_L|B|$, and element frequencies are at most
$d_R|\Sigma_U|$.
If $I$ is satisfiable, $\mathcal T_{\mathrm{SC}}(I)$ has a cover of size
$|U|$. On a family with common graph and alphabets, the transformation
and recovery have sensitivity-transfer parameters
\[
 C_{\mathcal T}=d_L|\Sigma_U||B|,\qquad
 C_{\mathcal R}=1+\frac{d_L}{d_R},\qquad D_{\mathcal R}=1.
\]
\end{lemma}
\begin{proof}
The sizes $N_0$ and $m_0$ are immediate from the construction. A set indexed by $(u,y)$ is a union of at most $d_L$ sets of size
at most $|B|$, so the set sizes are at most $d_L|B|$. Each element $(v,b)\in V\times B$ belongs only to
sets indexed by labels at neighbors of $v$, so each element appears in at most $d_R|\Sigma_U|$ sets. Given a satisfying labeling of $I$, select at each left vertex the set corresponding to its label. At a right vertex $v$ labeled $x$, these sets contain $\{v\}\times C_{x,i}$ for every $i\in[d_R]$ and therefore
cover $\{v\}\times B$.

To bound $ C_{\mathcal T}$, observe a projection swap on $e=(u,v)$ affects incidences only between sets
indexed by labels at $u$ and elements of $\{v\}\times B$.
A predicate swap at $u$ affects such incidences for at most $d_L$
choices of $v$. In either case, at most $|\Sigma_U||B|$ incidences
change for each $v$, so 
\[
C_{\mathcal T} \leq d_L|\Sigma_U||B|.
\]
To bound $C_{\mathcal R}$, observe that adding or removing a selected
set index $(u,y)$ changes only the list $L_u$. Couple the sampled edges
identically and couple the sampled labels identically whenever their
lists agree. The recovered label can then change only at $u$ and at
right vertices whose sampled edge has left endpoint $u$.
Summing over the $d_L$ edges incident to $u$, including parallel edges,
the expected number of changed right labels is at most $d_L/d_R$.
Thus
\[
C_{\mathcal R} \leq 1+\frac{d_L}{d_R}.
\]
Finally, the recovery rule does not depend on predicate swaps, and a projection swap 
can only change the recovered label at its right endpoint so
\[
D_{\mathcal R} \leq 1.
\]
\end{proof}

We now argue that applying our recovery procedure to a randomized solution to the set cover instance of small enough expected size results in a high value solution to the original label cover instance.

\begin{lemma}[Recovery Value]\label{lem:sc-soundness}
If a distribution over set covers $\bm{\mathcal S}$ has expected size at most
\[
 \E|\bm{\mathcal S}|\le (1-\theta)\frac{\ell}{d_R}|U|,
\]
then the expected recovered value is at least
\[
 \E\bigl[\val_I(R_{\mathrm{SC},I}(\bm{\mathcal S}))\bigr]
 \ge\theta^5\frac{d_R^2}{\ell^2}.
\]
\end{lemma}
\begin{proof}
Choose a uniform right vertex $v$ independent of the cover,
and write $T:=\sum_{e=(u,v)}|L_u|$. Biregularity gives
\[
 \E T=\frac{d_L}{|V|}\E|\bm{\mathcal S}|
       =\frac{d_R}{|U|}\E|\bm{\mathcal S}|
       \le(1-\theta)\ell.
\]
We will argue that, for every fixed cover and right vertex $v$ with
$T\le\ell$, the set-system property of \Cref{lem:m-l-set-system}
implies that the expected fraction of edges incident to $v$ satisfied
by the recovered label-cover assignment is at least
$\theta^4d_R^2/\ell^2$. Since $\Pr[T \leq \ell] \geq \theta$ by Markov, this implies
\[
 \E\bigl[\val_I(R_{\mathrm{SC},I}(\bm{\mathcal S}))\bigr]
 \ge\frac{\theta^4d_R^2}{\ell^2}\Pr[T\le\ell]
 \ge\frac{\theta^5d_R^2}{\ell^2}
\]
as desired.

When $T\le\ell$, the parts used by $\bm{\mathcal S}$ to cover
$\{v\}\times B$ form a cover of $B$ by at most $\ell$ indexed
parts of our system from \Cref{lem:m-l-set-system}. By the set-system property, these contain more than $\theta d_R$ parts from one partition, so there exists some $x\in\Sigma_V$ and
$H\subseteq[d_R]$ with $|H|>\theta d_R$ such that, for each $i\in H$,
some $y\in L_{u_i}$ satisfies $P_{u_i}(y)=1$ and $f_{e_{v,i}}(y)=x$,
where $u_i$ is the left endpoint of $e_{v,i}$.

For the fixed cover and right vertex $v$, define
\[
 p_v(x):=\frac1{d_R}\sum_{e=(u,v)}
 \frac{|\{y\in L_u:P_u(y)=1,\ f_e(y)=x\}|}{|L_u|},
\]
where a summand is zero when $L_u$ is empty. Thus $p_v(x)$ is the
probability that a uniform edge $e=(u,v)$ has $L_u\ne\varnothing$
and a uniform label from $L_u$ satisfies $P_u$ and projects to $x$.
By the defining property of $H$ and Cauchy--Schwarz,
\[
 p_v(x)
 \ge\frac1{d_R}\sum_{i\in H}\frac1{|L_{u_i}|}
 \ge\frac{|H|^2}{d_RT}
 \ge\frac{\theta^2d_R}{\ell}.
\]
Now observe that by construction, the probability that recovery outputs $x$ at $v$ is at least $p_v(x)$. Similarly, the probability over a random edge $(u,v)$ that the recovered label at $u$ satisfies its predicate and projects to $x$
is at least $p_v(x)$. These choices are independent conditioned on the cover and $v$, so the conditional expected fraction of satisfied edges incident to $v$
is at least $p_v(x)^2 \geq \theta^4\frac{d_R^2}{\ell^2}$ as desired.
\end{proof}

\paragraph{Set Cover of any Size}

Finally, we put these lemmas together with a basic padding procedure to produce a sensitivity lower bound for set cover of any (sufficiently large) size from a low-soundness label cover bound. In the following sections, we will instantiate this result with our two label cover bounds.

\begin{lemma}\label{lem:sc-exact-size-transfer}
Let $\mathcal I$ be a swap-closed Label Cover family with common biregular
graph and alphabets, with $d_L,d_R>0$ as above, and fix a common set
system with parameters $\ell>0$, $0<\theta<1$, and $m_0\le N_0$.
Suppose every algorithm on $\mathcal I$ that outputs a labeling of expected
value at least $\theta^5d_R^2/\ell^2$ on each satisfiable member has
swap sensitivity at least $L>1$.
Let $\Gamma,n$ be integers satisfying
\[
 \Gamma\ge\max\{d_L|B|,d_R|\Sigma_U|\},\qquad n\ge N_0.
\]
Then, for every $\alpha$-approximation algorithm $A$ defined on all Set Cover
instances of universe size $n$, with
\[
 \alpha
 \le(1-\theta)\frac{\ell}{d_R}\left(1+\frac{N_0/\Gamma+1}{|U|}\right)^{-1},
\]
there are two instances $J_0,J_1$ differing by one incidence edit such that
\[
 \EMD\bigl(A(J_0),A(J_1)\bigr)\ge
 \frac{L-1}{d_L|\Sigma_U||B|(1+d_L/d_R)}.
\]
The instances $J_0,J_1$ have a common universe and index set, at most
$n$ indexed sets, and set sizes and element frequencies at most $\Gamma$.
\end{lemma}
\begin{proof}
Let $A$ be the given Set Cover algorithm.
For $J_I=\mathcal T_{\mathrm{SC}}(I)$, put
$\rho=\lfloor n/N_0\rfloor$, $t=n-\rho N_0$, and
$r=\lceil t/\Gamma\rceil$. Form $J_I^{\mathrm{pad}}$ from $\rho$
disjoint copies of $J_I$ and $t$ additional elements partitioned into
$r$ blocks of size at most $\Gamma$, adding one set equal to each block.
The parameters $N_0,m_0,\rho,t,r$ and all padding choices are common to the family.
Since $m_0\le N_0$, the number of sets is
$\rho m_0+r\le\rho N_0+t=n$. Each added set has size at most $\Gamma$,
and each additional element has frequency one. Together with
\Cref{lem:sc-recovery}, this proves the set-size and frequency bounds.

Run $A$ on $J_I^{\mathrm{pad}}$, select a random one of the $\rho$ copies, and apply $R_{\mathrm{SC},I}$.
This defines a source algorithm $A_{\mathrm{LC}}$ on all of $\mathcal I$.
On satisfiable $I$, a cover of size $|U|$ exists before padding, and
$r\le N_0/\Gamma+1$. The selected indices $\bm{\mathcal S}$
in a uniform copy form a random set cover satisfying
\[
 \E|\bm{\mathcal S}|
 \le\alpha(|U|+r/\rho)
 \le\alpha(|U|+N_0/\Gamma+1)
 \le\frac{(1-\theta)\ell}{d_R}|U|.
\]
By \Cref{lem:sc-soundness}, the recovered labeling has expected value
at least $\theta^5d_R^2/\ell^2$ on every satisfiable source instance.

The padded transformation changes at most
$\rho C_{\mathcal T}$ incidences per source swap, where
$C_{\mathcal T}=d_L|\Sigma_U||B|$. Choosing the same uniform copy in
two recoveries shows that their Lipschitz constant $C_{\mathcal R}$ is at most
$(1+d_L/d_R)/\rho$ on all selected subsets, and the source dependence $D_{\mathcal R}$
contributes at most one. By the source lower bound, choose
neighbors $I,\widetilde I$ with
$\EMD(A_{\mathrm{LC}}(I),A_{\mathrm{LC}}(\widetilde I))\ge L$.
Their padded images $J,\widetilde J$ satisfy
\[
 \frac{\EMD(A(J),A(\widetilde J))}{\SwapDist(J,\widetilde J)}
 \ge\frac{L-1}{C_{\mathcal T}(1+d_L/d_R)}.
\]
Delete incidences present only in $J$, then insert those present only
in $\widetilde J$. Each intermediate incidence set is a subset of one
endpoint's, so the size and frequency bounds persist. This is a
shortest path; \Cref{lem:neighboring-witness} supplies the desired
neighboring pair.
\end{proof}
\paragraph{Remaining Proofs}

\begin{proof}[Proof of \Cref{lem:m-l-set-system}]
Consider the following random choice of the partitions $\{C_{i,j}\}$. For each point in $B$ and each of the $m$ partitions, independently choose a uniform
part in which to place the point. We will argue that for any fixed set of $\ell$ index pairs $(i,j) \in [m] \times [q]$ with at most $\theta q$ pairs per partition, the probability that the corresponding $\ell$ random indexed parts $\{C_{i,j}\}$ cover $B$ is small. Union bounding over all possible such index sets then gives the existence of the desired system.

Formally, fix any collection of at most $\ell$
indexed parts (formally, index pairs as above), and let $t_x$ be the number in the $x$-th partition. Suppose $t_x\le\theta q$ for every $x$. The probability over the choice of the $\{C_{i,j}\}$ that a fixed
point is uncovered by this collection is
\[
 p=\prod_x(1-t_x/q)
 \ge\exp\left(-\frac{\sum_x t_x}{q(1-\theta)}\right)
 \ge\exp\left(-\frac{\ell}{q(1-\theta)}\right),
\]
where we used $\ln(1-z)\ge-z/(1-\theta)$ for $0\le z\le\theta$.
The choices for different points are independent, so the probability
that the collection covers $B$ is at most $(1-p)^b\le e^{-bp}$.
There are at most $2^{mq}$ collections and by assumption $e^{-bp}2^{mq}<1$, so there exists a set system where no such collection covers $B$ as desired.
\end{proof}

\subsection{Sensitivity Bound via IKW}

We now use the IKW Label Cover family to prove a sensitivity lower bound for
Set Cover algorithms with approximation ratio up to $\log\Gamma$,
where $\Gamma$ bounds both set size and element frequency.  Choosing the product
parameter $k=\Theta((\log\log\Gamma)^2)$ transfers the Label Cover lower
bound to $n^{\Omega(1/(\log\log\Gamma)^2)}$ sensitivity, where $n$ is the
prescribed universe size, in the parameter range stated below.

\begin{theorem}[Sensitivity lower bound for Set Cover via IKW]\label{thm:app-set-cover-IKW}
    For every constant $\varepsilon\in(0,1)$, there exist a constant
    $c_0=c_0(\varepsilon)>0$ and an integer
    $\Gamma_0=\Gamma_0(\varepsilon)\ge 2$ such that the following holds.
    For all sufficiently large $n$ and every integer $\Gamma\ge \Gamma_0$ satisfying
    \[
        \log \Gamma\cdot (\log\log \Gamma)^2 \le c_0 \log n,
    \]
    any algorithm $A_{\mathrm{SC}}$ defined on all Set Cover instances of universe size $n$ that
    has approximation ratio $(1-\varepsilon)\ln\Gamma$ in expectation on every feasible
    instance must have sensitivity with respect to incidence edits
    \[
        2^{\Omega_{\varepsilon}\!\left(\frac{\log n}{(\log\log \Gamma)^2}\right)}.
    \]
    Moreover, the lower bound is already witnessed by neighboring instances $J_0,J_1$ on the same universe and with a common index set such that
    \[
        |\mathcal F(J_0)|=|\mathcal F(J_1)|=O(n),
    \]
    every set in $J_0$ and $J_1$ has size at most $\Gamma$, every element has frequency at most $\Gamma$, and
    \[
        \EMD\!\left(A_{\mathrm{SC}}(J_0),A_{\mathrm{SC}}(J_1)\right)
        =
        2^{\Omega_{\varepsilon}\!\left(\frac{\log n}{(\log\log \Gamma)^2}\right)}.
    \]
\end{theorem}

\begin{proof}[Proof of \Cref{thm:app-set-cover-IKW}]
Use the constants $C_{\mathrm{IKW}},a,b,\eta,k_0$ from
\Cref{thm:intro-label-cover-sensitivity-IKW}. Put $\theta=\varepsilon/10$.
Choosing $C_k=C_k(\varepsilon)$ sufficiently large let
\[
 k=\left\lceil C_k(\log\log\Gamma)^2\right\rceil
\]
such that the value requirement from \Cref{lem:sc-soundness} follows from
$\exp(-C_{\mathrm{IKW}}\sqrt{k})\le\theta^5/(\ln\Gamma)^2$.

Fix $c_{\mathrm{deg}}>0$ so that
$D_\Gamma:=2^{c_{\mathrm{deg}}k}$ bounds the source degrees and the
sum of its alphabet sizes. Choose $N$ (the source 2-CSP size in \Cref{thm:intro-label-cover-sensitivity-IKW}) to be
\[
 N=\left\lfloor
       \left(\frac{n}{100\Gamma}\right)^{1/(bk)}
    \right\rfloor.
\]
The stated parameter range gives for all $\Gamma\ge\Gamma_0$,
\[
 \log N=\Theta(\log n/k),\qquad
 k^2=o(\log n),\qquad k=o(\log N).
\]
so \Cref{thm:intro-label-cover-sensitivity-IKW} gives a common biregular source label-cover family with $s$ vertices and positive
degrees $d_L,d_R$ satisfying
\[
 N^{ak}\le s\le\frac{n}{100\Gamma}.
\]
Every algorithm on this family that outputs a labeling of expected value
at least $\exp(-C_{\mathrm{IKW}}\sqrt{k})$ on each satisfiable member
has sensitivity at least
$L=s^{\eta/k}=n^{\Omega(1/k)}$. Set
\[
 |B|=\left\lfloor\frac{\theta\Gamma}{D_\Gamma^4}\right\rfloor,
 \qquad
 \ell=d_R(1-2\theta)\ln|B|.
\]
Since $D_\Gamma=\Gamma^{o(1)}$, increasing $\Gamma_0$ ensures
\[
 \ln|B|\ge(1-\theta)\ln\Gamma,\qquad
 |B|\ge D_\Gamma^2,\qquad
 |B|^{\theta/(1-\theta)}>d_R|\Sigma_V|\ln2.
\]
The last inequality is the hypothesis of \Cref{lem:m-l-set-system},
so choose partitions of $B$ with these parameters. The value requirement
holds because
\[
 \frac{\theta^5d_R^2}{\ell^2}
 =\frac{\theta^5}{(1-2\theta)^2(\ln|B|)^2}
 \ge\frac{\theta^5}{(\ln\Gamma)^2}
 \ge\exp(-C_{\mathrm{IKW}}\sqrt{k}).
\]
We also have
\[
 d_L|\Sigma_U||B|(1+d_L/d_R)
 \le D_\Gamma^4|B|\le\theta\Gamma.
\]
In particular, $d_L|B|\le\Gamma$ and $d_R|\Sigma_U|\le\Gamma$, and
\Cref{lem:sc-recovery} gives
\[
 N_0=|V||B|\le s\Gamma\le n/100,\qquad
 m_0=|U||\Sigma_U|\le N_0,
\]
where the last inequality uses biregularity and $|B|\ge D_\Gamma^2$.

Put $\alpha=(1-\varepsilon)\ln\Gamma$.
For all sufficiently large $n$, we have
$|U|\ge s/(1+D_\Gamma)\ge1/\theta$. Hence
\[
 \frac{N_0/\Gamma+1}{|U|}
 =\frac{d_L|B|}{d_R\Gamma}+\frac1{|U|}
 \le2\theta.
\]
Since $\theta=\varepsilon/10$,
\[
 \begin{aligned}
 \alpha\left(1+\frac{N_0/\Gamma+1}{|U|}\right)
 &\le(1-\varepsilon)(1+2\theta)\ln\Gamma\\
 &\le(1-\theta)^2(1-2\theta)\ln\Gamma\\
 &\le (1-\theta)\frac{\ell}{d_R}.
 \end{aligned}
\]
All hypotheses of \Cref{lem:sc-exact-size-transfer} now hold, so there
are neighboring instances $J_0,J_1$ with the stated size and frequency
bounds such that
\[
 \EMD\bigl(A_{\mathrm{SC}}(J_0),A_{\mathrm{SC}}(J_1)\bigr)\ge
 \frac{L-1}{d_L|\Sigma_U||B|(1+d_L/d_R)}
 \ge\frac{s^{\eta/k}-1}{\Gamma}
 =n^{\Omega_{\varepsilon}(1/k)}
\]
as desired.
\end{proof}

\subsection{Sensitivity Bound via Dinur--Harsha}

We now transfer the polynomial sensitivity lower bound for Label Cover from
\Cref{thm:intro-label-cover} to Set Cover up to $\log^B n$-approximation for some constant $B \in (0,1)$.

\begin{theorem}[Sensitivity lower bound for Set Cover via Dinur--Harsha]\label{thm:intro-set-cover}
    There exist an integer $n_0\in\mathbb N$ and positive absolute constants
    $B,c,C,\delta$ such that the following holds.
    For every integer $n\ge n_0$ and every real $q$ with
    \[
        2 \le q \le \log^{B} n,
    \]
    any $C\sqrt q$-approximation algorithm $A$ defined on all Set Cover
    instances of universe size $n$ has sensitivity $\Omega(n^{\delta})$.
    Moreover, the lower bound already holds on instances in which every set
    has size at most $\exp(q^{c})$, every
    element belongs to at most
    $\exp(q^{c})$ sets, and the number of sets
    satisfies $|\mathcal F| = O(n)$.
    More specifically, for every such algorithm $A$, there are two instances
    $J_0,J_1$ with the same universe of size $n$ and a common index set that
    differ by one incidence edit, satisfy these structural bounds, and obey
    \[
        \EMD\bigl(A(J_0),A(J_1)\bigr)
        = \Omega(n^{\delta}).
    \]
\end{theorem}

\begin{proof}[Proof of~\Cref{thm:intro-set-cover}]
Let $B_0,\delta_{\mathrm{LC}}>0$ be the constants from
\Cref{thm:intro-label-cover}. Fix $C\ge1$, choose $c\ge1$ large enough, then take
$0<B<B_0/4$ with $2Bc<1$. Set $\delta=\delta_{\mathrm{LC}}/2$ and
\[
 g=\lceil64C\sqrt q\rceil^2,\qquad
 \Lambda=\lfloor\exp(q^c)\rfloor,\qquad
 s=\left\lfloor\frac{n}{8\Lambda^2}\right\rfloor.
\]
For $2\le q\le\log^B n$, we have
$\Lambda=n^{o(1)}$, $\log s\sim\log n$, and
$g=O(q)\le\log^{B_0} s$ for sufficiently large $n$.
\Cref{thm:intro-label-cover} therefore supplies a family with $|U|=s$ and source
degree and alphabet bounds
\[
 D,|\Sigma_V|=q^{O(1)},\qquad
 |\Sigma_U|\le\exp(q^{O(1)}),\qquad m:=|E|\le sD.
\]
Discard the isolated left vertices introduced by the final padding in
\Cref{sec:dh-iteration}, so that $|U|\le s$.
The remaining graph is biregular with positive degrees $d_L,d_R$:
this holds at initialization and is preserved by each composition
followed by alphabet and degree reduction.
Any algorithm on the remaining family extends to the padded family
by assigning fixed labels to the discarded vertices, with unchanged
value and no greater sensitivity.
Set $\theta=1/2$ and $\ell=d_R\sqrt g/8$, so that
$\theta^5d_R^2/\ell^2=2/g$.
Every algorithm on this
family that outputs a labeling of expected value at least $2/g$ on
each satisfiable member has sensitivity at least
$L=\Omega(s^{\delta_{\mathrm{LC}}})$.

Choose a universe $B_{\mathrm{sys}}$ of size $\exp(q^{O(1)})$ satisfying
\[
 |B_{\mathrm{sys}}|e^{-\sqrt g/4}>d_R|\Sigma_V|\ln2,
 \qquad
 |B_{\mathrm{sys}}|\ge\frac{d_R}{d_L}|\Sigma_U|.
\]
By \Cref{lem:m-l-set-system}, choose $|\Sigma_V|$ partitions of
$B_{\mathrm{sys}}$ into $d_R$ parts with parameters $\ell,\theta$.
The second inequality gives $m_0\le N_0$ and
$d_R|\Sigma_U|\le d_L|B_{\mathrm{sys}}|$.
The absolute constant $c$ can be chosen so that, for all $q\ge2$,
\[
 D,\quad D^2|B_{\mathrm{sys}}|,\quad
 d_L|\Sigma_U||B_{\mathrm{sys}}|(1+d_L/d_R)\le\Lambda.
\]
Consequently the constructed universe has size
\[
 N_0=|V||B_{\mathrm{sys}}|\le m|B_{\mathrm{sys}}|
 \le m\Lambda\le s\Lambda^2\le n/8.
\]
Apply \Cref{lem:sc-exact-size-transfer} with set-size and frequency bound
$\Gamma=\Lambda$, target universe size $n$, and approximation factor
$\alpha=C\sqrt q$. Its approximation condition holds because
\[
 \begin{aligned}
 \alpha\left(1+\frac{N_0/\Lambda+1}{|U|}\right)
 &=C\sqrt q\left(1+\frac{d_L|B_{\mathrm{sys}}|}{d_R\Lambda}
                    +\frac1{|U|}\right)\\
 &\le3C\sqrt q\le\frac{\sqrt g}{16}
 =\frac{(1-\theta)\ell}{d_R}.
 \end{aligned}
\]
We obtain neighboring instances $J_0,J_1$
with $O(n)$ sets and set sizes and frequencies at most
$\Lambda\le\exp(q^c)$, satisfying
\[
 \EMD\bigl(A(J_0),A(J_1)\bigr)\ge
 \frac{L-1}{d_L|\Sigma_U||B_{\mathrm{sys}}|(1+d_L/d_R)}
 \ge\frac{L-1}{\Lambda}
 =\Omega\left(\frac{s^{\delta_{\mathrm{LC}}}}{\Lambda}\right)
 =\Omega(n^\delta),
\]
where the last equality uses $s=\Theta(n/\Lambda^2)$ and
$\Lambda=n^{o(1)}$. This proves the theorem.
\end{proof}

 \section{PCP-Based Lower Bounds for Dominating Set}\label{sec:dominating-set}
This section proves two sensitivity lower bounds for Dominating Set on
graphs of maximum degree at most $\Delta$ and derives locality lower
bounds from them. The IKW-based bound in
\Cref{thm:intro-dominating-set-IKW} concerns approximation ratio
$(1-\varepsilon)\ln\Delta$ for every constant $\varepsilon\in(0,1)$.
The Dinur--Harsha-based bound in
\Cref{thm:intro-dominating-set} gives a polynomial sensitivity lower bound for
approximation ratios that are a small positive power of $\log\Delta$.
We first isolate the common reduction from Set Cover to Dominating Set and then specialize it in the two approximation regimes.
We then compare the resulting locality bounds with prior distributed results.
\Cref{tab:ds-sensitivity-comparison} summarizes the sensitivity bounds
and their degree ranges.

\begin{table}[ht!]
\centering
\caption{Sensitivity bounds for Dominating Set on $n$ vertices with
maximum degree at most $\Delta\le n-1$.
The first lower bound uses DH, with a sufficiently small absolute
constant $\beta>0$, and requires
sufficiently large $\Delta\le\exp(C_{\mathrm{range}}\log^\kappa n)$
for some $\kappa\in(0,1)$.
For every constant $\varepsilon\in(0,1)$, the second uses IKW and requires
$\Delta\ge\Delta_0(\varepsilon)$ and
$\log\Delta(\log\log\Delta)^2\le c_0(\varepsilon)\log n$.
The upper bound holds for every constant $\varepsilon>0$ and sufficiently
large $\Delta$ depending on $\varepsilon$.}
\label{tab:ds-sensitivity-comparison}
\small
\begin{tabular}{lll}
\toprule
Approximation factor & Sensitivity bound & Reference \\
\midrule
\multicolumn{3}{c}{Lower bounds} \\
$\log^\beta \Delta$ & $n^{\Omega(1)}$ & \Cref{thm:intro-dominating-set} \\
$(1-\varepsilon)\ln\Delta$ & $n^{\Omega_\varepsilon(1/(\log\log\Delta)^2)}$ & \Cref{thm:intro-dominating-set-IKW} \\
\midrule
\multicolumn{3}{c}{Upper bound} \\
$(1+\varepsilon)\ln\Delta$ & $\exp(O_\varepsilon(\log^3 \Delta))$
& \Cref{subsec:ds-logarithmic-approximation}, from \cite{kuhn2006price} \\
\bottomrule
\end{tabular}
\end{table}

\paragraph{Parameter guide.}
In the lower-bound proofs below, $N$ denotes the size of the universe in the source Set Cover instance, $m:=|\mathcal F|$ the number of sets, $n$ the prescribed number of vertices, and $\Gamma$ the parameter bounding both set size and element frequency in the Dominating Set reduction.
Here and below, sensitivity for graph algorithms is measured under single-edge edits.
For a Dominating Set algorithm $A$, we write $\SwapSens(A)$ for the
maximum EMD between its output distributions on graphs in the stated
input class that differ by one edge.  The distance between output vertex
sets is the size of their symmetric difference.

\subsection{Reduction from Set Cover to Dominating Set}

Fix an integer parameter $\Gamma \ge 1$.
We use the classical reduction of~\cite{chlebik2008approximation}.
We start with the element--set incidence graph.  Choosing the set-vertices
of a cover dominates all element-vertices.  Helper vertices will dominate
the remaining set-vertices, with at most $\Gamma$ set-vertices per helper.
For a Set Cover instance
$I=(U=\{u_1,\ldots,u_{N}\},\mathcal F=(S_1,\ldots,S_{m}))$
whose set sizes and element frequencies are at most $\Gamma$, let
$G_I=\mathcal T_{\mathrm{DS}}(I)$ be the graph on
$U\sqcup\{S_i:i\in[m]\}\sqcup W$, where
$W=\{w_1,\ldots,w_{\lceil m/\Gamma\rceil}\}$.
In graph notation, $S_i$ denotes a tagged vertex representing index $i$,
whose identity is independent of the contents of the subset $S_i\subseteq U$.
The three vertex classes are disjoint, and different indices give distinct
vertices even when their subsets coincide.
Its edges are the incidences $uS_i$ with $u\in S_i$ and the helper edges
$S_iw_{\lceil i/\Gamma\rceil}$.  Thus
\[
    \Delta(G_I)\le \Gamma+1,
    \qquad
    |V(G_I)|=N+m
      +\left\lceil\frac{m}{\Gamma}\right\rceil.
\]

\paragraph{Recovery and sensitivity.}
Given a vertex set $D\subseteq V(G_I)$, retain the indices of its
set-vertices and replace each element-vertex in $D$ by a set containing
that element. Ignore helper vertices.
Formally, let $\sigma_I(u)$ be the smallest index $i\in[m]$ with
$u\in S_i$, or $\bot$ if no such index exists. Define
$\phi_I:V(G_I)\to[m]\cup\{\bot\}$ by
\[
 \phi_I(S_i)=i,\qquad \phi_I(u)=\sigma_I(u),\qquad
 \phi_I(w)=\bot\quad(w\in W),
\]
and set $R_I(D):=\phi_I(D)\setminus\{\bot\}$ for every
$D\subseteq V(G_I)$. Here $R_I$ abbreviates $R_{\mathrm{DS},I}$.
If $I$ is feasible and $D$ dominates $G_I$, then every element-vertex
is either selected or adjacent to a selected set-vertex. In either case,
$R_I(D)$ contains a set covering that element, and $|R_I(D)|\le|D|$.
Adding or removing one vertex in $D$ changes at most one recovered index. An
incidence edit changes one graph edge and can change $\sigma_I$ only at
the edited element, replacing at most one recovered index by another.
Thus $C_{\mathcal T}=C_{\mathcal R}=1$ and $D_{\mathcal R}\le2$.

\paragraph{Padding to a prescribed size.}
For each pair of integers $\Delta\ge2$ and $t\ge0$, fix a labeled star
forest $\mathcal P_{\Delta,t}$ on a designated set of $t$ vertices.  Its
components have size at most $\Delta$ and its domination number is
$\lceil t/\Delta\rceil$: writing $t=a\Delta+b$ with $0\le b<\Delta$, take
$a$ copies of $K_{1,\Delta-1}$ and, if $b>0$, one copy of $K_{1,b-1}$,
where $K_{1,0}$ is an isolated vertex.

Fix $C_{\mathcal F}\ge1$, a positive integer target $n$, and an integer $\Gamma\ge1$.
We choose the source universe size to leave room for padding to exactly
$n$ vertices. The padding must increase the optimum by at most a
constant factor, so that recovery preserves the approximation guarantee
up to a constant loss. We use $C_{\mathcal F}$ to bound the ratio of the
number of sets to the universe size. Define a padding factor
$C_{\mathrm{pad}}$, an approximation-loss bound $C_{\mathrm{lift}}$,
and a source universe size $N$ by
\[
 C_{\mathrm{pad}}:=8(C_{\mathcal F}+1),
 \qquad C_{\mathrm{lift}}:=C_{\mathcal F}+2C_{\mathrm{pad}}+3,
 \qquad N:=\left\lfloor\frac{n}{C_{\mathrm{pad}}}\right\rfloor.
\]
Let $I=(U,\mathcal F=(S_1,\ldots,S_m))$ have
$|U|=N$, $m\le C_{\mathcal F}N$, and set sizes
and element frequencies at most $\Gamma$, and set
$t_I:=n-|V(G_I)|$.  Whenever $t_I\ge0$, define
$H_I:=G_I\sqcup\mathcal P_{\Gamma+1,t_I}$.
Extend $\phi_I$ to $V(H_I)$ by mapping every padding vertex to $\bot$,
and use the same formula $R_I(D)=\phi_I(D)\setminus\{\bot\}$ on all
$D\subseteq V(H_I)$. Since the padding is disjoint from $G_I$ and is
unchanged under incidence edits, the recovery and sensitivity guarantees
above continue to hold.

The next proposition shows that the padding is well-defined for all
sufficiently large target sizes and records the approximation and
perturbation guarantees of this construction.

\begin{proposition}[Reduction to Dominating Set with a prescribed number of vertices]\label{prop:ds-exact-size-lift}
For every $C_{\mathcal F}\ge1$, there is a threshold
$n_0=n_0(C_{\mathcal F})$ such that, for every $n\ge n_0$ and every
integer $\Gamma\ge1$, each instance $I$ in the bounded Set Cover family
above has $t_I\ge0$ and satisfies the following.
\begin{enumerate}
\item \emph{Size, degree, and optimum.} We have
\[
 \begin{gathered}
 |V(H_I)|=n,
 \qquad \Delta(H_I)\le\Gamma+1,
 \qquad n\le2C_{\mathrm{pad}}N,\\
 \opt(H_I)\le C_{\mathrm{lift}}\opt(I)
 \quad\text{if $I$ is feasible.}
 \end{gathered}
\]
\item \emph{Recovery from a fixed instance.} If $I$ is feasible and $D$ dominates $H_I$, then
$R_I(D)$ is a set cover of $I$ and
$|R_I(D)|\le|D|$.  For all
$D,\widetilde D\subseteq V(H_I)$,
\[
 |R_I(D)\triangle
   R_I(\widetilde D)|
 \le |D\triangle\widetilde D|.
\]

\item \emph{One incidence edit.} If $I$ and $\widetilde I$ belong to this bounded family, have the same
universe and a common index set, and differ by one incidence edit, then
$H_I$ and $H_{\widetilde I}$ have a common vertex set and differ by exactly
one edge.  Moreover, every vertex subset $D$ of this common vertex set
satisfies
\[
 |R_I(D)\triangle
   R_{\widetilde I}(D)|\le2.
\]
\end{enumerate}

\end{proposition}
\begin{proof}
Writing $m:=|\mathcal F|$, for all sufficiently large $n$ depending only
on $C_{\mathcal F}$ and uniformly over $\Gamma\ge1$,
\[
 |V(G_I)|=N+m+\left\lceil\frac m\Gamma\right\rceil
 \le(1+2C_{\mathcal F})N+1
 \le n/2,
 \qquad
 n\le2C_{\mathrm{pad}}N.
\]
This proves $t_I\ge0$ and $|V(H_I)|=n$; the bounds for $G_I$ and the star forest give the degree claim.

If $I$ is feasible, bounded set size and $m\le C_{\mathcal F}N$ give
\[
 \opt(I)\ge N/\Gamma,
 \qquad
 m/\Gamma+1\le(C_{\mathcal F}+1)\opt(I).
\]
Moreover,
\[
 \begin{aligned}
 \opt(G_I)&\le\opt(I)+m/\Gamma+1,\\
 \opt(\mathcal P_{\Gamma+1,t_I})
 &=\left\lceil\frac{t_I}{\Gamma+1}\right\rceil
 \le\frac{n}{\Gamma+1}+1
 \le(2C_{\mathrm{pad}}+1)\opt(I).
 \end{aligned}
\]
The first bound is obtained by selecting the set-vertices of an optimal cover
and all helper vertices.  Since domination numbers add over components, these
bounds give the claimed optimum comparison.

Suppose now that $I$ is feasible and $D$ dominates $H_I$.  Since there are no edges between
$G_I$ and the padding forest, the set $D\cap V(G_I)$ dominates $G_I$.
For each $u\in U$, if $u\in D$, then $\sigma_I(u)\in R_I(D)$ and
$u\in S_{\sigma_I(u)}$.  Otherwise, $D$ contains a set-vertex $S_i$
adjacent to $u$, so $i\in R_I(D)$ and $u\in S_i$.  Thus
$R_I(D)$ is a set cover of $I$.  Since this recovery is the image
of $D$ under $\phi_I$, with $\bot$ removed, it does not increase size.  For
arbitrary $D,\widetilde D\subseteq V(H_I)$,
\[
 R_I(D)\triangle
 R_I(\widetilde D)
 \subseteq \phi_I(D\triangle\widetilde D),
\]
which proves the fixed-instance Lipschitz bound.

Finally, suppose $\widetilde I$ is obtained from $I$ by one incidence edit
at $(u_0,S_j)$.  Then $m$ is unchanged, so all vertices, helper edges, and
the fixed padding forest are common; only the edge $u_0S_j$ changes.
The maps $\phi_I$ and $\phi_{\widetilde I}$ agree outside $u_0$.
For every $D\subseteq V(H_I)$, the set $R_{\widetilde I}(D)$ is therefore
obtained from $R_I(D)$ by deleting at most one index and adding at most
one index.  This proves the bound on the change in recovery.
\end{proof}

Use the parameters of \Cref{prop:ds-exact-size-lift}, with
$n$ sufficiently large. Let $A_{\mathrm{DS}}$ have
approximation ratio $r_{\mathrm{DS}}$ on $n$-vertex graphs
of maximum degree at most $\Gamma+1$. For Set Cover instances $I$ on $N$
elements with at most $C_{\mathcal F}N$ sets, each of size at most $\Gamma$
and each element occurring in at most $\Gamma$ sets, define
\[
 A_{\mathrm{SC}}(I):=
 R_I(A_{\mathrm{DS}}(H_I)).
\]
Extend $A_{\mathrm{SC}}$ to all instances on $N$ elements
by returning an optimal cover on feasible instances outside the family
and an arbitrary output on infeasible instances outside the family.

\begin{corollary}[Approximation and sensitivity transfer to Dominating Set]
\label{cor:ds-approximation-sensitivity}
The induced algorithm $A_{\mathrm{SC}}$ is a
$C_{\mathrm{lift}}r_{\mathrm{DS}}$-approximation on all
Set Cover instances of universe size $N$.
On instances with at most $C_{\mathcal F}N$ sets, set sizes at most
$\Gamma$, and element frequencies at most $\Gamma$, the reduction has transfer parameters
$C_{\mathcal T}=C_{\mathcal R}=1$ and $D_{\mathcal R}=2$.
For every neighboring pair $I,\widetilde I$ in that family with a common
universe and index set,
\[
 \EMD(A_{\mathrm{SC}}(I),A_{\mathrm{SC}}(\widetilde I))\ge L
 \quad\Longrightarrow\quad
 \SwapSens(A_{\mathrm{DS}})\ge L-2.
\]
\end{corollary}
\begin{proof}
The recovery and optimum bounds in \Cref{prop:ds-exact-size-lift} prove
the approximation guarantee on instances satisfying these three bounds;
outside this family, the
extension is optimal whenever a cover exists. The same proposition gives
$C_{\mathcal T}=C_{\mathcal R}=1$ and $D_{\mathcal R}=2$.
For each neighboring pair in the statement, \Cref{lem:recovery-lifting}
and the triangle inequality give
\[
 \EMD(A_{\mathrm{SC}}(I),A_{\mathrm{SC}}(\widetilde I))
 \le \EMD(A_{\mathrm{DS}}(H_I),A_{\mathrm{DS}}(H_{\widetilde I}))+2.
\]
Since $H_I$ and $H_{\widetilde I}$ differ by one edge, this proves the
sensitivity implication directly on the given pair.
\end{proof}

\subsection{Logarithmic Approximation in the Maximum Degree}
\label{subsec:ds-logarithmic-approximation}

We prove a sensitivity lower bound for Dominating Set algorithms with
approximation ratio $(1-\varepsilon)\ln\Delta$, for every constant
$\varepsilon\in(0,1)$, where $\Delta$ bounds the maximum degree.
Set $\Gamma:=\lfloor\Delta\rfloor-1$ and apply the IKW Set Cover
lower bound of \Cref{thm:app-set-cover-IKW}, with set size and element
frequency bounded by a sufficiently small $\varepsilon$-dependent
fraction of $\Gamma$. The common reduction above then gives
Dominating Set instances of maximum degree at most $\Delta$ while
preserving the sensitivity lower bound up to a constant factor.

\begin{theorem}[Sensitivity lower bound for Dominating Set via IKW]\label{thm:intro-dominating-set-IKW}
    For every constant $\varepsilon\in(0,1)$, there exist an integer
    $\Delta_0=\Delta_0(\varepsilon)\ge2$ and a constant
    $c_0=c_0(\varepsilon)>0$ such that, for all sufficiently large $n$
    and every $\Delta_0\le\Delta\le n-1$ satisfying
    \[
        \log \Delta\cdot (\log\log \Delta)^2 \le c_0 \log n,
    \]
    any algorithm for the Dominating Set problem on graphs with $n$
    vertices and maximum degree at most $\Delta$ with approximation
    ratio
    \[
        (1-\varepsilon)\ln\Delta
    \]
    in expectation must have sensitivity
    \[
        2^{\Omega_{\varepsilon}\!\left(
            \frac{\log n}{(\log\log \Delta)^2}\right)}.
    \]
\end{theorem}

\begin{proof}[Proof of \Cref{thm:intro-dominating-set-IKW}]
Let $c_1,\Gamma_0$ be the constants from
\Cref{thm:app-set-cover-IKW} with error parameter $\varepsilon/4$,
and fix $C_{\mathcal F}\ge1$ bounding its number of sets by
$C_{\mathcal F}N$.
Use the corresponding $C_{\mathrm{pad}}$ from
\Cref{prop:ds-exact-size-lift}.
For $n,\Delta$ in the stated range, put
\[
 N=\left\lfloor n/C_{\mathrm{pad}}\right\rfloor,\qquad
 \Gamma=\lfloor\Delta\rfloor-1,\qquad
 \gamma=\left\lfloor
     \frac{\varepsilon\Gamma}{16C_{\mathrm{pad}}}
 \right\rfloor.
\]
Taking $\Delta_0$ large enough and $c_0$ small enough ensures
$\gamma\ge\Gamma_0$, $\Gamma+1\le\Delta$, and
\[
 (1+\varepsilon/4)(1-\varepsilon)\ln\Delta
 \le(1-\varepsilon/4)\ln\gamma,
 \qquad
 \log\gamma(\log\log\gamma)^2\le c_1\log N.
\]
Here $N=\Theta(n)$ and $\ln\gamma=\ln\Delta-O_\varepsilon(1)$.
The stated degree range also gives $N\ge16\gamma/\varepsilon$
for sufficiently large $n$.

Consider Set Cover instances $I$ on $N$ elements with at most
$C_{\mathcal F}N$ sets, set sizes at most $\gamma$, and element
frequencies at most $\gamma$.
Apply the construction of \Cref{prop:ds-exact-size-lift} with parameter
$\Gamma$, obtaining a graph $H_I$ on exactly $n$ vertices with
maximum degree at most $\Gamma+1\le\Delta$.
If $I$ is feasible, then $\opt(I)\ge N/\gamma$ and
\[
 \begin{aligned}
 \opt(H_I)
 &\le\opt(I)
     +\left\lceil\frac{m}{\Gamma}\right\rceil
     +\left\lceil\frac{t_I}{\Gamma+1}\right\rceil\\
 &\le\opt(I)+\frac n\Gamma+2\\
 &\le(1+\varepsilon/4)\opt(I).
 \end{aligned}
\]
The second inequality uses $m+t_I\le n$.
For the last, $n\le2C_{\mathrm{pad}}N$ and the choice of $\gamma$
give $n/\Gamma\le(\varepsilon/8)\opt(I)$, while
$N\ge16\gamma/\varepsilon$ gives
$2\le(\varepsilon/8)\opt(I)$.

Let $A'$ be an algorithm as in the statement.
On the bounded family above, define
$A(I):=R_I(A'(H_I))$ using the recovery from
\Cref{prop:ds-exact-size-lift}.
Extend $A$ to all Set Cover instances on $N$ elements by returning
an optimal cover on feasible instances outside this family and an
arbitrary output on infeasible instances outside this family.
The recovery does not increase size, so the preceding optimum bound
and approximation comparison make $A$ a
$(1-\varepsilon/4)\ln\gamma$-approximation algorithm in expectation.
\Cref{thm:app-set-cover-IKW} gives neighboring instances $I,\widetilde I$
on the same universe and index set such that
\[
 \EMD(A(I),A(\widetilde I))\ge L
 =2^{\Omega_{\varepsilon}(\log N/(\log\log\gamma)^2)}.
\]
Each instance has at most $C_{\mathcal F}N$ sets, set sizes at most
$\gamma$, and element frequencies at most $\gamma$.
The sensitivity comparison from
\Cref{cor:ds-approximation-sensitivity} therefore gives
\[
 \SwapSens(A')\ge L-2
 =2^{\Omega_{\varepsilon}(\log n/(\log\log\Delta)^2)}.\qedhere
\]
\end{proof}

For comparison, approximation ratio $(1+\varepsilon)\ln\Delta$ permits
sensitivity bounded independently of $n$ for sufficiently large $\Delta$.
We use KMW's CONGEST algorithm to compute a fractional dominating set,
then round its coordinates independently and add every undominated vertex.
Let $G=(V,E)$ have maximum degree at most $\Delta$, and let $N[v]$ denote
the closed neighborhood of $v$.
For an LP accuracy parameter $0<\eta\le1$, KMW's
algorithm~\cite[Corollary~4.1]{kuhn2006price} computes a nonnegative vector
$x=(x_v)_{v\in V}$ in $O(\log^2(\Delta+1)/\eta^4)$ rounds, with
\[
 \sum_{u\in N[v]}x_u\ge1\quad(v\in V),
 \qquad
 \sum_{v\in V}x_v\le(1+\eta)\opt(G).
\]

For randomized rounding, select each vertex $v$ independently with
probability $\min\{1,\ln(\Delta+1)x_v\}$.
After one additional CONGEST round to exchange selection bits with their
neighbors, all vertices left undominated join the selected set, producing
a dominating set $D$.
By LP feasibility, each vertex is left undominated by the rounding with
probability at most $1/(\Delta+1)$.
Also, summing the LP constraints gives $n\le(\Delta+1)\sum_v x_v$.
Thus
\[
 \E[|D|]
 \le \ln(\Delta+1)\sum_v x_v+\frac{n}{\Delta+1}
 \le (1+\eta)(1+\ln(\Delta+1))\opt(G).
\]
For any constant $\varepsilon>0$, choose
$\eta=\min\{\varepsilon/2,1/2\}$ so that the LP approximation leaves room
for the additive term in $1+\ln(\Delta+1)$.
For sufficiently large $\Delta$ depending on $\varepsilon$,
\[
 (1+\eta)(1+\ln(\Delta+1))\le(1+\varepsilon)\ln\Delta.
\]
The entire procedure therefore gives a $(1+\varepsilon)\ln\Delta$
approximation in expectation in $t=O_\varepsilon(\log^2\Delta)$ \mbox{CONGEST}
rounds, including the final neighbor check.

To bound sensitivity, fix the vertex identifiers and the prescribed degree
bound across inputs, and couple executions on graphs differing by one edge
using the same independent random numbers at corresponding vertices.
By locality, each vertex's output is determined by its radius-$t$ input
neighborhood and the random numbers in that neighborhood.
As in the argument for \eqref{eq:sensitivity-locality}, these input
neighborhoods coincide outside distance $t+1$ from the edited edge's
endpoints, with distance measured in the graph containing the edge.
The coupled outputs therefore agree outside this region, which contains
at most $2(\Delta+1)^{t+1}$ vertices.
This bounds their symmetric difference under the coupling and hence gives
sensitivity at most
\begin{equation}\label{eq:ds-sensitivity-upper}
 \exp\!\bigl(O_\varepsilon(\log^3 \Delta)\bigr).
\end{equation}
For every constant $\varepsilon\in(0,1)$ and every sufficiently large
constant $\Delta$ depending on $\varepsilon$, this upper bound is
independent of $n$, whereas \Cref{thm:intro-dominating-set-IKW} requires
sensitivity $n^{\Omega_{\varepsilon,\Delta}(1)}$ for
$(1-\varepsilon)\ln\Delta$-approximation.
The bounds do not give matching dependence on $\Delta$ for the
sensitivity itself.

Applying the sensitivity-to-locality bound \eqref{eq:sensitivity-locality}
of Fleming and Yoshida~\cite{fleming2026sensitivity} to
\Cref{thm:intro-dominating-set-IKW} gives the following distributed lower bound.

\begin{corollary}[Locality lower bound via IKW]\label{thm:intro-locality-dominating-set-IKW}
For every constant $\varepsilon\in(0,1)$, there exist a constant
$c_0=c_0(\varepsilon)>0$ and an integer
$\Delta_0=\Delta_0(\varepsilon)\ge2$ such that, for all sufficiently
large $n$ and every $\Delta_0\le\Delta\le n-1$ satisfying
\[
    \log\Delta\,(\log\log\Delta)^2\le c_0\log n,
\]
any non-signaling algorithm for Dominating Set on $n$-vertex graphs of
maximum degree at most $\Delta$ with approximation ratio
$(1-\varepsilon)\ln\Delta$ in expectation has locality at least
\[
    \Omega_\varepsilon\!\left(
        \frac{\log n}{\log\Delta\,(\log\log\Delta)^2}
    \right).
\]
\end{corollary}

\begin{proof}
Fix $\varepsilon\in(0,1)$ and use the constants from
\Cref{thm:intro-dominating-set-IKW}, decreasing $c_0$ as needed below.
Apply that theorem and
\eqref{eq:sensitivity-locality}, using $\log(\Delta+1)=\Theta(\log\Delta)$.
For some $a_\varepsilon>0$ depending only on $\varepsilon$ and an
absolute constant $b>0$, the locality $t$ satisfies
\[
 t\ge a_\varepsilon
 \frac{\log n}{\log\Delta\,(\log\log\Delta)^2}-b.
\]
The degree restriction makes the fraction at least $1/c_0$.
Taking $c_0\le a_\varepsilon/(2b)$ absorbs the additive term uniformly
over the stated degree range.
\end{proof}

For every constant $\varepsilon\in(0,1)$, the CONGEST algorithm described
above gives a
$(1+\varepsilon)\ln\Delta$ approximation in
$O_\varepsilon(\log^2\Delta)$ CONGEST rounds for sufficiently large
$\Delta$ depending on $\varepsilon$.
The locality lower bound in \Cref{thm:intro-locality-dominating-set-IKW}
for $(1-\varepsilon)\ln\Delta$-approximation is
$\omega(\log^2\Delta)$ whenever
$\log^3 \Delta\cdot(\log\log\Delta)^2=o(\log n)$.
In particular, for every sufficiently large constant $\Delta$ depending
on $\varepsilon$, achieving a $(1-\varepsilon)\ln\Delta$ approximation
on graphs of maximum degree at most $\Delta$ requires
$\Omega_{\varepsilon,\Delta}(\log n)$ locality even in the non-signaling
model, whereas a
$(1+\varepsilon)\ln\Delta$ approximation takes a constant number of rounds
independent of $n$.

\subsection{Polylogarithmic Approximation in the Maximum Degree}

We use the Dinur--Harsha Set Cover family to prove a polynomial sensitivity
lower bound for Dominating Set algorithms whose approximation ratio is a
constant positive power of $\log\Delta$.  Starting from the instances in
\Cref{thm:intro-set-cover}, whose universe size can be prescribed exactly,
we choose the Set Cover parameter $q$ from the degree bound
$\Delta$ and set $\Gamma:=\Delta-1$.  We then apply the common reduction and
exact-size padding above.  The choice of $q$ ensures that set sizes and
element frequencies are at most $\Gamma$, and that the induced Set Cover
algorithm satisfies the approximation guarantee required by
\Cref{thm:intro-set-cover}.

\begin{theorem}[Sensitivity lower bound for Dominating Set via Dinur--Harsha]\label{thm:intro-dominating-set}
    There exist $\beta_0,\delta,\kappa\in(0,1)$, absolute constants $C_{\mathrm{DS}},C_{\mathrm{range}}>0$, and an integer $\Delta_0\ge2$ such that for all sufficiently large $n$ and every integer $\Delta$ satisfying
    \[
        \Delta_0\le\Delta\le\exp(C_{\mathrm{range}}\log^\kappa n),
    \]
    any algorithm that $C_{\mathrm{DS}}\log^{\beta_0} \Delta$-approximates Dominating Set on $n$-vertex graphs of maximum degree at most $\Delta$ must have sensitivity $\Omega(n^\delta)$.
\end{theorem}

\begin{proof}[Proof of \Cref{thm:intro-dominating-set}]
Let $B_0,c_0,C_0,\delta_0>0$ be the constants from
\Cref{thm:intro-set-cover}, decreasing $B_0$ so that $B_0<1$.
Set $\delta=\min\{\delta_0,1/2\}$, $C_{\mathrm{DS}}=1$, and choose
\[
 0<\beta_0<\min\{1/2,1/(4c_0),B_0/4\},\qquad \kappa=B_0/2.
\]
Fix $C_{\mathcal F}\ge1$ bounding the number of sets in that theorem by
$C_{\mathcal F}N$, and use the corresponding constants
$C_{\mathrm{pad}},C_{\mathrm{lift}}$ from \Cref{prop:ds-exact-size-lift}.
For $n,\Delta$ in the stated range, put
\[
 N=\lfloor n/C_{\mathrm{pad}}\rfloor,\qquad \Gamma=\Delta-1,\qquad
 q=\max\left\{2,\left\lceil
       (64C_{\mathrm{lift}}/C_0)^2\log^{2\beta_0} \Delta
       \right\rceil\right\}.
\]
Since $2\beta_0c_0<1$, choosing $\Delta_0$ sufficiently large gives
$\exp(q^{c_0})\le\Gamma$.
Also $2\beta_0\kappa<B_0/2$, so the degree range implies
$q\le\log^{B_0/2} n\le\log^{B_0} N$ for sufficiently large $n$,
after fixing a sufficiently small $C_{\mathrm{range}}>0$.
Finally, the choice of $q$ ensures
\[
 C_{\mathrm{lift}}\log^{\beta_0} \Delta\le C_0\sqrt q.
\]
Consequently \Cref{cor:ds-approximation-sensitivity} converts an algorithm
$A'$ as in the statement into a $C_0\sqrt q$-approximation algorithm
$A$ for Set Cover on $N$ elements.
\Cref{thm:intro-set-cover} gives neighboring instances $I,\widetilde I$
on the same universe and index set satisfying
$\EMD(A(I),A(\widetilde I))=\Omega(N^{\delta_0})$.
Each instance has at most $C_{\mathcal F}N$ sets, and set sizes and
element frequencies at most $\exp(q^{c_0})\le\Gamma$.
\Cref{cor:ds-approximation-sensitivity} therefore gives
\[
 \SwapSens(A')\ge\Omega(N^{\delta_0})-2
 =\Omega(n^\delta).\qedhere
\]
\end{proof}
Within the theorem's degree range, we apply the sensitivity-to-locality
argument directly. For larger values of $\Delta$, we first restrict the
algorithm to graphs satisfying a smaller upper bound on maximum degree
within that range. This gives the following consequence; extending to all
$\Delta\le n-1$ uses a smaller approximation exponent.

\begin{corollary}[Locality lower bound via Dinur--Harsha]\label{thm:intro-locality-dominating-set}
There exist absolute constants $C>0$ and $\beta\in(0,1)$ and an integer $\Delta_0\ge2$ such that, for all sufficiently large $n$ and every $\Delta_0\le\Delta\le n-1$, any $C\log^\beta \Delta$-approximate non-signaling algorithm for Dominating Set on $n$-vertex graphs of maximum degree at most $\Delta$ has locality at least
\[
    \Omega\!\left(\frac{\log n}{\log\Delta}\right).
\]
\end{corollary}

\begin{proof}[Proof of \Cref{thm:intro-locality-dominating-set}]
Use the constants of \Cref{thm:intro-dominating-set} and set
$\beta=\kappa\beta_0$. For sufficiently large $n$, write $L=\log n$ and put
\[
 \widehat\Delta:=\min\left\{\lfloor\Delta\rfloor,
       \left\lfloor\exp\left(\frac{C_{\mathrm{range}}}{2}L^\kappa\right)
       \right\rfloor\right\}.
\]
This integer lies in the degree range of \Cref{thm:intro-dominating-set}.
After increasing $\Delta_0$, there is an absolute $c_*>0$ such that
\[
 \log\widehat\Delta\ge c_*\min\{\log\Delta,L^\kappa\}
 \ge c_*\log^\kappa \Delta,
\]
where the second inequality uses $1\le\log\Delta\le L$.
Choose $C\le C_{\mathrm{DS}}c_*^{\beta_0}$. Then
\[
 C\log^\beta \Delta
 \le C_{\mathrm{DS}}\log^{\beta_0} \widehat\Delta.
\]
Let $A$ be a $t$-local algorithm with the stated approximation guarantee.
Restrict $A$ to graphs of degree at most $\widehat\Delta$.
\Cref{thm:intro-dominating-set} gives sensitivity $\Omega(n^\delta)$ on
this class. Since $\log\widehat\Delta=O(\log^\kappa n)$ with
$\kappa<1$, the ratio $\log n/\log\widehat\Delta$ tends to infinity.
The sensitivity-to-locality bound \eqref{eq:sensitivity-locality} gives,
after absorbing its additive constant,
\[
 t=\Omega\left(\frac{\log n}{\log\widehat\Delta}\right)
 =\Omega\left(\frac{\log n}{\log\Delta}\right),
\]
since $\widehat\Delta\le\Delta$.
\end{proof}

\subsection{Comparison with Existing Distributed Bounds}
\label{subsec:ds-distributed-comparison}

\Cref{tab:intro-tradeoffs} compares our two locality lower bounds with
prior upper and lower bounds for Dominating Set.

\begin{table}[t!]
\centering
\caption{Selected approximation and locality bounds for Dominating Set on $n$ vertices
with maximum degree at most $\Delta\le n-1$.
For every constant $\varepsilon\in(0,1)$, IKW has the
$\varepsilon$-dependent degree restrictions in
\Cref{tab:ds-sensitivity-comparison}.
DH allows all sufficiently large $\Delta\le n-1$, with $\beta\in(0,1)$.
Upper bounds are for constant $0<\varepsilon\le1$, and hidden constants may
depend on $\varepsilon$. The $(1+\varepsilon)\ln\Delta$
rows require sufficiently large $\Delta$ depending on $\varepsilon$.
The Chang--Li upper bound succeeds with high probability.}
\label{tab:intro-tradeoffs}
\small
\setlength{\tabcolsep}{3.25pt}
\begin{tabular}{@{}p{0.24\linewidth}p{0.38\linewidth}p{0.14\linewidth}p{0.17\linewidth}@{}}
\toprule
Approximation ratio & Locality lower bound or upper bound & Model & Reference \\
\midrule
\multicolumn{4}{c}{Lower bounds} \\
\addlinespace[0.25em]
$1+\varepsilon$,\newline $0<\varepsilon\le\varepsilon_0$ & $\Omega(\log n/\varepsilon)$ & LOCAL & \cite{chang2023complexity} \\
Polylogarithmic in $\Delta$ & $\Omega\!\left(\min\!\left\{\sqrt{\frac{\log n}{\log\log n}},\frac{\log\Delta}{\log\log\Delta}\right\}\right)$ & LOCAL & \cite{kuhn2016local,coupette2021breezing} \\
$(1-\varepsilon)\ln\Delta$ & $\Omega_\varepsilon\!\left(\frac{\log n}{\log\Delta\,(\log\log\Delta)^2}\right)$ & Non-signaling & \Cref{thm:intro-locality-dominating-set-IKW} \\
$C\log^\beta \Delta$ & $\Omega\!\left(\frac{\log n}{\log\Delta}\right)$ & Non-signaling & \Cref{thm:intro-locality-dominating-set} \\
\midrule
\multicolumn{4}{c}{Upper bounds} \\
\addlinespace[0.25em]
$1+\varepsilon$ & $\poly(\log n)$ & LOCAL & \cite{ghaffari2017complexity} \\
$1+\varepsilon$ & $O_\varepsilon(\log n(\log\log n)^3)$ & LOCAL & \cite{chang2023complexity} \\
$(1+\varepsilon)\ln\Delta$ & $O_\varepsilon(\log n)$ & LOCAL & \cite{kuhn2016local} \\
$(1+\varepsilon)\ln\Delta$ & $O_\varepsilon(\log^2\Delta)$ & CONGEST & From \cite{kuhn2006price} \\
\bottomrule
\end{tabular}
\end{table}

Chang and Li~\cite{chang2023complexity} establish logarithmic dependence
on $n$ for approximation ratios sufficiently close to one.
For approximation factors growing polylogarithmically with $\Delta$,
the KMW lower bound~\cite{kuhn2016local,coupette2021breezing} becomes
independent of $n$ when $\Delta$ is constant.
For every sufficiently large constant $\Delta$, our locality lower bounds
retain an $\Omega(\log n)$ dependence on graphs of maximum degree at most
$\Delta$, even in the non-signaling model.
For $\Delta=\log^a n$ with constant $a>0$,
\Cref{thm:intro-locality-dominating-set} gives
$\Omega(\log n/\log\log n)$ for
$C\log^\beta \Delta$-approximation, whereas the KMW bound in the table
is $\Omega(\log\log n/\log\log\log n)$.
For every constant $\varepsilon\in(0,1)$, our IKW locality lower bound applies to
$(1-\varepsilon)\ln\Delta$-approximation throughout its stated degree
range, while a $(1+\varepsilon)\ln\Delta$ approximation can be attained
in $O_\varepsilon(\log^2\Delta)$ CONGEST rounds for sufficiently large
$\Delta$ depending on $\varepsilon$, as discussed after
\Cref{thm:intro-locality-dominating-set-IKW}.
 
\part{Representative Selection and Directed Spanners}
\label{part:minrep-spanner}

This part proves sensitivity and distributed lower bounds for minimum
unweighted directed 3-spanner: selecting as few arcs as possible so that
every input arc has a directed path of length at most three in the
selected subgraph.  For every constant approximation factor, we obtain
a polynomial sensitivity lower bound on weakly connected DAGs of bounded total degree,
where the degree bound may depend on the approximation factor.
This gives an $\Omega(\log n)$ lower bound on LOCAL and non-signaling
locality for all sufficiently large graph sizes $n$.

\Cref{sec:minrep} first converts the low-soundness Label Cover bounds
into bounds for selecting representatives from groups.  \Cref{sec:directed-spanner}
encodes these selections by three-arc directed paths while controlling
both input edits and maximum degree.  It also gives a polynomial
sensitivity lower bound for a $\log^\beta n$ approximation factor on DAGs of degree
$\exp(O(\log^\kappa n))$, for absolute constants $\beta>0$ and
$\kappa\in(0,1/2)$.  \Cref{sec:spanner-distributed} uses the degree bounds
to convert sensitivity into LOCAL and non-signaling locality lower bounds.

\section{Sensitivity Lower Bounds for Min-Rep}
\label{sec:minrep}

We give a second application of the low-soundness Label Cover theorems,
independent of the covering reductions in \Cref{sec:set-cover,sec:dominating-set}.
The intermediate problem is \emph{Min-Rep}: choose a small collection of
representatives that contains a compatible pair for every pair of label
groups joined by at least one compatibility edge.  This is the standard
representative-selection problem used in spanner hardness reductions;
see, for example,
\cite{dinitz2012largegirth}.  We analyze both the value of a recovered
Label Cover solution and the effects of input edits and changes to the
supplied representative set.
The resulting reduction gives sensitivity $\Omega(n^\delta)$ for
$\log^\beta n$-approximation on instances with $n$ representatives,
for some absolute constants $\beta,\delta>0$.

There are two issues with applying Min-Rep directly to
\Cref{thm:intro-label-cover}.  The Label Cover graph need not be regular,
and it may contain parallel edges carrying different projections.
We address the first issue by copying label groups according to the
source degrees, and the second by requiring a compatible representative
pair to satisfy every parallel constraint between its source vertices.
Both operations allow a
recovery rule independent of the input predicates and projections.

\subsection{The Problem and the Input-Edit Convention}

A Min-Rep instance $J$ consists of disjoint groups
$(A_i)_{i\in L}$ and $(B_j)_{j\in R}$ and a simple bipartite
\emph{compatibility graph} with edge set
\[
 E_J\subseteq \bigcup_{i\in L,\,j\in R}A_i\times B_j.
\]
Its superedge set is
\[
 \mathcal E_J:=\{(i,j)\in L\times R:E_J\cap(A_i\times B_j)\ne\varnothing\}.
\]
The bipartite graph $(L,R,\mathcal E_J)$ is the \emph{supergraph};
its vertices index groups, and its edges are called \emph{superedges}.
A representative set $S\subseteq\bigcup_i A_i\cup\bigcup_j B_j$ is
feasible if, for every $(i,j)\in\mathcal E_J$, some compatibility edge
between $A_i$ and $B_j$ has both endpoints in $S$.  The objective is to
minimize $|S|$.  Groups
with no incident superedge need not contribute a representative.
Every instance is feasible, including the empty compatibility graph,
whose optimum is zero.

Throughout a family the representative vertices, groups, and their names
are fixed.  An input edit adds or deletes one compatibility edge.  In
particular, deleting the last compatibility edge between two groups
\emph{removes} that superedge; it does not leave an unsatisfiable demand.
This convention is sufficient for sensitivity transfer: the recovered
Label Cover algorithm needs a value guarantee only on satisfiable source
instances. On other source instances, it must remain defined and obey the
perturbation bounds; the reduction need not preserve a gap in optimum.
Outputs are subsets of the same representative universe and have Hamming
distance $|S\triangle\widetilde S|$.  Sensitivity is the maximum output
EMD over such single-edge edits within the input class under consideration.
Approximation guarantees follow the convention of \Cref{sec:pre}.

\subsection{Copies Proportional to Vertex Degrees}
\label{sec:minrep-construction}

We replace each source vertex's label group by as many copies as its degree.
If we choose a Label Cover edge uniformly, counting parallel edges
separately, and then independently choose one group copy uniformly at
each endpoint, every left group has the same probability of being chosen,
as does every right group.  This will let us
bound the recovered Label Cover value using only the total expected
number of selected representatives, even when source degrees vary.

Fix a swap-closed Label Cover family on a common bipartite multigraph
$(U,V,E)$ with $m=|E|\ge1$.  Degrees count parallel edges.  Put
\[
 d_z:=\deg(z),\qquad D:=\max_z d_z,\qquad
 Q:=\max\{|\Sigma_U|,|\Sigma_V|\}.
\]
The bounds $D,Q$ and all auxiliary choices below are common to the family.
Isolated source vertices receive a fixed default label in recovery and
create no group.  The spanner reduction in \Cref{sec:spanner-construction}
requires every group to have the same size.  We therefore pad groups to
size $Q$ with isolated dummy representatives.
For each nonisolated $u\in U$ create $d_u$ groups
$A_{u,h}$, $h\in[d_u]$.  In each group, the vertex corresponding to
$a\in\Sigma_U$ is named $(u,h,a)$; add $Q-|\Sigma_U|$ distinct dummy
vertices with fixed names.  Similarly, for each nonisolated $v\in V$
and $h\in[d_v]$, create $B_{v,h}$ with vertices $(v,h,b)$ for
$b\in\Sigma_V$ and $Q-|\Sigma_V|$ dummy vertices.  Dummy vertices
have no compatibility edges and cannot help satisfy a superedge, so they
leave the optimum unchanged.  Thus every group has size $Q$, while
the source alphabets, predicates, and projections retain their original
domains.
There are
\begin{equation}\label{eq:minrep-group-count}
 r=2m\quad\text{groups},\qquad N_{\rm rep}=2mQ\quad\text{representatives}.
\end{equation}

For a source instance $I$, each distinct adjacent pair $u,v$, and
each $h\in[d_u]$, $h'\in[d_v]$, insert the compatibility edge
$(u,h,a)\,(v,h',b)$ for $a\in\Sigma_U$, $b\in\Sigma_V$ exactly when
\begin{equation}\label{eq:minrep-parallel-filter}
 P_u(a)=1\quad\text{and}\quad
 f_e(a)=b
 \quad\text{for every edge $e$ joining $u,v$}.
\end{equation}
Let $\mathcal T_{\rm MR}(I)$ be the resulting Min-Rep instance.

\begin{lemma}[Structure and completeness]\label{lem:minrep-structure}
The construction has at most $mD^2Q$ compatibility edges.  Its supergraph
has maximum degree at most $D^2$, and its compatibility graph has maximum
degree at most $D^2Q$.  If $I$ is satisfiable, every group has an incident
superedge and
\[
 \opt(\mathcal T_{\rm MR}(I))=r=2m.
\]
\end{lemma}
\begin{proof}
For a fixed group $A_{u,h}$, the number of potentially adjacent right
groups is at most
$\sum_{v:\,uv\in E}d_v\le Dd_u\le D^2$, where the sum is over
\emph{distinct} neighbors.  The same bound holds on the right.
A left representative has at most one compatible label in each adjacent
right group.  A right representative has at most $Q$ neighbors in each
adjacent left group.  These observations prove the two degree bounds.
Summing $d_ud_vQ\le D^2Q$ over the at most $m$ distinct adjacent source
pairs proves the edge bound.

Let $\pi$ satisfy $I$.  Selecting the copy of $\pi_U(u)$ in every
$A_{u,h}$ and the copy of $\pi_V(v)$ in every $B_{v,h'}$ covers every
potential group pair by \eqref{eq:minrep-parallel-filter}.
This gives a feasible set of size
$r$.  Every group comes from a positive-degree source vertex, so it
has an incident superedge on a satisfiable input and must contribute
at least one representative.  The lower bound $r$ follows.
\end{proof}

\subsection{Value Recovery and Perturbation Bounds}

We next recover a Label Cover labeling from a representative set.
On satisfiable source instances, a feasible representative set of
expected size at most $ar$ will give expected value at least $a^{-2}$,
for $a\ge1$, where $r$ is the number of groups.
We then bound the effects of changes to the supplied representative set
and swaps of source predicates or projections
to transfer a sensitivity lower bound.

For every representative set, the recovery $R_{\rm MR}$ returns a random
source labeling as follows.
Independently for each nonisolated source vertex $z$, choose
$h\in[d_z]$ uniformly.  If its chosen group contains selected
representatives, choose one of them uniformly and return its source
label, mapping dummy vertices to fixed default source labels.  If the
chosen group contains no selected representative, return a fixed default
source label.  Isolated vertices also receive fixed labels.
The rule uses the source degrees and fixed label names, but no predicate
or projection.

We use the recovery parameters of
\Cref{def:law-sensitivity-pullback-data}: $C_{\mathcal R}$ bounds the
change in recovered output EMD per unit of representative-set distance,
and $D_{\mathcal R}$ bounds the change under one source swap with the
supplied representative set held fixed.

\begin{lemma}[Value and Lipschitz bounds for recovery]
\label{lem:minrep-recovery}
Suppose $I$ is satisfiable and a random set $\bm S$ is a feasible
solution to $\mathcal T_{\rm MR}(I)$ almost surely.  If
$\E|\bm S|\le ar$ for a real $a\ge1$, then
\begin{equation}\label{eq:minrep-value}
 \E\bigl[\val_I(R_{\rm MR}(\bm S))\bigr]\ge a^{-2}.
\end{equation}
Moreover, for arbitrary representative sets $S,\widetilde S$,
\begin{equation}\label{eq:minrep-recovery-lipschitz}
 \EMD(R_{\rm MR}(S),R_{\rm MR}(\widetilde S))
 \le |S\triangle\widetilde S|.
\end{equation}
Thus the recovery has $C_{\mathcal R}\le1$ and $D_{\mathcal R}=0$.
\end{lemma}
\begin{proof}
Condition first on a feasible set $S$.  Let $\ell_{z,h}$ be the number
of selected representatives in group $(z,h)$.  Every such number is
positive, by \Cref{lem:minrep-structure}.  Write $\bm\pi$ for the
labeling returned by $R_{\rm MR}(S)$, and choose an edge
$e=(u,v)$ uniformly from $E$, counting parallel edges separately,
independently of this recovery.  Conditional on $e$, the copy indices
$h\in[d_u]$ and $h'\in[d_v]$ used by the recovery at its endpoints are independent
and uniform.  Feasibility supplies a selected
compatible pair in these groups.  By
\eqref{eq:minrep-parallel-filter}, that pair satisfies every edge
joining $u,v$, including the sampled edge $e$.  Hence
\[
 \Pr[\text{$\bm\pi$ satisfies $e$}\mid S,e,h,h']
 \ge\frac1{\ell_{u,h}\ell_{v,h'}}
 \ge\frac4{(\ell_{u,h}+\ell_{v,h'})^2}.
\]
Each left group $A_{u,h}$ is chosen with probability
$(d_u/m)(1/d_u)=1/m$: the sampled edge has left endpoint $u$ with
probability $d_u/m$, and the recovery chooses copy $h$ with probability
$1/d_u$.  Each right group $B_{v,h'}$ is also chosen with probability
$1/m$ by the same calculation.  Consequently, now including the
randomness of $\bm S$,
\[
 \E[\ell_{u,h}+\ell_{v,h'}]
 =\frac{\E|\bm S|}{m}\le2a.
\]
Apply Jensen's inequality to the positive random variable
$\ell_{u,h}+\ell_{v,h'}$ and the convex function $z\mapsto4/z^2$.
This proves \eqref{eq:minrep-value} using only the expected size of
$\bm S$.

If $S,\widetilde S$ differ at one representative in a group of $z$,
share the copy indices and all recovery randomness outside $z$.
Only the recovered label at $z$ can differ.  A path that adds or deletes
one representative at a time and the triangle inequality prove
\eqref{eq:minrep-recovery-lipschitz}.  The fixed recovery rule proves
$D_{\mathcal R}=0$.  The same bounds on distributions follow from
\Cref{lem:recovery-lifting}.
\end{proof}

\begin{lemma}[Input edits]\label{lem:minrep-input-edits}
A single predicate or projection swap in the source changes at most
\begin{equation}\label{eq:minrep-forward-cost}
 J_0:=2D^3Q
\end{equation}
compatibility edges.  The two resulting Min-Rep instances admit a
shortest single-edge-edit path along which the upper bounds on
the number of edges and the two maximum degrees in
\Cref{lem:minrep-structure} continue to hold.
\end{lemma}
\begin{proof}
A predicate swap at $u$ changes only edges incident to copies of $u$;
a projection swap on $e=(u,v)$ affects only edges between copies of $u$ and $v$.
In either instance, the number of edges incident to copies of $u$ is at most
\[
 d_u Q\sum_{v:\,uv\in E}d_v\le D^3Q.
\]
The two compatibility edge sets therefore differ in at most $2D^3Q$ edges.
Delete edges present only in the first instance, then insert those
present only in the second.  Each intermediate compatibility graph
is a subgraph of one endpoint, and its supergraph is also a subgraph
of that endpoint's supergraph.  The upper bounds on the number of edges and the degrees
are therefore preserved.
\end{proof}

\begin{theorem}[Label Cover to Min-Rep sensitivity]
\label{thm:minrep-transfer}
Let $\mathcal I$ be a swap-closed Label Cover family on a common bipartite
multigraph $(U,V,E)$ and common alphabets, with $m=|E|\ge1$.
Let $D$ be its maximum degree, $Q=\max\{|\Sigma_U|,|\Sigma_V|\}$,
and $r=2m$ the number of groups in $\mathcal T_{\rm MR}(I)$.
Suppose every algorithm
on $\mathcal I$ with expected Label Cover value at least $1/g$ on
satisfiable members has swap sensitivity at least $L$, where $g>1$.
Let $\mathcal J$ consist of the images $\mathcal T_{\rm MR}(I)$,
$I\in\mathcal I$, together with the instances on a path from
\Cref{lem:minrep-input-edits} for each neighboring source pair.
Any Min-Rep algorithm defined on $\mathcal J$ that, for every satisfiable
$I\in\mathcal I$, returns an almost-surely feasible set of expected
size at most $\sqrt g\,r$ on $\mathcal T_{\rm MR}(I)$ has
single-compatibility-edge sensitivity on $\mathcal J$ at least
\begin{equation}\label{eq:minrep-transfer}
 \frac{L}{2D^3Q}.
\end{equation}
In particular this holds for any $\sqrt g$-approximation algorithm
for Min-Rep on $\mathcal J$.
\end{theorem}
\begin{proof}
Let $A$ be a Min-Rep algorithm satisfying the hypothesis, and let $B$ be
the Label Cover algorithm that, on input $I\in\mathcal I$, returns
$R_{\rm MR}(A(\mathcal T_{\rm MR}(I)))$.
By \Cref{lem:minrep-recovery}, $B$ has expected value at least $1/g$
on every satisfiable $I\in\mathcal I$.  For two neighboring source
instances $I,\widetilde I$, follow the path in $\mathcal J$ from
$\mathcal T_{\rm MR}(I)$ to $\mathcal T_{\rm MR}(\widetilde I)$ supplied
by \Cref{lem:minrep-input-edits}; its length is at most $2D^3Q$.
Couple the outputs of $A$ along this path and apply
\eqref{eq:minrep-recovery-lipschitz}.  The swap sensitivity of $B$ on
$\mathcal I$ is therefore at most $2D^3Q$ times the sensitivity of $A$
on $\mathcal J$.  Since $\SwapSens(B,\mathcal I)\ge L$ by hypothesis,
this proves the claimed lower bound on the sensitivity of $A$.
The final assertion uses
$\opt(\mathcal T_{\rm MR}(I))=r$ on satisfiable inputs.
\end{proof}

The spanner reduction need not provide a multiplicative Min-Rep
approximation on instances with many groups having no incident superedge.
To apply \Cref{thm:minrep-transfer}, it suffices that, on the image of
every satisfiable Label Cover instance, the spanner reduction returns
an almost-surely feasible representative set of expected size at most
$\sqrt g\,r$.  The recovered labeling needs no value guarantee on the
other source instances.

\subsection{Consequences of the Low-Soundness Theorems}

We first use the Dinur--Harsha theorem to obtain a polynomial sensitivity
lower bound for every sufficiently large number of representative vertices.
We then apply the IKW theorem, which trades a larger approximation factor
against a smaller sensitivity exponent at the sizes of its construction.

\begin{corollary}[Min-Rep with a prescribed number of representative vertices]
\label{cor:minrep-dh}
There are absolute constants $\beta,\delta>0$, $\kappa\in(0,1/2)$,
and $C>0$ such that, for all sufficiently large $n$, every
$\log^\beta n$-approximation algorithm for Min-Rep with exactly $n$
representative vertices has sensitivity $\Omega(n^\delta)$.
The bound is witnessed by two instances on the same $n$ representative
vertices, with the same partition into groups, that differ by one
compatibility edge.  In both instances, group sizes and the maximum
degrees of the compatibility graph and the supergraph are at most
$\exp(C\log^\kappa n)$.
\end{corollary}
\begin{proof}
Let $B,\delta_{\rm LC}>0$ be from \Cref{thm:intro-label-cover}.
For source parameters $s,g$ in that theorem's range, let
$\mathcal I_{s,g}$ be the supplied family on a common multigraph
$(U,V,E)$ with $|U|=s$.  Write $m=|E|$, and let $D,Q$ be its maximum
degree and maximum alphabet size as in \Cref{sec:minrep-construction}.
Since $m\le sD$, there are constants $a\ge1$, $C_0>0$ such that,
uniformly for these families with $g\ge2$,
\[
 D^3Q\le\exp(C_0g^a),\qquad
 2mQ\le s\exp(C_0g^a).
\]
Choose $0<\theta<B$ with $a\theta<1/2$, put
$s=\lfloor n^{1/2}\rfloor$, and set $g=\log^\theta s$.
For sufficiently large $n$, the source theorem applies and the
number of representative vertices is less than $n$.  Pad to $n$ by adding fixed
isolated singleton groups on the left.  Padding changes neither the
optimum nor the edge-edit distance; recovery ignores the new groups.
Choose any $0<\beta<\theta/2$.  Then
$\log^\beta n\le\sqrt g$ for all sufficiently large $n$.
By \Cref{thm:minrep-transfer}, the sensitivity is at least
\[
 \frac{\Omega(s^{\delta_{\rm LC}})}{\exp(C_0g^a)}
 =n^{\delta_{\rm LC}/2-o(1)}.
\]
Take a constant $\delta<\delta_{\rm LC}/2$ and
$\kappa=a\theta<1/2$.  Every group has size at most $Q$, and
\Cref{lem:minrep-structure} bounds the maximum degrees of the supergraph
and the compatibility graph by $D^2$ and $D^2Q$, respectively.
These three quantities are at most $D^3Q\le\exp(C_0g^a)$.
Since $g^a=\log^\kappa s\le\log^\kappa n$, they are at most
$\exp(C_0\log^\kappa n)$.
The padding and all names are independent of the source predicates
and projections, so the two instances witnessing the lower bound have
the same $n$ representative vertices and differ by one compatibility edge.
\end{proof}

For the IKW consequence, let $N$ be the number of original 2-CSP variables and let
$k=k(N)$ be sufficiently large with $k=o(\log N)$, as in
\Cref{thm:intro-label-cover-sensitivity-IKW}.  Write $C_{\rm IKW}$
for its constant $C$ in the value guarantee, and $\eta$ for its
sensitivity exponent.  For the resulting Label Cover family, let $s$
be the total number of vertices and $D,Q$ its maximum degree and alphabet
size.  With $g=\exp(C_{\rm IKW}\sqrt{k})$ and $D,Q\le2^{O(k)}$,
\Cref{thm:minrep-transfer} gives
\[
 \text{approximation }\exp(C_{\rm IKW}\sqrt{k}/2),\qquad
 \text{sensitivity at least }s^{\eta/k}/2^{O(k)}.
\]
Since $s=N^{\Theta(k)}$ and $k=o(\log N)$ in that theorem,
$k^2=o(\log s)$, so the denominator can be absorbed in a bound
$s^{\Omega(1/k)}$.  We record the corresponding spanner consequence
in \Cref{cor:spanner-ikw}.
 \section{Sensitivity Lower Bounds for Unweighted Directed 3-Spanner}
\label{sec:directed-spanner}

We reduce Min-Rep to minimum unweighted directed 3-spanner and use
the reduction to transfer the sensitivity lower bounds of \Cref{sec:minrep}.
All graphs in this section are simple directed graphs; all arcs have
unit length and unit cost.  For $G=(V,A)$, a \emph{directed 3-spanner}
is an arc set $H\subseteq A$ such that every $(u,v)\in A$ has a directed
$u$--$v$ path of length at most three in $(V,H)$.  This is equivalent
to $d_H(u,v)\le3d_G(u,v)$ for every reachable ordered pair, by replacing
each arc of a shortest path.  Thus $\opt(G)$ is the
minimum number of arcs in a directed 3-spanner of $G$.

We approximate the number of arcs while requiring stretch at most three.
Outputs are embedded in the common universe of all
ordered pairs of distinct vertex names, and their distance is
$|H\triangle\widetilde H|$.  Input sensitivity is under insertion or
deletion of one arc.  A degree bound means
$\Delta(G)=\max_v(\deg_G^+(v)+\deg_G^-(v))$.
The constructions below are directed acyclic graphs (DAGs), and we
also preserve weak connectivity, meaning connectivity after forgetting
arc directions.

For every constant approximation factor $\alpha\ge1$, we obtain sensitivity
$\Omega(n^\delta)$ on weakly connected $n$-vertex DAGs of degree at most
a constant $\Delta_\alpha$, for an absolute $\delta>0$ and all sufficiently
large $n$. We also obtain a polynomial sensitivity lower bound for approximation
factor $\log^\beta n$ with maximum degree $\exp(O(\log^\kappa n))$,
for absolute $\beta>0$ and $\kappa\in(0,1/2)$.
These degree bounds are essential for the next section: they limit how
many output arcs can be affected by a change within a local neighborhood.

Classical reductions from Min-Rep to undirected spanners use large
girth of the supergraph to exclude unwanted short routes
\cite{dinitz2012largegirth}.  We instead use directions and two-hop
auxiliary routes to characterize every three-hop alternative explicitly.
The auxiliary vertices are local to individual groups.  There is no
vertex adjacent to all representatives, and the maximum degree depends
only on the degrees and alphabet sizes of the source Label Cover instance.

\subsection{The Directed Construction}
\label{sec:spanner-construction}

The construction represents a compatible pair by a three-arc path from
a new source to a new sink. Selecting the two representatives corresponds
to including the first and last arcs; the middle arc records their
compatibility. Auxiliary two-arc routes span the first and last arcs
when they are unselected. Using either auxiliary route in place of one
of these arcs gives a path of length four, which cannot span the
source-to-sink demand arc within stretch three.
We repeat the sources and sinks so that the cost of selecting
representatives can absorb the cost of the other arcs.

Consider a Min-Rep instance $J$ with fixed groups $(A_i)_{i\in L}$ and
$(B_j)_{j\in R}$, where $r=|L|+|R|\ge2$, every group has size $Q\ge1$,
and $M=|E_J|$.  Let $M_{\max}$ be a common upper bound on the number
of compatibility edges in the family.  The number of repetitions is
an integer $x$ satisfying
\begin{equation}\label{eq:spanner-copy-choice}
 x\ge Q+\left\lceil M_{\max}/r\right\rceil+4.
\end{equation}
In particular, $x$ does not change when a compatibility edge is edited.
Retain all representative vertices.  For every left group $A_i$ create
sources $s_i^p$, $p\in[x]$, and a hub $h_i$.  For every right group
$B_j$ create sinks $t_j^p$, $p\in[x]$, and a hub $\bar h_j$.
The representatives and hubs are shared by all $x$ repetitions.
Insert the following arcs:
\[
\begin{aligned}
 A_{\rm comp}&:=\{a\to b:(a,b)\in E_J\},\\
 A_{\rm sel}&:=\bigcup_{i\in L}\{s_i^p\to a:a\in A_i,\ p\in[x]\}
       \ \cup\ \bigcup_{j\in R}\{b\to t_j^p:b\in B_j,\ p\in[x]\},\\
 A_{\rm aux}&:=\bigcup_{i\in L}\{s_i^p\to h_i,\ h_i\to a:a\in A_i,\ p\in[x]\}\\
       &\hspace{1.2em}\cup\ \bigcup_{j\in R}\{b\to\bar h_j,\ \bar h_j\to t_j^p:
                                      b\in B_j,\ p\in[x]\},\\
 A_{\rm dem}&:=\{s_i^p\to t_j^p:(i,j)\in\mathcal E_J,\ p\in[x]\}.
\end{aligned}
\]
Each arc is included only once, even if it occurs for several values of
$p$.  The selection and auxiliary arcs are fixed
throughout the family.  Demand arcs are present precisely for actual
superedges, including when $J$ is obtained from an unsatisfiable
Label Cover instance.

To ensure weak connectivity, order all $r$ hubs with the left hubs
first and the right hubs second.  Between each consecutive pair insert
a directed path of length four from the earlier hub to the later hub,
using three private new vertices.
Its length prevents it from participating in an alternative path of
length at most three from a source to a sink.
Let $A_{\rm con}$ be these $4(r-1)$ arcs, and let
$\mathcal T_{\rm sp}(J)$ be the resulting graph.
The private vertices and paths are fixed across the family.

\begin{figure}[t]
\centering
\begin{tikzpicture}[scale=1.05, every node/.style={font=\small}]
\node (s) at (0,0) {$s_i^p$};
\node (a) at (3,0) {$a\in A_i$};
\node (b) at (6,0) {$b\in B_j$};
\node (t) at (9,0) {$t_j^p$};
\node (h) at (1.5,-1.25) {$h_i$};
\node (k) at (7.5,-1.25) {$\bar h_j$};
\draw[->] (s) -- node[above] {selection} (a);
\draw[->] (a) -- node[above] {compatibility} (b);
\draw[->] (b) -- node[above] {selection} (t);
\draw[->] (s) to[bend left=23] node[above] {demand} (t);
\draw[->] (s) -- (h); \draw[->] (h) -- (a);
\draw[->] (b) -- (k); \draw[->] (k) -- (t);
\node at (1.5,-1.75) {two-hop auxiliary route};
\node at (7.5,-1.75) {two-hop auxiliary route};
\end{tikzpicture}
\caption{The arcs associated with a compatible pair $(a,b)\in E_J$
and a repetition index $p\in[x]$.
All arcs have unit length and unit cost.  Omitting a demand arc requires
the three-arc route through a compatible pair and both selection arcs;
using either auxiliary route would require at least four arcs.
The four-arc paths connecting consecutive hubs are not shown.}
\label{fig:directed-spanner-gadget}
\end{figure}

\begin{lemma}[Structure and three-hop routes]\label{lem:spanner-structure}
Suppose the compatibility graph has maximum degree at most $\Delta_{\rm rep}$
and the supergraph has maximum degree at most $\Delta_{\rm grp}$.
The graph $\mathcal T_{\rm sp}(J)$ is a weakly connected DAG whose
number of vertices $n_0$ and maximum total degree satisfy
\begin{equation}\label{eq:spanner-size-degree}
 n_0=r(Q+x+4)-3,\qquad
 \Delta(\mathcal T_{\rm sp}(J))\le x+Q+\Delta_{\rm rep}+\Delta_{\rm grp}+4.
\end{equation}
For each demand arc $s_i^p\to t_j^p$, any alternative directed path of
length at most three is exactly
\begin{equation}\label{eq:spanner-canonical-route}
 s_i^p\to a\to b\to t_j^p,
 \qquad a\in A_i,\ b\in B_j,\ (a,b)\in E_J.
\end{equation}
\end{lemma}
\begin{proof}
There are $rQ$ representatives, $rx$ sources and sinks in total, $r$ hubs, and
$3(r-1)$ private connector vertices.  Each group is connected to its
hub after forgetting directions, and the connector paths join all hubs.
A topological ordering puts sources first, then the left hubs in their
fixed order, then left representatives, right representatives, the right
hubs in their fixed order, and sinks.  The internal vertices of each
connector path are placed in path order between its endpoints.
Thus no directed cycle is introduced.

Each source or sink has degree at most $Q+1+\Delta_{\rm grp}$, each representative
at most $x+1+\Delta_{\rm rep}$, each hub at most $x+Q+2$, and each internal
vertex of a connector path at most two.  These give \eqref{eq:spanner-size-degree}.

An alternative path of length at most three from a source to a sink
cannot contain an internal vertex of a connector path, since traversing
that connector path requires four arcs.  Sinks have no outgoing arcs, so a
path using a demand arc must end immediately at that demand's sink.
A route using $h_i$ or $\bar h_j$ in place of a selection arc has length
at least four.  The only remaining possibility of length at most three
is \eqref{eq:spanner-canonical-route}.
\end{proof}

Let $B_J$ be the number of compatibility, auxiliary, and connector arcs:
\[
 B_J:=M+r(Q+x)+4(r-1).
\]
Our spanner will include all these arcs, with the remaining cost
determined by the selected representatives.  We also record bounds on the
number of vertices and the optimum to control the cost of padding later.

\begin{lemma}[A spanner from representatives]\label{lem:spanner-completeness}
If $S$ is a feasible Min-Rep set, the graph has a 3-spanner of size
\begin{equation}\label{eq:spanner-completeness-general}
 B_J+x|S|.
\end{equation}
In particular, if $\opt(J)=r$, then
\begin{equation}\label{eq:spanner-opt-satisfiable}
 \opt(\mathcal T_{\rm sp}(J))\le3xr.
\end{equation}
Also $n_0\le2rx$ and
$\opt(\mathcal T_{\rm sp}(J))\ge n_0-1$.
\end{lemma}
\begin{proof}
Include every compatibility, auxiliary, and connector arc.  Include
$s_i^p\to a$ exactly when $a\in S$, and $b\to t_j^p$ exactly when
$b\in S$.  Include no demand arc.  Each omitted selection arc has
its two-hop auxiliary route.  Each demand has a three-hop route through
a compatible pair in $S$.  All other input arcs are included, so this
is a 3-spanner of the stated size.  By
\eqref{eq:spanner-copy-choice}, $M+rQ+4(r-1)\le rx$, and hence
$B_J+xr\le3xr$.  The same inequality $x\ge Q+4$ gives $n_0\le2rx$.
Finally, replacing each input arc by a directed path preserves weak
connectivity, so every 3-spanner has at least $n_0-1$ arcs.
\end{proof}

\subsection{Recovery by Sampling a Repetition}
\label{sec:spanner-recovery}

We recover representatives by choosing an index $p\in[x]$ uniformly and
examining the selection and demand arcs with that index.  A single
compatibility edit can create or remove $x$ demand arcs, but each changed
selection or demand arc in the output is used with probability $1/x$.
This averaging will cancel the
factor $x$ in the sensitivity transfer.

For every left group $A_i$ and right group $B_j$, fix an ordering of
$A_i\times B_j$ independently of the compatibility graph.  Let
$w_J(i,j)\subseteq A_i\cup B_j$ be the set of endpoints of the first
edge in $E_J\cap(A_i\times B_j)$ under this ordering, or
$\varnothing$ if this block is empty.
For an arbitrary set $H$ of ordered pairs on the construction's common
vertex set and $p\in[x]$, define
\begin{align*}
 S_p(J,H):={}&\bigcup_{i\in L}\{a\in A_i:s_i^p\to a\in H\}
 \ \cup\ \bigcup_{j\in R}\{b\in B_j:b\to t_j^p\in H\}\\
 &\cup\bigcup_{\substack{i\in L,\,j\in R:\\s_i^p\to t_j^p\in H}}w_J(i,j).
\end{align*}
Pairs not of these forms are ignored.  In particular this rule is defined
also for outputs infeasible for the current input.  Let $R_{{\rm sp},J}(H)$
choose $p\in[x]$ uniformly and return $S_p(J,H)$.

\begin{lemma}[Feasibility, size, and recovery sensitivity]
\label{lem:spanner-recovery}
If $H$ is a 3-spanner of $\mathcal T_{\rm sp}(J)$, then every
$S_p(J,H)$ is feasible for Min-Rep, and
\begin{equation}\label{eq:spanner-recovery-size}
 \E|R_{{\rm sp},J}(H)|\le\frac{2|H|}{x}.
\end{equation}
For arbitrary output sets $H,\widetilde H$ on the common domain,
\begin{equation}\label{eq:spanner-recovery-fixed-input}
 \EMD(R_{{\rm sp},J}(H),R_{{\rm sp},J}(\widetilde H))
 \le\frac2x|H\triangle\widetilde H|.
\end{equation}
If $J,\widetilde J$ differ by one compatibility-edge edit, then for
every fixed $H$,
\begin{equation}\label{eq:spanner-recovery-input-change}
 \EMD(R_{{\rm sp},J}(H),R_{{\rm sp},\widetilde J}(H))\le4.
\end{equation}
\end{lemma}
\begin{proof}
Fix a superedge $(i,j)$ and an index $p\in[x]$.  If its demand arc is in $H$,
the rule inserts a compatible pair.  Otherwise
\Cref{lem:spanner-structure} supplies the route
\eqref{eq:spanner-canonical-route} in $H$, and both its representatives
are selected by the rule.  This proves feasibility.
Every selection arc is charged once when summing $|S_p|$ over $p$;
every demand arc is charged at most twice.  Thus
$\sum_p|S_p(J,H)|\le2|H|$, proving \eqref{eq:spanner-recovery-size}.

For \eqref{eq:spanner-recovery-fixed-input}, use the same $p$ in both
recoveries.  A change of one selection arc affects at most one
representative for one value of $p$, and a demand arc affects at most two
for one value of $p$.  Every other arc has no effect.  Summing over the arcs in
$H\triangle\widetilde H$ proves the bound, also for infeasible outputs.

An input edit changes $w_J(i,j)$ for at most one group pair.  The old
and new sets each contain at most two representatives.  Under the
same choice of $p$, at most four representatives can therefore change.
This includes deletion of the last compatibility edge, when one of
the two sets is empty.  Averaging proves
\eqref{eq:spanner-recovery-input-change}.
\end{proof}

For this reduction, $C_{\mathcal T}$ bounds the number of arc edits caused
by one edit to the compatibility graph.  The parameter $C_{\mathcal R}$
bounds the change in recovered output EMD per unit of arc-set distance.
The parameter $D_{\mathcal R}$ bounds the effect on recovery of one
compatibility-edge edit with the output arc set held fixed, as in
\Cref{def:law-sensitivity-pullback-data}.

\begin{lemma}[Sensitivity-transfer parameters]\label{lem:spanner-transfer-data}
A compatibility-edge edit changes at most $x+1$ arcs of the constructed
graph.  The transformation and recovery have parameters
\[
 C_{\mathcal T}\le x+1,\qquad C_{\mathcal R}\le2/x,
 \qquad D_{\mathcal R}\le4,
 \qquad C_{\mathcal T}C_{\mathcal R}\le4.
\]
If two instances $J,\widetilde J$ in the Min-Rep family have compatibility graphs of
maximum degree at most $\Delta_{\rm rep}$ and supergraphs of maximum
degree at most $\Delta_{\rm grp}$, then
$G=\mathcal T_{\rm sp}(J)$ and
$\widetilde G=\mathcal T_{\rm sp}(\widetilde J)$ can be joined by a
shortest arc-edit path of weakly connected DAGs satisfying the degree
bound in \eqref{eq:spanner-size-degree}.
\end{lemma}
\begin{proof}
The compatibility arc itself changes.  Only when its block becomes
empty or nonempty do the $x$ demand arcs change as well.
All other arcs and all vertices are fixed.  This proves the bound on
the number of arc edits; the recovery bounds follow from \Cref{lem:spanner-recovery}.
To obtain a shortest arc-edit path from $G$ to $\widetilde G$, first
delete the arcs present only in $G$, then insert those present only in
$\widetilde G$.  Each intermediate graph is a subgraph of $G$ or
$\widetilde G$ and contains the weakly connected spanning subgraph
with arc set $A_{\rm aux}\cup A_{\rm con}$.  All intermediate graphs
respect the same acyclic orientation and satisfy the common degree bound.
\end{proof}

Thus the product $C_{\mathcal T}C_{\mathcal R}$ stays bounded even when
a superedge appears or disappears.  The additive term $D_{\mathcal R}$
accounts for the input dependence of the compatible pair $w_J(i,j)$
selected for a demand arc in $H$.

\subsection{Connected Padding to a Prescribed Number of Vertices}
\label{sec:spanner-padding}

To obtain a lower bound at every prescribed graph size, we must add
vertices while controlling their contribution to the optimum.
Repeating the graph before adding a short tail keeps this extra cost
within a constant factor and allows recovery by sampling one copy.

Let $G$ be a weakly connected DAG on a fixed set of $n_0\ge2$ vertices,
and fix an anchor vertex.  For every $n\ge n_0$, write
$n=qn_0+h$, where $q\ge1$ and $0\le h<n_0$.
Construct $P_n(G)$ from $q$ disjoint named copies of $G$ by joining
their anchors with arcs oriented from copy $j$ to copy $j+1$ for
$j\in[q-1]$.  Append a directed path of $h$ new vertices starting at
the last anchor.  All names, anchors, and added arcs are fixed
independently of the input arcs of $G$.

For any set $H$ of ordered pairs of distinct vertices of $P_n(G)$,
let $U_n(H)$ choose one of the $q$ copies uniformly, retain only pairs
whose endpoints both lie in that copy, and return them under the
original vertex names of $G$.  Thus $U_n$ is defined on all arc sets
on the common vertex universe and is independent of the input arcs.

\begin{lemma}[Connected replication and padding]
\label{lem:spanner-connected-padding}
For the construction and recovery above, the following properties hold.
\begin{enumerate}
\item $P_n(G)$ is a weakly connected DAG on exactly $n$ vertices of
maximum total degree at most
$\Delta(G)+2$.
\item Every 3-spanner of $P_n(G)$ restricts to a 3-spanner on each of the $q$
fixed copies of $G$, and
\[
 \opt(P_n(G))
 =q\opt(G)+(q-1)+h
 \le2q\opt(G).
\]
\item For every arc set $H$ on the common vertex universe,
$\E|U_n(H)|\le |H|/q$.  The recovery $U_n$ is $1/q$-Lipschitz
in output EMD on this entire domain and is independent of the input arcs.
An edit of one arc of $G$ changes exactly $q$ arcs of $P_n(G)$.
\end{enumerate}
\end{lemma}
\begin{proof}
The arcs joining copies and forming the tail are bridges in the
underlying undirected graph, and every one must be
included in a 3-spanner.  A directed path that leaves a copy cannot
return to it.  Hence an input arc inside that copy can only be
spanned by arcs of that copy.  This proves the optimum identity and
the restriction property.  Conversely, the union of spanners of the
copies and all added arcs spans the padded graph.

Since $\opt(G)\ge n_0-1$,
\[
 (q-1)+h\le q+n_0-2\le q(n_0-1)
 \le q\opt(G),
\]
proving the factor-two bound.  The degree and DAG assertions follow
from the construction.  To compare $U_n(H)$ and $U_n(\widetilde H)$,
choose the same copy in both recoveries.  Each changed pair internal
to a copy is seen with probability $1/q$, while all other pairs are
ignored.  This proves the Lipschitz bound; the same counting applied
to $H$ proves the size bound.  A change to one input arc of $G$ changes
that arc in each of the $q$ copies, proving the final claim.
\end{proof}

For the graph family in \Cref{sec:spanner-construction}, the subgraph
with arc set $A_{\rm aux}\cup A_{\rm con}$ inside each copy remains
present throughout every deletion-then-insertion
path.  Thus the preceding padding also preserves weak connectivity,
acyclicity, and the degree bound along the paths needed for sensitivity
transfer.  One need not assume that intermediate graphs are themselves
images of Min-Rep instances.

\subsection{Sensitivity Transfer from Label Cover}

We now compose the three recoveries to obtain a Label Cover algorithm
from a spanner approximation algorithm.  To guarantee the required
expected Label Cover value, it suffices to bound the expected size of
the recovered representative set on images of satisfiable Label Cover
instances, as in \Cref{thm:minrep-transfer}.
Fix a common source graph and alphabets as in \Cref{sec:minrep-construction}:
$m$ is its number of edges, $D$ its maximum degree, $Q$ the larger alphabet
size, and $r=2m$ the number of Min-Rep groups.  Choose
\begin{equation}\label{eq:spanner-lc-local-parameters}
 x=(D^2+1)Q+4,\qquad
 n_0=r(Q+x+4)-3,\qquad
 \Delta_*:=x+Q+D^2Q+D^2+6.
\end{equation}
The choice of $x$ satisfies \eqref{eq:spanner-copy-choice}, since
$M_{\max}=mD^2Q$ is a common upper bound.
All graphs below have maximum total degree at most $\Delta_*$, including
padding and intermediate arc-edit paths.

\begin{proposition}[Sensitivity transfer to directed 3-spanners with a prescribed number of vertices]
\label{prop:spanner-lc-transfer}
Let $\mathcal I$ be a swap-closed Label Cover family on the common
source graph above, and suppose expected value $1/g$ on its satisfiable
members requires swap sensitivity at least $L$.
Fix any $n\ge n_0$ and any $\alpha\ge1$ with
\begin{equation}\label{eq:spanner-alpha-g}
 144\alpha^2\le g.
\end{equation}
Any $\alpha$-approximation algorithm for minimum directed 3-spanner on
all weakly connected $n$-vertex DAGs of maximum total degree at most
$\Delta_*$ has single-arc sensitivity at least
\begin{equation}\label{eq:spanner-end-to-end-sensitivity}
 \frac{L}{4J_0}-1,\qquad J_0=2D^3Q.
\end{equation}
The witnessing pair is in this same graph class.
\end{proposition}
\begin{proof}
Fix an $\alpha$-approximation algorithm $A$ on the stated graph class.
For each $I\in\mathcal I$, construct
\[
 J_I=\mathcal T_{\rm MR}(I),\qquad
 G_I=\mathcal T_{\rm sp}(J_I),\qquad
 \widehat G_I=P_n(G_I).
\]
All vertex sets, the number $q$ of padding copies, and the choices of
names and anchors are the same for every $I\in\mathcal I$.
Define a Label Cover algorithm $B$ by
\[
 B(I):=R_{\rm MR}\bigl(R_{{\rm sp},J_I}(U_n(A(\widehat G_I)))\bigr).
\]
Here $U_n$ restricts the spanner output to a uniformly chosen copy of
$G_I$ and restores its original vertex names.  Writing $H$ for this
restricted arc set, $R_{{\rm sp},J_I}$ chooses $p\in[x]$ uniformly
and returns $S_p(J_I,H)$, and
$R_{\rm MR}$ recovers a source labeling.  Each recovery uses fresh
randomness conditional on its supplied output.

If $I$ is satisfiable, \Cref{lem:minrep-structure} gives a Min-Rep
solution of size $r$, and
\Cref{lem:spanner-completeness,lem:spanner-connected-padding} give
$\opt(\widehat G_I)\le6qxr$.
The composition $R_{{\rm sp},J_I}\circ U_n$ produces an almost-surely feasible Min-Rep set
$\bm S$ with
\[
 \E|\bm S|\le\frac{2}{qx}\,
       \E|A(\widehat G_I)|\le12\alpha r.
\]
By \Cref{lem:minrep-recovery}, the expected value of the recovered
labeling $B(I)$ is at least $(12\alpha)^{-2}\ge1/g$.
On other source instances $B$ remains defined, but we require no bound
on the expected value of $B(I)$.  The approximation guarantee for
$A$ continues to hold on the entire stated graph class.

For one compatibility-edge edit, the padded graphs differ in at
most $q(x+1)$ arcs. The composition $R_{{\rm sp},J_I}\circ U_n$ has Lipschitz
constant $2/(qx)$.  At a fixed supplied arc set, this input edit changes
the recovered distribution by EMD at most four.
The subsequent map $R_{\rm MR}$ is 1-Lipschitz and
input-independent.  A source swap can be realized by at most $J_0$
compatibility-edge edits.  If $S_A$ is the spanner algorithm's
sensitivity on the stated graph class, triangle inequalities along
the corresponding valid graph paths therefore give
\[
 \SwapSens(B,\mathcal I)
 \le 2J_0(1+1/x)S_A+4J_0
 \le4J_0S_A+4J_0.
\]
This is the sensitivity-transfer argument of
\Cref{lem:law-sensitivity-pullback}, with paths confined to the graph
class just verified.  Its left-hand side is at least $L$.
Rearranging proves \eqref{eq:spanner-end-to-end-sensitivity}.
\end{proof}

\subsection{Consequences of the Low-Soundness Theorems}
\label{sec:spanner-dh-results}

Our Label Cover lower bound via Dinur--Harsha
(\Cref{thm:intro-label-cover}) gives a constant positive sensitivity exponent
at every sufficiently large graph size, including for approximation
factors that are small powers of $\log n$.  Our lower bound via IKW
(\Cref{thm:intro-label-cover-sensitivity-IKW}) gives a different trade-off
for the graph sizes produced by the construction: approximation factor
$\exp(c_1\sqrt{k})$ requires sensitivity $n^{c_2/k}$ for absolute
constants $c_1,c_2>0$, where $k$ is the IKW direct-product parameter.
\Cref{tab:spanner-sensitivity-comparison} summarizes the sensitivity
and degree guarantees.

\begin{table}[ht!]
\centering
\caption{Sensitivity lower bounds for minimum unweighted directed
3-spanner on weakly connected $n$-vertex DAGs.
The first two rows use DH and hold for every sufficiently large $n$,
with a constant sensitivity exponent and $\kappa\in(0,1/2)$.
The third row uses IKW and holds at constructed sizes $n=N^{\Theta(k)}$,
with $k$ sufficiently large and $k=o(\log N)$ as in
\Cref{thm:intro-label-cover-sensitivity-IKW}.
Here $c_1>0$ is an absolute constant.}
\label{tab:spanner-sensitivity-comparison}
\small
\begin{tabular}{llll}
\toprule
Approximation factor & Maximum total degree & Sensitivity lower bound & Reference \\
\midrule
Any constant $\alpha\ge1$ & $O_\alpha(1)$ & $n^{\Omega(1)}$ & \Cref{cor:spanner-dh-specializations} \\
$\log^\beta n$ & $\exp(O(\log^\kappa n))$ & $n^{\Omega(1)}$ & \Cref{cor:spanner-dh-specializations} \\
$\exp(c_1\sqrt{k})$ & $2^{O(k)}$ & $n^{\Omega(1/k)}$ & \Cref{cor:spanner-ikw} \\
\bottomrule
\end{tabular}
\end{table}

\begin{theorem}[Directed 3-spanner sensitivity via Dinur--Harsha]
\label{thm:spanner-dh-parameterized}
There are absolute constants $a\ge1$, $C>0$, $\delta>0$, and
$\theta>0$ with $a\theta<1/2$ such that the following holds uniformly
for all sufficiently large $n$ and all real
\[
 144\le g\le\log^\theta n.
\]
Every $(\sqrt g/12)$-approximation algorithm for minimum unweighted
directed 3-spanner on weakly connected $n$-vertex DAGs of maximum total
degree at most
\begin{equation}\label{eq:spanner-dh-degree}
 \Delta_g:=\left\lceil\exp(Cg^a)\right\rceil
\end{equation}
has single-arc sensitivity $\Omega(n^\delta)$.
Both input graphs witnessing the sensitivity lie in the stated class.
\end{theorem}
\begin{proof}
Let $B,\delta_{\rm LC}>0$ be from \Cref{thm:intro-label-cover}.
For source parameters $s,g$ in its range, let $\mathcal I_{s,g}$ be
the supplied family on a common graph $(U,V,E)$ with $|U|=s$.
Write $m=|E|$ and let $D,Q$ be its maximum degree and alphabet size.
Use these parameters to define $J_0,\Delta_*,n_0$ as in
\Cref{prop:spanner-lc-transfer} and \eqref{eq:spanner-lc-local-parameters}.
Since $m\le sD$, the source theorem's uniform structural bounds allow
constants $a\ge1,C_1>0$ such that, for these families with $g\ge144$,
\[
 J_0\le\exp(C_1g^a),\qquad
 \Delta_*\le\exp(C_1g^a),\qquad
 n_0\le s\exp(C_1g^a).
\]
Increase $a,C_1$ once to absorb all absolute constants.
Choose $0<\theta<B$ with $a\theta<1/2$ and put
$s=\lfloor n^{1/2}\rfloor$.
For all sufficiently large $n$, uniformly over the claimed $g$-range,
$g\le\log^B s$ and
$\exp(C_1g^a)=n^{o(1)}$.  Thus $n_0\le n$ and the source theorem
supplies $L=\Omega(s^{\delta_{\rm LC}})$.
Choose $C\ge C_1$, so that the degree bound $\Delta_*$ of
\Cref{prop:spanner-lc-transfer} is at most $\Delta_g$.
Fix an algorithm $A$ with approximation factor $\alpha=\sqrt g/12$
on the graph class in the statement, and let $S_A$ be its sensitivity
on that class.  Applying the
proposition to the restriction of $A$ to degree at most $\Delta_*$ gives
\[
 S_A\ge
 \frac{\Omega(s^{\delta_{\rm LC}})}{4\exp(C_1g^a)}-1
 =n^{\delta_{\rm LC}/2-o(1)}.
\]
Choose any constant $0<\delta<\delta_{\rm LC}/2$.
The graph paths in the proposition stay within degree $\Delta_*\le \Delta_g$,
which proves the claimed witness assertion and uniformity.
\end{proof}

\begin{corollary}[Growing and constant approximation factors]
\label{cor:spanner-dh-specializations}
There are absolute constants $\beta\in(0,1)$,
$\kappa\in(0,1/2)$, $\delta>0$, and $C>0$ such that
$\log^\beta n$-approximation for minimum unweighted directed
3-spanner has sensitivity $\Omega(n^\delta)$ already on weakly
connected DAGs of maximum total degree at most
$\exp(C\log^\kappa n)$.
Moreover, for every constant $\alpha\ge1$ there is a constant degree
bound $\Delta_\alpha$ such that $\alpha$-approximation on weakly connected
$n$-vertex DAGs of degree at most $\Delta_\alpha$ has sensitivity
$\Omega(n^\delta)$ for all sufficiently large $n$.
\end{corollary}
\begin{proof}
For the first assertion set $g=\log^\theta n$ in
\Cref{thm:spanner-dh-parameterized}, take
$\kappa=a\theta<1/2$, and choose any $0<\beta<\theta/2$.
For large $n$, $\log^\beta n\le\frac{\log^{\theta/2} n}{12}$, so the
same lower bound applies to this stronger approximation guarantee.
For the second assertion take $g=(12\alpha)^2$ and
$\Delta_\alpha=\Delta_g$.  This constant $g$ is in the theorem's range once $n$
is sufficiently large.  The exponent $\delta$ can be chosen uniformly;
the size threshold and degree bound may depend on $\alpha$.
\end{proof}

\begin{corollary}[The IKW trade-off]\label{cor:spanner-ikw}
Let $N$ denote the number of variables in the source 2-CSP.
There are absolute constants $c_1,c_2,C_2,a_1,a_2>0$ with the following
property. Let $k=k(N)$ be integer-valued, with $k$ sufficiently large and
$k=o(\log N)$. For all sufficiently large $N$, the IKW construction
yields a common graph size $n$ satisfying
\[
 N^{a_1k}\le n\le N^{a_2k}
\]
for which every $\exp(c_1\sqrt{k})$-approximation algorithm for
minimum unweighted directed 3-spanner has sensitivity at least
$n^{c_2/k}$ on weakly connected DAGs of maximum total degree at most
$2^{C_2k}$.
\end{corollary}
\begin{proof}
Take the Label Cover family used in the proof of
\Cref{thm:intro-label-cover-sensitivity-IKW}, constructed in
\Cref{sec:ikw-sensitivity-analysis}.  Its parameter $N$ counts the
original 2-CSP variables, and $k$ is the direct-product parameter
specified in the statement.
Let $s$ be its total number of vertices, $m$ its number of edges, and $D,Q$ its
maximum degree and alphabet size.  This family has positive degrees,
so $s/2\le m\le sD$.
Write $C_{\rm IKW}$ for the source theorem's constant $C$ in the value
guarantee and $\eta$ for its sensitivity exponent.  With $L$ denoting
the source sensitivity lower bound and $g$ the reciprocal value
guarantee, its parameters satisfy
\[
 s=N^{\Theta(k)},\quad D,Q\le2^{O(k)},\quad
 L\ge s^{\eta/k},\quad g=\exp(C_{\rm IKW}\sqrt{k}).
\]
Take $n=n_0$ in \Cref{prop:spanner-lc-transfer}.  Then
$s\le n\le s2^{O(k)}$, $J_0\le2^{O(k)}$, and $\Delta_*\le2^{O(k)}$.
Choose $c_1<C_{\rm IKW}/2$; for sufficiently large $k$,
$144\exp(2c_1\sqrt{k})\le g$.
The resulting sensitivity is at least $s^{\eta/k}/2^{O(k)}-1$.
Since $k^2=o(\log s)$, the denominator and additive constant can be
absorbed, giving $n^{c_2/k}$ for some $c_2>0$.
The source size bounds and $n\le s2^{O(k)}$ imply the stated bounds
on $n$ after changing $a_1,a_2$.  This assertion concerns the sizes
generated by the IKW construction. The bound for every sufficiently large
graph size $n$ is given in \Cref{thm:spanner-dh-parameterized}.
\end{proof}

\paragraph{Role of the degree bound.}
To derive locality lower bounds from sensitivity, we also bound the
maximum degree. A single arc deletion can force a large increase in the
optimum when degrees are large: take
disjoint nonempty vertex sets $A,B$ and two fresh vertices $u,v$.
Include all arcs from $A$ to $B$, together with
$a\to u$, $u\to v$, and $v\to b$ for all $a\in A,b\in B$.
With $u\to v$, the $|A|+|B|+1$ arcs through $u,v$ form a 3-spanner.
After its deletion, all $|A||B|$ arcs in $A\times B$ are necessary.
This elementary example has large degree.  In contrast,
\Cref{thm:spanner-dh-parameterized} controls the degree uniformly,
and \Cref{cor:spanner-dh-specializations} keeps it constant for each
constant approximation factor.  The next section uses precisely this
additional control.
 \section{Distributed Consequences for Directed 3-Spanner}
\label{sec:spanner-distributed}

We convert the arc-set sensitivity bounds into distributed lower bounds.
For every constant approximation factor, we obtain an $\Omega(\log n)$
lower bound on LOCAL rounds and non-signaling locality on DAGs of constant
maximum total degree.  The degree bound and size threshold may
depend on the approximation factor.  We also derive bounds for
approximation factors that grow as a small power of $\log n$.

The communication network is the underlying undirected graph: vertices
can communicate in both directions, independently of input-arc directions,
as in the distributed spanner formulation of
\cite{censorhillel2021spanner}.  Each vertex initially knows its incident
input arcs and their endpoints.  Identifiers and global parameters such
as $n$ and a prescribed degree bound are fixed across compared inputs.
Local computation and message size are unrestricted.

Each vertex may mark a subset of its incident input arcs, and the union
of all marked arcs is the output spanner.  This includes the convention
that an arc is owned by its tail, as well as the convention that either
endpoint may select it.  All unmarked or absent arcs have output bit
zero.

\subsection{From Arc Sensitivity to Non-signaling Locality}

At each vertex, the output is the vector of indicators specifying which
incident arcs it selects.  We apply the definition of non-signaling
locality in \Cref{sec:pre} to the joint distribution of these vertex outputs.
The following bound accounts for the fact that one affected vertex may
change several output arcs.

\begin{lemma}[Locality bound for arc-set outputs]
\label{lem:spanner-nonsignaling}
For a $t$-local non-signaling algorithm whose outputs are subsets of the
input arcs, its single-arc sensitivity on graphs of maximum total degree
at most $\Delta\ge2$ is at most
\begin{equation}\label{eq:spanner-nonsignaling-bound}
 4\Delta(\Delta+1)^{t+1}.
\end{equation}
\end{lemma}
\begin{proof}
Let $G,\widetilde G$ be graphs on a common vertex set $V$ that differ
by one arc with endpoint set $X$, $|X|=2$.
In the underlying undirected graph of $G\cup\widetilde G$, let $W$
be the vertices at distance at most $t+1$ from $X$.  The union is
one of the two inputs, so its degree is at most $\Delta$, and
$|W|\le2(\Delta+1)^{t+1}$.

The radius-$t$ input neighborhood of $V\setminus W$ is identical in
$G$ and $\widetilde G$.  The extra layer in the definition of $W$
ensures this equality even when the input labels at boundary vertices
specify all their incident arcs.
Non-signaling therefore gives identical joint output marginals on
$V\setminus W$.  Couple those vectors to agree, and extend the coupling
to full outputs by conditional sampling.  Every differing selected arc
then has at least one endpoint in $W$.  Each vertex in $W$ is incident
to at most $2\Delta$ arcs present in either $G$ or $\widetilde G$.
The Hamming distance between the coupled arc sets is at most
$2\Delta|W|$, proving
\eqref{eq:spanner-nonsignaling-bound}.
\end{proof}

\subsection{Locality and Round Lower Bounds}

We combine the polynomial sensitivity lower bound of
\Cref{thm:spanner-dh-parameterized} with the sensitivity upper bound of
\Cref{lem:spanner-nonsignaling}.  Taking logarithms
gives the following lower bound on non-signaling locality and LOCAL rounds.

\begin{theorem}[Parameterized distributed lower bounds]
\label{thm:spanner-distributed-parameterized}
Use $a,C,\theta,\delta$ from
\Cref{thm:spanner-dh-parameterized}.  Let $t$ denote the locality radius
in the non-signaling model or the number of communication rounds in LOCAL.
For every sufficiently large $n$
and every $144\le g\le\log^\theta n$, a
$(\sqrt g/12)$-approximation for minimum unweighted directed 3-spanner
on weakly connected $n$-vertex DAGs of maximum total degree at most
$\Delta_g=\lceil\exp(Cg^a)\rceil$ requires
\begin{equation}\label{eq:spanner-local-lower}
 t=\Omega\!\left(\frac{\log n}{g^a}\right)
 \quad\text{in the non-signaling and LOCAL models}.
\end{equation}
The implicit constant is absolute.
\end{theorem}
\begin{proof}
The sensitivity lower bound is $\Omega(n^\delta)$, while
$\log \Delta_g=O(g^a)=o(\log n)$ uniformly in the stated range.
Taking logarithms in \Cref{lem:spanner-nonsignaling} gives
\[
 \delta\log n-O(1)
 \le\log(4\Delta_g)+(t+1)\log(\Delta_g+1).
\]
The additive $O(g^a)$ terms can be absorbed because $g^a=o(\log n)$.
This proves \eqref{eq:spanner-local-lower}.
\end{proof}

\begin{corollary}[Polylogarithmic approximation factors]
\label{cor:spanner-polylog-rounds}
There are constants $\beta\in(0,1)$, $\kappa\in(0,1/2)$, and
$\tau=1-\kappa\in(1/2,1)$ such that every
$\log^\beta n$-approximation for minimum unweighted directed
3-spanner requires non-signaling locality at least
\[
 \Omega(\log^\tau n)
\]
and the same lower bound holds for the number of LOCAL rounds.
The lower bound already holds on weakly connected DAGs of maximum total
degree $\exp(O(\log^\kappa n))$, with no bound on local computation.
\end{corollary}
\begin{proof}
Set $g=\log^\theta n$ and $\kappa=a\theta<1/2$ in
\Cref{thm:spanner-distributed-parameterized}, and choose
$0<\beta<\theta/2$ as in
\Cref{cor:spanner-dh-specializations}.
Since $g^a=\log^\kappa n$, this gives the stated bounds.
\end{proof}

\begin{corollary}[Every constant approximation factor]
\label{cor:spanner-constant-rounds}
For every constant $\alpha\ge1$, there is a constant $\Delta_\alpha$ such that
$\alpha$-approximation for minimum unweighted directed 3-spanner on
weakly connected $n$-vertex DAGs of maximum total degree at most
$\Delta_\alpha$ requires non-signaling locality $\Omega_\alpha(\log n)$,
and the same lower bound holds for the number of LOCAL rounds.
\end{corollary}
\begin{proof}
Use $g=(12\alpha)^2$ in
\Cref{thm:spanner-distributed-parameterized}.
\end{proof}

For the IKW family, take $N,k,n$ and the constant $c_1$ from
\Cref{cor:spanner-ikw}, with $k$ sufficiently large and
$k=o(\log N)$.  At the graph sizes $n=N^{\Theta(k)}$ supplied by
that construction, let $t_{\rm LOCAL}$ denote the number of LOCAL rounds
and $t_{\rm NS}$ the non-signaling locality.  The same conversion gives
\begin{equation}\label{eq:spanner-ikw-locality}
 t_{\rm LOCAL},t_{\rm NS}
 =\Omega\!\left(\frac{\log n}{k^2}\right)
\end{equation}
for $\exp(c_1\sqrt{k})$-approximation algorithms.  The maximum total degree
bound $\Delta$ from \Cref{cor:spanner-ikw} satisfies $\log\Delta=O(k)$.
The prefactor $4\Delta$ in \eqref{eq:spanner-nonsignaling-bound} therefore
contributes $O(k)$ to the logarithm of the sensitivity upper bound.
The logarithm of the sensitivity lower bound is $\Omega(\log n/k)$,
which dominates $O(k)$ because $k^2=o(\log n)$.

\subsection{Algorithms Allowed to Fail with Small Probability}

The preceding results use almost-sure feasibility and an expected-cost
approximation guarantee.  We now consider algorithms that, with high
probability, output a directed 3-spanner of $G$ with at most
$\alpha\opt(G)$ arcs, where $\alpha$ is the required
approximation factor.  We show that the same locality and round lower
bounds hold when the failure probability is at most $n^{-3}$.

\begin{lemma}[Sensitivity with a small failure probability]
\label{lem:spanner-small-failure}
Suppose that, on every graph $G$ in an $n$-vertex input class, a
randomized algorithm $A$ outputs a directed 3-spanner with at most
$\alpha\opt(G)$ arcs with probability at least $1-p$.
There is an $\alpha$-approximation algorithm $B$ on the same input class
such that, for every graph,
\[
 \EMD(A(G),B(G))\le p\,n(n-1).
\]
Consequently, if all such $B$ have sensitivity at least $L$, then
$A$ has sensitivity at least $L-2pn(n-1)$.
\end{lemma}
\begin{proof}
Define $B$ by retaining $A$'s output when it satisfies both conditions
and replacing it otherwise by a fixed optimal spanner of $G$.
The algorithm $B$ is used only in the analysis and need not be distributed.
Under this coupling, $A(G)$ and $B(G)$ differ in at most $n(n-1)$
output coordinates on an event of probability at most $p$.
The output $B(G)$ always has at most $\alpha \cdot \opt(G)$ arcs.
Apply the triangle inequality to the neighboring pair
witnessing $B$'s sensitivity.
\end{proof}

Thus \Cref{thm:spanner-distributed-parameterized,cor:spanner-polylog-rounds,cor:spanner-constant-rounds}
and the IKW bound \eqref{eq:spanner-ikw-locality} remain valid for
algorithms with failure probability at most $n^{-3}$: the sensitivity
loss is at most $2/n$, which is negligible compared with both the DH
bound $n^{\Omega(1)}$ and the IKW bound $n^{\Omega(1/k)}$.
The locality and round conclusions use
a worst-case round bound for LOCAL and the stated locality bound
for non-signaling output families.  The argument does not by itself extend the
statements to constant failure probability or to expected termination
time.  Outputs may first be restricted locally to actual input arcs,
so this convention causes no communication overhead.

\paragraph{Comparison with existing distributed results.}
\Cref{tab:spanner-locality-comparison} summarizes selected upper and lower bounds.

\begin{table}[t!]
\centering
\caption{Selected approximation and locality bounds for minimum unweighted directed
3-spanner on $n$ vertices. The three lower-bound rows use the DAG
families and parameter ranges of \Cref{tab:spanner-sensitivity-comparison},
with $\tau=1-\kappa\in(1/2,1)$.
All three lower bounds also allow failure probability at most $n^{-3}$.
The upper bound holds on general directed graphs with high probability,
for $\varepsilon>0$.}
\label{tab:spanner-locality-comparison}
\small
\setlength{\tabcolsep}{3.25pt}
\begin{tabular}{@{}p{0.24\linewidth}p{0.24\linewidth}p{0.23\linewidth}p{0.22\linewidth}@{}}
\toprule
Approximation ratio & Maximum total degree & Locality bound & Reference \\
\midrule
\multicolumn{4}{c}{Lower bounds (LOCAL and non-signaling)} \\
\addlinespace[0.25em]
Any constant $\alpha\ge1$ & $O_\alpha(1)$ & $\Omega_\alpha(\log n)$ & \Cref{cor:spanner-constant-rounds} \\
$\log^\beta n$ & $\exp(O(\log^\kappa n))$ & $\Omega(\log^\tau n)$ & \Cref{cor:spanner-polylog-rounds} \\
$\exp(c_1\sqrt{k})$ & $2^{O(k)}$ & $\Omega(\log n/k^2)$ & \eqref{eq:spanner-ikw-locality} \\
\midrule
\multicolumn{4}{c}{Upper bound (LOCAL)} \\
\addlinespace[0.25em]
$1+\varepsilon$ & Unrestricted & $\operatorname{poly}(\log n/\varepsilon)$ & \cite{censorhillel2021spanner} \\
\bottomrule
\end{tabular}
\end{table}

Censor-Hillel and Dory~\cite{censorhillel2021spanner} obtain a
$(1+\varepsilon)$-approximation in
$\operatorname{poly}(\log n/\varepsilon)$ LOCAL rounds with high probability
for constant stretch, allowing unrestricted local computation. Their directed
unweighted CONGEST lower bound is stated for stretch $k\ge5$.
Dinitz and Nazari \cite{dinitz2018distributednetwork} give a randomized
algorithm that, with high probability, computes an
$O(\sqrt n\log n)$-approximation for directed $k$-spanner in
$O(k\log n)$ LOCAL rounds using polynomial local computation.
For any constant approximation ratio $\alpha\ge1$,
our lower bound requires logarithmic locality even on DAGs whose
maximum degree is bounded by a constant depending only on $\alpha$.
For approximation ratio $\log^\beta n$, our lower bound requires
$\Omega(\log^\tau n)$ locality and thus rules out
$(\log\log n)^{O(1)}$ locality.
The CONGEST lower bound of Censor-Hillel and Dory for directed unweighted
$k$-spanner with $k\ge5$ and their LOCAL lower bound for weighted
2-spanner~\cite{censorhillel2021spanner} do not establish these conclusions
for unweighted directed 3-spanner.  Our bounds hold with unrestricted message size
and local computation and extend to non-signaling algorithms.

A direction for future work is to extend these sensitivity-based arguments
to the HYBRID model~\cite{augustine2020hybrid}, which augments LOCAL with
limited global communication.
 
\section*{Acknowledgments}
We thank Mohsen Ghaffari for suggesting that we clarify the role of the constant
in the logarithmic approximation factor and the resulting distinction between
local and non-local approximation for Dominating Set.

\section*{Statement of AI Usage}

The main ideas in this work are due entirely to the authors. GPT Astra was used after the initial publication of the manuscript to simplify several of the arguments, audit for minor errors, and to help draft corresponding new versions of the main statements and proofs.

\bibliographystyle{abbrv}\bibliography{main}

\end{document}